\PassOptionsToPackage{unicode}{hyperref}
\PassOptionsToPackage{hyphens}{url}
\PassOptionsToPackage{table,dvipsnames,svgnames,x11names}{xcolor}

\documentclass[12pt]{article}
\usepackage{amsmath}
\usepackage{graphicx,psfrag,epsf}
\usepackage{enumerate}
\usepackage{natbib}
\usepackage{url} 

\usepackage[T1]{fontenc}
\usepackage[utf8]{inputenc}
\usepackage{textcomp} 
\usepackage{lmodern}

\IfFileExists{upquote.sty}{\usepackage{upquote}}{}
\IfFileExists{microtype.sty}{
  \usepackage[]{microtype}
  \UseMicrotypeSet[protrusion]{basicmath} 
}{}
\makeatletter
\@ifundefined{KOMAClassName}{
  \IfFileExists{parskip.sty}{%
    \usepackage{parskip}
  }{
    \setlength{\parindent}{0pt}
    \setlength{\parskip}{6pt plus 2pt minus 1pt}}
}{
  \KOMAoptions{parskip=half}}
\makeatother
\usepackage{xcolor}
\makeatletter
\ifx\paragraph\undefined\else
  \let\oldparagraph\paragraph
  \renewcommand{\paragraph}{
    \@ifstar
      \xxxParagraphStar
      \xxxParagraphNoStar
  }
  \newcommand{\xxxParagraphStar}[1]{\oldparagraph*{#1}\mbox{}}
  \newcommand{\xxxParagraphNoStar}[1]{\oldparagraph{#1}\mbox{}}
\fi
\ifx\subparagraph\undefined\else
  \let\oldsubparagraph\subparagraph
  \renewcommand{\subparagraph}{
    \@ifstar
      \xxxSubParagraphStar
      \xxxSubParagraphNoStar
  }
  \newcommand{\xxxSubParagraphStar}[1]{\oldsubparagraph*{#1}\mbox{}}
  \newcommand{\xxxSubParagraphNoStar}[1]{\oldsubparagraph{#1}\mbox{}}
\fi
\makeatother

\usepackage{longtable,booktabs,array,tabularx}
\usepackage{calc} 
\usepackage{etoolbox}
\makeatletter
\patchcmd\longtable{\par}{\if@noskipsec\mbox{}\fi\par}{}{}
\makeatother
\IfFileExists{footnotehyper.sty}{\usepackage{footnotehyper}}{\usepackage{footnote}}
\makesavenoteenv{longtable}
\usepackage{graphicx}
\makeatletter
\def\maxwidth{\ifdim\Gin@nat@width>\linewidth\linewidth\else\Gin@nat@width\fi}
\def\maxheight{\ifdim\Gin@nat@height>\textheight\textheight\else\Gin@nat@height\fi}
\makeatother
\setkeys{Gin}{width=\maxwidth,height=\maxheight,keepaspectratio}
\makeatletter
\def\fps@figure{htbp}
\makeatother

\makeatletter
\@ifpackageloaded{caption}{}{\usepackage{caption}}
\AtBeginDocument{%
\ifdefined\contentsname
  \renewcommand*\contentsname{Table of contents}
\else
  \newcommand\contentsname{Table of contents}
\fi
\ifdefined\listfigurename
  \renewcommand*\listfigurename{List of Figures}
\else
  \newcommand\listfigurename{List of Figures}
\fi
\ifdefined\listtablename
  \renewcommand*\listtablename{List of Tables}
\else
  \newcommand\listtablename{List of Tables}
\fi
\ifdefined\figurename
  \renewcommand*\figurename{Figure}
\else
  \newcommand\figurename{Figure}
\fi
\ifdefined\tablename
  \renewcommand*\tablename{Table}
\else
  \newcommand\tablename{Table}
\fi
}
\@ifpackageloaded{float}{}{\usepackage{float}}
\floatstyle{ruled}
\@ifundefined{c@chapter}{\newfloat{codelisting}{h}{lop}}{\newfloat{codelisting}{h}{lop}[chapter]}
\floatname{codelisting}{Listing}

\makeatother
\makeatletter
\@ifpackageloaded{caption}{}{\usepackage{caption}}
\@ifpackageloaded{subcaption}{}{\usepackage{subcaption}}
\makeatother

\newcommand{\blind}{0}

\usepackage[]{natbib}
\usepackage{bookmark}

\hypersetup{
  pdftitle={Quadruply robust methods for causal mediation analysis},
  pdfauthor={\if0\blind Zhen Qi; Yuqian Zhang\fi},
  pdfkeywords={Natural (in)direct effect; Mediation analysis; Cross-world counterfactual; Robust inference; High-dimensional statistics},
  colorlinks=true,
  linkcolor={blue},
  filecolor={Maroon},
  citecolor={Blue},
  urlcolor={Blue},
  pdfcreator={LaTeX via pandoc}}

\usepackage{pdflscape} 
\usepackage{rotating} 
\usepackage{threeparttable} 

\usepackage{perpage} 
\MakePerPage[1]{footnote} 

\usepackage{scalefnt}

\usepackage{amsthm, amssymb, amsfonts, bbm, bm}
\usepackage{amsmath}
\usepackage{graphicx}
\usepackage{mathrsfs}
\usepackage{threeparttable}
\usepackage[title]{appendix}

\usepackage{caption}
\usepackage{tikz}
\usepackage{algorithm}
\usepackage{algpseudocode}
\usepackage{booktabs}
\usepackage{chngpage}
\usepackage{bbding}

\usepackage{pifont}
\usepackage{enumerate}
\usepackage{enumitem}
\newcounter{globalcondition}
\usepackage{tabu}
\usepackage{multirow}
\usepackage{multicol}
\usepackage{float}
\usepackage{makecell}
\pgfdeclarelayer{background}
\pgfsetlayers{background,main}
\usetikzlibrary{arrows,positioning}
\usepackage{arydshln}
\tikzset{
>=stealth',
punkt/.style={
rectangle,
rounded corners,
draw=black, very thick,
text width=6.5em,
minimum height=2em,
text centered},
pil/.style={
->,
thick,
shorten <=2pt,
shorten >=2pt,}
}
\newcommand{\Vertex}[2]
{\node[minimum width=0.6cm,inner sep=0.05cm] (#2) at (#1) {$\textstyle#2$};}
\newcommand{\Vertexunsee}[2]
{\node[minimum width=0.6cm,inner sep=0.05cm] (#2) at (#1) {$ $};
}
\newcommand{\Vertexr}[2]
{\node[rectangle, draw, minimum width=0.6cm,inner sep=0.05cm] (#2) at (#1) {$\textstyle#2$};
}
\newcommand{\ArrowR}[3]%
{ \begin{pgfonlayer}{background}
\draw[->,#3] (#1) to[bend right=30] (#2);
\end{pgfonlayer}
}
\newcommand{\ArrowL}[3]%
{ \begin{pgfonlayer}{background}
\draw[->,#3] (#1) to[bend left=45] (#2);
\end{pgfonlayer}
}
\newcommand{\EdgeL}[3]%
{ \begin{pgfonlayer}{background}
\draw[dashed,#3] (#1) to[bend right=-45] (#2);
\end{pgfonlayer}
}

\newcommand{\Arrow}[3]%
{ \begin{pgfonlayer}{background}
\draw[->,#3] (#1) -- +(#2);
\end{pgfonlayer}
}

\newtheorem{theorem}{Theorem}[section]
\newtheorem{corollary}{Corollary}[section]
\newtheorem{lemma}{Lemma}[section]
\newtheorem{remark}{Remark}
\newtheorem{assumption}{Assumption}
\newtheorem{example}{Example}

\numberwithin{equation}{section} 
\def\idf{\mathbbm{1}}

\def\d{\mathrm{d}}

\def\R{\mathbb{R}}

\def\N{\mathbb{N}}
\def\E{\mathbb{E}}
\def\P{\mathbb{P}}
\def\K{\mathbb{K}}

\def\I{\mathcal{I}}

\newcommand{\bs}[1]{\boldsymbol{#1}}

\newcommand{\cmark}{\ding{51}}
\def\br{\bar{r}}
\newcommand{\eqr}{\bar{e}}
\newcommand{\emqr}{e}
\newcommand{\bemqr}{\bar{e}}
\newcommand{\rmqr}{\bar{r}}
\def\suk{\sum\limits_{i\in \I_{-k}}}

\def\bphi{\bs{\Phi}}

\def\bnabla{\bs{\nabla}}
\def\bbeta{\bs{\beta}}
\def\bnu{\bs{\nu}}
\def\hbnu{\hat{\bnu}}
\def\bpsi{\phi}
\def\mpsi{\phi}

\def\ssigma{\sigma}
\def\rsigma{\sigma_{\mathrm{m}}}

\def\bepi{\bbeta_{\pi}}
\def\bepia{\bbeta_{\pi,a}}
\def\bepib{\bbeta_{\pi,b}}

\def\bemu{\bbeta_\mu}

\def\betau{\bbeta_\tau}
\def\bettau{\bbeta_{\tau_n}}
\def\beq{\bbeta_{q}}
\def\hbepi{\hat{\bbeta}_{\pi}}
\def\hbepia{\hat{\bbeta}_{\pi,a}}
\def\hbepib{\hat{\bbeta}_{\pi,b}}

\def\hbemu{\hat{\bbeta}_\mu}

\def\hbetau{\hat{\bbeta}_\tau}
\def\hbettau{\hat{\bbeta}_{\tau_n}}
\def\hbeq{\hat{\bbeta}_{q}}
\def\hbettau{\hat{\bbeta}_{\tau_n}}

\def\tbepia{\tilde{\bbeta}_{\pi,a}}
\def\tbepib{\tilde{\bbeta}_{\pi,b}}

\def\tbemu{\tilde{\bbeta}_\mu}

\def\tbetau{\tilde{\bbeta}_\tau}
\def\tbeq{\tilde{\bbeta}_{q}}
\def\tbettau{\tilde{\bbeta}_{\tau_n}}

\def\tvarepsilon{\tilde{\varepsilon}}
\def\tzeta{\tilde{\zeta}}
\def\tvarrho{\tilde{\varrho}}

\def\bepir{\bbeta_{\pi}^*}
\def\bepiar{\bbeta_{\pi,a}^*}
\def\bepibr{\bbeta_{\pi,b}^*}
\def\beqr{\bbeta_{q}^*}
\def\bemur{\bbeta_{\mu}^*}

\def\betaur{\bbeta_{\tau}^*}
\def\bettaur{\bbeta_{\tau_{n}}^*}
\def\hbepiar{\hat{\bbeta}_{\pi,a}}
\def\hbepibr{\hat{\bbeta}_{\pi,b}}
\def\hbeqr{\hat{\bbeta}_{q}}
\def\hbemur{\hat{\bbeta}_{\mu}}

\def\hbetaur{\hat{\bbeta}_{\tau}}
\def\hbettaur{\hat{\bbeta}_{\tau_n}}

\def\norme#1{\left\|#1\right\|_{\psi_{1}}}
\def\normg#1{\left\|#1\right\|_{\psi_2}}
\def\normp#1#2{\left\|#1\right\|_{\P,#2}}

\def\suj#1{\sum\limits_{j=1}^{#1}}
\def\supiar{\sum\limits_{i\in \I_{\pi}}}
\def\supibr{\sum\limits_{i\in \I_{\pi}}}
\def\suqr{\sum\limits_{i\in \I_{q}}}
\def\sumur{\sum\limits_{i\in \I_{\mu}}}
\def\suttaur{\sum\limits_{i\in \I_{\tau_n}}}
\def\sutaur{\sum\limits_{i\in \I_{\tau}}}

\def\bx{\bs{x}}
\def\bU{\bs{U}}
\def\balpha{\bs{\alpha}}
\def\halpha{\hat{\balpha}}
\def\bX{\bs{X}}
\def\bM{\bs{M}}

\def\bS{\bs{S}}
\def\bzS{\bS_0}

\def\bziS{\bS_{i,0}}
\def\bW{\bs{W}}
\def\bv{\bs{v}}
\def\e{\bs{e}}
\def\bDelta{\bs{\Delta}}
\def\bz{\bs{0}}
\def\Var{\mathrm{Var}}
\def\Ktil{\breve{\mathbb{K}}}
\def\Ctil{\breve{\mathbb{C}}}
\def\independenT#1#2{\mathrel{\rlap{$#1#2$}\mkern2mu{#1#2}}}
\newcommand\independent{\protect\mathpalette{\protect\independenT}{\perp}}

\title{Quadruply robust methods for causal mediation analysis}
\author{Zhen Qi$^{1}$ and Yuqian Zhang$^{1,}$\thanks{Corresponding author (email: yuqianzhang@ruc.edu.cn)}
\\	{\small \it $^{1}$ Institute of Statistics and Big Data, Renmin University of China, Beijing, China}\\~
}

\begin{document}

\def\spacingset#1{\renewcommand{\baselinestretch}%
{#1}\small\normalsize} \spacingset{1}


\if0\blind
{
  \title{\bf Quadruply robust methods for causal mediation analysis}
  \author{Zhen Qi
    \hspace{.2cm}\\
    Institute of Statistics and Big Data, Renmin University of China\\
    and \\
    Yuqian Zhang\thanks{
Corresponding author (email: yuqianzhang@ruc.edu.cn)} \\
    Institute of Statistics and Big Data, Renmin University of China}
  \maketitle
} \fi

\if1\blind
{
  \bigskip
  \bigskip
  \bigskip
  \begin{center}
    {\LARGE\bf Quadruply robust methods for causal mediation analysis}
\end{center}
  \medskip
} \fi

\bigskip
\begin{abstract}
Estimating natural effects is a core task in causal mediation analysis. Existing triply robust (TR) frameworks \citep{SemiparametricTheoryCausal2012} and their extensions have been developed to estimate the natural effects. In this work, we introduce a new quadruply robust (QR) framework that enlarges the model class for unbiased identification. We study two modeling strategies. The first is a nonparametric modeling approach, under which we propose a general QR estimator that supports the use of machine learning methods for nuisance estimation. We also study high-dimensional settings, where the dimensions of covariates and mediators may both be large. In these settings, we adopt a parametric modeling strategy and develop a model quadruply robust (MQR) estimator to limit the impact of model misspecification. Simulation studies and a real data application demonstrate the finite-sample performance of the proposed methods.
\end{abstract}

\noindent%
\vfill

\newpage
\spacingset{1.8} 

\section{Introduction}\label{sec:intro}

Causal mediation analysis studies how treatment effects arise through intermediate variables, or mediators. It decomposes the total treatment effect into direct and indirect components. This decomposition helps characterize causal pathways and supports mechanism-based scientific conclusions \citep{GeneralApproachCausal2010,InterpretationIdentificationCausal2014}.

A key challenge in causal mediation analysis is the misspecification of nuisance functions that characterize the treatment assignment, mediator process, and outcome mechanism. This challenge is amplified in high-dimensional settings, where a large number of covariates may need to be adjusted for to mitigate confounding bias, while a large number of mediators may need to be considered to capture numerous potential indirect pathways \citep{EpigenomicsGestationalProgramming2015,DNAMethylationMediator2018}. Such high dimensionality makes accurate estimation of all nuisance functions difficult, motivating methods that remain reliable even in the presence of substantial nuisance estimation errors. Robustness to such errors has been studied, for example, by \citet{SemiparametricTheoryCausal2012}, \citet{TargetedMaximumLikelihood2012}, and \citet{CausalMediationAnalysis2022}. In this work, we develop methods with stronger robustness properties.

\paragraph*{Problem setup}
Suppose we collect $N$ independent and identically distributed (i.i.d.) samples $(\bW_i)_{i=1}^N = (Y_i,A_i,\bX_i,\bM_i)_{i=1}^N$, with $\bW = (Y,A,\bX,\bM)$ being an independent copy. Here, $Y\in\R$ denotes the outcome, $A\in\{0,1\}$ is a binary treatment indicator, $\bX\in\R^{d_1}$ represents the pre-treatment covariates, and $\bM\in\mathcal{S}_M\subseteq \R^{d_2}$ denotes the mediators that lie on the causal pathway from $A$ to $Y$. A directed acyclic graph in Figure~\ref{Fig:DAG} shows the data-generating structure. We write $\bS:=(\bX^\top,\bM^\top)^\top$ for the collection of covariates and mediators, which has dimension $d:=d_1+d_2$. Throughout this paper, we allow both $d_1$ and $d_2$ to increase with $N$, covering high-dimensional cases where the numbers of covariates and mediators may \emph{both} be large.

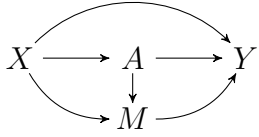
\begin{figure}[!b]
\centering
\caption{A directed acyclic graph illustrating the causal relationships among treatment $A$, mediator $\bM$, outcome $Y$, and confounder $\bX$.}\label{Fig:DAG}
\begin{tikzpicture}[scale=1, every node/.style={transform shape}]
    \node[minimum width=0.6cm,inner sep=0.05cm] (X) at (-1.5, 0) {$\textstyle\bX$};
    \node[minimum width=0.6cm,inner sep=0.05cm] (M) at (0, -0.8) {$\textstyle\bM$};
    \Vertex{0, 0}{A}
    \Vertex{1.5, 0}{Y}
    \Arrow{X}{A}{black}
    \Arrow{A}{M}{black}
    \ArrowR{M}{Y}{black}
    \ArrowL{X}{Y}{black}
    \Arrow{A}{Y}{black}
    \ArrowR{X}{M}{black}
\end{tikzpicture}
\end{figure}
We adopt the potential-outcomes framework to define the causal effects of interest. For $a\in\{0,1\}$ and $\bs{m}\in\mathcal{S}_M$, let $\bM(a)$ denote the potential mediator under treatment $a$, and $Y(a,\bs{m})$ the potential outcome under treatment $a$ and mediator value $\bs{m}$. In causal inference, the average treatment effect (ATE), $\mathrm{ATE} := \theta_{1,1} - \theta_{0,0}$, measures the total effect of $A$ on $Y$, where $\theta_{a,a^\prime} := \E[Y(a,\bM(a^\prime))]$ for $a,a^\prime \in \{0,1\}$. To study the underlying causal pathways, the ATE is usually decomposed into the natural direct effect (NDE) and the natural indirect effect (NIE), defined as $\mathrm{NDE} := \theta_{1,0} - \theta_{0,0}$ and $\mathrm{NIE} := \theta_{1,1} - \theta_{1,0}$ \citep{InterpretationIdentificationCausal2014}. The NDE captures the effect along paths that bypass the mediator ($A \rightarrow Y$), while the NIE captures the effect along paths that operate through $\bM$ ($A \rightarrow \bM \rightarrow Y$).

Identification and estimation of the mediation effects are usually more challenging than the total effect, since both the NDE and NIE rely on the quantity $\theta_{1,0}=\E[Y(1,\bM(0))]$. The variable $Y(1,\bM(0))$ is a cross-world counterfactual, representing the potential outcome under treatment $A=1$ and the mediator values that would occur under $A=0$, and it is \emph{never observed} in the data. In contrast, the ATE is defined through $Y(1,\bM(1))$ and $Y(0,\bM(0))$, which are observed when $A=1$ and $A=0$, respectively. While robust inference for $\theta_{1,1}$ and $\theta_{0,0}$ (and thus the ATE) has been studied extensively \citep{HighdimensionalInferenceAverage2022,SparsityDoubleRobust2019,InferenceTreatmentEffect2021,UnifyingApproachDoublyrobust2019}, robust inference for $\theta_{1,0}$ is more challenging and is the main focus of this work. For simplicity, the target parameter $\theta_{1,0}=\E[Y(1,\bM(0))]$ is denoted as $\theta$ throughout the sequel.

Since the potential outcome $Y(1,\bM(0))$ is never observed, $\theta=\E[Y(1,\bM(0))]$ cannot be identified directly from the observable variables $\bW$. Instead, we identify $\theta$ through observable quantities together with several related nuisance functions. Define the conditional mean outcome as $\mu(\bS) := \E[Y \mid A = 1, \bS]$ and the cross-world conditional mean outcome as $\tau(\bX) := \E[Y(1,\bM(0)) \mid \bX]$. We also define the treatment propensity score as $\pi(\bX) := \P(A = 1 \mid \bX)$, let $f_a(\bS) := f(\bM \mid A = a, \bX)$ be the conditional density of $\bM$ given $\bX$ under treatment $a \in \{0,1\}$, and set $q(\bS):=f_0(\bS)/ f_1(\bS)$. We use the superscript ``$^*$'' to denote working models for these functions, such as $\mu^*$ for $\mu$, which represent population-level approximations. A working model is said to be \emph{correctly specified} if it coincides with the corresponding true function, that is, $\mu^*=\mu$; otherwise, it is \emph{misspecified}.

The following standard identification conditions are assumed throughout \citep{GeneralApproachCausal2010,SemiparametricTheoryCausal2012,MediationAnalysisMultiple2014}.

\begin{assumption}[Basic assumptions]\label{cond:basic}
(a) Consistency: $\bM = \bM(A)$, $Y = Y(A,\bM)$. (b) Ignorability: For $\bs{m} \in \mathcal{S}_M$, $\{Y(1,\bs{m}),\bM(0)\} \independent A \mid \bX$ and $Y(1,\bs{m}) \independent \{\bM(1),\bM(0)\} \mid A,\bX$. (c) Positivity: $\pi(\bX)\in[c_0,1-c_0]$ and $q(\bS)\in[c_0^2,c_0^{-2}]$ almost surely with some $c_0\in(0,1/2)$.
\end{assumption}

Under Assumption~\ref{cond:basic}, \citet{SemiparametricTheoryCausal2012}, \citet{TargetedMaximumLikelihood2012}, \citet{SemiparametricEstimationPathspecific2019}, \citet{CausalMediationAnalysis2022}, and \citet{DeepMedSemiparametricCausal2022} proposed ``triply robust'' methods, which allow unbiased identification of $\theta$ under three model scenarios, where each scenario requires correct specification of two nuisance models; see Section~\ref{sec:back} for a survey. However, in high-dimensional settings, existing approaches generally require correct specification of all nuisance models to establish $\sqrt{N}$-consistency and asymptotic normality; under model misspecification, they typically provide only consistency. In this work, we propose a ``quadruply robust'' framework that enlarges the model class for unbiased identification and, within this broader class, further develops methods that enable valid inference in high-dimensional settings under model misspecification.

\paragraph*{Our contribution}\label{par:contributions}
We consider two modeling approaches:
\begin{itemize}
\item\emph{Nonparametric modeling}: We propose a new quadruply robust (QR) estimator of $\theta$ in Section~\ref{sec:QR}, which extends consistency guarantees to a strictly broader class of data-generating processes (DGPs). Methodologically, this robustness is achieved by jointly leveraging a doubly robust (DR) representation of $\tau$ and a new QR score that augments the canonical mediation influence function with an additional correction for potential misspecification bias. The resulting estimator accommodates flexible machine-learning-based nuisance estimators and achieves asymptotic normality when the nuisance estimation errors satisfy suitable product-rate conditions.

\item\emph{Parametric modeling}: We also study a parametric approach in Section~\ref{sec:MQR} for settings where nonparametric methods suffer from the curse of dimensionality. In such cases, we use regularized parametric working models for nuisance estimation and propose a \emph{model quadruply robust}\footnote{The term ``model quadruply robust'' is inspired by the ``model doubly robust'' concept in ATE inference \citep{UnifyingApproachDoublyrobust2019,InferenceTreatmentEffect2021}.} (MQR) estimator. Building on the proposed QR strategy, MQR further employs specialized loss functions for nuisance estimation to reduce the impact of model misspecification on the final estimator. The resulting method enables valid inference under model misspecification, even when certain nuisance functions cannot be consistently estimated, provided that appropriate sparsity and regularity conditions are satisfied.
\end{itemize}

\section{Background and Related Works}
\label{sec:back}

The semiparametric theory for natural effects has become a main focus in causal mediation analysis and has led to the development of robust methods. A foundational contribution is that of \citet{SemiparametricTheoryCausal2012}, who derived the representation $\theta = \E[\bpsi_\mathrm{eff}(\bW)]$ with the efficient influence function
\begin{align}
\bpsi_\mathrm{eff}(\bW)=\frac{A  f_0(\bS)\{Y - \mu(\bS)\}}{\pi(\bX)  f_1(\bS)} + \frac{(1 - A)\{\mu(\bS) - \tau(\bX)\}}{1 - \pi(\bX)}+ \tau(\bX).\label{def:EIF}
\end{align}
Estimation of $\theta$ can be carried out by estimating the nuisance functions in \eqref{def:EIF} and taking the empirical average of the resulting score. Existing robust estimators share this general representation but differ in how the nuisance functions are constructed and estimated, particularly in the identification and estimation of $\tau$ and $q$.

The triply robust (TR) approach \citep{SemiparametricTheoryCausal2012,CausalMediationAnalysis2022,DeepMedSemiparametricCausal2022} directly models $(f_0,f_1)$ and uses
\begin{align}
\tau(\bX)=
\int_{\mathcal S_M}\mu(\bX,\bs{m})f_0(\bX,\bs{m})\,\mathrm d\bs{m}
\label{eq:joint2}
\end{align}
to estimate $\tau$ through numerical integration. The odds-modeling triply robust (O-TR) approach \citep{TargetedMaximumLikelihood2012,SemiparametricEstimationPathspecific2019,CausalMediationAnalysis2022} constructs $\tau$ through nested regression and represents the density ratio $q$ through treatment odds:
\begin{align}
\tau(\bX)
&=
\E[\mu(\bS)\mid A=0,\bX],
\qquad
q(\bS)
=
\frac{\pi(\bX)}{1-\pi(\bX)}
\frac{1-p(\bS)}{p(\bS)},
\label{eq:joint1}
\end{align}
where $p(\bS):=\P(A=1\mid\bS)$ is the inverse-causality treatment probability.

\begin{table}[!b]
\centering
\spacingset{1}
\caption{Model-correctness scenarios under which the TR, O-TR, and proposed QR approaches achieve unbiased identification of $\theta$. In each scenario, \cmark\ indicates that the listed nuisance working models are correctly specified. The scenarios satisfy $\mathcal M_b\subset\mathcal M_b'$ and $\mathcal M_c\subset\mathcal M_c'=\mathcal M_c''$.}
\label{table:robustness_comparison}
\setlength{\tabcolsep}{14pt}
\renewcommand{\arraystretch}{1.25}
\begin{tabular}{l|l|l}
\Xhline{1pt}
\multicolumn{1}{c|}{TR} & \multicolumn{1}{c|}{O-TR} & \multicolumn{1}{c}{QR}\\
\hline
$\mathcal M_a:\ \pi^*,\mu^*$\ \cmark & $\mathcal M_a:\ \pi^*,\mu^*$\ \cmark & $\mathcal M_a:\ \pi^*,\mu^*$\ \cmark\\
\hline
$\mathcal M_b:\ (f_0^*,f_1^*),\mu^*$\ \cmark & $\mathcal M_b':\ \mu^*,\tau_n^*$\ \cmark & $\mathcal M_b':\ \mu^*,\tau_n^*$\ \cmark\\
\hline
$\mathcal M_c:\ \pi^*,(f_0^*,f_1^*)$\ \cmark & $\mathcal M_c':\ \pi^*,p^*$\ \cmark & $\mathcal M_c'':\ \pi^*,q^*$\ \cmark\\
\hline
& & $\mathcal M_d:\ q^*,\tau_n^*,\tau^*$\ \cmark\\
\Xhline{1pt}
\end{tabular}
\end{table}

These different nuisance constructions lead to distinct model-specification scenarios for unbiased identification and consistent estimation. As summarized in Table~\ref{table:robustness_comparison}, we define the following model classes according to which working models are correctly specified: $\mathcal M_a$ for $(\pi^*,\mu^*)$, $\mathcal M_b$ for $(f_0^*,f_1^*,\mu^*)$, $\mathcal M_b'$ for $(\mu^*,\tau_n^*)$, $\mathcal M_c$ for $(\pi^*,f_0^*,f_1^*)$, $\mathcal M_c'$ for $(\pi^*,p^*)$, $\mathcal M_c''$ for $(\pi^*,q^*)$, and $\mathcal M_d$ for $(q^*,\tau_n^*,\tau^*)$, where $\tau_n^*$ denotes the working model for $\tau_n(\bX):=\E[\mu^*(\bS)\mid A=0,\bX]$, the conditional mean induced by a possibly misspecified outcome working model $\mu^*$. Note that when $\mu^*$ is misspecified ($\mu^*\neq\mu$), we typically have $\tau_n\neq\tau$.

Under these definitions, the TR and O-TR approaches yield unbiased estimation of $\theta$ over $\mathcal M_{\mathrm{TR}}:=\mathcal M_a\cup\mathcal M_b\cup\mathcal M_c$ and $\mathcal M_{\mathrm{OTR}}:=\mathcal M_a\cup\mathcal M_b'\cup\mathcal M_c'$, respectively. A further review of these approaches is provided in Section~\ref{apdx:existing_tr} of the Supplementary Materials. Since TR relies on modeling the conditional densities $(f_0,f_1)$, whereas O-TR avoids doing so, the corresponding scenarios satisfy $\mathcal M_b\subset\mathcal M_b'$ and $\mathcal M_c\subset\mathcal M_c'$, and hence $\mathcal M_{\mathrm{TR}}\subset\mathcal M_{\mathrm{OTR}}$. Thus, O-TR achieves unbiased identification over a richer model class.

However, as shown in Table~\ref{table:robustness_comparison}, the model class $\mathcal M_{\mathrm{OTR}}$ still exhibits an asymmetry across the nuisance functions: $\mu^*$ and $\pi^*$ each appear in two admissible scenarios, and at least one of them must be correctly specified, whereas $\tau_n^*$ and $p^*$ each appear only once. This asymmetry naturally raises the question of whether the robustness achieved by existing methods can be further extended.

Robust statistical inference becomes even more challenging in high-dimensional settings. Within the double machine learning framework \citep{DoubleDebiasedMachine2018}, for example, \citet{CausalMediationAnalysis2022} and \citet{DeepMedSemiparametricCausal2022} establish valid inference only when all nuisance models are correctly specified; fully parametric approaches \citep{BayesianShrinkageEstimation2020,HighdimensionalMediationAnalysis2022,CausalMediationAnalysis2023} are even more restrictive, requiring full correct specification merely for consistency. In practice, when the dimension greatly exceeds the sample size, nuisance functions are typically estimated via regularized parametric working models, making the assumption that all models accurately capture the true DGP largely unrealistic.

\section{Quadruply Robust Estimation}\label{sec:QR}

In this section, we consider a nonparametric modeling approach and develop a flexible and robust framework for estimating mediation effects.

\subsection{A quadruply robust estimator}\label{sec:dr_tau}

\paragraph*{Doubly robust estimation for $\tau$}
We first consider estimation of the cross-world conditional mean $\tau(\bX)=\E[Y(1,\bM(0))\mid\bX]$, a key quantity for identifying $\theta$. Since the cross-world counterfactual is unobserved, $\tau$ must be recovered through observed-data representations. The existing TR and O-TR methods use the integral representation in \eqref{eq:joint2} and the nested representation in \eqref{eq:joint1}, respectively. Both approaches depend on an accurate initial estimate of $\mu$. Consequently, an inconsistent estimate of $\mu$ generally leads to an inconsistent estimator of $\tau$. Alternatively, $\tau$ can also be identified through importance weighting \citep{CorrectingSampleSelection2006} as $\tau(\bX) = \E[q(\bS) Y \mid A = 1, \bX]$, where the conditional density ratio $q(\bS)$ adjusts for the conditional distribution shift from the $A=1$ group to the $A=0$ group. This strategy instead depends on a reliable initial estimate of $q$.

To reduce reliance on any single working model, we bridge the nested-regression and importance-weighting approaches and propose a DR representation for $\tau$. This representation requires only one of the working models for $\mu$ or $q$ to be correctly specified, and therefore provides a more reliable identification of $\tau$.

\begin{theorem}[A DR representation for $\tau$]\label{thm:dr_tau}
Let Assumption~\ref{cond:basic} hold. If either $\mu^*=\mu$ or $q^*=q$, then $\tau(\bX)=\tau_n(\bX)+\E[q^*(\bS)\{Y-\mu^*(\bS)\}\mid A=1,\bX]$.
\end{theorem}
Based on Theorem~\ref{thm:dr_tau}, a DR estimation strategy for $\tau$ is constructed from initial estimators $\hat\mu$, $\hat q$, and a nested-regression estimator $\hat\tau_n$.
Let $\hat\mu$ denote an estimator of $\mu(\bS)=\E[Y \mid A = 1, \bS]$, obtained by fitting any regression model of $Y$ on $\bS$ using the subsample with $A=1$. Similarly, let $\hat q$ be an estimator of the density ratio $q(\bS)$, which can be obtained by separately estimating the conditional densities $f_0$ and $f_1$ or directly under additional structural assumptions, such as through the empirical-likelihood approach of \citet{zhang2024semiparametric}.
Given $\hat\mu$ and $\hat q$, we first estimate $\tau_n$ using $\hat\mu$ and observations from the $A=0$ group:
\begin{align*}
\hat\tau_{n}(\cdot)\in\arg\min_{f \in \mathcal{G}}N^{-1}\sum_{i=1}^N (1-A_i)\left\{\hat\mu(\bS_i)-f(\bX_i)\right\}^2,
\end{align*}
where $\mathcal G$ is a prespecified function class, such as linear models, trees, or neural networks. We then estimate $\tau$ using the DR representation to correct the bias:
\begin{align}\label{def:tauhat}
\hat\tau(\cdot)\in\arg\min_{f \in \mathcal{G}}N^{-1}\sum_{i=1}^NA_i\left[\hat\tau_{n}(\bX_i)+\hat q(\bS_i)\left\{Y_i-\hat\mu(\bS_i)\right\}-f(\bX_i)\right]^2.
\end{align}
Alternatively, $\tau(\bX)$ can be estimated by directly modeling the difference $\tau(\bX)-\tau_n(\bX)$; see Section~\ref{apdx:tau_alter} of the Supplementary Materials.

The DR estimation of $\tau$ is closely related to the standard augmented inverse-weighting principle for conventional finite-dimensional target parameters, such as the ATE. However, DR estimation of the infinite-dimensional function $\tau$ relies on estimating a \emph{conditional density ratio}, rather than the usual conditional treatment probability, to characterize the \emph{conditional} distributional shift across mediator groups while avoiding adjustment for the direct marginal difference between the treatment groups. DR estimation of intermediate outcome models has also been considered in dynamic settings \citep{MultiplyRobustGfunctional2017,SequentialDoubleRobustness2018,NonparametricCausalEffects2023,HighdimensionalInferenceDynamic2024,DynamicTreatmentEffects2023}. These settings do not involve cross-world outcomes, and plugging a DR estimator of the intermediate outcome regression function into the standard influence function is sufficient to enhance robustness for the corresponding single-world target parameters. In contrast, in our setting, DR estimation of $\tau$ alone does not improve robustness for the cross-world target $\theta=\E[Y(1,\bM(0))]$. In the following, we show that further improving robustness for the final target parameter requires an additional debiasing correction to the canonical mediation influence function that jointly incorporates the DR and nested estimators of $\tau$.

\paragraph*{A QR representation for $\theta$}
We now turn to estimation of the target parameter $\theta$. Consider a new scenario $\mathcal M_d$, under which $(q^*,\tau^*,\tau_n^*)=(q,\tau,\tau_n)$ while both $\mu^*$ and $\pi^*$ may be misspecified. In this setting, neither the TR nor the O-TR approach guarantees unbiased identification; see Table~\ref{table:robustness_comparison}.

Let $\hat\pi$ be an estimator of the propensity score $\pi(\bX)=\P(A=1\mid\bX)$. Following the plug-in strategy described after \eqref{def:EIF}, we substitute the nuisance estimates $\hat\pi$, $\hat q$, and $\hat\mu$, together with the new DR estimator $\hat\tau$ defined in \eqref{def:tauhat}, into the score and then take the empirical average. The resulting estimator is denoted by $\hat\theta_{\mathrm{PI}}$. However, as shown in the analysis below, this direct plug-in approach is generally insufficient to provide improved robustness.

Replace the true nuisance functions in \eqref{def:EIF} with their working models and define
\begin{align*}
\bpsi^\dagger(\bW):= \frac{Aq^*(\bS)\{Y-\mu^*(\bS)\}}{\pi^*(\bX)}+ \frac{(1-A)\{\mu^*(\bS)-\tau^*(\bX)\}}{1-\pi^*(\bX)}+ \tau^*(\bX).
\end{align*}
Then under standard regularity conditions, $\hat\theta_{\mathrm{PI}}=\E[\bpsi^\dagger(\bW)]+o_p(1)$.
Under the newly introduced scenario $\mathcal M_d$, its asymptotic bias becomes
\begin{align*}
\E[\bpsi^\dagger(\bW)]-\theta =& \E[\varDelta(\bW)],\quad\text{where}\\
\varDelta(\bW) :=& \Big\{\frac{\pi(\bX)}{\pi^*(\bX)} - \frac{1-\pi(\bX)}{1-\pi^*(\bX)}\Big\}\E[\mu(\bS)-\mu^*(\bS)\mid A=0,\bX].
\end{align*}
The bias term $\E[\varDelta(\bW)]$ is generally nonzero when both $\pi^*$ and $\mu^*$ are misspecified. Therefore, even if $\tau=\tau^*$ is accurately identified by the newly proposed DR approach, the canonical influence function generally remains biased. To address this issue, we derive the following equivalent representation of the bias under $\mathcal M_d$:
$$\E[\varDelta(\bW)]=\E[\varDelta^*(\bW)],\;\;\varDelta^*(\bW) := \Big\{\frac{A}{\pi^*(\bX)} - \frac{1-A}{1-\pi^*(\bX)}\Big\} \{\tau^*(\bX)-\tau_n^*(\bX)\},$$
where $\pi(\bX)$ is replaced by $A$, and $\E[\mu(\bS)-\mu^*(\bS)\mid A=0,\bX]=\tau(\bX)-\tau_n(\bX)$ is replaced by the difference between working models $\tau^*(\bX)-\tau_n^*(\bX)$. By subtracting this bias term $\varDelta^*(\bW)$ from $\bpsi^\dagger(\bW)$, we introduce a new QR score function $\bpsi^*$:
\begin{align}
\bpsi^*(\bW)&:=\frac{A q^*(\bS)\{Y-\mu^*(\bS)\}}{\pi^*(\bX)}+\frac{(1-A)\{\mu^*(\bS)-\tau_{n}^*(\bX)\}}{1-\pi^*(\bX)}\nonumber\\
&\qquad+\frac{A\{\tau_{n}^*(\bX)-\tau^*(\bX)\}}{\pi^*(\bX)}+\tau^*(\bX).\label{def:QR_form}
\end{align}
This adjustment yields the identification $\theta=\E[\bpsi^*(\bW)]$ under the new scenario $\mathcal M_d$. At the same time, $\E[\varDelta^*(\bW)]=0$ under the existing TR and O-TR classes $\mathcal M_\mathrm{TR}$ and $\mathcal M_\mathrm{OTR}$. Hence, subtracting $\varDelta^*(\bW)$ does not create extra bias when the correction is unnecessary.

The following theorem shows the QR property of the proposed representation.

\begin{theorem}[A QR representation for $\theta$]\label{thm:QR_robust}
Let Assumption \ref{cond:basic} hold. Then $\E[\bpsi^*(\bW)]=\theta$ over the model class
$\mathcal M_\mathrm{QR}:=\mathcal M_a\cup\mathcal M_b'\cup\mathcal M_c''\cup\mathcal M_d$; that is, the representation holds if at least one of the following conditions is satisfied: (a) $\mu^*=\mu$ and $\pi^*=\pi$; (b) $\mu^*=\mu$ and $\tau_n^*=\tau_n$; (c) $\pi^*=\pi$ and $q^*=q$; (d) $q^*=q$ and $(\tau_n^*,\tau^*)=(\tau_n,\tau)$.
\end{theorem}

By incorporating the additional regime $\mathcal M_d$, the class $\mathcal M_\mathrm{QR}=\mathcal M_\mathrm{OTR}\cup\mathcal M_d\supset\mathcal M_\mathrm{TR}\cup\mathcal M_d$ strictly enlarges the robustness classes of both existing approaches; see Table~\ref{table:robustness_comparison} and concrete examples in Section~\ref{apdx:example} of the Supplementary Materials. Moreover, under $\mathcal M_b'$, $\mu^*=\mu$ implies $\tau_n=\tau$ and therefore also ensures correct specification of $\tau$. If $(\tau_n^*,\tau^*)$ is viewed jointly as a single nuisance component, then $\mathcal M_\mathrm{QR}$ exhibits a \emph{symmetric robustness structure}: unbiased identification requires one nuisance component involving $\bX$, either $\pi^*$ or $(\tau_n^*,\tau^*)$, to be correctly specified, together with one nuisance component involving $\bS$, either $q^*$ or $\mu^*$.

\paragraph*{The general QR estimator}
We propose the \emph{quadruply robust estimator} of $\theta$ as $\hat{\theta}_{\mathrm{QR}}=N^{-1}\sum_{i=1}^N\hat\bpsi(\bW_i)$, where $\hat\bpsi$ is obtained from the QR score in \eqref{def:QR_form} by replacing the nuisance working models with their estimators; see Algorithm~\ref{alg:QR} for the cross-fitted implementation.

\begin{algorithm}[!t] \caption{The quadruply robust (QR) estimator for $\theta$}\label{alg:QR}
\begin{algorithmic}[1]
\Require Observations $\mathbb{S}=(\bW_i)_{i=1}^N$. Let $\K\geq2$ be the number of folds.
\State Let $\I_1,\dots,\I_{\K}$ be a disjoint partition of $\I=\{1,\dots,N\}$ with equal sizes $n=N/\K$.
\For{$k\in\{1,2,\dots,\K\}$}
\State Let $\I_{-k}^0,\I_{-k}^1, \I_{-k}^2$ be a disjoint partition of $\I_{-k}=\I\setminus\I_k$ with equal sizes $M$.
\State Using $\I_{-k}$, construct $\hat{\pi}^{-k}$; using $\I_{-k}^0\cup\I_{-k}^1$, construct $\hat{q}^{-k}$; using $\I_{-k}^0$, construct $\hat{\mu}^{-k}$.
\State Using $\I_{-k}^1$, construct $\hat{\tau}_n^{-k}$ by solving
\begin{align*}
    \hat{\tau}_n^{-k}\in\arg\min\limits_{f\in\mathcal{G}}M^{-1}\sum\limits_{i\in \I_{-k}^1}(1-A_i)\{\hat{\mu}^{-k}(\bS_i)-f(\bX_i)\}^2.
\end{align*}
\State Using $\I_{-k}^2$, construct $\hat{\tau}^{-k}$ by solving
\begin{align*}
    \hat{\tau}^{-k}\in\arg\min\limits_{f\in\mathcal{G}}M^{-1}\sum\limits_{i\in \I_{-k}^2}A_i[\hat{\tau}_n^{-k}(\bX_i)+\hat{q}^{-k}(\bS_i)\{Y_i-\hat{\mu}^{-k}(\bS_i)\}-f(\bX_i)]^2.
\end{align*}
\EndFor\\
\Return The QR estimator is defined as
$\hat{\theta}_{\mathrm{QR}}=N^{-1}\sum_{k=1}^{\K}\sum_{i\in\I_k}\hat\bpsi^{(-k)}(\bW_i),$
where $\hat\bpsi^{(-k)}$ is obtained from \eqref{def:QR_form} by replacing the working models $(\pi^*,q^*,\mu^*,\tau_{n}^*,\tau^*)$ with their cross-fitted estimates $(\hat{\pi}^{-k},\hat{q}^{-k},\hat{\mu}^{-k},\hat{\tau}_{n}^{-k},\hat{\tau}^{-k})$.
\end{algorithmic}
\end{algorithm}

The proposed QR estimator differs from existing approaches in two key aspects: (a) a new DR estimator for $\tau$, and (b) a new QR score function \eqref{def:QR_form} that combines both the DR estimation and the nested-regression estimation for $\tau$. Both components are needed to improve robustness. Part (a) reduces the dependence of accurate $\tau$ estimation on an accurate initial estimate of $\mu$, while part (b) further reduces bias in the final estimation of $\theta$.

\begin{remark}[Implementation variant]\leavevmode\label{remark:imple_qr}
    To mitigate overfitting bias from nonparametric nuisance estimation, we employ cross-fitting \citep{DoubleDebiasedMachine2018} and consider two implementations of the QR estimator. The procedure in Algorithm~\ref{alg:QR} performs an additional within-fold split for the DR estimation of $\tau$, ensuring that the fitted nuisance estimates $\hat{\mu}$, $\hat{q}$, and $\hat{\tau}_n$ in \eqref{def:tauhat} remain independent of the subsample used to estimate $\tau$. Alternatively, Section~\ref{apdx:qr2} of the Supplementary Materials provides a variant, termed QR2, which omits this sub-splitting and estimates all nuisance functions on the entire training fold. We recommend QR2 when the sample size is limited to support further splitting. When a distinction from QR2 is necessary, we refer to the implementation in Algorithm~\ref{alg:QR} as QR1.
\end{remark}

\subsection{Theoretical properties of the QR estimator}\label{sec:QR_thm}

In this section, we present the theoretical properties of the proposed QR estimator. Let $\mathcal{H}=\{\pi, q,\mu,\tau_n,\tau\}$ index the nuisance functions under consideration. Unless stated otherwise, let $h$ denote a generic nuisance function. For a random variable $U$ and $0<p<\infty$, we write $\|U\|_{\P,p}:=(\int |U|^p\,\mathrm d\P)^{1/p}$, where $\P$ denotes the joint law of an independent copy of the underlying variables. Working-model misspecification is quantified via the $L_6(\P)$ approximation errors $\eqr_h := \normp{h^* - h}{6}$. Throughout our asymptotic analysis, we adopt a triangular-array formulation that allows the dimensions $(d_1,d_2)$, the DGP, and the associated nuisance functions and error rates to depend on $N$; see Remark~\ref{rem:triangular_array} of the Supplementary Materials for details. Our theoretical analysis further relies on the following regularity conditions \citep{DoubleDebiasedMachine2018,CausalMediationAnalysis2022}.

\begin{assumption}[Regularity conditions]\label{cond:QR_basic}
Denote $\bzS=(\bX^\top,\bM(0)^\top)^\top$. Assume the following conditions hold for some constants $c_0\in(0,1/2)$, $G_0>0$, $c_1>0$, and $r>2$:
(a) $\P(c_0\leq \pi^*(\bX)\leq 1-c_0)=1$ and $\P(c_0^2\leq q^*(\bS)\leq c_0^{-2})=1$; with probability approaching one, $\P(c_0\leq \hat{\pi}(\bX)\leq 1-c_0)=1$ and $\P(c_0^2\leq \hat{q}(\bS)\leq c_0^{-2})=1$.
(b) $\|\hat\pi-\pi^*\|_{\P,2}=O_p(\br_\pi)$, $\|\hat q-q^*\|_{\P,2}=O_p(\br_q)$, $\|\hat\mu-\mu^*\|_{\P,2}=O_p(\br_\mu)$, $\|\hat\tau_n-\tau_n^*\|_{\P,2}=O_p(\br_{\tau_n})$, and $\|\hat\tau-\tau^*\|_{\P,2}=O_p(\br_\tau)$ as $N\to\infty$, with positive sequences $\br_{\pi},\br_{q},\br_{\mu},\br_{\tau_n},\br_{\tau}=o(1)$.
(c) Define $\xi:=\tau(\bX)-\theta$, $\varepsilon:=Y(1,\bM)-\mu^*(\bS)$, $\zeta:=\mu^*(\bzS)-\tau^*(\bX)$, and $\varrho:=\mu^*(\bzS)-\tau_n^*(\bX)$. Assume that $\E[|\varepsilon|^r\mid\bS]+\E[|\varrho|^r+|\zeta|^r+|\xi|^r\mid\bX]\leq G_0$ almost surely, and $\sigma^2:=\Var[\bpsi^*(\bW)]>c_1$.
\end{assumption}

\begin{theorem}[Convergence]\label{thm:QR_consistent}
Let Assumptions~\ref{cond:basic} and~\ref{cond:QR_basic} hold. Then as $N\to\infty$,
$\hat\theta_{\mathrm{QR}}-\theta=O_p(N^{-1/2}+\bar{\mathcal R}_{\mathrm{QR}})$, where
\begin{align}
\bar{\mathcal R}_{\mathrm{QR}}:=&\underbrace{\eqr_\pi(\eqr_{\tau_n}+\eqr_\tau)+\eqr_q\eqr_\mu}_{\text{approximation errors}}
\!+\!\underbrace{\br_\pi(\br_\mu+\br_{\tau_n}+\br_\tau)+\br_q\br_\mu}_{\text{stochastic errors}}\nonumber\\
&+\underbrace{\eqr_\pi(\br_\mu+\br_{\tau_n}+\br_\tau)+\eqr_q\br_\mu+\eqr_\mu\br_q+(\eqr_{\tau_n}+\eqr_\tau)\br_\pi}_{\text{joint effect of both errors}}.\label{decom:qr}
\end{align}
\end{theorem}
When all models are correctly specified, $\hat\theta_{\mathrm{QR}}-\theta=O_p(N^{-1/2}+\br_\pi\br_{\tau_n}+\br_\pi\br_\tau+\br_\mu\br_q+\br_\mu\br_\pi)$, so $\sqrt{N}$-consistency holds as long as all the involved product rates are $O(N^{-1/2})$. More generally, the approximation-error products in \eqref{decom:qr} reveal the quadruply robust structure of $\hat\theta_{\mathrm{QR}}$: each cross-product vanishes whenever at least one model in the corresponding pair is correctly specified. The same conclusion also extends to mild misspecification across all models simultaneously, for example, when $\eqr_h=o(N^{-1/2})$ for all $h\in\mathcal{H}$. When a model exhibits nonvanishing misspecification with $\eqr_h\asymp1$, $\sqrt{N}$-consistency can still be maintained in low-dimensional parametric settings, where $\br_h=O(N^{-1/2})$. In high-dimensional or nonparametric settings, however, nuisance estimation rates are typically slower than $N^{-1/2}$, which may degrade the overall convergence rate under model misspecification. In Section~\ref{sec:MQR}, we introduce the MQR approach to mitigate this additional bias arising from model misspecification, with particular emphasis on high-dimensional settings.

\begin{theorem}[Asymptotic normality]\label{thm:QR_normal}
Let Assumptions~\ref{cond:basic} and~\ref{cond:QR_basic} hold, and further assume that $\bar{\mathcal R}_{\mathrm{QR}}=o(N^{-1/2})$.
Then, as $N\to\infty$, we have
$\sigma^{-1}N^{1/2}(\hat{\theta}_{\mathrm{QR}}-\theta)\leadsto \mathcal N(0,1)$
and $\hat{\sigma}^2=\sigma^2\{1+o_p(1)\}$, where
$\hat{\sigma}^2:=N^{-1}\sum_{k=1}^{\K}\sum_{i\in\I_k}\{\hat\bpsi^{(-k)}(\bW_i)-\hat{\theta}_{\mathrm{QR}}\}^2.$
\end{theorem}
Since nuisance estimation rates $\br_h$ are generally not expected to be faster than the $N^{-1/2}$ rate, their appearance in the joint-effect terms in \eqref{decom:qr} implies that the requirement $\bar{\mathcal R}_{\mathrm{QR}}=o(N^{-1/2})$ in Theorem~\ref{thm:QR_normal} essentially requires all nuisance models to be consistently estimable, with vanishing approximation errors $\eqr_h=o(1)$. Moreover, when all working models are correctly specified, valid inference requires $\br_\pi(\br_\mu+\br_{\tau_n}+\br_\tau)+\br_q\br_\mu=o(N^{-1/2})$. In this case, $\bpsi^*=\bpsi_\mathrm{eff}$, and $\hat{\theta}_{\mathrm{QR}}$ is semiparametrically efficient, with its asymptotic variance attaining the semiparametric efficiency bound derived in \citet{SemiparametricTheoryCausal2012}.

\section{Model Quadruply Robust Estimation}\label{sec:MQR}

In this section, we consider high-dimensional settings, where the dimensions of covariates and mediators $(d_1,d_2)$ can both be large relative to the sample size $N$. Fully nonparametric methods suffer from the curse of dimensionality, which leads to slow convergence rates or even inconsistent estimation when $d=d_1+d_2$ exceeds $N$. We therefore focus on regularized parametric nuisance estimation and aim to reduce the bias caused by model misspecification.

\subsection{A model quadruply robust estimator}\label{sec:r_display}

We adopt linear estimators for all outcome-related nuisance functions: $\hat\mu(\bS)=\bS^\top\hbemu$, $\hat\tau_n(\bX)=\bX^\top\hbettau$, and $\hat\tau(\bX)=\bX^\top\hbetau$. Let $g(x)=\{1+\exp(-x)\}^{-1}$ denote the logistic function. We consider two separate logistic models for the propensity score, $\hat\pi_a(\bX):= g(\bX^\top\hbepia)$ and $\hat\pi_b(\bX):= g(\bX^\top\hbepib)$; the motivation for using two distinct estimates is discussed later. The conditional density ratio is modeled using an exponential form, $\hat q(\bS)=\exp(\bS^\top\hbeq)$. Our procedure can also be applied when the covariates $\bX$ and $\bS$ are replaced with basis expansions, such as polynomials, splines, or wavelets. Here, a nuisance model is said to be correctly specified if the corresponding true nuisance function belongs to the chosen working model family. In what follows, we construct a collection of specialized nuisance estimators $\hbnu=(\hbepia,\hbepib,\hbeq,\hbemu,\hbettau,\hbetau)$ with the following properties: (a) when a nuisance model is correctly specified, the corresponding estimator converges to the true parameter; (b) when a nuisance model is misspecified, the resulting error has a limited impact on the estimation of the target parameter $\theta$.

\paragraph*{A moment condition for robust inference}
Let $\bnu=(\bepia,\bepib,\beq,\bemu,\bettau,\betau)$ be an arbitrary collection of vectors. We use the QR score function in \eqref{def:QR_form}, but modify the weighting scheme as follows. For the $A=1$ observations, we use the inverse probability $g^{-1}(\bX^\top\bepia)$, and for the $A=0$ observations, we use $[1-g(\bX^\top\bepib)]^{-1}$:
\begin{align}\label{def:psihat-MQR}
\bpsi(\bW,\bnu)&:=\frac{A\exp(\bS^\top\beq)(Y-\bS^\top\bemu)}{g(\bX^\top\bepia)}+\frac{(1-A)(\bS^\top\bemu-\bX^\top\bettau)}{1-g(\bX^\top\bepib)}\nonumber\\
&\qquad+\frac{A(\bX^\top\bettau-\bX^\top\betau)}{g(\bX^\top\bepia)}+\bX^\top\betau.
\end{align}

As shown in Theorem~\ref{thm:QR_consistent}, the convergence rate of the plug-in estimator $\hat{\theta}=N^{-1}\sum_{i=1}^N\bpsi(\bW_i,\hat\bnu)$ typically depends linearly on the nuisance stochastic errors $\|\hat\bnu-\bnu^*\|_2$ when model misspecification occurs. To remove this first-order dependence, we construct a specialized set of nuisance estimates $\hbnu$ whose limits $\bnu^*=(\bepiar,\bepibr,\beqr,\bemur,\bettaur,\betaur)$ satisfy the moment condition
\begin{align}\label{eq:moment}
\E[\bnabla_{\bnu}\bpsi(\bW,\bnu^*)]=\bz\quad\text{\emph{even under model misspecification}}.
\end{align}
When \eqref{eq:moment} holds, Taylor expansion gives
\begin{align*}
\E[\bpsi(\bW,\hat\bnu)]=\E[\bpsi(\bW,\bnu^*)]+\E[\bnabla_{\bnu}\bpsi(\bW,\bnu^*)]^\top(\hbnu-\bnu^*)+\mathrm{Rem}=\E[\bpsi(\bW,\bnu^*)]+\mathrm{Rem},
\end{align*}
where $\E[\bpsi(\bW,\bnu^*)]=\theta$ under the assumptions of Theorem~\ref{thm:QR_robust} and the remainder term satisfies $\mathrm{Rem}=O(\|\hbnu-\bnu^*\|_2^2)$ under regularity conditions. Hence the bias is of second order.

It is worth noting that Neyman orthogonality \citep{DoubleDebiasedMachine2018}, which is a property of the score function $\bpsi$, is not sufficient to obtain this result. Neyman orthogonality only guarantees the moment condition $\E[\bnabla_{\bnu}\bpsi(\bW,\bnu^*)]=\bz$ when all models are correctly specified. To maintain this orthogonality under partial model misspecification, we construct nuisance estimates $\hbnu$ that converge to a misspecified but well-performing target $\bnu^*$, which ensures \eqref{eq:moment} even when some working models are incorrect.

\paragraph*{Specialized loss functions for nuisance estimation}
To reach the desired moment condition, we introduce a sequence of loss functions:
\begin{align}
\ell_1(\bW,\bepia)&:=(1-A)\bX^\top\bepia+A\exp(-\bX^\top\bepia),\label{eq:r_loss1}\\
\ell_2(\bW,\bepib)&:=(1-A)\exp(\bX^\top\bepib)-A\bX^\top\bepib,\label{eq:r_loss2}\\
\ell_3(\bW,\bnu_3,\beq)&:=Ag^{-1}(\bX^\top\bepia)\exp(\bS^\top\beq)-(1-A)\{1-g(\bX^{\top}\bepib)\}^{-1}\bS^\top\beq,\label{eq:r_loss3}\\
\ell_4(\bW,\bnu_4,\bemu)&:=A g^{-1}(\bX^\top\bepia)\exp(\bS^\top\beq)(\bS^\top\bemu-Y)^2,\label{eq:r_loss4}\\
\ell_5(\bW,\bnu_5,\bettau)&:=(1-A)\exp(\bX^\top\bepib)(\bX^\top\bettau-\bS^\top\bemu)^2,\label{eq:r_loss5}\\
\ell_6(\bW,\bnu_6,\betau)&:=A\exp(-\bX^\top\bepia)\left[\exp(\bS^\top\beq)(Y-\bS^\top\bemu)+\bX^\top\bettau-\bX^\top\betau\right]^2,\label{eq:r_loss6}
\end{align}
where $\bnu_3=(\bepia,\bepib)$, $\bnu_4=(\bepia,\beq)$, $\bnu_5=(\bepib,\bemu)$, and $\bnu_6=(\bepia,\beq,\bemu,\bettau)$. Let $\bnu^*=(\bepiar,\bepibr,\beqr,\bemur,\bettaur,\betaur)$ denote the population-level minimizers, and let $(\pi_a^*,\pi_b^*,q^*,\mu^*,\tau_n^*,\tau^*)$ be the corresponding target nuisance functions induced by these parameters. For instance, $\bepiar:=\arg\min_{\bepia\in\R^{d_1}}\E[\ell_1(\bW,\bepia)]$ and $\pi_a^*(\bX)=g(\bX^\top\bepiar)$. These target parameters are defined sequentially in the order listed above, and $(\bnu_3^*,\bnu_4^*,\bnu_5^*,\bnu_6^*)$ denote the corresponding intermediate collections of population-level minimizers used to define subsequent nuisance targets.

The loss $\ell_1$ has been used to achieve covariate balancing, $\E[A g^{-1}(\bX^\top\bepiar)\bX]=\E(\bX)$, in the ATE estimation literature \citep{RobustEstimationCausal2020,tan2020model}. The loss $\ell_2$ is introduced to obtain the complementary balancing condition $\E[(1-A)\{1-g(\bX^\top\bepibr)\}^{-1}\bX]=\E(\bX)$ for the $A=0$ group. The loss $\ell_3$ is designed to estimate the conditional density ratio $q(\bS)=f_0(\bS)/f_1(\bS)$. Since accurately estimating the conditional densities $f_a(\bS)=f(\bM\mid A=a,\bX)$ ($a\in\{0,1\}$) is typically challenging when both $\bX$ and $\bM$ are high-dimensional, we propose to estimate the ratio directly by adapting the loss $A\exp(\bS^\top\beq)-(1-A)\bS^\top\beq$ used for \emph{unconditional} density ratio estimation \citep{SemisupervisedTriplyRobust2025}, while further introducing inverse probability weights derived from $\ell_1$ and $\ell_2$ to account for conditioning on covariates. Since the estimation now involves initial estimates for the propensity score, it might appear that an accurate estimate of $q$ requires accurate estimates for $\pi$ as well. However, as discussed in Section~\ref{apdx:mqr_supp} of the Supplementary Materials, correct specification of $q^*$ can still hold even when $\pi^*$ is misspecified, as long as the conditional mean $\bs{m}(\bX):=\E[\bM\mid A=0,\bX]$ is linear in $\bX$. Importantly, this does not require explicitly estimating the function $\bs{m}(\bX)$.

The outcome-related nuisance functions are estimated using the weighted squared losses $\ell_4$, $\ell_5$, and $\ell_6$. These weights arise directly from the construction of the score function in \eqref{def:psihat-MQR}. For example, the residual $Y-\bS^\top\bemu$ enters \eqref{def:psihat-MQR} with sampling weight $A g^{-1}(\bX^\top\bepia)\exp(\bS^\top\beq)$, and the same weight is therefore used in $\ell_4$. In addition, motivated by the QR strategy in Section~\ref{sec:dr_tau}, we use a nested regression in $\ell_5$ and a DR regression in $\ell_6$ to correct the estimation error of $\tau$ induced by an inaccurate initial estimate of $\mu$. Relative to the previous QR construction, we further incorporate the newly designed weights to reduce the impact of model misspecification on estimation of the final target parameter. These weights may, however, introduce additional finite-sample variability. This reflects a familiar robustness--stability trade-off: for example, incorporating propensity-score weighting in DR procedures can improve robustness while potentially increasing finite-sample instability.

All constructions above ensure the following connections:
{\small
\begin{align*}
&\E[\bnabla_{\betau}\bpsi(\bW,\bnu^*)]=\E[\bnabla_{\bepia}\ell_1(\bW,\bepiar)],\\
&\E[\bnabla_{\bettau}\bpsi(\bW,\bnu^*)]=-\E[\bnabla_{\bepia}\ell_1(\bW,\bepiar)+\bnabla_{\bepib}\ell_2(\bW,\bepibr)],\\
&\E[\bnabla_{\bemu}\bpsi(\bW,\bnu^*)]=-\E[\bnabla_{\beq}\ell_3(\bW,\bnu^*_3,\beqr)],\;\; \E[\bnabla_{\beq}\bpsi(\bW,\bnu^*)]=-\tfrac12\E[\bnabla_{\bemu}\ell_4(\bW,\bnu^*_4,\bemur)],\\
&\E[\bnabla_{\bepib}\bpsi(\bW,\bnu^*)]=-\tfrac12\E[\bnabla_{\bettau}\ell_5(\bW,\bnu^*_5,\bettaur)],\;\; \E[\bnabla_{\bepia}\bpsi(\bW,\bnu^*)]=\tfrac12\E[\bnabla_{\betau}\ell_6(\bW,\bnu^*_6,\betaur)].
\end{align*}}

By the first-order conditions defining the population minimizers, all right-hand sides are zero vectors, and hence the corresponding gradients of the QR score also vanish, which yields the moment condition \eqref{eq:moment} even under model misspecification.

The idea of preserving orthogonality under partial model misspecification has been studied for ATE inference \citep{UnifyingApproachDoublyrobust2019,InferenceTreatmentEffect2021,tan2020model,RobustEstimationCausal2020}. For the mediation functional considered here, however, the cross-world structure introduces nested and interdependent nuisance targets. We therefore introduce new loss functions tailored to this cross-world setting. Specifically, $\ell_3$ is designed for conditional density ratio estimation, which is used to adjust for the shift in the mediator distribution, while $\ell_5$ and $\ell_6$ are designed for cross-world conditional mean outcome estimation based on the nested and DR approaches, respectively. The remaining loss functions, $\ell_1$, $\ell_2$, and $\ell_4$, are directly adapted from existing proposals.

Nevertheless, we introduce two separate loss functions, $\ell_1$ and $\ell_2$, for estimating the same propensity score function. This construction also arises naturally from the QR score, where $\tau$ needs to be estimated through both the nested and DR approaches. Specifically, the loss $\ell_1$ is constructed so that the final estimator is first-order insensitive to the error $\hbetau-\betaur$ by enforcing $\E[\bnabla_{\betau}\bpsi(\bW,\bnu^*)]=\E[\{1-A g^{-1}(\bX^\top\bepiar)\}\bX]=\bz$, whereas the loss $\ell_2$ is constructed so that the estimator is also first-order insensitive to the error $\hbettau-\bettaur$ by enforcing
$\E[\bnabla_{\bettau}\bpsi(\bW,\bnu^*)]=\E[\{Ag^{-1}(\bX^\top\bepiar)-(1-A)\{1-g(\bX^\top\bepibr)\}^{-1}\}\bX]=\E[\{1-(1-A)\{1-g(\bX^\top\bepibr)\}^{-1}\}\bX]=\bz.$ In general, there need not exist a single coefficient vector $\bepir=\bepiar=\bepibr$ that satisfies both conditions simultaneously unless the propensity score model is correctly specified. To overcome this constraint, we introduce two distinct targets $(\bepiar,\bepibr)$ defined through the separate loss functions $(\ell_1,\ell_2)$. When the true propensity score function is logistic, $\bepiar=\bepibr$, and both targets coincide with the true logistic coefficient; see Section~\ref{apdx:mqr_supp} of the Supplementary Materials for further justification.

Moreover, the entire procedure is built upon the QR framework proposed in Section~\ref{sec:QR}, where we introduce a QR score for the target parameter that incorporates an additional debiasing term, together with both nested and DR estimation of $\tau$. These constructions are specific to the mediation functional and are not required for ATE estimation, because of the cross-world structure and the resulting interdependence among the nuisance functions.

\begin{algorithm}[b!] \caption{The model quadruply robust (MQR) estimator for $\theta$}\label{alg:MQR}
\begin{algorithmic}[1]
\Require Observations $\mathbb{S}=(\bW_i)_{i=1}^N$. Let $\K\geq2$ be the number of folds.
\State Let $\I_1,\dots,\I_{\K}$ be a disjoint partition of $\I=\{1,\dots,N\}$ with equal sizes $n=N/\K$.
\For{$k\in\{1,2,\dots,\K\}$}
\State Let $\I_{\pi}, \I_{q}, \I_{\mu}, \I_{\tau_n}, \I_{\tau}$ be a disjoint partition of $\I_{-k}=\I\setminus\I_k$ with equal sizes $M$.
\State Construct the nuisance estimates:
\begin{align*}
\hbepia^{-k}&=\mathop{\arg\min}_{\bepia\in\R^{d_1}}\Big\{M^{-1}\supiar\ell_1(\bW_i,\bepia)+\lambda_{\pi_a}\|\bepia\|_1\Big\},\\
\hbepib^{-k}&=\mathop{\arg\min}_{\bepib\in\R^{d_1}}\Big\{M^{-1}\supibr\ell_2(\bW_i,\bepib)+\lambda_{\pi_b}\|\bepib\|_1\Big\},\\
\hbeq^{-k}&=\mathop{\arg\min}_{\beq\in\R^{d}}\Big\{M^{-1}\suqr\ell_3(\bW_i,\hbepia^{-k},\hbepib^{-k},\beq)+\lambda_{q}\|\beq\|_1\Big\},\\
\hbemu^{-k}&=\mathop{\arg\min}_{\bemu\in\R^d}\Big\{M^{-1}\sumur\ell_4(\bW_i,\hbepia^{-k},\hbeq^{-k},\bemu)+\lambda_{\mu}\|\bemu\|_1\Big\},\\
\hbettau^{-k}&=\mathop{\arg\min}_{\bettau\in\R^{d_1}}\Big\{M^{-1}\suttaur\ell_5(\bW_i,\hbepib^{-k},\hbemu^{-k},\bettau)+\lambda_{\tau_n}\|\bettau\|_1\Big\},\\
\hbetau^{-k}&=\mathop{\arg\min}_{\betau\in\R^{d_1}}\Big\{M^{-1}\sutaur\ell_6(\bW_i,\hbepia^{-k},\hbeq^{-k},\hbemu^{-k},\hbettau^{-k},\betau)+\lambda_{\tau}\|\betau\|_1\Big\},
\end{align*}
where $\lambda_{\pi_a},\lambda_{\pi_b},\lambda_{q},\lambda_{\mu},\lambda_{\tau_n},\lambda_{\tau}>0$ are tuning parameters selected by cross-validation. Each tuning parameter is chosen sequentially, tuning one parameter at a time while treating all previously selected parameters and nuisance estimates as fixed.
\EndFor\\
\Return The MQR estimator is proposed as
$
\hat{\theta}_{\mathrm{MQR}}=N^{-1}\sum_{k=1}^{\K}\sum_{i\in\I_k}\bpsi(\bs{W}_i,\hbnu^{-k}),
$
where $\bpsi$ is defined in \eqref{def:psihat-MQR} and $\hbnu^{-k} = (\hbepia^{-k},\hbepib^{-k},\hbeq^{-k},\hbemu^{-k},\hbettau^{-k},\hbetau^{-k})$.
\end{algorithmic}
\end{algorithm}

As discussed previously, correct specification of $q^*$ depends on the form of either $\pi$ or $\bs{m}$. We therefore introduce a new model class $\mathcal M_d'$, under which $q$, $\tau$, and $\bs{m}$ are of the chosen parametric forms. In general, unless $\bM$ is one-dimensional, $\mathcal M_d'\subset\mathcal M_d$. We restrict $\mathcal M_d$ to $\mathcal M_d'$ because, under $\mathcal M_d$ alone, multivariate density-ratio estimation does not generally guarantee that $\hat q$ converges to the true $q$. By imposing $\mathcal M_d'$, this consistency property for $q$ is restored.
Further justification is provided in Section~\ref{apdx:mqr_supp} of the Supplementary Materials.

Consider the model class $\mathcal M_\mathrm{MQR}:=\mathcal M_a\cup\mathcal M_b'\cup\mathcal M_c''\cup\mathcal M_d'$. That is, at least one of the following conditions holds: (a) $\mu$ is linear and $\pi$ is logistic; (b) $\mu$ is linear and $\tau_n$ is linear; (c) $\pi$ is logistic and $q$ is exponential; or (d) $q$ is exponential and $(\tau,\bs{m})$ are linear. Theorem~\ref{thm:r_infer} below shows that valid $\sqrt N$-inference can be established under $\mathcal M_\mathrm{MQR}$ when additional sparsity and regularity conditions are satisfied. Under the chosen working models, $\mathcal M_\mathrm{MQR}$ is a subset of $\mathcal M_\mathrm{QR}$ studied in Section~\ref{sec:QR}. Nevertheless, $\mathcal M_\mathrm{MQR}=\mathcal M_\mathrm{OTR}\cup\mathcal M_d'\supset\mathcal M_\mathrm{TR}\cup\mathcal M_d'$ remains strictly broader than the corresponding model classes $\mathcal M_\mathrm{TR}$ and $\mathcal M_\mathrm{OTR}$.
More importantly, whereas TR and O-TR establish only consistency over the smaller model classes $\mathcal M_\mathrm{TR}$ and $\mathcal M_\mathrm{OTR}$ in high-dimensional settings, MQR provides valid $\sqrt{N}$-inference under potential model misspecification.

\paragraph*{The MQR estimator}
Let $\hbnu$ denote the collection of nuisance estimates obtained by solving the empirical minimization problems associated with the loss functions $\ell_1,\ell_2,\dots,\ell_6$. Combined with the QR score \eqref{def:psihat-MQR}, we define the MQR estimator for $\theta$ as
$
\hat{\theta}_{\mathrm{MQR}}=N^{-1}\sum_{i=1}^N\bpsi(\bW_i,\hbnu).
$
The full construction is described in Algorithm~\ref{alg:MQR}.

\subsection{Theoretical properties of the MQR estimator}\label{sec:r_thm}

The model class $\mathcal M_\mathrm{MQR}$ provides an illustrative setting in which each working model is treated as either exactly correct or completely misspecified. More generally, following Section~\ref{sec:QR_thm}, we quantify the $L_6(\P)$ parametric approximation errors of the true nuisance functions. Specifically, define $\emqr_\pi:=\inf_{\bbeta\in\R^{d_1}}\|\pi(\bX)-g(\bX^\top\bbeta)\|_{\P,6}$, $\emqr_q:=\inf_{\bbeta\in\R^d}\|q(\bS)-\exp(\bS^\top\bbeta)\|_{\P,6}$, $\emqr_\mu:=\inf_{\bbeta\in\R^d}\|\mu(\bS)-\bS^\top\bbeta\|_{\P,6}$, $\emqr_{\tau_n}:=\inf_{\bbeta\in\R^{d_1}}\|\tau_n(\bX)-\bX^\top\bbeta\|_{\P,6}$, $\emqr_{\tau}:=\inf_{\bbeta\in\R^{d_1}}\|\tau(\bX)-\bX^\top\bbeta\|_{\P,6}$, and $\emqr_{\bs{m}}:=\inf_{\bphi\in\R^{d_1\times d_2}}\sup_{\|\bv\|_2=1}\|\{\bs{m}(\bX)-\bphi^\top\bX\}^\top\bv\|_{\P,6}$. These errors are allowed to decay with the sample size, without requiring any of them to be exactly zero.

Let $(s_{\pi_a},s_{\pi_b},s_{q},s_{\mu},s_{\tau_n},s_{\tau})$ denote the sparsity levels of the corresponding target nuisance parameters collected in $\bnu^*$, and define $s_\pi:=s_{\pi_a}+s_{\pi_b}$. We further set
$$
r_{\pi}:=\sqrt{\frac{s_{\pi}\log d_1}{N}},\;
r_{q}:=\sqrt{\frac{s_{q}\log d}{N}},\;
r_{\mu}:=\sqrt{\frac{s_{\mu}\log d}{N}},\;
r_{\tau_n}:=\sqrt{\frac{s_{\tau_n}\log d_1}{N}},\;
r_{\tau}:=\sqrt{\frac{s_{\tau}\log d_1}{N}}.
$$
All $r$-values above reflect the complexity of the corresponding nuisance models.
Because the nuisance estimators are constructed sequentially, the actual approximation and stochastic error rates $(\bar e_h,\bar r_h)$ in Section~\ref{sec:QR_thm} may be slower than the corresponding raw rates $(e_h,r_h)$ and may depend on the convergence rates of previously constructed estimators. See Lemma~\ref{lem:r_approx_bound} and Theorem~\ref{thm:r_nuisance} of the Supplementary Materials for the precise connections.

In the analysis below, we assume that $\bX$ includes an intercept as its first entry and choose the tuning parameters as $\lambda_{\pi_a},\lambda_{\pi_b},\lambda_{\tau_n},\lambda_{\tau}\asymp \sqrt{\log d_1/N}$ and $\lambda_{q},\lambda_{\mu}\asymp \sqrt{\log d/N}$.

\begin{assumption}[Regularity conditions]\label{cond:basic2}
Let the following conditions hold for constants $c_0\in(0,1/2)$ and $\sigma_{\bS},\sigma_0,c_{\min},c_Y>0$ independent of $N,d_1,d_2$:

(a) $\P(c_0\leq \pi_a^*(\bX)\leq 1-c_0)=1$, $\P(c_0\leq \pi_b^*(\bX)\leq 1-c_0)=1$, and $\P(c_0^2\leq q^*(\bS)\leq c_0^{-2})=1$.

(b) The random vector $\bS$ is sub-Gaussian and satisfies $\|\bv^\top\bS\|_{\psi_2}\leq\sigma_{\bS}\|\bv\|_2$ for all $\bv\in\R^{d}$, where $\|U\|_{\psi_2}:=\inf\{t>0:\E[\exp(U^2/t^2)]\leq 2\}$. Let $\varepsilon:=Y(1,\bM)-\bS^\top\bemur$, $\zeta:=\bzS^\top\bemur-\bX^\top\betaur$, and $\varrho:=\bzS^\top\bemur-\bX^\top\bettaur$. Assume that $\varepsilon$, $\zeta$, and $\varrho$ are sub-Gaussian with parameter $\sigma_0$.
In addition, the smallest eigenvalues of $\Var(A\bS)$ and $\Var\{(1-A)\bX\}$ are bounded below by $c_{\min}$. Let $\E[A\{Y(1,\bM)-\mu(\bS)\}^2]>c_Y$.
\end{assumption}

\begin{assumption}[Sparsity]
Let $r_{\pi}+r_{q}+r_{\mu}+r_{\tau_n}+r_{\tau}=o(1)$ and $(r_{\pi}^2+r_{q}^2)\log d=O(1)$.\label{cond:sparsity}
\end{assumption}

Sparsity conditions of the form $r=o(1)$, or equivalently $s=o(N/\log d)$, are standard in the high-dimensional statistics literature and are used to ensure consistent estimation. The additional requirement $(r_{\pi}^2+r_{q}^2)\log d=O(1)$ can alternatively be replaced by the boundedness condition $\|\bS\|_\infty< C$, as seen in \citet{UnifyingApproachDoublyrobust2019} and \citet{HighdimensionalInferenceDynamic2024}.

\begin{theorem}[Improved convergence]\label{thm:r_bias}
Let Assumptions~\ref{cond:basic}, \ref{cond:basic2}, and \ref{cond:sparsity} hold. Let $\bs{e}_1=(1,0,\ldots,0)^\top\in\R^{d_1}$ and $\rsigma^2:=\E\{\bpsi(\bW,\bnu^*)-\theta\}^2$. Then, as $N,d_1,d_2\to\infty$, $\rsigma\asymp 1+\|\betaur-\theta\bs{e}_1\|_2$ and $\hat\theta_{\mathrm{MQR}}-\theta=O_p(\rsigma N^{-1/2}+\mathcal R_{\mathrm{MQR}})$, where
\begin{align}
\mathcal R_{\mathrm{MQR}}:=&\underbrace{\emqr_\pi(\emqr_{\tau_n}+\emqr_\tau)+(\emqr_q+\emqr_\pi\emqr_{\bs{m}})\emqr_\mu}_{\text{approximation errors}}
+\underbrace{r_\pi(r_\mu+r_{\tau_n}+r_\tau)+r_q r_\mu}_{\text{stochastic errors}}\nonumber\\
&+\underbrace{\emqr_\mu(r_\pi+r_q)^2+(\emqr_{\tau_n}+\emqr_\tau)r_\pi^2}_{\text{joint effect of both errors}}\quad\text{and}\quad\emqr_{\tau_n}=O(\emqr_{\bs{m}}).\label{def:r_mqr}
\end{align}
\end{theorem}

\begin{table}[!t]
\centering
\spacingset{1}
\captionsetup{font=footnotesize,skip=6pt}
\caption{Convergence-rate bounds for QR and MQR. For simplicity, we assume $\sigma\asymp\rsigma\asymp1$ and ignore potential decay in the approximation errors, treating $\eqr_h\asymp 1$ and $\emqr_h\asymp 1$ whenever the corresponding function $h^*$ is misspecified. Here, $\bs m^*$ denotes the working model for $\bs m(\bX)$, $\mathcal R_0:=N^{-1/2}+r_\pi(r_\mu+r_{\tau_n}+r_\tau)+r_qr_\mu$, and $\bar{\mathcal R}_0$ is defined analogously by replacing each $r_h$ with $\br_h$.}
\label{table:MQR_convergence}
\fontsize{10}{12}\selectfont
\setlength{\tabcolsep}{3pt}
\setlength{\arrayrulewidth}{0.4pt}
\renewcommand{\arraystretch}{1.45}
\newlength{\MQRcomparisoncolwidth}
\settowidth{\MQRcomparisoncolwidth}{$\bar{\mathcal R}_0+\br_\mu+\br_{\tau_n}+\br_\tau+\br_q$}
\makebox[\textwidth][c]{%
\begin{tabular}{|l|*{2}{>{\centering\arraybackslash}p{\MQRcomparisoncolwidth}|}}
\hline
\multicolumn{1}{|c|}{\multirow{2}{*}{Model class}} & \multicolumn{2}{c|}{Convergence rates} \\
\cline{2-3}
& \multicolumn{1}{c|}{QR} & \multicolumn{1}{c|}{MQR} \\
\hline
$\mathcal M_a:(\pi^*,\mu^*)$\ \cmark
& \mbox{$\bar{\mathcal R}_0+\br_\pi+\br_\mu$}
& \mbox{$\mathcal R_0+r_\pi^2$} \\
\hline
$\mathcal M_b':(\mu^*,\tau_n^*)$\ \cmark
& \mbox{$\bar{\mathcal R}_0+\br_\mu+\br_{\tau_n}+\br_\tau$}
& \mbox{$\mathcal R_0$} \\
\hline
$\mathcal M_c'':(\pi^*,q^*)$\ \cmark
& \mbox{$\bar{\mathcal R}_0+\br_\pi+\br_q$}
& \mbox{$\mathcal R_0+(r_\pi+r_q)^2$} \\
\hline
$\mathcal M_d':(q^*,\tau^*,\bs m^*)$\ \cmark
& \mbox{$\bar{\mathcal R}_0+\br_\mu+\br_{\tau_n}+\br_\tau+\br_q$}
& \mbox{$\mathcal R_0+(r_\pi+r_q)^2$} \\
\hline
\end{tabular}%
}
\end{table}

The TR, O-TR, and QR methods contain bias terms involving products of approximation and stochastic errors, such as $\eqr_\pi \br_\mu$. In contrast, MQR reduces these joint effects to cubic terms in which approximation errors are multiplied by squared stochastic errors. Table~\ref{table:MQR_convergence} compares the convergence rates of QR and MQR across the scenarios in $\mathcal{M}_{\mathrm{MQR}}$. Unlike $\bar{\mathcal R}_{\mathrm{QR}}$ in \eqref{decom:qr}, which is expressed in terms of the actual errors $(\bar e_h,\bar r_h)$ and may therefore contain additional terms beyond the raw errors $(e_h,r_h)$ due to the sequential nuisance construction, $\mathcal R_{\mathrm{MQR}}$ is expressed directly in terms of the raw errors $(e_h,r_h)$ and is thus more transparent. As a result, $\mathcal{R}_{\mathrm{MQR}}$ contains an additional cubic term $\emqr_\pi\emqr_\mu\emqr_{\bs{m}}$ involving $\emqr_{\bs{m}}$. Its contribution is nevertheless of higher order because it is multiplied by the product $\emqr_\pi\emqr_\mu$.

\begin{theorem}[Robust inference]\label{thm:r_infer}
Let Assumptions~\ref{cond:basic}, \ref{cond:basic2}, and \ref{cond:sparsity} hold. Define $\hat\sigma_{\mathrm m}^2:=N^{-1}\sum_{k=1}^{\K}\sum_{i\in\I_k}\{\bpsi(\bW_i,\hbnu^{-k})-\hat\theta_{\mathrm{MQR}}\}^2$. If $\mathcal R_{\mathrm{MQR}}=o(\rsigma N^{-1/2})$ as $N,d_1, d_2\to\infty$, then $\rsigma^{-1}\sqrt N(\hat\theta_{\mathrm{MQR}}-\theta)\leadsto\mathcal N(0,1)$ and $\hat\sigma_{\mathrm m}^2=\rsigma^2\{1+o_p(1)\}$.
\end{theorem}

Despite the reduction in the joint-effect terms, \eqref{def:r_mqr} still contains products of approximation errors and products of stochastic errors. Hence, the asymptotic normality results continue to require certain approximation or stochastic errors to shrink at rates of order $o(N^{-1/4})$.

\begin{table}[!b]
\centering
\spacingset{1}
\caption{Sparsity conditions required for the MQR estimator to be asymptotically normal. For simplicity, we assume $\rsigma\asymp1$, $\|\bS\|_\infty<C$, and $d_1\asymp d$, and ignore the possible decay of approximation errors; $\bs m^*$ denotes the working model for $\bs m(\bX)$.
Replacing $o(\cdot)$ by $O(\cdot)$ gives the sufficient conditions for root-$N$ consistency.}
\label{table:MQR_sparsity}
\footnotesize
\setlength{\tabcolsep}{4pt}
\renewcommand{\arraystretch}{1.2}
\begin{tabularx}{\textwidth}{@{}l>{\centering\arraybackslash}X@{}}
\toprule
Model class & Sparsity conditions \\
\midrule
$\mathcal M_a:(\pi^*,\mu^*)$\ \cmark
& $s_\pi=o(\sqrt N/\log d),\quad s_\pi(s_\mu+s_{\tau_n}+s_\tau)+s_qs_\mu=o(N/\log^2 d)$ \\[2mm]
$\mathcal M_b':(\mu^*,\tau_n^*)$\ \cmark
& $s_\pi(s_\mu+s_{\tau_n}+s_\tau)+s_qs_\mu=o(N/\log^2 d)$ \\[2mm]
$\mathcal M_c'':(\pi^*,q^*)$\ \cmark
& $s_\pi+s_q=o(\sqrt N/\log d),\quad s_\pi(s_\mu+s_{\tau_n}+s_\tau)+s_qs_\mu=o(N/\log^2 d)$ \\[2mm]
$\mathcal M_d':(q^*,\tau^*,\bs m^*)$\ \cmark
& $s_\pi+s_q=o(\sqrt N/\log d),\quad s_\pi(s_\mu+s_{\tau_n}+s_\tau)+s_qs_\mu=o(N/\log^2 d)$ \\
\bottomrule
\end{tabularx}
\end{table}

In contrast to Theorem~\ref{thm:QR_normal}, which requires vanishing approximation errors for all nuisance models, Theorem~\ref{thm:r_infer} allows for severe model misspecification with nonvanishing approximation errors. In such cases, we may additionally require individual $o(N^{-1/4})$ convergence rates for the \emph{stochastic errors} $r_\pi$ or $r_q$ as the price for accommodating such misspecification; see the corresponding sparsity conditions in Table~\ref{table:MQR_sparsity}. By contrast, TR, O-TR, and QR do not provide valid inference guarantees in high-dimensional settings unless all approximation errors vanish. For instance, when $\emqr_\mu\asymp1$ while all other approximation errors are exactly zero, MQR requires $r_\pi+r_q=o(N^{-1/4})$. In this case, the usual product-rate condition on the \emph{total estimation errors} is generally violated, since $(\emqr_\mu+r_\mu)(\emqr_q+r_q)=o(N^{-1/4})$, which does not generally satisfy the stronger $o(N^{-1/2})$ requirement imposed in Theorem~\ref{thm:QR_normal} or in standard double machine learning procedures \citep{DoubleDebiasedMachine2018,CausalMediationAnalysis2022,DeepMedSemiparametricCausal2022}.
Moreover, if $\eqr_\mu=o(1)$ is not exactly zero, then the QR method in Theorem~\ref{thm:QR_normal} requires $\br_q=o(N^{-1/2}/\eqr_\mu\land1)$. In comparison, when $\emqr_\mu=o(1)$ is nonzero, Theorem~\ref{thm:r_infer} requires only $\br_q=o(\sqrt{N^{-1/2}/\emqr_\mu}\land1)$ for the MQR method, where we may take $\br_q=r_\pi+r_q$; see Theorem~\ref{thm:r_nuisance} of the Supplementary Materials. For example, if $\eqr_\mu\asymp\emqr_\mu\asymp N^{-1/6}$, then QR requires $\br_q=o(N^{-1/3})$, whereas MQR requires only $\br_q=o(N^{-1/6})$.

\begin{remark}[Practical guidelines]\label{rem:implementation_variants}
Despite these theoretical guarantees, MQR does not uniformly dominate the QR approach in Section~\ref{sec:QR}. Specifically, MQR relies on parametric working models for the nuisance functions, whereas QR allows more flexible nonparametric procedures. The choice between the two therefore reflects an approximation--estimation trade-off. In low-dimensional settings where nonparametric estimation is feasible, QR may be preferable when the sample size is sufficiently large. In contrast, in high-dimensional or small-sample settings where regularized parametric models are more practical and full correct specification is unrealistic, MQR may be preferable because it can preserve valid inference under partial misspecification. MQR can also be applied after basis transformations, which may reduce approximation errors at the cost of potentially increasing stochastic errors.
\end{remark}

Analogously to Remark~\ref{remark:imple_qr}, we introduce a variant, termed MQR2, that uses the entire training fold for all nuisance estimation. Its detailed construction and asymptotic properties are provided in Section~\ref{apdx:mqr2_results} of the Supplementary Materials. When a distinction from MQR2 is necessary, we refer to the implementation in Algorithm~\ref{alg:MQR} as MQR1.

\section{Numerical Experiments}\label{sec:num}
\subsection{Simulation results}\label{sec:simu}

In this section, we evaluate the finite-sample performance of eight methods through simulation studies. These include the four proposed methods QR1, QR2, MQR1, and MQR2.
The remaining four methods are included for comparison:
\textbf{Oracle}, an oracle estimator computed from \eqref{def:EIF} using the true nuisance functions;
\textbf{TR}, a cross-fitted version of the triply robust estimator \citep{SemiparametricTheoryCausal2012};
\textbf{O-TR}, the odds-modeling triply robust estimator \citep{CausalMediationAnalysis2022};
and \textbf{O-TR'}, a parametric variant of O-TR in which all nuisance functions are estimated using regularized (generalized) linear models.

We divide all considered methods into two groups: (a) nonparametric modeling, including QR1, QR2, O-TR, and TR, where random forests are used for nuisance estimation; and (b) parametric modeling, including MQR1, MQR2, and O-TR', where parametric working models are adopted. In particular, for MQR1 and MQR2, the calibration losses $\ell_1$, $\ell_2$, and $\ell_3$ are minimized via the \textsf{RCAL} package \citep{RCALSoftware2020}, and the weighted regression losses $\ell_4$, $\ell_5$, and $\ell_6$ are solved as weighted Lasso regressions via the \textsf{glmnet} package \citep{glmnetSoftware2010}. All methods are implemented using 5-fold cross-fitting ($\K=5$), with all tuning parameters selected via additional cross-validation.
Further implementation details are provided in Section~\ref{apdx:implementation} of the Supplementary Materials.

We now describe the DGPs. Let $\mathcal N_{\mathrm{trun},c}$ denote the distribution of $V\mid |V|\le c$, where $V\sim\mathcal N(0,1)$ and $c>0$. Define $\tilde g(x)=(0.1+|x|+\sin x)/(0.7+|x|+\sin x)$. Generate $\varepsilon_i\sim\mathcal N(0,1)$ independently of the other variables. Let $\bs{\Sigma}$ have unit diagonal, $\Sigma_{12}=\Sigma_{21}=0.3$, and all other off-diagonal entries zero. For integers $1\leq k\leq d$, define $\bar{\bs1}_{k,d}:=(\bs1_k^\top,\bz_{d-k}^\top)^\top$. Define $\mathbf B\in\R^{d_1\times d_2}$ by $B_{jk}=\idf_{\{j=k\}}+0.5\idf_{\{|j-k|=1\}}$. We choose $(d_1,d_2)=(100,50)$ and $N\in\{600,1000,1400,1800\}$. Each experiment is repeated 100 times.

\textbf{Setting (i)}: Non-logistic $\pi$ and non-linear $\mu$. Generate $X_{ij}\overset{\mathrm{i.i.d.}}{\sim}\mathcal N_{\mathrm{trun},1.2}$, $A_i\mid\bX_i\sim\mathrm{Bernoulli}\{\tilde g(1+\bar{\bs1}_{3,d_1}^\top\bX_i)\}$, and $\bM_i=0.6\mathbf B^\top\bX_i+\bU_i$, where $\bU_i\mid A_i\sim\mathcal N_{d_2}\{0.6A_i\bar{\bs1}_{1,d_2}+(0.4+0.2A_i)\bar{\bs1}_{3,d_2},\bs{\Sigma}\}$. Let $Y_i=1+\sum_{j=1}^3w_{ij}+\bar{\bs1}_{3,d_2}^\top\bM_i-(1-A_i)+\varepsilon_i$, where $w_{ij}=X_{ij}(U_{i,r_j}+\idf_{\{U_{i1}\geq0.4\}})+0.5X_{ij}^2(U_{ij}-0.4)$ for $j=1,2,3$, with $(r_1,r_2,r_3)=(3,1,2)$.

\textbf{Setting (ii)}: Non-logistic $\pi$ and non-log-linear $q$. Generate $\bX_i\sim\mathcal N_{d_1}(\bz,\mathbf I_{d_1})$, $A_i\mid\bX_i\sim\mathrm{Bernoulli}\{\tilde g(1+\bar{\bs1}_{3,d_1}^\top\bX_i)\}$, and $\bM_i=0.6\mathbf B^\top\bX_i+\bU_i$, where $\bU_i\mid(A_i,\bX_i)\sim\mathcal N_{d_2}\{0.4\bar{\bs1}_{3,d_2}+A_i(2\idf_{\{X_{i4}>0\}}-1)(0.6g(X_{i1}),0.2g(X_{i2}),0.2g(X_{i3}),\allowbreak\bz_{d_2-3}^\top)^\top,\bs{\Sigma}\}$, and $Y_i=1-0.6\bar{\bs1}_{3,d_1}^\top\bX_i+\bar{\bs1}_{3,d_2}^\top\bM_i-(1-A_i)+\varepsilon_i$.

\begin{table}[!t]
\footnotesize
\centering
\spacingset{1}
\caption{Simulation results under Setting (i). Bias: median estimation error; RMSE: root mean square error; AL: median length of the $95\%$ confidence intervals; AC: empirical coverage of the $95\%$ confidence intervals. \textit{Italics} denote the proposed methods. Shaded rows denote parametric modeling methods.}
\begin{tabular}{lccccclcccc}
\Xhline{1pt}
\multicolumn{5}{c}{$N=600$, $(d_1,d_2)=(100,50)$}&\quad&\multicolumn{5}{c}{$N=1000$, $(d_1,d_2)=(100,50)$}\\
\cmidrule(lr){1-5}\cmidrule(lr){7-11}
\textbf{Method}&\textbf{Bias}&\textbf{RMSE}&\textbf{AL}&\textbf{AC}&\quad&\textbf{Method}&\textbf{Bias}&\textbf{RMSE}&\textbf{AL}&\textbf{AC}\\
Oracle&0.001&0.251&1.070&0.97&\quad&Oracle&0.005&0.218&0.818&0.97\\
\hdashline
TR&-0.724&0.869&1.468&0.50&\quad&TR&-0.579&0.647&1.037&0.39\\
O-TR&0.361&0.456&0.795&0.51&\quad&O-TR&0.307&0.368&0.619&0.51\\
\textit{QR1}&-0.285&0.458&1.267&0.83&\quad&\textit{QR1}&-0.229&0.296&0.849&0.87\\
\textit{QR2}&-0.117&0.336&1.047&0.92&\quad&\textit{QR2}&-0.074&0.228&0.733&0.90\\
\hdashline
\rowcolor{black!5}
O-TR'&-0.197&0.405&1.053&0.77&\quad&O-TR'&-0.189&0.284&0.845&0.89\\
\rowcolor{black!5}
\textit{MQR1}&-0.549&0.707&1.006&0.46&\quad&\textit{MQR1}&-0.382&0.467&0.780&0.52\\
\rowcolor{black!5}
\textit{MQR2}&-0.065&0.269&0.957&0.93&\quad&\textit{MQR2}&-0.038&0.206&0.733&0.92\\
\Xhline{1pt}
\multicolumn{5}{c}{$N=1400$, $(d_1,d_2)=(100,50)$}&\quad&\multicolumn{5}{c}{$N=1800$, $(d_1,d_2)=(100,50)$}\\
\cmidrule(lr){1-5}\cmidrule(lr){7-11}
\textbf{Method}&\textbf{Bias}&\textbf{RMSE}&\textbf{AL}&\textbf{AC}&\quad&\textbf{Method}&\textbf{Bias}&\textbf{RMSE}&\textbf{AL}&\textbf{AC}\\
Oracle&0.000&0.175&0.698&0.96&\quad&Oracle&-0.010&0.151&0.612&0.93\\
\hdashline
TR&-0.478&0.526&0.839&0.40&\quad&TR&-0.431&0.463&0.701&0.24\\
O-TR&0.282&0.326&0.527&0.47&\quad&O-TR&0.210&0.248&0.469&0.61\\
\textit{QR1}&-0.110&0.227&0.661&0.86&\quad&\textit{QR1}&-0.055&0.175&0.556&0.88\\
\textit{QR2}&-0.032&0.180&0.592&0.88&\quad&\textit{QR2}&-0.012&0.149&0.508&0.87\\
\hdashline
\rowcolor{black!5}
O-TR'&-0.143&0.238&0.716&0.87&\quad&O-TR'&-0.163&0.223&0.633&0.81\\
\rowcolor{black!5}
\textit{MQR1}&-0.218&0.327&0.666&0.66&\quad&\textit{MQR1}&-0.180&0.248&0.565&0.74\\
\rowcolor{black!5}
\textit{MQR2}&-0.026&0.170&0.632&0.92&\quad&\textit{MQR2}&-0.025&0.140&0.545&0.96\\
\Xhline{1pt}
\end{tabular}
\label{table:config1}
\end{table}

\begin{table}[!t]
\footnotesize
\centering
\spacingset{1}
\caption{Simulation results under Setting (ii). The rest of the caption details remain the same as those in Table~\ref{table:config1}.}
\begin{tabular}{lccccclcccc}
\Xhline{1pt}
\multicolumn{5}{c}{$N=600$, $(d_1,d_2)=(100,50)$}&\quad&\multicolumn{5}{c}{$N=1000$, $(d_1,d_2)=(100,50)$}\\
\cmidrule(lr){1-5}\cmidrule(lr){7-11}
\textbf{Method}&\textbf{Bias}&\textbf{RMSE}&\textbf{AL}&\textbf{AC}&\quad&\textbf{Method}&\textbf{Bias}&\textbf{RMSE}&\textbf{AL}&\textbf{AC}\\
Oracle&0.025&0.156&0.666&0.99&\quad&Oracle&0.033&0.122&0.523&0.96\\
\hdashline
TR&-0.826&0.879&1.297&0.27&\quad&TR&-0.623&0.692&0.872&0.19\\
O-TR&-0.338&0.356&0.516&0.29&\quad&O-TR&-0.328&0.351&0.401&0.14\\
\textit{QR1}&-0.315&0.372&0.763&0.63&\quad&\textit{QR1}&-0.294&0.325&0.492&0.44\\
\textit{QR2}&-0.280&0.318&0.628&0.58&\quad&\textit{QR2}&-0.201&0.258&0.445&0.55\\
\hdashline
\rowcolor{black!5}
O-TR'&-0.232&0.267&0.518&0.58&\quad&O-TR'&-0.180&0.217&0.405&0.57\\
\rowcolor{black!5}
\textit{MQR1}&-0.321&0.383&0.544&0.34&\quad&\textit{MQR1}&-0.255&0.300&0.436&0.39\\
\rowcolor{black!5}
\textit{MQR2}&-0.099&0.171&0.569&0.91&\quad&\textit{MQR2}&-0.049&0.128&0.458&0.95\\
\Xhline{1pt}
\multicolumn{5}{c}{$N=1400$, $(d_1,d_2)=(100,50)$}&\quad&\multicolumn{5}{c}{$N=1800$, $(d_1,d_2)=(100,50)$}\\
\cmidrule(lr){1-5}\cmidrule(lr){7-11}
\textbf{Method}&\textbf{Bias}&\textbf{RMSE}&\textbf{AL}&\textbf{AC}&\quad&\textbf{Method}&\textbf{Bias}&\textbf{RMSE}&\textbf{AL}&\textbf{AC}\\
Oracle&0.016&0.116&0.441&0.96&\quad&Oracle&0.009&0.102&0.392&0.96\\
\hdashline
TR&-0.591&0.621&0.681&0.05&\quad&TR&-0.545&0.560&0.552&0.02\\
O-TR&-0.312&0.324&0.342&0.09&\quad&O-TR&-0.299&0.314&0.302&0.05\\
\textit{QR1}&-0.220&0.255&0.392&0.45&\quad&\textit{QR1}&-0.210&0.236&0.333&0.36\\
\textit{QR2}&-0.170&0.210&0.363&0.52&\quad&\textit{QR2}&-0.164&0.193&0.310&0.45\\
\hdashline
\rowcolor{black!5}
O-TR'&-0.138&0.177&0.343&0.63&\quad&O-TR'&-0.125&0.154&0.303&0.61\\
\rowcolor{black!5}
\textit{MQR1}&-0.189&0.229&0.379&0.51&\quad&\textit{MQR1}&-0.173&0.200&0.342&0.53\\
\rowcolor{black!5}
\textit{MQR2}&-0.020&0.113&0.388&0.92&\quad&\textit{MQR2}&-0.032&0.095&0.344&0.94\\
\Xhline{1pt}
\end{tabular}
\label{table:config2}
\end{table}

As reported in Tables~\ref{table:config1} and~\ref{table:config2}, the QR2 and MQR2 versions of the proposed methods outperform the QR1 and MQR1 versions in RMSE, reflecting more efficient use of samples for nuisance training. For instance, MQR1 uses only $N(\K-1)/(5\K)=288$ samples for nuisance training when $N=1800$ and $\K=5$, whereas MQR2 utilizes $N(\K-1)/\K=1440$ training samples.

Among the nonparametric modeling methods, the QR2 method achieves the best performance in estimation accuracy across both settings. Nevertheless, under Setting~(ii), all nonparametric methods exhibit relatively large bias and under-coverage of confidence intervals, reflecting the slow convergence rates of nonparametric nuisance estimates in high dimensions. In contrast, the parametric modeling methods O-TR' and MQR2 generally perform better than the nonparametric approaches under this setting. Under Setting (i), however, QR2 has smaller bias and root mean squared error than O-TR' at all sample sizes, as the O-TR' estimator is inconsistent when $\pi^*$ and $\mu^*$ are misspecified.
Overall, MQR2 achieves the best performance, with RMSE comparable to the oracle estimate and empirical coverage close to the nominal $95\%$ level.

Additional simulation DGPs for Settings~(iii)--(vii), together with the corresponding results, are presented in Section~\ref{apdx:detail_simu} of the Supplementary Materials.

\subsection{Real data application}\label{sec:real}

We apply the proposed methods to data from the AIDS Clinical Trials Group Protocol 175 (ACTG175) \citep{TrialComparingNucleoside1996}, a randomized trial evaluating antiretroviral regimens in HIV-1 infected adults with baseline CD4 T-cell counts between 200 and 500 cells/mm$^3$. Clinical evidence from the DELTA trial, a randomized double-blind controlled study, demonstrated that combination therapy leads to a significantly larger increase in CD4 counts compared with zidovudine (ZDV) monotherapy \citep{DeltaRandomisedDoubleblind1996}, with this benefit emerging early and persisting over the long term.
The goal of our analysis is to quantify the extent to which this immunological benefit is mediated through early immune recovery.

The outcome $Y$ is defined as the log-transformed CD4 count at week 96, and the mediator vector $\bM$ consists of the log-transformed CD4 and CD8 counts at week 20, representing early immune recovery. We consider two treatment groups: $A=1$ for combination therapy and $A=0$ for ZDV monotherapy. The baseline covariate vector $\bX$ includes 12 variables. To improve the flexibility of parametric methods and reduce model misspecification, we augment $\bX$ and $\bM$ with higher-order terms and interactions, resulting in $d_1 = 188$ covariates and $d_2 = 101$ mediators; see Section~\ref{apdx:detail_real} of the Supplementary Materials for details. Nonparametric methods, in contrast, use the original covariates and mediators. After excluding observations with missing values, the final sample includes $N=780$ patients.

The average treatment effect, $\mathrm{ATE} = \E[Y(1,\bM(1))] - \E[Y(0,\bM(0))]$, for the treatment regimes in this study has been extensively examined in previous work \citep{MultiplyRobustEstimation2014,DoublyRobustInference2017, MultipleRobustEstimation2019}.
Our focus is on the mediation effects: the natural indirect effect $\mathrm{NIE} = \E[Y(1,\bM(1))] - \E[Y(1,\bM(0))]$ and the natural direct effect $\mathrm{NDE} = \E[Y(1,\bM(0))] - \E[Y(0,\bM(0))]$, as well as the mediation proportion $\mathrm{NIE}/\mathrm{ATE}$, which quantifies the fraction of the total effect mediated through early immune recovery.

To ensure a fair comparison, we use the same doubly robust estimator \citep{DoubleRobustEstimatorAverage2006,DoubleDebiasedMachine2018} to estimate $(\theta_{0,0}, \theta_{1,1})$, incorporating random forest nuisance estimates with cross-fitting. For $\theta=\theta_{1,0}$, we implement the proposed QR2 and MQR2 methods, along with the other comparison methods described in Section~\ref{sec:simu}. To account for the randomness introduced by the cross-fitting procedure, we repeat cross-fitting 100 times and use the median of the resulting estimates as the final point estimate \citep{DoubleDebiasedMachine2018}.

\begin{table}[!t]
    \centering
    \spacingset{1}
    \caption{Mediation effect estimation based on the ACTG175 dataset. Est.: point estimate; Len.: length of the $95\%$ confidence intervals; $p$-value: $p$-value for testing the null hypothesis that the effect equals zero. \textit{Italics} denote the proposed methods. Shaded rows denote parametric modeling methods.}\label{table:ACTG}
    \footnotesize
    \begin{tabular}{l cc ccc ccc c}
    \toprule
    \textbf{Method} & \multicolumn{2}{c}{$\theta_{1,0}$} & \multicolumn{3}{c}{\textbf{NDE}} & \multicolumn{3}{c}{\textbf{NIE}} & \textbf{NIE}/\textbf{ATE} \\
    \cmidrule(lr){2-3} \cmidrule(lr){4-6} \cmidrule(lr){7-9}
    & Est. & Len. & Est. & Len. & $p$-value & Est. & Len. & $p$-value & \\
    \midrule

    TR & 5.657 & 0.207 & 0.204 & 0.258 & 0.002 & 0.135 & 0.162 & 0.0011 & 38.99\% \\
    O-TR & 5.689 & 0.142 & 0.241 & 0.211 & 8.0e-06 & 0.104 & 0.097 & 2.5e-05 & 30.26\% \\
    \textit{QR2} & 5.668 & 0.134 & 0.222 & 0.205 & 2.1e-05 & 0.122 & 0.084 & 1.3e-08 & 35.44\% \\
    \hdashline
    \rowcolor{black!5}
    O-TR' & 5.696 & 0.112 & 0.249 & 0.203 & 1.6e-06 & 0.098 & 0.062 & 6.3e-10 & 28.37\% \\
    \rowcolor{black!5}
    \textit{MQR2} & 5.675 & 0.126 & 0.224 & 0.205 & 1.8e-05 & 0.118 & 0.081 & 1.5e-08 & 34.11\% \\

    \midrule
    \textbf{ATE} & \multicolumn{2}{c}{Est.: 0.345} & \multicolumn{2}{c}{Len.: 0.213} & \multicolumn{2}{c}{$p$-value: 1.9e-10} & & & \\
    \bottomrule
    \end{tabular}
\end{table}

As shown in Table~\ref{table:ACTG}, all estimated NDEs and NIEs are positive and statistically significant at the $0.001$ level, except for the NDE and NIE estimated by TR, which has the widest confidence interval due to the unstable conditional density estimation. The remaining four methods, O-TR, QR2, O-TR', and MQR2, lead to comparable confidence interval lengths. The proposed QR2 and MQR2 methods provide smaller point estimates for $\theta_{1,0}$ than O-TR and O-TR', thereby leading to smaller NDE estimates and larger NIE estimates. The mediation proportions from QR2 and MQR2 range from 34.11\% to 35.44\%, suggesting that approximately one-third of the treatment effect is mediated through early immune recovery.

\section{Discussion}
\label{sec:discuss}

In this work, we develop a new QR framework for causal mediation analysis. Within this framework, we propose two estimation approaches. The first is a general QR estimator based on a nonparametric modeling strategy that is compatible with flexible machine learning methods for nuisance estimation. The second is an MQR parametric modeling approach, designed to provide robust inference even when some regularized (generalized) linear working models are misspecified. Overall, the proposed framework contributes to causal mediation analysis along two fundamental fronts: it expands the admissible model classes from triple to quadruple robustness, while advancing high-dimensional mediation analysis from consistent estimation to valid inference under model misspecification.

Our work provides a principled approach for robust causal inference involving cross-world counterfactuals. Beyond natural effects, it is also of interest to extend our methods to a broader class of cross-world functionals, including conditional and marginal interventional effects \citep{EffectDecompositionPresence2014,LongitudinalMediationAnalysis2017}.

\section*{Supplementary Materials}
The Supplementary Materials contain additional discussions, numerical studies, and technical proofs. Section~\ref{apdx:method} provides further methodological justification, and Section~\ref{apdx:variants} presents the algorithms and theoretical guarantees for QR2 and MQR2. Section~\ref{apdx:num} presents supplementary numerical results, and Sections~\ref{apdx:QR} and~\ref{apdx:MQR} provide proofs of the main results.

\section*{Data and Code Availability}
The ACTG175 dataset is available at \url{https://r-packages.io/datasets/ACTG175}, and the code is available at \url{https://github.com/zhenqi-stat/QR_causal_mediation}.
\bibliography{mediation}

\newpage
\begin{appendices}
\allowdisplaybreaks[3]
\def\norme#1{{\mathchoice{\left\|#1\right\|}{\|#1\|}{\|#1\|}{\|#1\|}}_{\psi_{1}}}
\def\normg#1{{\mathchoice{\left\|#1\right\|}{\|#1\|}{\|#1\|}{\|#1\|}}_{\psi_2}}
\def\normp#1#2{{\mathchoice{\left\|#1\right\|}{\|#1\|}{\|#1\|}{\|#1\|}}_{\P,#2}}
\begin{center}
    {\large\bf SUPPLEMENTARY MATERIAL}
\end{center}

This supplementary document provides additional methodological discussions, theoretical results, numerical studies, and technical proofs supporting the main findings. Section~\ref{apdx:method} offers further methodological discussions and theoretical justifications for the proposed QR and MQR frameworks, and Section~\ref{apdx:variants} details the algorithms and theoretical guarantees for the variants QR2 and MQR2. Section~\ref{apdx:num} presents implementation details and supplementary numerical results. Finally, Sections~\ref{apdx:QR} and~\ref{apdx:MQR} collect all technical proofs for the QR and MQR methods, respectively.

\paragraph*{Notation.}
We use the following notation throughout. For any positive integer $k$, let $[k]:=\{1,2,\ldots,k\}$. Constants $c,C>0$, independent of $N$ and $d$, may change from line to line. For $r\geq1$, define the $\ell_r$-norm of a vector $\bs{z}=(z_1,\ldots,z_p)^\top\in\R^p$ as $\|\bs{z}\|_r:=(\sum_{j=1}^p\lvert z_j\rvert^r)^{1/r}$, and let $\|\bs{z}\|_0 := \lvert\{j:z_j\neq0\}\rvert$ and $\|\bs{z}\|_\infty:=\max_j\lvert z_j\rvert$. Additionally, let $\bs{z}_S$ denote the subvector of $\bs{z}$ indexed by the set $S$. Denote $\bs{e}_j$ as the canonical unit vector with $1$ at the $j$-th position and $0$ elsewhere. The $d$-dimensional vectors of all ones and all zeros are denoted by $\bs{1}_{d}$ and $\bz_{d}$, respectively. For any symmetric matrices $\mathbf{A}$ and $\mathbf{B}$, $\mathbf{A}\succ\mathbf{B}$ means that $\mathbf{A}-\mathbf{B}$ is positive definite, and $\mathbf{A}\succeq\mathbf{B}$ means that $\mathbf{A}-\mathbf{B}$ is positive semidefinite; we let $\lambda_{\min}(\mathbf{A})$ denote the smallest eigenvalue of $\mathbf{A}$. We denote the logistic function by $g(u)= {\exp(u)}/\{1+\exp(u)\}$ for all $u\in\R$.

Let $\P$ and $\E$ denote the probability measure and expectation induced by the joint distribution of the underlying random vector $(\{Y(a,\bs{m})\}_{a\in\{0,1\}, \bs{m}\in \mathcal{S}_M},\{\bM(a)\}_{a\in\{0,1\}}, A, \bX)$ (independent of the observed samples), respectively. To simplify the notation, we write $\bM_0:=\bM(0)$ and $\bzS:=(\bX^\top,\bM(0)^\top)^\top$. For any $\alpha>0$, let $\psi_{\alpha}(x):=\exp(x^\alpha)-1$ for all $x\geq 0$. Specifically, $\psi_2(x)=\exp(x^2)-1$ and $\psi_{1}(x)=\exp(x)-1$. The $\psi_{\alpha}$-Orlicz norm $\|\cdot\|_{\psi_{\alpha}}$ of a random variable $X\in\R$ is defined as $\|X\|_{\psi_{\alpha}}:=\inf\{c>0:\E[\psi_{\alpha}(\lvert X\rvert/c)]\leq 1\}$. For any $r>0$ and a random variable $U$, define $\|U\|_{\P,r}:=\{\int |U|^r\,\mathrm d\P\}^{1/r}$, suppressing the argument of a function when clear. For any $\tilde{\mathbb{S}}\subseteq\mathbb{S}=(\bW_i)_{i=1}^N$, define $\P_{\tilde{\mathbb{S}}}$ as the joint distribution of $\tilde{\mathbb{S}}$ and $\E_{\tilde{\mathbb{S}}}(f)=\int f\,\mathrm{d}\P_{\tilde{\mathbb{S}}}$. Specifically, let $\mathbb{S}_k$ be the subset of $\mathbb{S}$ corresponding to the index set $\mathcal I_k$ for $k\in[\K]$.

In the notation above, dependence on the sample size $N$ is left implicit for ease of exposition. In our asymptotic analysis, however, the covariate and mediator dimensions $(d_1,d_2)$ as well as the underlying DGP $\P$ are allowed to vary with $N$, as detailed in the following remark.

\begin{remark}[Triangular-array formulation]\label{rem:triangular_array}
For each $N$, let $\bW_{1,N},\ldots,\bW_{N,N}\overset{\mathrm{i.i.d.}}{\sim}\P_N$, with covariate and mediator dimensions $d_{1,N}$ and $d_{2,N}$. For a nuisance function $h$, let $h_N$ denote its $N$-indexed version and let $h_N^*$ denote its working-model target under $\P_N$. When defined through loss minimization, these targets minimize the corresponding population risks under $\P_N$. Correct specification of a component means $h_N^*=h_N$ for each $N$. Both approximation and stochastic errors are defined accordingly under $\P_N$, with $\eqr_{h,N}=\|h_N^*-h_N\|_{\P_N,6}$ and $\|\hat h_N-h_N^*\|_{\P_N,2}=O_p(\br_{h,N})$, as in Section \ref{sec:QR_thm}.
\end{remark}
All subsequent asymptotic results and proofs for the QR1, QR2, MQR1, and MQR2 estimators are established under the triangular-array formulation in Remark~\ref{rem:triangular_array}. For notational simplicity, we suppress the subscript $N$ and write $\P,\bW_i,d_1,d_2,h,h^*,\hat h,\eqr_h,\br_h,\bnu^*,\hbnu,\emqr_h,\emqr_{\bs{m}},r_h$ for their $N$-indexed counterparts.

\section{Additional Discussions and Justifications}\label{apdx:method}

In this section, we provide additional discussions and theoretical justifications for the proposed QR and MQR methods. Section~\ref{apdx:existing_tr} reviews the existing triply robust approaches. Section~\ref{apdx:tau_alter} presents an alternative difference-based DR estimator for $\tau$. Illustrative examples for the newly introduced scenarios $\mathcal M_d$ and $\mathcal M_d'$ are given in Section~\ref{apdx:example}. Finally, Section~\ref{apdx:mqr_supp} characterizes the target nuisance parameters, working-model approximation errors, and nuisance convergence rates for the MQR method.

\subsection{Existing triply robust approaches}\label{apdx:existing_tr}
Depending on how the nuisance functions are constructed and estimated, existing robust methods can be broadly grouped into two categories: the triply robust (TR) and odds-modeling triply robust (O-TR) approaches. Recalling the model classes defined in Table~\ref{table:robustness_comparison}, these distinct estimation strategies also lead to different forms of triple robustness. We review the two approaches and their respective robustness properties below.

The TR approach \citep{SemiparametricTheoryCausal2012,CausalMediationAnalysis2022,DeepMedSemiparametricCausal2022} directly estimates the conditional densities $f_0$ and $f_1$. Given estimators $(\hat\mu,\hat f_0)$, it estimates the cross-world conditional mean through the integral representation in \eqref{eq:joint2}:
\begin{align*}
\hat\tau(\bX)=\int_{\mathcal S_M}\hat\mu(\bX,\bs{m})\hat f_0(\bX,\bs{m})\,\mathrm d\bs{m}.
\end{align*}
The nuisance estimates are then substituted into the influence function in \eqref{def:EIF}, and the resulting score is averaged over the observations. This estimator is consistent over $\mathcal M_\mathrm{TR}=\mathcal M_a\cup\mathcal M_b\cup\mathcal M_c$. Its implementation, however, requires estimation of the two conditional densities and numerical evaluation of the integral defining $\hat\tau$, which can be difficult and unstable when the mediator dimension is moderate or large.

The O-TR approach \citep{TargetedMaximumLikelihood2012,SemiparametricEstimationPathspecific2019,CausalMediationAnalysis2022} avoids direct estimation of $(f_0,f_1)$ by modeling $p(\bS)=\P(A=1\mid\bS)$ and using the Bayes identity in \eqref{eq:joint1}. Given $\hat\mu$, it estimates $\tau$ by regressing $\hat\mu(\bS)$ on $\bX$ among observations with $A=0$, according to the nested representation in \eqref{eq:joint1}. In both methods, the estimation of $\tau$ is based on an initial estimator $\hat\mu$ of the outcome regression $\mu$. For the probability limit $\mu^*$ of $\hat\mu$, define the nested mean $\tau_n(\bX):=\E[\mu^*(\bS)\mid A=0,\bX]$ and let $\tau_n^*$ denote its working model. The O-TR estimator is consistent over $\mathcal M_\mathrm{OTR}=\mathcal M_a\cup\mathcal M_b'\cup\mathcal M_c'$. Under $\mathcal M_b'$, $\mu^*=\mu$ and hence $\tau_n=\tau$; when $\mu^*\neq\mu$, however, $\tau_n$ generally differs from $\tau$ even if the working model for $\tau_n$ is correctly specified.

The representations in \eqref{eq:joint2} and \eqref{eq:joint1} imply $\mathcal M_b\subset\mathcal M_b'$ and $\mathcal M_c\subset\mathcal M_c'$, when the working models are chosen correspondingly, so that $\mathcal M_\mathrm{TR}\subset\mathcal M_\mathrm{OTR}$. This gain avoids conditional-density estimation but requires modeling the inverse-causality treatment probability $p(\bS)=\P(A=1\mid\bX,\bM)$. Since $\bM$ is downstream of $A$, it is more natural to impose smoothness or sparsity conditions on the forward DGP and to regard $p$ as an induced quantity whose complexity is determined by $(\pi,q)$ through \eqref{eq:joint1}.

\subsection{An alternative DR estimator for \texorpdfstring{$\tau$}{tau} in the QR framework}\label{apdx:tau_alter}
In addition to the DR estimator of $\tau$ in \eqref{def:tauhat}, an alternative strategy is to model the difference $\Delta(\bX):=\tau(\bX)-\tau_n(\bX)$. This difference can be estimated by solving
\begin{align*}
\hat\Delta(\cdot)\in\arg\min_{f \in \mathcal{G}_{\Delta}} N^{-1}\sum_{i=1}^N A_i \left[\hat q(\bS_i)\left\{Y_i-\hat\mu(\bS_i)\right\} - f(\bX_i)\right]^2,
\end{align*}
for a pre-specified function class $\mathcal{G}_{\Delta}$. An estimator of $\tau(\bX)$ is then obtained as $\hat\tau_n(\bX) + \hat\Delta(\bX)$. This difference-based approach can be advantageous when $\Delta(\bX)$ has a simpler structure and can be estimated more accurately; otherwise, directly modeling the target $\tau(\bX)$ as in \eqref{def:tauhat} may be more appropriate.

\subsection{Examples for the new scenarios \texorpdfstring{$\mathcal M_d$}{Md} and \texorpdfstring{$\mathcal M_d'$}{Md'}}\label{apdx:example}
We now present two examples to illustrate the newly introduced scenarios $\mathcal M_d$ and $\mathcal M_d'$. First, Example~\ref{ex:general_residual} introduces a model class in which $q$, $\tau_n$, and $\tau$ share low-dimensional structures and are therefore more likely to be correctly specified under $\mathcal M_d$. Second, Example~\ref{ex:bounded_q} constructs a concrete DGP in $\mathcal M_d'\setminus\mathcal M_\mathrm{OTR}$. Since $\mathcal M_d\supset\mathcal M_d'$ and $\mathcal M_\mathrm{OTR}\supset\mathcal M_\mathrm{TR}$, this example directly establishes that both $\mathcal M_\mathrm{QR}=\mathcal M_\mathrm{OTR}\cup\mathcal M_d$ and, with nuisance functions specified within the considered parametric families, $\mathcal M_\mathrm{MQR}=\mathcal M_\mathrm{OTR}\cup\mathcal M_d'$ strictly contain $\mathcal M_\mathrm{TR}$ and $\mathcal M_\mathrm{OTR}$.

We first study the scenario $\mathcal M_d$, which requires correct specification of the nuisance models $(q,\tau_n,\tau)$. In practice, the specification conditions for these three nuisance models are closely related. Since $q$ characterizes the cross-group shift in the mediator distribution and $\tau_n$ averages the working outcome model $\mu^*$ over the mediator distribution under $A=0$, an accurate characterization of the conditional mediator mechanism naturally supports the joint correctness of $q$ and $\tau_n$. Given $(q^*,\tau_n^*)=(q,\tau_n)$, Theorem~\ref{thm:dr_tau} further ensures that our DR regression directly targets $\tau(\bX)$, so that $\tau^*=\tau$ whenever the working model for $\tau$ captures its true functional form. The following example illustrates a general model class in which scenario $\mathcal M_d$ naturally arises.

\begin{example}[A general model class motivating $\mathcal M_d$]
\label{ex:general_residual}
Consider a partially linear mediation model with a mediator vector $\bM\in\R^{d_2}$ of fixed or moderate dimension and a low-dimensional covariate summary $Z=T(\bX)$. Let $\bM(a)=\bs r(Z)+\bs R_a$ and $Y(a,\bs m)=\bs b(\bX)^\top\bbeta_a+g_a\{Z,\bs m-\bs r(Z)\}+\varepsilon_Y$ for $a\in\{0,1\}$, so that the observed data satisfy
\begin{align*}
\bM&=\bs r(Z)+\bs R_A,\\
Y&=\bs b(\bX)^\top\bbeta_A+g_A(Z,\bs R_A)+\varepsilon_Y,
\end{align*}
where $\bs r(Z)$ is an unknown, possibly nonlinear baseline, $g_a$ is an unknown function allowing nonlinear effects and interactions, and $\bs b(\bX)$ contains covariate transformations and an intercept. Assume that $\bs R_a\sim F_a$, $\E(\varepsilon_Y)=0$, $(\bs R_0,\bs R_1)\independent\bX$, $\varepsilon_Y\independent(\bX,\bs R_0,\bs R_1)$, and $A\independent(\bs R_0,\bs R_1,\varepsilon_Y)\mid\bX$. This model encompasses several existing mediation specifications in terms of its regression form, including the nonlinear mediator-response model of \citet{MediationShapeRestricted2024} and the varying-coefficient model of \citet{VaryingCoefficientMediation2022}.

Here, $Y$ depends on the covariates and the mediator residual $\bs R_A=\bM-\bs r(Z)$, whose distribution may shift under treatment. For instance, in biomedical studies where $\bM$ denotes a panel of clinical biomarkers, normal reference levels $\bs r(Z)$ typically vary with demographic characteristics such as age and sex, so clinical outcomes are driven by the deviations $\bM-\bs r(Z)$ and their interactions rather than the raw measurements $\bM$.

Suppose $F_0$ and $F_1$ have common-support densities $f_{R,0}$ and $f_{R,1}$. Then $q(\bS)=f_{R,0}\{\bM-\bs r(Z)\}/f_{R,1}\{\bM-\bs r(Z)\}$ inherits a low-dimensional structure, as it depends only on $(Z,\bM)$ through the residual $\bs R_A=\bM-\bs r(Z)$.
Consider a possibly misspecified working outcome regression $\mu^*(\bS)=\bs b(\bX)^\top\bbeta^*+\eta^*(Z,\bM)$. Assuming the relevant integrals are finite, we obtain
\begin{align*}
\tau_n(\bX)&=\bs b(\bX)^\top\bbeta^*+\int \eta^*\{Z,\bs r(Z)+\bs u\}\,\mathrm dF_0(\bs u),\\
\tau(\bX)&=\bs b(\bX)^\top\bbeta_1+\int g_1(Z,\bs u)\,\mathrm dF_0(\bs u).
\end{align*}
Although the nonlinear component of \(\mu\) (and possibly \(\mu^*\)) depends jointly on \(Z\) and \(\bM\), averaging over \(F_0\) integrates out \(\bs R_0\), leaving \(\tau\) and \(\tau_n\) with nonlinear components depending only on the low-dimensional \(Z\). Under suitable smoothness and integrability conditions, these conditional means may therefore be substantially easier to specify and estimate than the full outcome regression, even when \(\mu^*\) is misspecified. Moreover, \(\tau\) and \(\tau_n\) generally differ when \(\mu^*\neq\mu\). If the treatment mechanism is sufficiently complex that the working model for $\pi$ is misspecified, this model belongs to $\mathcal{M}_d$ but not to $\mathcal{M}_{\mathrm{OTR}}$ or $\mathcal{M}_{\mathrm{TR}}$.
\end{example}

\begin{example}[A concrete DGP in $\mathcal M_d'\setminus\mathcal M_\mathrm{OTR}$]\label{ex:bounded_q}
Consider a simple setting where the mediator $M$ is univariate. Let $X_1,\ldots,X_{d_1-1}\sim\mathrm{Unif}[-1,1]$, $U\sim\mathrm{Unif}(0,1)$, and $\varepsilon_2\sim\mathcal N(0,1)$ be mutually independent, and set $\bX=(1,X_1,\ldots,X_{d_1-1})^\top$ to include an intercept. Given $\bX$, generate $A$ independently of $(U,\varepsilon_2)$ with propensity score $\pi(\bX)=0.3X_1+0.5$, which is modeled using a misspecified logistic working model. For $a\in\{0,1\}$, let $F_a$ be the distribution function of a $\mathcal N(a,1)$ variable truncated to the common interval $[-3,3]$, and write $Z_a=\Phi(3-a)-\Phi(-3-a)$, where $\phi$ and $\Phi$ denote the standard normal density and distribution functions. Define
\begin{align*}
M(a)=X_1+F_a^{-1}(U),\qquad M=M(A),\qquad
Y=X_1+M+1-A+(M-X_1)^2+\varepsilon_2.
\end{align*}
Thus, $Y(1,m)=X_1+m+(m-X_1)^2+\varepsilon_2$, and the conditional distributions of $M-X_1$ in both treatment groups have the same bounded support. The density ratio is
\begin{align*}
q(\bS)=\frac{Z_1}{Z_0}\exp(-M+X_1+1/2),
\qquad 0.08<q(\bS)<32.45\quad\text{almost surely}.
\end{align*}
Hence, $q$ belongs to $\mathcal G_q$ (for the parametric nuisance classes defined in \eqref{def:parametrization} of Section~\ref{apdx:mqr_supp}) and is bounded away from zero and infinity. The outcome regression $\mu(\bS)=X_1+M+(M-X_1)^2$ remains nonlinear. Since $F_0$ is symmetric about zero, the mediator conditional mean and the cross-world conditional mean satisfy
\begin{align*}
m(\bX)=X_1,\qquad \tau(\bX)=2X_1+v,
\qquad v:=1-\frac{6\phi(3)}{2\Phi(3)-1}\approx0.9733.
\end{align*}
Both are linear, and $\tau_n$ is also linear for any linear working model $\mu^*$.

Consequently, we have $q\in\mathcal G_q$, $m\in\mathcal G_m$, and $\tau\in\mathcal G_\tau$, but $\mu\notin\mathcal G_\mu$ and $\pi\notin\mathcal G_\pi$. The DGP therefore belongs to both $\mathcal M_\mathrm{MQR}$ and $\mathcal M_\mathrm{QR}$, but not to $\mathcal M_\mathrm{TR}$ or $\mathcal M_\mathrm{OTR}$ under these working models. Under the regularity and nuisance-rate conditions of Theorem~\ref{thm:QR_consistent}, the QR estimator is consistent; it is $\sqrt N$-consistent when $d_1$ is fixed and the nuisance estimators converge at the parametric rate. Moreover, Theorem~\ref{thm:r_infer} yields $\sqrt N$-consistency and asymptotic normality of the MQR estimator under the corresponding regularity and sparsity/rate conditions, including those for $\mathcal M_d'$ in Table~\ref{table:MQR_sparsity}.

To illustrate the limitation of the existing nuisance constructions, let $\hat\mu_{\mathrm{otr}}$ be a linear estimator obtained by regressing $Y$ on $\bS$ given $A=1$, and let $\hat\tau_{\mathrm{otr}}$ be a nested linear estimator obtained by regressing $\hat\mu_{\mathrm{otr}}(\bS)$ on $\bX$ given $A=0$, both constructed without the specialized weights used in our approach. Let $R_1=F_1^{-1}(U)$ and define
\begin{align*}
\gamma:=\frac{\operatorname{Cov}(R_1,R_1^2)}{\Var(R_1)}\approx1.6839,
\qquad \alpha:=\E(R_1^2)-\gamma\E(R_1)\approx0.1876.
\end{align*}
The population regression limits are
\begin{align*}
\mu_{\mathrm{otr}}^*(\bS)=\alpha+(1-\gamma)X_1+(1+\gamma)M,
\qquad \tau_{\mathrm{otr}}^*(\bX)=2X_1+\alpha.
\end{align*}
If $\hat\tau_{\mathrm{tr}}$ is instead constructed through the integral representation in \eqref{eq:joint2}, it still converges to the same limit $\tau_{\mathrm{tr}}^*(\bX)=\tau_{\mathrm{otr}}^*(\bX)$, even when $f_0$ is known. The corresponding approximation errors are
\begin{align*}
\normp{\mu_{\mathrm{otr}}^*-\mu}{2}\approx1.7954,
\qquad \normp{\tau_{\mathrm{otr}}^*-\tau}{2}=|\alpha-v|\approx0.7857.
\end{align*}
Thus, although the true target $\tau$ is linear, its nested estimator remains inconsistent because of the inaccurate initial working model for $\mu$. With both $\pi$ and this nested approximation of $\tau$ misspecified, the existing TR and O-TR robustness guarantees do not ensure consistency of the final estimator of $\theta$.
\end{example}

\subsection{Properties of the nuisance parameters and estimators for the MQR method}\label{apdx:mqr_supp}

\paragraph*{Uniqueness of $\bnu^*$}
The following lemma establishes the uniqueness of the target nuisance parameters $\bnu^*=(\bepiar,\bepibr,\beqr,\bemur,\bettaur,\betaur)$ defined through the loss functions $\ell_1,\ldots,\ell_6$ in \eqref{eq:r_loss1}--\eqref{eq:r_loss6}.

\begin{lemma}[Uniqueness of $\bnu^*$]\label{lem:unique_nu}
Let Assumption~\ref{cond:basic2} hold. Then the expectations of the loss functions in \eqref{eq:r_loss1}--\eqref{eq:r_loss6} admit unique minimizers $\bepiar,\bepibr,\beqr,\bemur,\bettaur,$ and $\betaur$.
\end{lemma}

\paragraph*{Correct specification of the nuisance models}
Consider the following parametric nuisance classes:
\begin{align}
&\mathcal{G}_\pi=\{g(\bX^\top\bepi):\bepi\in\R^{d_1}\},\quad
\mathcal{G}_q=\{\exp(\bS^\top\beq):\beq\in\R^d\},\quad
\mathcal{G}_\mu=\{\bS^\top\bemu:\bemu\in\R^d\},\nonumber\\
&\mathcal{G}_{\tau_n}=\mathcal{G}_{\tau}=\{\bX^\top\betau:\betau\in\R^{d_1}\},\quad
\mathcal{G}_{\bs{m}}=\{\bphi^\top\bX:\bphi\in\R^{d_1\times d_2}\}. \label{def:parametrization}
\end{align}
In terms of \eqref{def:parametrization}, the model class $\mathcal{M}_{\mathrm{MQR}}$ consists of DGPs satisfying at least one of the following conditions: (a)~$\mu\in\mathcal{G}_\mu$ and $\pi\in\mathcal{G}_\pi$; (b)~$\mu\in\mathcal{G}_\mu$ and $\tau_n\in\mathcal{G}_{\tau_n}$; (c)~$\pi\in\mathcal{G}_\pi$ and $q\in\mathcal{G}_q$; or (d)~$q\in\mathcal{G}_q$, $\tau\in\mathcal{G}_\tau$, and $\bs{m}\in\mathcal{G}_{\bs{m}}$. Note that $h\in\mathcal{G}_h$ means that the true nuisance function $h$ admits the parametric form in \eqref{def:parametrization}, which implies $\emqr_h=0$.

In what follows, we distinguish two notions of correct specification for a nuisance function $h\in\{\pi,q,\mu,\tau_n,\tau\}$: whether $h$ admits the corresponding parametric form ($h\in\mathcal{G}_h$), and whether it coincides with its target nuisance function ($h^*=h$, with $\pi_a^*=\pi_b^*=\pi$ for $h=\pi$). By characterizing how approximation errors propagate across the sequential nuisance construction, we verify that these two notions are equivalent across the regimes of $\mathcal{M}_{\mathrm{MQR}}$. This justifies using the two notions interchangeably when describing the robustness properties of the MQR estimator.

To separate individual model misspecification from sequential error propagation in \eqref{eq:r_loss1}--\eqref{eq:r_loss6}, we define the following oracle nuisance parameters by replacing the preceding nuisance working models in $\ell_3$--$\ell_6$ with their true counterparts, with the sole exception of $\mu^*$ in $\ell_5$, where $\mu^*$ serves as the ``true'' response in the nested regression for $\tau_n$:
\begin{align*}
\beq^0&:=\mathop{\arg\min}_{\beq\in\R^d}\E\!\left[\frac{A}{\pi(\bX)}\exp(\bS^\top\beq)-\frac{1-A}{1-\pi(\bX)}\bS^\top\beq\right],\\
\bemu^0&:=\mathop{\arg\min}_{\bemu\in\R^d}\E\!\left[\frac{Aq(\bS)}{\pi(\bX)}(\bS^\top\bemu-Y)^2\right],\\
\bettau^0&:=\mathop{\arg\min}_{\bettau\in\R^{d_1}}\E\!\left[(1-A)\frac{\pi(\bX)}{1-\pi(\bX)}\{\bX^\top\bettau-\mu^*(\bS)\}^2\right],\\
\betau^0&:=\mathop{\arg\min}_{\betau\in\R^{d_1}}\E\!\left[A\frac{1-\pi(\bX)}{\pi(\bX)}\big[q(\bS)\{Y-\mu(\bS)\}+\tau(\bX)-\bX^\top\betau\big]^2\right].
\end{align*}
Similarly, we define the oracle nuisance matrix by the $L_2(\P)$ linear projection of the mediator conditional mean $\bs{m}(\bX):=\E[\bM\mid A=0,\bX]$ onto $\bX$:
\begin{align*}
\bphi^0&:=\mathop{\arg\min}_{\bphi\in\R^{d_1\times d_2}}\E\big[\|\bs{m}(\bX)-\bphi^\top\bX\|_2^2\big].
\end{align*}
Under Assumptions~\ref{cond:basic} and~\ref{cond:basic2}, each objective above is strictly convex and admits a unique minimizer. We denote the induced functions by $q^0(\bS):=\exp(\bS^\top\beq^0)$, $\mu^0(\bS):=\bS^\top\bemu^0$, $\tau_n^0(\bX):=\bX^\top\bettau^0$, $\tau^0(\bX):=\bX^\top\betau^0$, and $\bs{m}^0(\bX):=\bphi^{0\top}\bX$.

\begin{lemma}[Oracle target approximation errors]\label{lem:oracle_equiv}
Let Assumptions~\ref{cond:basic} and~\ref{cond:basic2} hold. Then the oracle targets $(q^0,\mu^0,\tau_n^0,\tau^0,\bs{m}^0)$ satisfy:
(a) $\|\beq^0\|_2=O(1)$ and $\normp{q^0-q}{6}\asymp \emqr_q$;
(b) $\normp{\mu^0-\mu}{6}\asymp \emqr_\mu$;
(c) $\normp{\tau_n^0-\tau_n}{6}\asymp \emqr_{\tau_n}$ and $\sup_{\|\bv\|_2=1}\normp{\{\bs{m}^0(\bX)-\bs{m}(\bX)\}^\top\bv}{6}\asymp \emqr_{\bs{m}}$;
(d) $\normp{\tau^0-\tau}{6}\asymp \emqr_\tau$.
\end{lemma}
Lemma~\ref{lem:oracle_equiv} shows that the $L_6(\P)$ parametric approximation errors $\emqr_q,\emqr_\mu,\emqr_{\tau_n},\emqr_\tau,\emqr_{\bs{m}}$ in Section~\ref{sec:r_thm} are equivalent to the $L_6(\P)$ deviations of the oracle targets $(q^0,\mu^0,\tau_n^0,\tau^0,\bs{m}^0)$ from the true nuisance functions, i.e., the individual model misspecification errors in the absence of sequential error propagation. In contrast, following Section~\ref{sec:QR_thm}, we quantify the deviations of the target nuisance functions $(\pi_a^*,\pi_b^*,q^*,\mu^*,\tau_n^*,\tau^*)$, which additionally incorporate propagated errors, by $\bemqr_\pi:=\max_{j\in\{a,b\}}\|\pi_j^*-\pi\|_{\P,6}$ and $\bemqr_h:=\|h^*-h\|_{\P,6}$ for $h\in\{q,\mu,\tau_n,\tau\}$. Building on Lemma~\ref{lem:oracle_equiv}, the next lemma traces how the individual approximation errors $\emqr_\pi,\emqr_q,\emqr_\mu,\emqr_{\tau_n},\emqr_\tau,\emqr_{\bs{m}}$ propagate across the sequential steps.

\begin{lemma}[Approximation error bounds for the target nuisance functions]\label{lem:r_approx_bound}
Let Assumptions~\ref{cond:basic} and~\ref{cond:basic2} hold. Then the approximation errors $\bemqr_\pi,\bemqr_q,\bemqr_\mu,\bemqr_{\tau_n},\bemqr_\tau$ of the target nuisance functions satisfy:
(a) $\bemqr_\pi=O(\emqr_\pi)$;
(b) $\bemqr_q=O(\emqr_q+\bemqr_\pi \emqr_{\bs{m}})=O(\emqr_q+\emqr_\pi \emqr_{\bs{m}})$;
(c) $\bemqr_\mu=O(\emqr_\mu)$;
(d) $\bemqr_{\tau_n}=O(\emqr_{\tau_n})$ and $\emqr_{\tau_n}=O(\emqr_{\bs{m}})$;
(e) $\bemqr_\tau=O\big(\emqr_\tau+\min\{\emqr_{\tau_n},\emqr_\mu\}+\bemqr_q \emqr_\mu\big)=O\big(\emqr_\tau+\min\{\emqr_{\tau_n},\emqr_\mu\}+(\emqr_q+\emqr_\pi \emqr_{\bs{m}})\emqr_\mu\big)$.
\end{lemma}

The target deviation $\bemqr_\pi$ is controlled directly by $\emqr_\pi$. Analogously, $\bemqr_\mu$ and $\bemqr_{\tau_n}$ are controlled by $\emqr_\mu$ and $\emqr_{\tau_n}$, respectively, even though $\mu^*$ and $\tau_n^*$ depend on preceding nuisance functions.
However, this direct relationship does not hold for $q^*$ and $\tau^*$, for which the sequential construction introduces additional propagated approximation error terms.
Specifically, focusing on the regime where each approximation error is either zero ($\emqr_h=0$) or of constant order ($\emqr_h\asymp 1$), we see that $\bemqr_q=O(\emqr_q)$ holds when either $\emqr_\pi=0$ or $\emqr_{\bs{m}}=0$, and $\bemqr_\tau=O(\emqr_\tau)$ holds when either $\emqr_\mu=0$ or $\bemqr_q=\emqr_{\tau_n}=0$ (which holds in particular when $\emqr_q=\emqr_{\bs{m}}=0$). More generally, specializing Lemma~\ref{lem:r_approx_bound} to exact specification yields the following corollary.
\begin{corollary}[Correct specification of the nuisance models]\label{lem:r_identi}
Let Assumptions~\ref{cond:basic} and~\ref{cond:basic2} hold. Then the following statements hold:
(a) $\pi_a^*=\pi_b^*=\pi\Longleftrightarrow\pi\in\mathcal{G}_\pi$;
(b) If either $\pi\in\mathcal{G}_\pi$ or $\bs{m}\in\mathcal{G}_{\bs{m}}$, then $q^*=q\Longleftrightarrow q\in\mathcal{G}_q$;
(c) $\mu^*=\mu\Longleftrightarrow\mu\in\mathcal{G}_\mu$;
(d) $\tau_n^*=\tau_n\Longleftrightarrow\tau_n\in\mathcal{G}_{\tau_n}\Longleftarrow\bs{m}\in\mathcal{G}_{\bs{m}}$;
(e) If either $\mu\in\mathcal{G}_\mu$ or $(q^*,\tau_n^*)=(q,\tau_n)$, then $\tau^*=\tau\Longleftrightarrow\tau\in\mathcal{G}_{\tau}$;
(f) $\tau^*=\tau$ if $\mu\in\mathcal{G}_\mu$ and $\tau_n\in\mathcal{G}_{\tau_n}$.
\end{corollary}
In each regime of $\mathcal{M}_{\mathrm{MQR}}$ that requires $q\in\mathcal{G}_q$ or $\tau\in\mathcal{G}_\tau$, the corresponding auxiliary conditions in Corollary~\ref{lem:r_identi}(b) and~(e) are automatically met, so $q^*=q$ or $\tau^*=\tau$ holds accordingly. Therefore, $\mathcal{M}_{\mathrm{MQR}}$ can be equivalently characterized as the class of DGPs satisfying at least one of the following conditions: (a)~$\mu^*=\mu$ and $\pi_a^*=\pi_b^*=\pi$; (b)~$\mu^*=\mu$ and $\tau_n^*=\tau_n$; (c)~$\pi_a^*=\pi_b^*=\pi$ and $q^*=q$; or (d)~$q^*=q$, $\tau^*=\tau$, and $\bs{m}\in\mathcal{G}_{\bs{m}}$.

To illustrate the role of $\bs{m}(\bX)$ in Lemma~\ref{lem:r_approx_bound}(b), we show how the product term $\emqr_\pi \emqr_{\bs{m}}$ arises in the approximation error bound for $q^*$. Under Assumptions~\ref{cond:basic} and~\ref{cond:basic2}, $\bemqr_q$ is controlled by $\normp{q^0-q}{6}+\|\beqr-\beq^0\|_2=O(\emqr_q+\|\beqr-\beq^0\|_2)$. Moreover, by the definition of $\bbeta_q^*$ and Taylor's theorem, $\|\beqr-\beq^0\|_2^2$ is bounded, up to a constant factor, by $|(\beqr-\beq^0)^\top\E[\bnabla_{\beq}\ell_3(\bW,\bnu_3^*,\beq^0)]|$. It therefore suffices to bound the population score evaluated at $\beq^0$, which admits the decomposition
\begin{align*}
\E[\bnabla_{\beq}\ell_3(\bW,\bnu_3^*,\beq^0)]
&\overset{(i)}{=}\E\!\left[\frac{A\{q^0(\bS)-q(\bS)\}}{\pi_a^*(\bX)}\bS\right]
+\E\!\left[\Big\{\frac{A}{\pi_a^*(\bX)}-\frac{1-A}{1-\pi_b^*(\bX)}\Big\}
\begin{pmatrix}\bX\\ \bs{m}(\bX)\end{pmatrix}\right] \\
&\overset{(ii)}{=}\E\!\left[\frac{A\{q^0(\bS)-q(\bS)\}}{\pi_a^*(\bX)}\bS\right]
+\begin{pmatrix}\bm{I}_{d_1}\\ \bphi^{0\top}\end{pmatrix}
\E[\bnabla_{\bettau}\bpsi(\bW,\bnu^*)]\\
&\qquad+\begin{pmatrix}\bz_{d_1}\\ \E\!\left[\Big\{\frac{\pi(\bX)}{\pi_a^*(\bX)}-\frac{1-\pi(\bX)}{1-\pi_b^*(\bX)}\Big\}\{\bs{m}(\bX)-\bs{m}^0(\bX)\}\right]\end{pmatrix},
\end{align*}
where (i) follows from Assumption~\ref{cond:basic}, and (ii) holds by $\bs{m}^0(\bX)=\bphi^{0\top}\bX$. By the first-order conditions for $(\ell_1,\ell_2)$, the middle term vanishes since $\E[\bnabla_{\bettau}\bpsi(\bW,\bnu^*)]=\E[\{A/\pi_a^*(\bX)-(1-A)/(1-\pi_b^*(\bX))\}\bX]=\E[\{1-(1-A)/(1-\pi_b^*(\bX))\}\bX]=\bz_{d_1}$, while the remaining two terms are controlled by $\emqr_q$ and $\bemqr_\pi \emqr_{\bs{m}}=O(\emqr_\pi \emqr_{\bs{m}})$, respectively.

The function $\bs{m}(\bX)$ and the class $\mathcal{G}_{\bs{m}}$ are introduced solely to characterize the approximation error and correct specification of $q^*$. Unlike other nuisance functions, $\bs{m}(\bX)$ is not estimated in our procedure. Instead, it is sufficient to estimate $\tau_n(\bX)=\E(\bS^\top\bemur\mid A=0,\bX)=(\bX^\top\;\bs{m}(\bX)^\top)\bemur$, which only requires estimation of a linear combination of $\bs{m}(\bX)$. Even when $\bs{m}$ is linear with $\bs{m}(\bX)=\bphi^{0\top}\bX$, estimating the full matrix $\bphi^0$, whose row and column dimensions are both large, is typically difficult, and our method avoids imposing any sparse structure on this matrix.
At the same time, correct specification of $q$ may rely on the linearity of $\bs{m}$, and linearity of $\tau_n$ alone is not sufficient, except in the special case where the mediator $\bM$ is one-dimensional.

\paragraph*{Convergence rates of the nuisance estimators}
We now turn to the theoretical guarantees for the nuisance estimators $\hbepiar,\hbepibr,\hbeqr,\hbemur,\hbettaur,$ and $\hbetaur$ in Algorithm~\ref{alg:MQR}, accommodating both sequential estimation and model misspecification. For brevity, we suppress the fold superscript $-k$ throughout.

\begin{theorem}\label{thm:r_nuisance}
Let Assumptions \ref{cond:basic}, \ref{cond:basic2} and \ref{cond:sparsity} hold. As $N$, $d_1$, $d_2\to\infty$, the nuisance estimators converge at the following rates:
(a) $\|\hbepiar-\bepiar\|_2=O_p(r_{\pi})$,
$\|\hbepibr-\bepibr\|_2=O_p(r_{\pi})$;
(b) $\|\hbeqr-\beqr\|_2=O_p(r_{\pi}+r_{q})$;
(c) $\|\hbemur-\bemur\|_2=O_p(r_{\mu}+\emqr_{\mu}(r_{\pi}+r_{q}))$;
(d) $\|\hbettaur-\bettaur\|_2=O_p((\emqr_{\mu}+\emqr_{\tau_n})r_{\pi}+\emqr_{\mu}r_{q}+r_{\mu}+r_{\tau_n})$;
and (e) $\|\hbetaur-\betaur\|_2=O_p((\emqr_{\mu}+\emqr_{\tau_n}+\emqr_{\tau})r_{\pi}+\emqr_{\mu}r_{q}+r_{\mu}+r_{\tau_n}+r_{\tau})$.
\end{theorem}
Under the sub-Gaussian assumptions, the actual stochastic errors in Assumption~\ref{cond:QR_basic} can be taken as follows: (a) $\bar r_\pi=r_{\pi}$; (b) $\bar r_q=r_{\pi}+r_{q}$; (c) $\bar r_\mu=r_{\mu}+\emqr_{\mu}(r_{\pi}+r_{q})$; (d) $\bar r_{\tau_n}=(\emqr_{\mu}+\emqr_{\tau_n})r_{\pi}+\emqr_{\mu}r_{q}+r_{\mu}+r_{\tau_n}$; and (e) $\bar r_\tau=(\emqr_{\mu}+\emqr_{\tau_n}+\emqr_{\tau})r_{\pi}+\emqr_{\mu}r_{q}+r_{\mu}+r_{\tau_n}+r_{\tau}$.

Theorem~\ref{thm:r_nuisance} reveals how stochastic and approximation errors accumulate across the nuisance estimates. Since the nuisance estimators are constructed sequentially, later estimators depend on all previous ones. As a result, when all nuisance models are misspecified, each estimator's convergence rate depends on the sparsity levels of all nuisance parameters up to itself. When the regression models are correctly specified, however, the accumulation of the stochastic errors from $\pi$ and $q$ vanishes. This benefit does not extend to the stochastic errors within the regression sequence itself: even under correct specification, the convergence rate of $\hbettaur$ remains influenced by $s_{\mu}$, and that of $\hbetaur$ by both $s_{\mu}$ and $s_{\tau_n}$. This phenomenon originates from the nested regression structure, which performs a weighted regression of the imputed values $\bS^\top\hbemu$ on $\bX$, rendering error accumulation across the regression sequence unavoidable.

\section{The QR2 and MQR2 methods and their theoretical justification}\label{apdx:variants}
This section presents the algorithms and theoretical justifications for QR2 and MQR2. We refer to the primary implementations in Algorithms~\ref{alg:QR} and~\ref{alg:MQR} as QR1 and MQR1 whenever distinguishing them from QR2 and MQR2, while using QR and MQR for the general frameworks when no distinction is needed. For notational simplicity, we adopt the same notation for the nuisance estimators in Algorithms~\ref{alg:qr2} and~\ref{alg:mqr2} as in QR1 and MQR1, respectively, suppressing the fold index $-k$ when clear from context.

\subsection{The QR2 method and its theoretical justification}\label{apdx:qr2}
The variant QR2 shares the same score function and cross-fitting partition $(\I_k)_{k\in[\K]}$ as QR1, but estimates all nuisance functions on the entire training set $\I_{-k}$ without further within-fold splitting. Algorithm~\ref{alg:qr2} summarizes the complete procedure.

\begin{algorithm}[t!] \caption{The quadruply robust (QR2) estimator for $\theta$}\label{alg:qr2}
\begin{algorithmic}[1]
\Require Observations $\mathbb S=(\bW_i)_{i\in[N]}$ and the number of folds $\K\geq2$.
\State Let $(\I_k)_{k\in[\K]}$ be a disjoint partition of $\I=[N]$ with equal sizes $n=N/\K$.
\For{$k\in[\K]$}
\State Let $\I_{-k}=\I\setminus\I_k$.
\State Using $\I_{-k}$, construct $\hat\pi^{-k}$, $\hat q^{-k}$, and $\hat\mu^{-k}$.
\State Using $\I_{-k}$, construct $\hat\tau_n^{-k}$ by solving
\begin{align*}
\hat\tau_n^{-k}\in\arg\min_{f\in\mathcal G}
\frac{1}{|\I_{-k}|}\sum_{i\in\I_{-k}}(1-A_i)
\{\hat\mu^{-k}(\bS_i)-f(\bX_i)\}^2.
\end{align*}
\State Using $\I_{-k}$, construct $\hat\tau^{-k}$ by solving
\begin{align*}
\hat\tau^{-k}\in\arg\min_{f\in\mathcal G}
\frac{1}{|\I_{-k}|}\sum_{i\in\I_{-k}}A_i
[\hat\tau_n^{-k}(\bX_i)+\hat q^{-k}(\bS_i)
\{Y_i-\hat\mu^{-k}(\bS_i)\}-f(\bX_i)]^2.
\end{align*}
\EndFor\\
\Return The QR2 estimator:
$\hat\theta_{\mathrm{QR2}}=N^{-1}\sum_{k=1}^{\K}\sum_{i\in\I_k}\hat\bpsi^{(-k)}(\bW_i),$
where $\hat\bpsi^{(-k)}$ is obtained from \eqref{def:QR_form} by replacing the working models $(\pi^*,q^*,\mu^*,\tau_n^*,\tau^*)$ with the nuisance estimates $(\hat\pi^{-k},\hat q^{-k},\hat\mu^{-k},\hat\tau_n^{-k},\hat\tau^{-k})$ obtained above.
\end{algorithmic}
\end{algorithm}

The theoretical guarantees in Theorems~\ref{thm:QR_consistent} and~\ref{thm:QR_normal} also hold for QR2 once the nuisance estimators in Algorithm~\ref{alg:qr2} satisfy the required conditions. This is because all nuisance estimators are still fitted on $(\bW_i)_{i\in\I_{-k}}$ and thus remain independent of the evaluation fold $(\bW_i)_{i\in\I_k}$; conditional on $(\bW_i)_{i\in\I_{-k}}$, the decomposition~\eqref{decom:QR_theta} and the bounds in Lemmas~\ref{lem:Delta_0}, \ref{lem:Delta_k_2}, and~\ref{lem:Delta_k_3} apply without change. However, without the sub-splitting, the nested regression for $\tau_n$ and the DR estimation of $\tau$ evaluate $\hat\mu^{-k}$, $\hat q^{-k}$, and $\hat\tau_n^{-k}$ on their own training fold $\I_{-k}$ and may therefore incur overfitting bias \citep{DoubleDebiasedMachine2018}. To attain the nuisance convergence rates in Assumption~\ref{cond:QR_basic}(b) and the product-rate condition in Theorem~\ref{thm:QR_normal}, QR2 generally requires stronger smoothness or sparsity assumptions than QR1.

\subsection{The MQR2 method and its theoretical justification}\label{apdx:mqr2_results}
We next investigate the asymptotic properties of the MQR2 estimator introduced in Algorithm~\ref{alg:mqr2}. Unlike MQR1, MQR2 estimates all nuisance functions on $\I_{-k}$ without further splitting. While using all observations in $\I_{-k}$ improves finite-sample performance (as shown in Section~\ref{sec:num}), the induced statistical dependence across nuisance estimators demands stronger sparsity conditions.

\begin{algorithm}[t!] \caption{The model quadruply robust (MQR2) estimator for $\theta$}\label{alg:mqr2}
\begin{algorithmic}[1]
\Require Observations $\mathbb{S}=(\bW_i)_{i\in[N]}$ and the number of folds $\K\geq2$.
\State Let $(\I_k)_{k\in[\K]}$ be a disjoint partition of $\I=[N]$ with equal sizes $n=N/\K$.
\For{$k\in[\K]$}
\State Construct the nuisance estimates sequentially:
\begin{align*}
\hbepia^{-k}&=\mathop{\arg\min}_{\bepia\in\R^{d_1}}\Big\{|\mathcal I_{-k}|^{-1}\suk\ell_1(\bW_i,\bepia)+\lambda_{\pi_a}\|\bepia\|_1\Big\},\\
\hbepib^{-k}&=\mathop{\arg\min}_{\bepib\in\R^{d_1}}\Big\{|\mathcal I_{-k}|^{-1}\suk\ell_2(\bW_i,\bepib)+\lambda_{\pi_b}\|\bepib\|_1\Big\},\\
\hbeq^{-k}&=\mathop{\arg\min}_{\beq\in\R^{d}}\Big\{|\mathcal I_{-k}|^{-1}\suk\ell_3(\bW_i,\hbepia^{-k},\hbepib^{-k},\beq)+\lambda_{q}\|\beq\|_1\Big\},\\
\hbemu^{-k}&=\mathop{\arg\min}_{\bemu\in\R^d}\Big\{|\mathcal I_{-k}|^{-1}\suk\ell_4(\bW_i,\hbepia^{-k},\hbeq^{-k},\bemu)+\lambda_{\mu}\|\bemu\|_1\Big\},\\
\hbettau^{-k}&=\mathop{\arg\min}_{\bettau\in\R^{d_1}}\Big\{|\mathcal I_{-k}|^{-1}\suk\ell_5(\bW_i,\hbepib^{-k},\hbemu^{-k},\bettau)+\lambda_{\tau_n}\|\bettau\|_1\Big\},\\
\hbetau^{-k}&=\mathop{\arg\min}_{\betau\in\R^{d_1}}\Big\{|\mathcal I_{-k}|^{-1}\suk\ell_6(\bW_i,\hbepia^{-k},\hbeq^{-k},\hbemu^{-k},\hbettau^{-k},\betau)+\lambda_{\tau}\|\betau\|_1\Big\},
\end{align*}
where $\lambda_{\pi_a},\lambda_{\pi_b},\lambda_q,\lambda_\mu,
\lambda_{\tau_n},\lambda_\tau>0$ are tuning parameters selected by cross-validation. Each tuning parameter is chosen sequentially, tuning one parameter at a time while treating all previously selected parameters and nuisance estimates as fixed.
\EndFor\\
\Return The MQR2 estimator:
$
\hat{\theta}_{\mathrm{MQR2}}=N^{-1}\sum_{k=1}^{\K}\sum_{i\in\I_k}\bpsi(\bs{W}_i,\hbnu^{-k}),
$
where $\bpsi$ is defined in \eqref{def:psihat-MQR} and $\hbnu^{-k} = (\hbepia^{-k},\hbepib^{-k},\hbeq^{-k},\hbemu^{-k},\hbettau^{-k},\hbetau^{-k})$.
\end{algorithmic}
\end{algorithm}

For MQR2, we choose the tuning parameters $\lambda_{\pi_a},\lambda_{\pi_b},\lambda_{\tau_n},\lambda_\tau \asymp\sqrt{\log d_1/N}$ and $\lambda_q,\lambda_\mu\asymp\sqrt{\log d/N}$. The estimation rates $r_\pi,r_q,r_\mu,r_{\tau_n}$, and $r_\tau$ are defined identically to those in Section~\ref{sec:r_thm}. For ease of exposition, we further assume $N=O(d_1)$ and formalize the required sparsity conditions as follows.
\begin{assumption}[Stronger sparsity]\label{cond:mqr2_sparsity}
    Let $r_{\mu}+r_{\tau_n}+r_{\tau}=o(1)$, $r_{\pi}+r_q=o(N^{-1/4})$, and $(r_{\pi}^2+r_{q}^2)\log d=O(1)$.
\end{assumption}

\begin{theorem}[Nuisance rates for the MQR2 method]\label{thm:mqr2_r_nuisance}
Let Assumptions~\ref{cond:basic}, \ref{cond:basic2}, and \ref{cond:mqr2_sparsity} hold.
Then, as $N,d_1,d_2\to\infty$, (a)
$\|\hbepiar-\bepiar\|_2=O_p(r_{\pi})$ and
$\|\hbepibr-\bepibr\|_2=O_p(r_{\pi})$; (b)
$\|\hbeqr-\beqr\|_2=O_p(r_{\pi}+r_q)$; (c)
$\|\hbemur-\bemur\|_2=O_p(r_{\pi}+r_q+r_{\mu})$; (d)
$\|\hbettaur-\bettaur\|_2=O_p(r_{\pi}+r_q+r_{\mu}+r_{\tau_n})$; and (e)
$\|\hbetaur-\betaur\|_2=O_p(r_{\pi}+r_q+r_{\mu}+r_{\tau_n}+r_{\tau})$.
\end{theorem}

Parts~(a) and (b) reproduce the corresponding statements of Theorem~\ref{thm:r_nuisance}. Because all nuisance functions are estimated on $\I_{-k}$, the preceding nuisance stochastic errors in parts~(c)--(e) enter additively without being scaled by the corresponding approximation errors, and hence remain even when those approximation errors are zero.

The next two theorems give the bias decomposition and robust inference results
for MQR2, corresponding to Theorems~\ref{thm:r_bias} and
\ref{thm:r_infer} for MQR1.

\begin{theorem}[Convergence rate]\label{thm:mqr2_bias}
Let Assumptions~\ref{cond:basic}, \ref{cond:basic2}, and \ref{cond:mqr2_sparsity} hold. Let $\rsigma^2:=\E\{\bpsi(\bW,\bnu^*)-\theta\}^2$. Then as $N,d_1,d_2\to\infty$,
$\rsigma\asymp 1+\|\betaur-\theta\bs{e}_1\|_2$ and
$\hat{\theta}_{\mathrm{MQR2}}-\theta
=O_p(\rsigma N^{-1/2}+\mathcal R_{\mathrm{MQR2}})$, where
$\mathcal R_{\mathrm{MQR2}}:=\emqr_\pi(\emqr_{\tau_n}+\emqr_\tau)+(\emqr_q+\emqr_\pi\emqr_{\bs{m}})\emqr_\mu+r_\pi(r_\mu+r_{\tau_n}+r_\tau)+r_qr_\mu+r_\pi^2+r_q^2.$
\end{theorem}

Table~\ref{table:mqr2_convergence} extends Table~\ref{table:MQR_convergence} to include MQR2 under the four scenarios in $\mathcal{M}_{\mathrm{MQR}}$.

\begin{table}[!t]
\centering
\spacingset{1}
\captionsetup{font=footnotesize,skip=6pt}
\caption{Convergence-rate bounds for QR, MQR1, and MQR2. For simplicity, $\sigma\asymp\rsigma\asymp1$, and we ignore the possible decay of approximation errors, omitting their nonzero factors. Here $\mathcal R_0:=N^{-1/2}+r_\pi(r_\mu+r_{\tau_n}+r_\tau)+r_qr_\mu$, and $\bar{\mathcal R}_0$ replaces each $r_h$ with $\br_h$.}
\label{table:mqr2_convergence}
\fontsize{10}{12}\selectfont
\setlength{\tabcolsep}{3pt}
\setlength{\arrayrulewidth}{0.4pt}
\renewcommand{\arraystretch}{1.45}
\newlength{\MQRtwocomparisoncolwidth}
\settowidth{\MQRtwocomparisoncolwidth}{$\mathcal R_0+(r_\pi+r_q)^2$}
\makebox[\textwidth][c]{%
\begin{tabular}{|l|c|*{2}{>{\centering\arraybackslash}p{\MQRtwocomparisoncolwidth}|}}
\hline
\multicolumn{1}{|c|}{\multirow{2}{*}{Model class}} & \multicolumn{3}{c|}{Convergence rates} \\
\cline{2-4}
& QR & \multicolumn{1}{c|}{MQR1} & \multicolumn{1}{c|}{MQR2} \\
\hline
$\mathcal M_a:(\pi^*,\mu^*)$\ \ding{51}
& \mbox{$\bar{\mathcal R}_0+\br_\pi+\br_\mu$}
& \mbox{$\mathcal R_0+r_\pi^2$}
& \mbox{$\mathcal R_0+(r_\pi+r_q)^2$} \\
\hline
$\mathcal M_b':(\mu^*,\tau_n^*)$\ \ding{51}
& \mbox{$\bar{\mathcal R}_0+\br_\mu+\br_{\tau_n}+\br_\tau$}
& \mbox{$\mathcal R_0$}
& \mbox{$\mathcal R_0+(r_\pi+r_q)^2$} \\
\hline
$\mathcal M_c'':(\pi^*,q^*)$\ \ding{51}
& \mbox{$\bar{\mathcal R}_0+\br_\pi+\br_q$}
& \mbox{$\mathcal R_0+(r_\pi+r_q)^2$}
& \mbox{$\mathcal R_0+(r_\pi+r_q)^2$} \\
\hline
$\mathcal M_d':(q^*,\tau^*,\bs m^*)$\ \ding{51}
& \mbox{$\bar{\mathcal R}_0+\br_\mu+\br_{\tau_n}+\br_\tau+\br_q$}
& \mbox{$\mathcal R_0+(r_\pi+r_q)^2$}
& \mbox{$\mathcal R_0+(r_\pi+r_q)^2$} \\
\hline
\end{tabular}%
}
\end{table}

\begin{theorem}[Robust inference]\label{thm:mqr2_infer}
Let Assumptions \ref{cond:basic}, \ref{cond:basic2}, and \ref{cond:mqr2_sparsity} hold. Additionally, assume
$\mathcal R_{\mathrm{MQR2}}=o(\rsigma N^{-1/2})$, where $\mathcal R_{\mathrm{MQR2}}$ is defined in Theorem~\ref{thm:mqr2_bias}.
Then, as $N,d_1,d_2\to\infty$,
$
\rsigma^{-1}N^{1/2}(\hat{\theta}_{\mathrm{MQR2}}-\theta)\leadsto\mathcal N(0,1)
$
with $\rsigma\asymp1+\|\betaur-\theta\bs{e}_1\|_2$ and $\hat{\sigma}_{\mathrm{m},2}^2=\rsigma^2\{1+o_p(1)\}$, where
$\hat{\sigma}_{\mathrm{m},2}^2:=N^{-1}\sum_{k=1}^{\K}\sum_{i\in\I_k}\{\bpsi(\bW_i,\hbnu^{-k})-\hat{\theta}_{\mathrm{MQR2}}\}^2.$
\end{theorem}

Compared with MQR1 in Theorem~\ref{thm:r_bias}, MQR2 retains the same approximation-error and pairwise product-rate terms in $\mathcal R_{\mathrm{MQR2}}$, but inflates the mixed term $\emqr_\mu(r_\pi+r_q)^2+(\emqr_{\tau_n}+\emqr_\tau)r_\pi^2$ to the squared term $r_\pi^2+r_q^2$. Consequently, although MQR2 remains consistent over $\mathcal M_{\mathrm{MQR}}$, correct specification of the regression models no longer compensates for the stochastic errors of $\pi$ and $q$. Even when all nuisance models are correctly specified, MQR2 still requires the individual rate condition $r_\pi+r_q=o(N^{-1/4})$ to achieve root-$N$ consistency and asymptotic normality, whereas MQR1 requires only the pairwise product-rate conditions; for example, if $r_\pi+r_q\asymp N^{-1/4+c}$ for some $c\in(0,1/4)$ while the product-rate conditions hold, MQR1 still achieves valid $\sqrt{N}$-inference, but this guarantee fails for MQR2. When $\rsigma\asymp 1$, $\|\bS\|_\infty<C$, and $d_1\asymp d$ as in Table~\ref{table:MQR_sparsity}, MQR2 therefore imposes $s_\pi+s_q=o(\sqrt{N}/\log d)$ uniformly across all regimes. Similar uniform rate requirements have appeared in \citet{RobustEstimationCausal2020}, \citet{HighdimensionalInferenceAverage2022}, and \citet{HighdimensionalModelassistedInference2024}.

Although MQR1 has more favorable asymptotic properties, MQR2 may still perform better in finite samples because each $\ell_1$-regularized problem is trained using five times as many observations as in MQR1. This substantially larger training sample may yield more accurate nuisance estimation despite the weaker independence structure, as observed in the simulation studies in Sections~\ref{sec:simu} and~\ref{apdx:detail_simu}.

\section{Supplementary Numerical Results}
\label{apdx:num}
This section provides implementation details and supplementary numerical results. Section~\ref{apdx:implementation} describes the simulation implementation details, Section~\ref{apdx:detail_simu} reports additional simulation results under Settings~(iii)--(vii), and Section~\ref{apdx:detail_real} provides implementation details for the real data analysis.

\subsection{Simulation implementation details}\label{apdx:implementation}
\paragraph*{Parametric modeling methods.}
In MQR1 and MQR2, tuning parameters in \textsf{RCAL} \citep{RCALSoftware2020} for $\hat{\pi}_a$, $\hat{\pi}_b$, and $\hat q$ are selected via five-fold cross-validation, and those in \textsf{glmnet} \citep{glmnetSoftware2010} for $\hat{\mu}$, $\hat{\tau}_n$, and $\hat{\tau}$ are chosen by default ten-fold cross-validation. Both packages select the penalty minimizing the mean cross-validation loss, and the fitted parametric weights are used directly without truncation. For O-TR', the same procedures are applied to estimate $\hat{\pi}$ and the inverse-causality treatment probability $p(\bS)=\P(A=1\mid \bS)$ via \textsf{RCAL} and the unweighted regressions $\hat{\mu}$ and $\hat{\tau}_n$ via \textsf{glmnet}.

\paragraph*{Nonparametric modeling methods.}
For all nonparametric methods, the nuisance functions $\pi$ (as well as $p(\bS)=\P(A=1\mid \bS)$ for O-TR), $\mu$, $\tau_n$, and $\tau$ are estimated using random forests implemented via \textsf{caret} and \textsf{ranger} \citep{caretSoftware2008,rangerSoftware2017}. Tuning parameters are selected by five-fold cross-validation, using the area under the ROC curve for classification probabilities and RMSE for regression functions. For numerical stability, $\hat{\pi}$ in all nonparametric methods and $\hat p$ in O-TR are truncated to $[0.01,0.99]$.

For the QR1 and QR2 methods, we estimate the density ratio $q$ via the empirical-likelihood approach of \citet{zhang2024semiparametric}, with its ridge penalty selected by five-fold cross-validation. Prior to fitting the density-ratio model, we apply a three-step screening procedure. First, to screen the mediators, we fit treatment-specific multivariate random forests via \textsf{randomForestSRC} \citep{randomForestSRCSoftware2026} and rank each mediator by a Gaussian symmetric Kullback--Leibler divergence statistic constructed from the fitted conditional means and out-of-bag residual variances across treatment groups. Second, we screen the covariates by taking the union of those with high forest impurity importance across the retained mediators. Finally, we screen the mediator basis terms by retaining all linear terms and selecting quadratic and pairwise interaction terms using a stability-selection scheme \citep{meinshausen2010stability} based on a cross-validated logistic elastic net of $A$ on the selected covariates and candidate basis. Subsequently, since $\hat{q}$ is only evaluated within the $A=1$ group, we truncate the raw estimates at the $1$st and $99$th empirical percentiles of the treated training observations and normalize them by their treated-group sample average, motivated by the identity $\E\{q(\bS)\mid A=1\}=1$. This post-processing is applied  to ensure stability both when sequentially estimating $\hat{\tau}$ and when evaluating the estimator on the evaluation fold $\I_k$.

The TR method requires conditional mediator density estimation and numerical integration. For continuous mediators, we fit treatment-specific multivariate Gaussian working models \citep[see, e.g.,][]{CausalMediationUncausallyRelated2021}, where conditional means are estimated by multivariate random forests \citep{MultivariateRandomForests2011} tuned via three-fold cross-validation, and covariance matrices are estimated from centered out-of-bag residuals via the analytic correlation and variance shrinkage estimators of \citet{ShrinkageCovariance2005} and \citet{DistributionFreeShrinkage2007} (\textsf{corpcor::cov.shrink}). For the binary mediator in Setting~(iv), the conditional densities reduce to conditional probabilities. We instead estimate them directly by random forests. The resulting density ratios are truncated and normalized as above. Finally, the integral is approximated using $50$ Monte Carlo draws when the mediator is continuous, and evaluated exactly by summing over $\{0,1\}$ in Setting~(iv).

\paragraph*{Variance estimation and confidence intervals.}
For the QR1, MQR1, and MQR2 estimators, the asymptotic variances are estimated by $\hat\sigma^2$ in Theorem~\ref{thm:QR_normal}, $\hat\sigma_{\mathrm m}^2$ in Theorem~\ref{thm:r_infer}, and $\hat\sigma_{\mathrm m,2}^2$ in Theorem~\ref{thm:mqr2_infer}, respectively. For QR2, the asymptotic variance is given by the analogous plug-in estimator constructed from $\hat\theta_{\mathrm{QR2}}$ and $\hat\bpsi^{(-k)}$ in Algorithm~\ref{alg:qr2}. Similarly, for TR, O-TR, O-TR', and Oracle, the asymptotic variances are estimated analogously by the empirical variances of their respective cross-fitted or oracle scores based on $\bpsi_{\mathrm{eff}}$ in~\eqref{def:EIF}. Dividing the corresponding asymptotic variance estimate by $N$ and taking the square root yields the standard error $\widehat{\mathrm{SE}}$ of each estimator $\hat\theta$. Two-sided $95\%$ Wald confidence intervals are then constructed as $[\hat\theta - z_{0.975}\widehat{\mathrm{SE}},\, \hat\theta + z_{0.975}\widehat{\mathrm{SE}}]$, where $z_{0.975}$ denotes the $0.975$-quantile of $\mathcal N(0,1)$.

\subsection{Additional simulation results}\label{apdx:detail_simu}
We now present the additional simulation results under Settings~(iii)--(vii). Following the implementation details described above, we retain the notation, competing estimators, and evaluation metrics defined in Section~\ref{sec:simu}. We set $(d_1,d_2)=(100,50)$ throughout, except in Setting~(iv) where $(d_1,d_2)=(100,1)$. We choose $N\in\{600,1000,1400,1800\}$ in Settings~(iii) and~(iv), fix $N=1800$ in Settings~(v) and~(vi) while varying the nonlinear outcome signal strength $\gamma\in\{0.25,0.5,1,2\}$ and nuisance-model sparsity levels $(s_\pi,s_q,s_\mu)\in\{(4,12,12),\allowbreak(7,18,16),\allowbreak(13,30,22),\allowbreak(17,39,25)\}$, respectively, and choose $N\in\{1000,1800,2500,5000\}$ in Setting~(vii). Each experiment is repeated $100$ times.

Let $\mathbf C\in\R^{d_1\times d_2}$ have entries $C_{j1}=0$, $C_{j2}=\operatorname{sgn}\{\sin(2j)\}$, $C_{j3}=-C_{j2}$, $C_{j,r+3}=\operatorname{sgn}[\sin\{j\lceil r/3\rceil\}]\{1-3\idf_{\{(j+r)\bmod3=0\}}\}$ for $1\leq r\leq6$, and $C_{jk}=\operatorname{sgn}\{\sin(jk)\}$ for $10\leq k\leq d_2$, where $1\leq j\leq d_1$. Set $\mathbf D=0.15\mathbf C$. Define $\bs m_2(\bX_i)=0.6(\mathbf B+\mathbf D)^\top\bX_i+0.06\mathbf C^\top(\bX_i-\bs1_{d_1})+0.45\bar{\bs1}_{3,d_2}$, and let $\bs{\Sigma}_2$ have entries $(\Sigma_2)_{jk}=\idf_{\{j=k\}}+0.5\idf_{\{|j-k|=1\}}$. Let $\bs{\Sigma}_X$ have entries $(\Sigma_X)_{jk}=0.3^{|j-k|}$, and set $\bs{\Sigma}_a=\operatorname{diag}\{(0.4+0.2a)\bs1_3^\top,\bs1_{d_2-3}^\top\}$ for $a\in\{0,1\}$. Let $\bs c=0.1(\bs1_4^\top,-\bs1_2^\top,-2)^\top$ and $\mathbf H_{128}$ be the $128\times128$ Sylvester--Hadamard matrix. Define $\mathbf D_{\mathrm d}\in\R^{d_1\times d_2}$ by $(D_{\mathrm d})_{jk}=0$ for $1\leq k\leq3$, $(D_{\mathrm d})_{jk}=c_{1+\{(j+k-5)\bmod7\}}$ for $4\leq k\leq10$, and $(D_{\mathrm d})_{jk}=0.1(H_{128})_{j,k-9}$ for $11\leq k\leq d_2$, where $1\leq j\leq d_1$.

\textbf{Setting (iii)}: All parametric models correctly specified. Generate $X_{ij}\overset{\mathrm{i.i.d.}}{\sim}\mathcal N_{\mathrm{trun},5}$ and $A_i\mid\bX_i\sim\mathrm{Bernoulli}\{g(0.5+0.4X_{i1}+0.2\bar{\bs1}_{5,d_1}^\top\bX_i)\}$. Generate $\bM_i=\bs m_2(\bX_i)+\bU_i$, where $\bU_i\mid(A_i,\bX_i)\sim\mathcal N_{d_2}\{0.55A_i(\bar{\bs1}_{1,d_2}+\bar{\bs1}_{2,d_2}),\bs{\Sigma}_2\}$. Let $Y_i=1+0.5\bar{\bs1}_{5,d_1}^\top\bX_i+\bar{\bs1}_{3,d_2}^\top\{\bM_i-\bs m_2(\bX_i)\}-(1-A_i)+\varepsilon_i$.

\textbf{Setting (iv)}: All parametric models misspecified, with a univariate binary mediator. Generate $\bX_i\sim\mathcal N_{d_1}(\bz,\bs{\Sigma}_X)$. Generate $A_i\mid\bX_i\sim\mathrm{Bernoulli}\{\tilde g(1+\bar{\bs1}_{3,d_1}^\top\bX_i)\}$. Generate $M_i\mid(A_i,\bX_i)\sim\mathrm{Bernoulli}\{0.2+0.15(1-A_i)(1+\idf_{\{X_{i2}>0\}})+0.3g(1+\bar{\bs1}_{3,d_1}^\top\bX_i)\}$. Let $Y_i=1+(1-M_i)(2X_{i1}X_{i3}+X_{i2}^2)+X_{i3}\idf_{\{X_{i1}\geq0\}}+X_{i2}\idf_{\{X_{i2}\geq0\}}+X_{i1}\idf_{\{X_{i3}\geq0\}}+0.5\sum_{j=1}^3|X_{ij}-1|-(1-A_i)(1+0.5\bar{\bs1}_{3,d_1}^\top\bX_i)+\varepsilon_i$.

\textbf{Setting (v)}: All parametric models misspecified, with varying nonlinear outcome signal strength. Generate $X_{ij}\overset{\mathrm{i.i.d.}}{\sim}\mathcal N_{\mathrm{trun},1.2}$ and $A_i\mid\bX_i\sim\mathrm{Bernoulli}\{g(0.36+0.55\bar{\bs1}_{3,d_1}^\top\bX_i-0.2(X_{i1}^2+X_{i2}^2)+0.1X_{i1}X_{i2}+0.1X_{i3}X_{i4})\}$. Generate $\bM_i=0.5(\mathbf B+2\mathbf D_{\mathrm d})^\top\bX_i+\bs\nu_i+\bU_i$, where $\bs\nu_i=(-0.4X_{i1}^2,0.4X_{i2}^2,\allowbreak0.4X_{i1}X_{i2},-0.4X_{i3}X_{i4},\bz_{d_2-4}^\top)^\top$ and $\bU_i\mid(A_i,\bX_i)\sim\mathcal N_{d_2}\{0.4\bar{\bs1}_{3,d_2}+A_i(0.6\bar{\bs1}_{1,d_2}+0.4\bar{\bs1}_{3,d_2}),\bs{\Sigma}_{A_i}\}$. Let $Y_i=1+0.2\bar{\bs1}_{5,d_1}^\top\bX_i+\bar{\bs1}_{3,d_2}^\top\bM_i+\gamma(X_{i1}^2-X_{i1}M_{i1}+X_{i2}M_{i2}+M_{i1}M_{i2})-(1-A_i)+\varepsilon_i$.

\textbf{Setting (vi)}: Correctly specified $\pi$ and $q$, misspecified $(\mu,\tau_n,\tau)$, and weak overlap, with varying sparsity levels $(s_\pi,s_q,s_\mu)$. Generate $X_{ij}\overset{\mathrm{i.i.d.}}{\sim}\mathcal N_{\mathrm{trun},1.2}$, $A_i\mid\bX_i\sim\mathrm{Bernoulli}\{g(1+\sqrt{3/|P_\ell|}\sum_{j\in P_\ell}X_{ij})\}$, and $\bM_i=0.6\bphi_{\mathrm s}^\top\bX_i+\bs\nu_i+\bU_i$, where $\bU_i\mid(A_i,\bX_i)\sim\mathcal N_{d_2}\{0.4\bar{\bs1}_{3,d_2}+A_i\sqrt{0.6/|Q_\ell|}\sum_{k\in Q_\ell}\bs\Sigma_{\cdot k},\bs\Sigma\}$ and $\nu_{ik}=0.4\{X_{i,k+20}^2-\E(X_{i,k+20}^2)\}$ for $31\leq k\leq45$ and zero otherwise. Let $Y_i=1+\bX_i^\top\bbeta_{X,\ell}+\bM_i^\top\bbeta_{M,\ell}+0.4\eta_{p,i}^2-(1-A_i)+\varepsilon_i$, where $\eta_{p,i}=0.7+\sqrt{3/|P_\ell|}\sum_{j\in P_\ell}X_{ij}+\sqrt{0.6/|Q_\ell|}\{\sum_{k\in Q_\ell}[M_{ik}-0.6(\bphi_{\mathrm s}^\top\bX_i)_k]-0.4\}$. Uniformly over $\ell\in[4]$, the empirical range of $\pi(\bX_i)$ covers $[0.054,0.991]$, with at least $13.7\%$ of values below $0.1$ or above $0.9$.

Here, for $\ell\in[4]$, let $(\bbeta_{X,\ell})_j=\sqrt{0.27/s_{\mu,X}}\idf_{\{71\leq j\leq 70+s_{\mu,X}\}}$ for $j\in[d_1]$ and $(\bbeta_{M,\ell})_k=0.31\idf_{\{k\in R_\ell\}}+0.8\idf_{\{k\in\{31,32\}\}}$ for $k\in[d_2]$. Set $\bphi_{\mathrm s}=\mathbf B+c_\ell\sum_{(j,k)\in\mathcal J_\ell}\mathbf E_{jk}$, where $\mathbf E_{jk}\in\R^{d_1\times d_2}$ has $1$ at entry $(j,k)$ and $0$ elsewhere, $c_\ell=\sqrt{0.4}\idf_{\{\ell=3\}}+\sqrt{31/48}\idf_{\{\ell=4\}}$, and $\mathcal J_\ell\subseteq[d_1]\times[d_2]$ satisfies $|\mathcal J_\ell|=5\idf_{\{\ell=3\}}+6\idf_{\{\ell=4\}}$ (exact indices in the replication code).
Across configurations $\ell\in[4]$, the active sets are defined recursively via $(P_\ell,Q_\ell,R_\ell) = (P_{\ell-1}\cup\Delta P_\ell,\, Q_{\ell-1}\cup\Delta Q_\ell,\, R_{\ell-1}\cup\Delta R_\ell)$ with $P_0=Q_0=R_0=\varnothing$, as summarized below:
\begin{center}
\small
\spacingset{1}
\begin{tabular}{cclll}
\toprule
$\ell$ & $s_{\mu,X}$ & $\Delta P_\ell$ & $\Delta Q_\ell$ & $\Delta R_\ell$ \\
\midrule
$1$ & $3$  & $\{2,3,4,30\}$ & $\{3,4,7,8\}$ & $Q_1\cup\{10,12,14\}$ \\
$2$ & $7$  & $\{5,6,7\}$ & $\{10,12\}$ & $\varnothing$ \\
$3$ & $10$ & $\{8,9,10,81,82,83\}$ & $\{14,17\}$ & $\{6,9,17\}$ \\
$4$ & $10$ & $\{11,12,13,14\}$ & $\{20,23\}$ & $\{15,20,23\}$ \\
\bottomrule
\end{tabular}
\end{center}

\textbf{Setting (vii)}: Non-logistic $\pi$ and non-log-linear $q$. Generate $X_{ij}\overset{\mathrm{i.i.d.}}{\sim}\mathcal N_{\mathrm{trun},2}$ and $A_i\mid\bX_i\sim\mathrm{Bernoulli}\{0.1+0.8g(0.25+0.8X_{i1}+0.7L_i+0.75H_i+0.3\sin L_i)\}$. Generate $\bM_i=A_i a(\bX_i)\bs c_M+\bU_i$, where $\bU_i\mid(A_i,\bX_i)\sim\mathcal N_{d_2}\{\bz,\bs\Sigma_U(X_{i1})\}$. Let $Y_i=1+\beta_YX_{i,d_1}+\gamma_YM_{i1}-(1-A_i)+\sigma_Y\varepsilon_i$.

Here, let $v_X=\operatorname{Var}(V)$, $r_X=\P(V>1)$, and $c_X=\{\E[(1+0.45V)^{-1}]\}^{-1}$, where $V\sim\mathcal N_{\mathrm{trun},2}$. Set $L_i=(50v_X)^{-1/2}\sum_{j=4}^{53}X_{ij}$, $H_i=0.75X_{i2}X_{i3}+0.5(X_{i54}^2-v_X)$, and $a(\bX_i)=1.2+0.3\tanh(0.75X_{i2}+0.5X_{i3})$. Let $\bs c_M=0.35\bs e_1+\sqrt{(1-0.35^2)/2}(\bs e_2+\bs e_3)\in\R^{d_2}$ and $\bs h=(\bs e_1-0.35\bs c_M)/\sqrt{1-0.35^2}$. Set $v(x)=0.2+0.8\idf_{\{x>1\}}/r_X$ and $\bs\Sigma_U(x)=\operatorname{diag}\{\bs1_3^\top,[c_X/(1+0.45x)]\bs1_{20}^\top,\bs1_{d_2-23}^\top\}+\{v(x)-1\}\bs h\bs h^\top$. Set $\gamma_Y=\sqrt{1.74}$, $\beta_Y=0.4\gamma_Y$, and $\sigma_Y=0.2$.

\begin{table}[!t]
\footnotesize
\centering
\spacingset{1}
\caption{Simulation results under Setting (iii). The rest of the caption details remain the same as those in Table~\ref{table:config1}.}
\begin{tabular}{lccccclcccc}
\Xhline{1pt}
\multicolumn{5}{c}{$N=600$, $(d_1,d_2)=(100,50)$}&\quad&\multicolumn{5}{c}{$N=1000$, $(d_1,d_2)=(100,50)$}\\
\cmidrule(lr){1-5}\cmidrule(lr){7-11}
\textbf{Method}&\textbf{Bias}&\textbf{RMSE}&\textbf{AL}&\textbf{AC}&\quad&\textbf{Method}&\textbf{Bias}&\textbf{RMSE}&\textbf{AL}&\textbf{AC}\\
Oracle&-0.020&0.201&0.749&0.93&\quad&Oracle&-0.016&0.143&0.585&0.96\\
\hdashline
TR&-0.857&1.087&3.001&0.87&\quad&TR&-0.456&0.650&1.775&0.84\\
O-TR&0.591&0.603&0.576&0.09&\quad&O-TR&0.460&0.486&0.445&0.07\\
\textit{QR1}&0.419&0.569&1.441&0.77&\quad&\textit{QR1}&0.180&0.324&0.975&0.85\\
\textit{QR2}&0.182&0.352&1.207&0.89&\quad&\textit{QR2}&0.160&0.257&0.718&0.79\\
\hdashline
\rowcolor{black!5}
O-TR'&-0.116&0.232&0.686&0.85&\quad&O-TR'&-0.085&0.167&0.535&0.89\\
\rowcolor{black!5}
\textit{MQR1}&0.076&0.261&0.723&0.85&\quad&\textit{MQR1}&0.016&0.220&0.575&0.85\\
\rowcolor{black!5}
\textit{MQR2}&-0.004&0.221&0.722&0.90&\quad&\textit{MQR2}&0.014&0.144&0.557&0.93\\
\Xhline{1pt}
\multicolumn{5}{c}{$N=1400$, $(d_1,d_2)=(100,50)$}&\quad&\multicolumn{5}{c}{$N=1800$, $(d_1,d_2)=(100,50)$}\\
\cmidrule(lr){1-5}\cmidrule(lr){7-11}
\textbf{Method}&\textbf{Bias}&\textbf{RMSE}&\textbf{AL}&\textbf{AC}&\quad&\textbf{Method}&\textbf{Bias}&\textbf{RMSE}&\textbf{AL}&\textbf{AC}\\
Oracle&-0.012&0.120&0.495&0.95&\quad&Oracle&-0.014&0.103&0.436&0.97\\
\hdashline
TR&-0.343&0.424&1.252&0.92&\quad&TR&-0.230&0.344&1.023&0.88\\
O-TR&0.399&0.417&0.374&0.07&\quad&O-TR&0.382&0.387&0.329&0.04\\
\textit{QR1}&0.200&0.262&0.664&0.73&\quad&\textit{QR1}&0.164&0.222&0.558&0.74\\
\textit{QR2}&0.156&0.203&0.548&0.81&\quad&\textit{QR2}&0.150&0.187&0.440&0.72\\
\hdashline
\rowcolor{black!5}
O-TR'&-0.062&0.128&0.452&0.94&\quad&O-TR'&-0.045&0.120&0.400&0.91\\
\rowcolor{black!5}
\textit{MQR1}&-0.010&0.148&0.484&0.90&\quad&\textit{MQR1}&-0.009&0.139&0.409&0.89\\
\rowcolor{black!5}
\textit{MQR2}&0.005&0.115&0.476&0.92&\quad&\textit{MQR2}&0.014&0.098&0.414&0.98\\
\Xhline{1pt}
\end{tabular}
\label{table:config3}
\end{table}

\begin{table}[t!]
\footnotesize
\centering
\spacingset{1}
\caption{Simulation results under Setting (iv). The rest of the caption details remain the same as those in Table~\ref{table:config1}.}
\begin{tabular}{lccccclcccc}
\Xhline{1pt}
\multicolumn{5}{c}{$N=600$, $(d_1,d_2)=(100,1)$}&\quad&\multicolumn{5}{c}{$N=1000$, $(d_1,d_2)=(100,1)$}\\
\cmidrule(lr){1-5}\cmidrule(lr){7-11}
\textbf{Method}&\textbf{Bias}&\textbf{RMSE}&\textbf{AL}&\textbf{AC}&\quad&\textbf{Method}&\textbf{Bias}&\textbf{RMSE}&\textbf{AL}&\textbf{AC}\\
Oracle&0.012&0.167&0.646&0.94&\quad&Oracle&-0.029&0.133&0.503&0.94\\
\hdashline
TR&0.230&0.265&0.587&0.70&\quad&TR&0.153&0.179&0.431&0.74\\
O-TR&0.289&0.319&0.729&0.73&\quad&O-TR&0.151&0.221&0.540&0.77\\
\textit{QR1}&0.277&0.327&0.692&0.73&\quad&\textit{QR1}&0.183&0.234&0.510&0.70\\
\textit{QR2}&0.165&0.243&0.654&0.86&\quad&\textit{QR2}&0.081&0.148&0.463&0.90\\
\hdashline
\rowcolor{black!5}
O-TR'&0.498&0.524&0.725&0.19&\quad&O-TR'&0.440&0.453&0.556&0.05\\
\rowcolor{black!5}
\textit{MQR1}&0.502&0.576&0.771&0.30&\quad&\textit{MQR1}&0.508&0.530&0.608&0.12\\
\rowcolor{black!5}
\textit{MQR2}&0.425&0.452&0.676&0.26&\quad&\textit{MQR2}&0.353&0.383&0.504&0.14\\
\Xhline{1pt}
\multicolumn{5}{c}{$N=1400$, $(d_1,d_2)=(100,1)$}&\quad&\multicolumn{5}{c}{$N=1800$, $(d_1,d_2)=(100,1)$}\\
\cmidrule(lr){1-5}\cmidrule(lr){7-11}
\textbf{Method}&\textbf{Bias}&\textbf{RMSE}&\textbf{AL}&\textbf{AC}&\quad&\textbf{Method}&\textbf{Bias}&\textbf{RMSE}&\textbf{AL}&\textbf{AC}\\
Oracle&-0.024&0.116&0.427&0.99&\quad&Oracle&-0.001&0.100&0.390&0.97\\
\hdashline
TR&0.121&0.162&0.355&0.67&\quad&TR&0.110&0.136&0.304&0.67\\
O-TR&0.125&0.178&0.433&0.79&\quad&O-TR&0.094&0.132&0.375&0.88\\
\textit{QR1}&0.142&0.189&0.417&0.72&\quad&\textit{QR1}&0.110&0.139&0.353&0.79\\
\textit{QR2}&0.043&0.116&0.364&0.85&\quad&\textit{QR2}&0.028&0.088&0.311&0.93\\
\hdashline
\rowcolor{black!5}
O-TR'&0.410&0.425&0.461&0.01&\quad&O-TR'&0.376&0.392&0.394&0.01\\
\rowcolor{black!5}
\textit{MQR1}&0.471&0.513&0.515&0.07&\quad&\textit{MQR1}&0.434&0.456&0.446&0.08\\
\rowcolor{black!5}
\textit{MQR2}&0.349&0.368&0.413&0.05&\quad&\textit{MQR2}&0.332&0.344&0.358&0.04\\
\Xhline{1pt}
\end{tabular}
\label{table:config4}
\end{table}

\begin{table}[!t]
\footnotesize
\centering
\spacingset{1}
\caption{Simulation results under Setting (v) at $N=1800$ and $(d_1,d_2)=(100,50)$. The rest of the caption details remain the same as those in Table~\ref{table:config1}.}
\begin{tabular}{lccccclcccc}
\Xhline{1pt}
\multicolumn{5}{c}{$\gamma=0.25$}&\quad&\multicolumn{5}{c}{$\gamma=0.5$}\\
\cmidrule(lr){1-5}\cmidrule(lr){7-11}
\textbf{Method}&\textbf{Bias}&\textbf{RMSE}&\textbf{AL}&\textbf{AC}&\quad&\textbf{Method}&\textbf{Bias}&\textbf{RMSE}&\textbf{AL}&\textbf{AC}\\
Oracle&0.001&0.097&0.379&0.95&\quad&Oracle&0.001&0.101&0.393&0.94\\
\hdashline
TR&-0.221&0.283&0.708&0.81&\quad&TR&-0.221&0.290&0.765&0.83\\
O-TR&0.323&0.328&0.339&0.22&\quad&O-TR&0.374&0.379&0.373&0.16\\
\textit{QR1}&-0.053&0.159&0.579&0.95&\quad&\textit{QR1}&-0.084&0.187&0.642&0.95\\
\textit{QR2}&-0.020&0.131&0.486&0.95&\quad&\textit{QR2}&-0.035&0.145&0.528&0.94\\
\hdashline
\rowcolor{black!5}
O-TR'&-0.077&0.151&0.381&0.87&\quad&O-TR'&-0.107&0.195&0.463&0.86\\
\rowcolor{black!5}
\textit{MQR1}&-0.011&0.157&0.378&0.87&\quad&\textit{MQR1}&-0.027&0.181&0.422&0.86\\
\rowcolor{black!5}
\textit{MQR2}&-0.053&0.129&0.387&0.92&\quad&\textit{MQR2}&-0.047&0.143&0.426&0.93\\
\Xhline{1pt}
\multicolumn{5}{c}{$\gamma=1$}&\quad&\multicolumn{5}{c}{$\gamma=2$}\\
\cmidrule(lr){1-5}\cmidrule(lr){7-11}
\textbf{Method}&\textbf{Bias}&\textbf{RMSE}&\textbf{AL}&\textbf{AC}&\quad&\textbf{Method}&\textbf{Bias}&\textbf{RMSE}&\textbf{AL}&\textbf{AC}\\
Oracle&-0.003&0.112&0.431&0.93&\quad&Oracle&0.009&0.139&0.528&0.94\\
\hdashline
TR&-0.208&0.303&0.881&0.91&\quad&TR&-0.216&0.368&1.175&0.93\\
O-TR&0.469&0.480&0.450&0.11&\quad&O-TR&0.646&0.668&0.623&0.12\\
\textit{QR1}&-0.142&0.244&0.781&0.93&\quad&\textit{QR1}&-0.268&0.384&1.105&0.92\\
\textit{QR2}&-0.052&0.177&0.635&0.92&\quad&\textit{QR2}&-0.108&0.263&0.902&0.92\\
\hdashline
\rowcolor{black!5}
O-TR'&-0.226&0.326&0.712&0.81&\quad&O-TR'&-0.447&0.605&1.314&0.78\\
\rowcolor{black!5}
\textit{MQR1}&-0.003&0.226&0.555&0.85&\quad&\textit{MQR1}&0.005&0.387&0.872&0.83\\
\rowcolor{black!5}
\textit{MQR2}&-0.022&0.201&0.554&0.91&\quad&\textit{MQR2}&0.051&0.367&0.911&0.92\\
\Xhline{1pt}
\end{tabular}
\label{table:config_decay}
\end{table}

\begin{table}[!t]
\footnotesize
\centering
\spacingset{1}
\caption{Simulation results under Setting (vi) at $N=1800$ and $(d_1,d_2)=(100,50)$. Panel labels give $(s_\pi,s_q,s_\mu)$, where $s_\mu=s_{\mu,X}+|R_\ell|+2$ counts the nonzero coefficients in $\bX_i^\top\bbeta_{X,\ell}+\bM_i^\top\bbeta_{M,\ell}$. The rest of the caption details remain the same as those in Table~\ref{table:config1}.}
\begin{tabular}{lccccclcccc}
\Xhline{1pt}
\multicolumn{5}{c}{$\ell=1$: $(s_\pi,s_q,s_\mu)=(4,12,12)$}&\quad&\multicolumn{5}{c}{$\ell=2$: $(s_\pi,s_q,s_\mu)=(7,18,16)$}\\
\cmidrule(lr){1-5}\cmidrule(lr){7-11}
\textbf{Method}&\textbf{Bias}&\textbf{RMSE}&\textbf{AL}&\textbf{AC}&\quad&\textbf{Method}&\textbf{Bias}&\textbf{RMSE}&\textbf{AL}&\textbf{AC}\\
Oracle&0.029&0.124&0.500&0.96&\quad&Oracle&0.020&0.141&0.513&0.97\\
\hdashline
TR&-0.148&0.197&0.479&0.80&\quad&TR&-0.141&0.218&0.505&0.80\\
O-TR&0.278&0.299&0.337&0.12&\quad&O-TR&0.434&0.448&0.345&0.00\\
\textit{QR1}&-0.074&0.148&0.407&0.84&\quad&\textit{QR1}&-0.008&0.114&0.406&0.93\\
\textit{QR2}&-0.078&0.131&0.368&0.86&\quad&\textit{QR2}&-0.020&0.109&0.372&0.93\\
\hdashline
\rowcolor{black!5}
O-TR'&-0.151&0.190&0.494&0.79&\quad&O-TR'&-0.193&0.222&0.507&0.75\\
\rowcolor{black!5}
\textit{MQR1}&-0.239&0.264&0.466&0.48&\quad&\textit{MQR1}&-0.332&0.395&0.471&0.31\\
\rowcolor{black!5}
\textit{MQR2}&0.015&0.118&0.438&0.93&\quad&\textit{MQR2}&-0.018&0.127&0.451&0.96\\
\Xhline{1pt}
\multicolumn{5}{c}{$\ell=3$: $(s_\pi,s_q,s_\mu)=(13,30,22)$}&\quad&\multicolumn{5}{c}{$\ell=4$: $(s_\pi,s_q,s_\mu)=(17,39,25)$}\\
\cmidrule(lr){1-5}\cmidrule(lr){7-11}
\textbf{Method}&\textbf{Bias}&\textbf{RMSE}&\textbf{AL}&\textbf{AC}&\quad&\textbf{Method}&\textbf{Bias}&\textbf{RMSE}&\textbf{AL}&\textbf{AC}\\
Oracle&0.023&0.154&0.558&0.97&\quad&Oracle&0.024&0.141&0.568&0.97\\
\hdashline
TR&-0.183&0.249&0.601&0.76&\quad&TR&-0.346&0.384&0.650&0.48\\
O-TR&0.582&0.591&0.369&0.00&\quad&O-TR&0.619&0.619&0.380&0.00\\
\textit{QR1}&0.066&0.144&0.459&0.89&\quad&\textit{QR1}&-0.052&0.170&0.551&0.89\\
\textit{QR2}&0.033&0.126&0.417&0.94&\quad&\textit{QR2}&-0.020&0.134&0.484&0.93\\
\hdashline
\rowcolor{black!5}
O-TR'&-0.270&0.293&0.542&0.53&\quad&O-TR'&-0.299&0.317&0.583&0.55\\
\rowcolor{black!5}
\textit{MQR1}&-0.528&0.549&0.488&0.12&\quad&\textit{MQR1}&-0.643&0.678&0.506&0.03\\
\rowcolor{black!5}
\textit{MQR2}&-0.062&0.140&0.489&0.92&\quad&\textit{MQR2}&-0.084&0.169&0.508&0.87\\
\Xhline{1pt}
\end{tabular}
\label{table:config_sparse_new}
\end{table}

\begin{table}[!t]
\footnotesize
\centering
\spacingset{1}
\caption{Simulation results under Setting (vii). The rest of the caption details remain the same as those in Table~\ref{table:config1}.}
\begin{tabular}{lccccclcccc}
\Xhline{1pt}
\multicolumn{5}{c}{$N=1000$, $(d_1,d_2)=(100,50)$}&\quad&\multicolumn{5}{c}{$N=1800$, $(d_1,d_2)=(100,50)$}\\
\cmidrule(lr){1-5}\cmidrule(lr){7-11}
\textbf{Method}&\textbf{Bias}&\textbf{RMSE}&\textbf{AL}&\textbf{AC}&\quad&\textbf{Method}&\textbf{Bias}&\textbf{RMSE}&\textbf{AL}&\textbf{AC}\\
Oracle&0.009&0.087&0.326&0.95&\quad&Oracle&0.001&0.059&0.238&0.95\\
\hdashline
TR&-0.001&0.093&0.299&0.93&\quad&TR&0.026&0.078&0.238&0.91\\
O-TR&0.090&0.113&0.212&0.60&\quad&O-TR&0.073&0.087&0.165&0.59\\
\textit{QR1}&0.107&0.127&0.250&0.64&\quad&\textit{QR1}&0.064&0.086&0.183&0.67\\
\textit{QR2}&0.061&0.092&0.238&0.82&\quad&\textit{QR2}&0.039&0.067&0.177&0.82\\
\hdashline
\rowcolor{black!5}
O-TR'&0.006&0.072&0.262&0.92&\quad&O-TR'&-0.001&0.053&0.209&0.95\\
\rowcolor{black!5}
\textit{MQR1}&0.013&0.064&0.232&0.94&\quad&\textit{MQR1}&0.011&0.051&0.182&0.96\\
\rowcolor{black!5}
\textit{MQR2}&0.008&0.075&0.276&0.93&\quad&\textit{MQR2}&-0.002&0.054&0.227&0.95\\
\Xhline{1pt}
\multicolumn{5}{c}{$N=2500$, $(d_1,d_2)=(100,50)$}&\quad&\multicolumn{5}{c}{$N=5000$, $(d_1,d_2)=(100,50)$}\\
\cmidrule(lr){1-5}\cmidrule(lr){7-11}
\textbf{Method}&\textbf{Bias}&\textbf{RMSE}&\textbf{AL}&\textbf{AC}&\quad&\textbf{Method}&\textbf{Bias}&\textbf{RMSE}&\textbf{AL}&\textbf{AC}\\
Oracle&-0.005&0.058&0.208&0.93&\quad&Oracle&-0.002&0.034&0.147&0.95\\
\hdashline
TR&0.022&0.064&0.212&0.89&\quad&TR&0.022&0.048&0.168&0.91\\
O-TR&0.057&0.076&0.143&0.61&\quad&O-TR&0.046&0.056&0.104&0.57\\
\textit{QR1}&0.055&0.076&0.155&0.71&\quad&\textit{QR1}&0.026&0.048&0.111&0.74\\
\textit{QR2}&0.028&0.061&0.154&0.81&\quad&\textit{QR2}&0.022&0.037&0.111&0.86\\
\hdashline
\rowcolor{black!5}
O-TR'&-0.004&0.054&0.187&0.94&\quad&O-TR'&0.000&0.034&0.138&0.94\\
\rowcolor{black!5}
\textit{MQR1}&0.002&0.049&0.167&0.92&\quad&\textit{MQR1}&0.003&0.031&0.131&0.95\\
\rowcolor{black!5}
\textit{MQR2}&0.001&0.056&0.199&0.95&\quad&\textit{MQR2}&-0.002&0.034&0.146&0.96\\
\Xhline{1pt}
\end{tabular}
\label{table:config7}
\end{table}

Tables~\ref{table:config3}--\ref{table:config7} summarize the results for Settings~(iii)--(vii), respectively. In Setting~(iii), all parametric working models are correctly specified; although the mediator mean $\bs m_2(\bX_i)$ depends densely on the covariates, both $q$ and $\mu$ remain sparse. Across all sample sizes, the parametric modeling methods (O-TR', MQR1, and MQR2) outperform the nonparametric approaches and achieve coverage close to the nominal $95\%$ level as $N$ grows, while MQR2 uniformly attains smaller RMSE than both O-TR' and MQR1, comparable to the oracle estimator. Notably, O-TR' incurs larger bias than MQR2 because it relies on estimating $p$, which is more difficult since $p$ is jointly determined by $\pi$ and $q$ and has a higher sparsity level than $q$. In contrast, the nonparametric methods exhibit relatively large bias or wide confidence intervals due to slow nuisance convergence in high dimensions, while QR2 still achieves the smallest RMSE among the nonparametric approaches.

In Setting~(iv), all parametric modeling methods (O-TR', MQR1, and MQR2) exhibit relatively large bias and RMSE due to high approximation errors under complete model misspecification, though MQR2 still attains smaller estimation error than O-TR' and MQR1. In comparison, QR2 achieves the best overall performance, with the smallest RMSE and coverage close to the nominal $95\%$ level across all $N$. This highlights both the advantage of the QR framework over the existing TR and O-TR frameworks and the benefit of nonparametric modeling when the mediator is low-dimensional and the true DGP deviates from the assumed parametric models.

In Setting~(v), all parametric working models are misspecified, and we vary the strength of the nonlinear outcome signal $\gamma$. Under this setting, the O-TR and O-TR' methods exhibit substantial bias due to the nested regression used to estimate $\tau$. The initial estimator $\hat\mu$ is fitted using observations from the $A=1$ group and is therefore optimized to reduce approximation error within that group. However, the subsequent nested regression is performed using the $A=0$ group, so extrapolation of the fitted outcome model from $A=1$ to $A=0$ can induce substantial bias in estimating $\tau$. Among the nonparametric and parametric methods, QR2 and MQR2 achieve the best performance, respectively. By incorporating the conditional density ratio $q$ into the DR estimation of $\tau$, both methods effectively mitigate this extrapolation error and achieve smaller estimation bias across values of $\gamma$. When $\gamma$ is relatively small, MQR2 yields a slightly smaller RMSE than QR2; when $\gamma$ is large and model misspecification becomes more severe, the nonparametric approach QR2 instead achieves the smaller RMSE.

In Setting~(vi), despite weak overlap, QR2 and MQR2 remain relatively stable and consistently achieve the best performance among the nonparametric and parametric methods, respectively. When $\ell=1$ and the sparsity levels are relatively low, the regularized parametric method MQR2 performs best overall. In contrast, when $\ell\in\{2,3,4\}$ and the models become denser, QR2 achieves the smallest RMSE.

In Setting~(vii), the nonparametric methods (QR1, QR2, O-TR, and TR) exhibit relatively large bias and RMSE, together with undercoverage, due to slow nuisance convergence in high dimensions, although QR2 achieves the highest estimation accuracy among them. In contrast, all parametric methods (O-TR', MQR1, and MQR2) maintain smaller bias and coverage close to the nominal $95\%$ level across all values of $N$. Under this setting, MQR1 achieves a smaller RMSE than MQR2, as well as O-TR', illustrating a regime in which the additional sample splitting in Algorithm~\ref{alg:MQR} leads to more accurate estimation than Algorithm~\ref{alg:mqr2}. This occurs because the regression functions $\mu$, $\tau_n$, and $\tau$ have simple sparse structures with low noise levels, resulting in small stochastic errors that offset the efficiency loss from additional sample splitting in MQR1. Meanwhile, $\pi$ contains a dense signal and is therefore more difficult to estimate accurately. As shown in Theorems~\ref{thm:r_bias} and~\ref{thm:mqr2_bias}, together with the comparison in Table~\ref{table:mqr2_convergence}, under $\mathcal M_b'$, the additional sample splitting in MQR1 allows the correctly specified regression models to compensate for the stochastic error in $\pi$, whereas in MQR2 this error propagates through the sequential estimation steps and contributes an $r_\pi^2$ bias term. Moreover, although estimating $\pi$ on smaller sub-folds may reduce its pointwise accuracy, the resulting stronger regularization produces less variable weights and further stabilizes the estimator. Lastly, O-TR' performs comparably to MQR2, although still less favorably than MQR1, and achieves near-nominal coverage despite the lack of a theoretical guarantee that its bias is asymptotically negligible under model misspecification.

\subsection{Implementation details for the real data analysis}\label{apdx:detail_real}

The original mediator vector $\bM$ consists of the log-transformed CD4 (CD420) and CD8 (CD820) cell counts at week 20 ($\pm$ 5 weeks). The original baseline covariate vector $\bX = (\bX_{\mathrm{c}}^\top, \bX_{\mathrm{b}}^\top)^\top$ includes five continuous variables $\bX_{\mathrm{c}}$, namely age, weight, Karnofsky score (karnof), and baseline CD4 (CD40) and CD8 (CD80) counts, as well as seven binary indicators $\bX_{\mathrm{b}}$ for hemophilia (hemo), homosexual activity (homo), history of intravenous drug use (drugs), race, gender, antiretroviral history (str2), and symptomatic status (symptom).

To improve the flexibility of parametric modeling methods (O-TR' and MQR2), we expand the baseline covariate vector $\bX$ to include:

(i) the linear and quadratic terms of $\bX_{\mathrm{c}}$:
\begin{align*}
    \mathrm{I} := (\mathrm{age} + \mathrm{weight} + \mathrm{karnof} + \mathrm{CD40} + \mathrm{CD80} + 1)^2;
\end{align*}

(ii) the linear and interaction terms within $\bX_{\mathrm{b}}$:
\begin{align*}
    \mathrm{II} &:= (\mathrm{hemo} + \mathrm{homo} + \mathrm{drugs} + \mathrm{race} + \mathrm{gender} + \mathrm{str2} + \mathrm{symptom} + 1)^2 \\
    &\quad - \mathrm{hemo}^2 - \mathrm{homo}^2 - \mathrm{drugs}^2 - \mathrm{race}^2 - \mathrm{gender}^2 - \mathrm{str2}^2 - \mathrm{symptom}^2;
\end{align*}

(iii) the interactions between $\bX_{\mathrm{b}}$ and components of (i):
\begin{align*}
    \mathrm{III} := \mathrm{I} \times (\mathrm{hemo} + \mathrm{homo} + \mathrm{drugs} + \mathrm{race} + \mathrm{gender} + \mathrm{str2} + \mathrm{symptom}).
\end{align*}

The quadratic terms of the binary variables are excluded, as they are identical to the linear terms. The expanded covariate set is defined as $\mathrm{I} + \mathrm{II} + \mathrm{III}$, with a total dimension of 188.

Similarly, the mediator vector $\bM$ is augmented with its linear and quadratic components, as well as interactions with (i) and (ii), i.e.,
\begin{align*}
    (\mathrm{CD420} + \mathrm{CD820} + 1)^2 + (\mathrm{CD420} + \mathrm{CD820}) \times (\mathrm{I} + \mathrm{II}).
\end{align*}
This results in an expanded mediator set of dimension 101.

The parametric modeling methods use the expanded bases in nuisance estimation to enhance flexibility, whereas the nonparametric modeling methods use the original covariates and mediators. Nuisance estimation, including tuning and stabilization, follows Section~\ref{apdx:implementation}, except for the conditional mediator density estimator used by TR. Within each treatment group, we instead estimate $f(\bM\mid A,\bX)$ by applying Gaussian product kernels to the mediator values of nearest neighbors in standardized covariate space. The number of neighbors is selected by hold-out cross-validation. Within each five-fold split, the asymptotic variance for $\theta_{1,0}$ is estimated as in Section~\ref{apdx:implementation}, while those for the NDE, NIE, and ATE are estimated from the sample variances of the differences between the corresponding cross-fitted scores. Across the $100$ splits, given split-specific point and asymptotic variance estimates $\{(\hat\vartheta_t,\hat v_t^2)\}_{t=1}^{100}$ for a target effect, we report $\hat\vartheta^{\mathrm{med}}=\mathrm{median}_{t\in[100]}(\hat\vartheta_t)$ and construct the $95\%$ Wald confidence intervals and two-sided $p$-values in Table~\ref{table:ACTG} using the median-adjusted asymptotic variance estimate $\mathrm{median}_{t\in[100]}\{\hat v_t^2+(\hat\vartheta_t-\hat\vartheta^{\mathrm{med}})^2\}$ \citep{DoubleDebiasedMachine2018}.

\section{Proofs of the results for the quadruply robust (QR) estimator}
\label{apdx:QR}
This section focuses on the QR estimator $\hat{\theta}_{\mathrm{QR}}$ for $\theta$, supplementing the results established in Section~\ref{sec:QR}. Note that the construction of $\hat{\theta}_{\mathrm{QR}}$ and the choices of parameters associated with the sample splitting procedure follow Algorithm~\ref{alg:QR}.

\subsection{Auxiliary lemmas}\label{apdx:QR_aux}
\begin{lemma}
Under Assumptions~\ref{cond:basic} and \ref{cond:QR_basic} (a), we have\label{lem:mu_relation}
\begin{align*}
\normp{\mu^*(\bzS)-\hat{\mu}(\bzS)}{2}\leq c_0^{-1/2}\normp{\mu^*(\bS)-\hat{\mu}(\bS)}{2}.
\end{align*}
\end{lemma}

\begin{lemma}\label{lem:QR_lower}
Under Assumptions~\ref{cond:basic} and \ref{cond:QR_basic}, we have

\emph{(a)} $\max\{\E[\varepsilon^2\mid \bS]\text{, }\E[\varrho^2\mid \bX]\text{, }\E[\zeta^2\mid \bX]\}\leq G_0^{2/r}$ holds almost surely;

\emph{(b)} $\max\{\E[\varepsilon^2]\text{, }\E[\varrho^2]\text{, }\E[\zeta^2]\text{, }\E[\xi^2]\}\leq G_0^{2/r}$;

\emph{(c)} $\ssigma^2\leq 64c_0^{-6}G_0^{2/r}$.
\end{lemma}

\begin{lemma}\label{lem:QR_consistent_lower}
Let Assumptions~\ref{cond:basic} and \ref{cond:QR_basic} hold. Then for any $t>0$,
\begin{align}
    \E[|\bpsi^*(\bW)-\theta|^{2+t}]&\leq \frac{2^{4+3t}}{c_0^{6+3t}}\E[|\xi|^{2+t}+|\varrho|^{2+t}+|\zeta|^{2+t}+|\varepsilon|^{2+t}].\label{bound:QR_lyapunov}
\end{align}
\end{lemma}

\subsection{Proof of the results in Section \ref{sec:dr_tau}}
\label{apdx:dr_rep}
\begin{proof}[Proof of Theorem~\ref{thm:dr_tau}]

Recall that $f_0(\bS)=f(\bM\mid A=0,\bX)$, $\mu(\bS)=\E[Y\mid A=1,\bS]$ and $\tau(\bX)=\E[\mu(\bS)\mid A=0,\bX]$. For any $\bx\in\R^{d_1}$, $\bs{m}\in\mathcal{S}_M$ and $\bs{s}=(\bx^\top,\bs{m}^\top)^\top$, we have
\begin{align*}
    f_0(\bs{s})&\overset{(i)}{=}f(\bM(0)=\bs{m}\mid A=0,\bX=\bx)\overset{(ii)}{=}f(\bM(0)=\bs{m}\mid \bX=\bx)\\
    &\overset{(iii)}{=}f(\bM(0)=\bs{m}\mid A=1,\bX=\bx),
\end{align*}
where (i) follows from Assumption~\ref{cond:basic} (a); (ii) and (iii) hold by Assumption~\ref{cond:basic} (b). Similarly,
\begin{align*}
    &\E[Y(1,\bM(0))\mid \bM(0)=\bs{m}, A=1,\bX=\bx]\\
    &\overset{(i)}{=}\E[Y(1,\bs{m})\mid \bM(0)=\bs{m}, A=1,\bX=\bx]\overset{(ii)}{=}\E[Y(1,\bs{m})\mid A=1,\bX=\bx]\\
    &\overset{(iii)}{=}\E[Y(1,\bs{m})\mid \bM=\bs{m}, A=1,\bX=\bx]=\mu(\bs{s}),
\end{align*}
where (i) holds by Assumption~\ref{cond:basic} (a); (ii) and (iii) follow from $Y(1,\bs{m})\independent \{\bM(0),\bM(1)\}\mid A=1,\bX$.
By the tower rule,
\begin{align*}
    &\E[Y(1,\bM(0))\mid A=1,\bX=\bx]\\
    =&\int_{\mathcal{S}_M}\E[Y(1,\bM(0))\mid \bM(0)=\bs{m}, A=1,\bX=\bx]\,f(\bM(0)=\bs{m}|A=1,\bX=\bx)\,\d \bs{m}\\
    =&\int_{\mathcal{S}_M}\mu(\bs{s})\,f_0(\bs{s})\,\d \bs{m}=\tau(\bx).
\end{align*}

Next, we show the double robustness of the proposed representation.
Recall the definition $\tau_{n}(\bX)=\E[\mu^*(\bS)\mid A=0,\bX]$.
When $\mu^*=\mu$, we have $\tau_{n}(\bX)=\E[\mu^*(\bS)\mid A=0,\bX]=\E[\mu(\bS)\mid A=0,\bX]=\tau(\bX),$ such that
\begin{align*}
    &\tau_{n}(\bX)+\E[q^*(\bS)\{Y-\mu^*(\bS)\}\mid A=1, \bX]\\
    =&\tau(\bX)+\E[q^*(\bS)\{Y-\mu(\bS)\}\mid A=1,\bX]\\
    \overset{(i)}{=}&\tau(\bX)+\E\left[q^*(\bS)\{\E[Y\mid A=1,\bM,\bX]-\mu(\bS)\}\mid A=1,\bX\right]\\
    \overset{(ii)}{=}&\tau(\bX),
\end{align*}
where (i) holds by the tower rule, and (ii) holds by the definition of $\mu$.

Otherwise, if $q^*=q$, we have
\begin{align*}
    &\tau_{n}(\bX)+\E[q^*(\bS)\{Y-\mu^*(\bS)\}\mid A=1, \bX]\\
    \overset{(i)}{=}&\tau_{n}(\bX)+\E\left[q(\bS)\E\big[Y-\mu^*(\bS)\mid A=1,\bM,\bX\big]\mid A=1,\bX\right]\\
    \overset{(ii)}{=}&\tau_{n}(\bX)+\E[\mu(\bS)-\mu^*(\bS)\mid A=0,\bX]\\
    \overset{(iii)}{=}&\tau_n(\bX)+\tau(\bX)-\tau_n(\bX)=\tau(\bX),
\end{align*}
where (i) holds by the tower rule; (ii) follows from the definition of $\mu$; (iii) holds by the definition of $\tau_n$ and $\tau$.
\end{proof}

The following lemma validates that the estimator defined in \eqref{def:tauhat} imposes no extra requirement for correct model specification, as its population-level limit $\tau^*$ satisfies $\tau^*=\tau$ when $\mu^*=\mu$ and $\tau_n^*=\tau_n$.
\begin{lemma}
\label{lem:compatible}
Let Assumption~\ref{cond:basic} hold. Consider nuisance regression functions $\tau_n^*$ and $\tau^*$ defined as
\begin{align*}
    \tau_n^*(\bX)&\in\mathop{\mathrm{argmin}}\limits_{\tau_n\in\mathcal{G}} \E[(1-A)\{\mu^*(\bS)-\tau_n(\bX)\}^2],\\
    \tau^*(\bX)&\in\mathop{\mathrm{argmin}}\limits_{\tau\in\mathcal{G}} \E\big[A\big\{q^*(\bS)\{Y-\mu^*(\bS)\}+\tau_n^*(\bX)-\tau(\bX)\big\}^2\big].
\end{align*}
If $\mu^*(\bS)=\mu(\bS)$, then $\tau^*(\bX)=\tau_{n}^*(\bX)$ almost surely. Additionally, if $\tau_n^*=\tau_n$, we have $\tau^*=\tau$.
\end{lemma}

\begin{proof}[Proof of Lemma~\ref{lem:compatible}]
If $\mu^*(\bS)=\mu(\bS)$, then for any $f\in\mathcal{G}$,
\begin{align*}
    &\E\big[A\big\{q^*(\bS)\{Y-\mu^*(\bS)\}+\tau_{n}^*(\bX)-f(\bX)\big\}^2\big]\\
    =&\E[A\{q^*(\bS)\{Y-\mu^*(\bS)\}\}^2]+2\E[Aq^*(\bS)\{Y-\mu^*(\bS)\}\{\tau_{n}^*(\bX)-f(\bX)\}]\\
    &+\E[A\{\tau_{n}^*(\bX)-f(\bX)\}^2],
\end{align*}
where the interaction term is zero, since
\begin{align*}
    &\E[Aq^*(\bS)\{Y-\mu^*(\bS)\}\{\tau_{n}^*(\bX)-f(\bX)\}]\\
    =&\E\big[Aq^*(\bS)\E[Y-\mu^*(\bS)\mid A=1,\bS]\{\tau_{n}^*(\bX)-f(\bX)\}\big]\\
    =&\E\big[Aq^*(\bS)\E[Y-\mu(\bS)\mid A=1,\bS]\{\tau_{n}^*(\bX)-f(\bX)\}\big]\\
    =&0.
\end{align*}
Therefore, $\tau^*$ is the minimizer of $\E[A\{\tau_{n}^*(\bX)-f(\bX)\}^2]$ over $f\in\mathcal{G}$. By the tower rule,
\begin{align*}
    \E[A\{\tau_{n}^*(\bX)-f(\bX)\}^2]=\E[\pi(\bX)\{\tau_{n}^*(\bX)-f(\bX)\}^2].
\end{align*}
Since $\tau_n^*\in\mathcal{G}$ and $\P(c_0\leq\pi(\bX)\leq1-c_0)=1$, any minimizer $\tau^*$ satisfies $\tau^*(\bX)=\tau_{n}^*(\bX)$ almost surely. Additionally, if $\tau_n^*=\tau_n$, then $\tau_n(\bX)=\E[\mu^*(\bS)\mid A=0,\bX]=\E[\mu(\bS)\mid A=0,\bX]=\tau(\bX)$. Therefore
$\tau^*(\bX)=\tau_n^*(\bX)=\tau_n(\bX)=\tau(\bX)$.
\end{proof}

\begin{proof}[Proof of Theorem~\ref{thm:QR_robust}]
Recall that $\bpsi^*(\bW)=\bpsi^\dagger(\bW)-\varDelta^*(\bW)$; the proof for every case is provided below.

\textbf{(a)} If $\mu^*=\mu$ and $\pi^*=\pi$, then
\begin{align*}
    \E[\bpsi^\dagger(\bW)]
    \overset{(i)}=&\E\left[\frac{\pi(\bX)}{\pi(\bX)}\E\big[\E\left[q^*(\bS)\{Y-\mu(\bS)\}\mid A=1,\bS\right]\mid A=1,\bX\big]\right]\\
    &+\E\left[\frac{1-\pi(\bX)}{1-\pi(\bX)}\E\left[\mu(\bS)-\tau^*(\bX)\mid A=0,\bX\right]\right]+\E[\tau^*(\bX)]\\
    \overset{(ii)}{=}&0+\E[\tau(\bX)-\tau^*(\bX)]+\E[\tau^*(\bX)]=\theta,
\end{align*}
where (i) holds by the tower rule, and (ii) holds by the definition of $\mu(\bS)$. Additionally,
\begin{align*}
    \E[\varDelta^*(\bW)]&=\E\left[\E\left[\frac{A}{\pi(\bX)}-\frac{1-A}{1-\pi(\bX)}\middle|\bX\right]\{\tau^*(\bX)-\tau_{n}^*(\bX)\}\right]=0.
\end{align*}
Hence, $\E[\bpsi^*(\bW)]=\E[\bpsi^\dagger(\bW)]-\E[\varDelta^*(\bW)]=\theta$.

\textbf{(b)} If $\mu^*=\mu$ and $\tau_n^*=\tau_n$, then Lemma~\ref{lem:compatible} implies $\tau^*=\tau_n^*=\tau_n=\tau$ almost surely, and we have
\begin{align*}
    \E[\bpsi^\dagger(\bW)]
    \overset{(i)}=&\E\left[\frac{\pi(\bX)}{\pi^*(\bX)}\E\big[\E\left[q^*(\bS)\{Y-\mu(\bS)\}\mid A=1,\bS\right]\mid A=1,\bX\big]\right]\\
    &+\E\left[\E\left[\frac{1-\pi(\bX)}{1-\pi^*(\bX)}\{\mu(\bS)-\tau(\bX)\}\mid A=0,\bX\right]\right]+\E[\tau(\bX)]\\
    \overset{(ii)}{=}&0+0+\E[\tau(\bX)]=\theta,
\end{align*}
where (i) holds by the tower rule, and (ii) follows from the definitions of $\mu(\bS)$ and $\tau(\bX)$. Additionally, since $\tau^*(\bX)=\tau_n^*(\bX)$ almost surely,
\begin{align*}
    \E[\varDelta^*(\bW)]&=\E\left[\E\left[\frac{A}{\pi^*(\bX)}-\frac{1-A}{1-\pi^*(\bX)}\middle |\bX\right]\{\tau^*(\bX)-\tau_n^*(\bX)\}\right]=0.
\end{align*}
Therefore, $\E[\bpsi^*(\bW)]=\E[\bpsi^\dagger(\bW)]-\E[\varDelta^*(\bW)]=\theta$.

\textbf{(c)} If $\pi^*=\pi$ and $q^*=q$, then
\begin{align*}
    \E[\bpsi^\dagger(\bW)]
    &\overset{(i)}=\E\left[\frac{\pi(\bX)}{\pi(\bX)}\E\big[\E\left[q(\bS)\{Y-\mu^*(\bS)\}\mid A=1,\bS\right]\mid A=1,\bX\big]\right]\\
    &\;\;\;+\E\left[\frac{1-\pi(\bX)}{1-\pi(\bX)}\E\left[\mu^*(\bS)-\tau^*(\bX)\mid A=0,\bX\right]\right]+\E[\tau^*(\bX)]\\
    &\overset{(ii)}{=}\E[\tau(\bX)-\tau_{n}(\bX)]+\E[\tau_{n}(\bX)-\tau^*(\bX)]+\E[\tau^*(\bX)]=\theta,
\end{align*}
where (i) holds by the tower rule, and (ii) holds by the definition of $\tau_{n}$ and the identity $\tau(\bX)=\E[q(\bS)Y\mid A=1,\bX]$. Recalling the proof of (a), $\E[\varDelta^*(\bW)]=0$ holds when $\pi^*=\pi$. Thus, $\E[\bpsi^*(\bW)]=\E[\bpsi^\dagger(\bW)]-\E[\varDelta^*(\bW)]=\theta$.

\textbf{(d)} When $q^*=q$ and $(\tau_{n}^*,\tau^*)=(\tau_{n},\tau)$ hold, we have
\begin{align*}
    \E[\bpsi^\dagger(\bW)]
    \overset{(i)}=&\E\left[\frac{\pi(\bX)}{\pi^*(\bX)}\E\big[\E\left[q(\bS)\{Y-\mu^*(\bS)\}\mid A=1,\bS\right]\mid A=1,\bX\big]\right]\\
    &+\E\left[\E\left[\frac{1-\pi(\bX)}{1-\pi^*(\bX)}\{\mu^*(\bS)-\tau(\bX)\}\mid A=0,\bX\right]\right]+\E[\tau(\bX)]\\
    \overset{(ii)}=&\E\left[\frac{\pi(\bX)}{\pi^*(\bX)}\{\tau(\bX)-\tau_n(\bX)\}\right]\\
    &+\E\left[\frac{1-\pi(\bX)}{1-\pi^*(\bX)}\{\tau_n(\bX)-\tau(\bX)\}\right]+\E[\tau(\bX)]\\
    =&\E\left[\left\{\frac{\pi(\bX)}{\pi^*(\bX)}-\frac{1-\pi(\bX)}{1-\pi^*(\bX)}\right\}\{\tau(\bX)-\tau_{n}(\bX)\}\right]+\theta\\
    \overset{(iii)}=&\E[\varDelta^*(\bW)]+\theta,
\end{align*}
where (i) holds by the tower rule; (ii) follows from the definition of $\tau_n$ and the identity $\tau(\bX)=\E[q(\bS)Y\mid A=1,\bX]$; (iii) holds since
\begin{align*}
    \E[\varDelta^*(\bW)]=&\E\left[\E\left[\frac{A}{\pi^*(\bX)}-\frac{1-A}{1-\pi^*(\bX)}\middle|\bX\right]\{\tau^*(\bX)-\tau_{n}^*(\bX)\}\right]\\
    =&\E\left[\left\{\frac{\pi(\bX)}{\pi^*(\bX)}-\frac{1-\pi(\bX)}{1-\pi^*(\bX)}\right\}\{\tau(\bX)-\tau_{n}(\bX)\}\right].
\end{align*}
Hence, $\E[\bpsi^*(\bW)]=\E[\bpsi^\dagger(\bW)]-\E[\varDelta^*(\bW)]=\theta$.
\end{proof}

\subsection{Proof of the results in Section \ref{sec:QR_thm}}\label{apdx:QR_thm}

Following the construction of $\hat{\theta}_{\mathrm{QR}}$ in Algorithm~\ref{alg:QR}, we have
\begin{align}
    \hat{\theta}_{\mathrm{QR}}-\theta&=N^{-1}\sum\limits_{k=1}^{\K}\sum\limits_{j\in\I_k}\{\hat{\bpsi}^{(-k)}(\bW_j)-\bpsi^*(\bW_j)\}
    +N^{-1}\sum\limits_{k=1}^{\K}\sum\limits_{j\in\I_k}\{\bpsi^*(\bW_j)-\theta\}\nonumber\\
    &=\Delta_{0}+\K^{-1}\sum\limits_{k=1}^{\K}\Delta_{k,1}+\K^{-1}\sum\limits_{k=1}^{\K}\Delta_{k,2}+\K^{-1}\sum\limits_{k=1}^{\K}\Delta_{k,3},\label{decom:QR_theta}
\end{align}
where
\begin{align}
    \Delta_{0}\;\;&=\E[\phi^*(\bW)]-\theta,\label{def:QE_delta_0}\\
    \Delta_{k,1}&=n^{-1}\sum\limits_{j\in\I_k}\{\bpsi^*(\bW_j)-\E[\bpsi^*(\bW)]\},\label{def:QR_delta_k_1}\\
    \Delta_{k,2}&=n^{-1}\sum\limits_{j\in\I_k}\left\{\hat{\bpsi}^{(-k)}(\bW_j)-\bpsi^*(\bW_j)-\E[\hat{\bpsi}^{(-k)}(\bW_j)-\bpsi^*(\bW_j)]\right\},\label{def:QR_delta_k_2}\\
    \Delta_{k,3}&=n^{-1}\sum\limits_{j\in\I_k}\E[\hat{\bpsi}^{(-k)}(\bW_j)-\bpsi^*(\bW_j)]=\E[\hat{\bpsi}^{(-k)}(\bW)-\bpsi^*(\bW)].\label{def:QR_delta_k_3}
\end{align}
The term $\Delta_{k,1}$ represents the oracle stochastic error that would remain even if the true nuisance functions were known, and it is naturally controlled through the standard deviation $\ssigma$ of the score function. The remaining terms, $\Delta_0$, $\Delta_{k,2}$ and $\Delta_{k,3}$, are analyzed in Lemmas~\ref{lem:Delta_0}, \ref{lem:Delta_k_2} and \ref{lem:Delta_k_3}, respectively.

\begin{lemma}\label{lem:Delta_0}
Let Assumptions \ref{cond:basic} and \ref{cond:QR_basic} hold. Then we have
$
    \Delta_0=O(\eqr_\pi\{\eqr_{\tau_n}+\eqr_\tau\}+\eqr_q\eqr_\mu),
$
where $\Delta_0$ is defined in \eqref{def:QE_delta_0}.
\end{lemma}
\begin{proof}
Noting that $\theta=\E[\tau(\bX)]$ and $\E[\mu^*(\bzS)]=\E[\tau_n(\bX)]$, we have
\begin{align*}
    \Delta_0=\E[\phi^*(\bW)]-\theta=\E[I_1+I_2+I_3],
\end{align*}
where
\begin{align*}
    I_1&=\frac{A}{\pi^*(\bX)}\left\{q^*(\bS)\{Y(1,\bM)-\mu^*(\bS)\}+\tau_n(\bX)-\tau(\bX)\right\},\\
    I_2&=\left\{\frac{1-A}{1-\pi^*(\bX)}-1\right\}\{\mu^*(\bzS)-\tau_n^*(\bX)\},\\
    I_3&=\frac{A}{\pi^*(\bX)}\{\tau_n^*(\bX)-\tau^*(\bX)-\tau_n(\bX)+\tau(\bX)\} +\mu^*(\bS_0)-\tau_n^*(\bX)+\tau^*(\bX)-\tau(\bX).
\end{align*}

Note that under Assumption~\ref{cond:QR_basic} (a), $\P(c_0\leq\pi^*(\bX)\leq1-c_0)=1$ and $\P(c_0^2\leq q^*(\bS)\leq c_0^{-2})=1$.
Recalling \eqref{eq:repre_taun} and \eqref{eq:repre_tau}, under Assumption~\ref{cond:QR_basic}, we have
\begin{align*}
    |\E[I_1]|&=\left|\E\left[\frac{A}{\pi^*(\bX)}\left\{q^*(\bS)\{Y(1,\bM)-\mu^*(\bS)\}+q(\bS)\{\mu^*(\bS)-\mu(\bS)\}\right\}\right]\right|\\
    &=\left|\E\left[\frac{A}{\pi^*(\bX)}\left\{q^*(\bS)\{\mu(\bS)-\mu^*(\bS)\}+q(\bS)\{\mu^*(\bS)-\mu(\bS)\}\right\}\right]\right|\\
    &\leq c_0^{-1}\normp{q^*-q}{2}\normp{\mu^*-\mu}{2}=O(\eqr_q\eqr_\mu)
\end{align*}
and
\begin{align*}
    |\E[I_2]|&=\left|\E\left[\left\{\frac{1-A}{1-\pi^*(\bX)}-1\right\}\{\mu^*(\bzS)-\tau_n^*(\bX)\}\right]\right|\\
    &=\left|\E\left[\frac{\pi^*(\bX)-\pi(\bX)}{1-\pi^*(\bX)}\{\tau_n(\bX)-\tau_n^*(\bX)\}\right]\right|\\
    &\leq c_0^{-1}\normp{\pi^*-\pi}{2}\normp{\tau_n^*-\tau_n}{2}=O(\eqr_\pi\eqr_{\tau_n}).
\end{align*}
Similarly, we have
\begin{align*}
    |\E[I_3]|&=\Big|\E\Big[\frac{A}{\pi^*(\bX)}\{\tau_n^*(\bX)-\tau^*(\bX)-\tau_n(\bX)+\tau(\bX)\}\\
    &\qquad\quad+\tau_n(\bX)-\tau_n^*(\bX)+\tau^*(\bX)-\tau(\bX)\Big]\Big|\\
    &=\left|\E\left[\frac{\pi(\bX)-\pi^*(\bX)}{\pi^*(\bX)}\{\tau_n^*(\bX)-\tau^*(\bX)-\tau_n(\bX)+\tau(\bX)\}\right]\right|\\
    &\leq c_0^{-1}\normp{\pi-\pi^*}{2}\left\{\normp{\tau_n^*-\tau_n}{2}+\normp{\tau^*-\tau}{2}\right\}\\
    &=O(\eqr_\pi \{\eqr_{\tau_n}+\eqr_{\tau}\}).
\end{align*}
Collectively, we obtain $\Delta_0=O(\eqr_\pi\{\eqr_{\tau_n}+\eqr_\tau\}+\eqr_q\eqr_\mu)$.
\end{proof}

\begin{lemma}\label{lem:Delta_k_2}
Let Assumptions~\ref{cond:basic} and \ref{cond:QR_basic} hold. For $\Delta_{k,2}$ defined in \eqref{def:QR_delta_k_2}, we have that for each $k\in[\K]$, as $N\to\infty$,
    \begin{align*}
        \Delta_{k,2}=O_p\left(n^{-1/2}\{\br_{\mu}+\br_{\tau}+\br_{\tau_{n}}+\br_{\pi}+\br_{q}\}\right).
    \end{align*}
\end{lemma}
\begin{proof}
By Chebyshev's inequality, it suffices to show that
\begin{align*}
    \{\E_{\mathbb S_k}[\Delta_{k,2}^2]\}^{1/2}=O_p\left(n^{-1/2}\{\br_{\mu}+\br_{\tau}+\br_{\tau_{n}}+\br_{\pi}+\br_{q}\}\right).
\end{align*}
Note that
\begin{align*}
    \{\E_{\mathbb S_k}[\Delta_{k,2}^2]\}^{1/2}\leq n^{-1/2}\sqrt{\Var[\hat{\bpsi}^{(-k)}(\bW)-\bpsi^*(\bW)]}\leq n^{-1/2}\sqrt{\E[\{\hat{\bpsi}^{(-k)}(\bW)-\bpsi^*(\bW)\}^2]}.
\end{align*}
We consider $\hat{\bpsi}^{(-k)}(\bW)-\bpsi^*(\bW)=\sum_{l=1}^{4}T_l$, where
\begin{align*}
    T_1&:=\frac{A}{\hat{\pi}(\bX)}\hat{q}(\bS)\{Y-\hat{\mu}(\bS)\}-\frac{A}{\pi^*(\bX)}q^*(\bS)\{Y-\mu^*(\bS)\},\\
    T_2&:=\frac{1-A}{1-\hat{\pi}(\bX)}\{\hat{\mu}(\bS)-\hat{\tau}_n(\bX)\}-\frac{1-A}{1-\pi^*(\bX)}\{\mu^*(\bS)-\tau^*_n(\bX)\},\\
    T_3&:=\frac{A}{\hat{\pi}(\bX)}\{\hat{\tau}_n(\bX)-\hat{\tau}(\bX)\}-\frac{A}{\pi^*(\bX)}\{\tau_{n}^*(\bX)-\tau^*(\bX)\},\\
    T_4&:=\hat{\tau}(\bX)-\tau^*(\bX).
\end{align*}
Next, we analyze these terms sequentially. The following arguments condition on the event
\begin{align}
    \mathcal{E}_0:=\{\P(c_0\leq \hat{\pi}(\bX)\leq 1-c_0)=1\text{, } \P(c_0^{2}\leq \hat{q}(\bS)\leq c_0^{-2})=1\},\label{def:E0}
\end{align}
which occurs with probability approaching $1$ under Assumption~\ref{cond:QR_basic}(a).
Note that by Lemma~\ref{lem:QR_lower}, we have $\P(\sqrt{\E[\varepsilon^2\mid\bS]}\leq G_0^{1/r})=1$, $\P(\sqrt{\E[\varrho^2\mid\bX]}\leq G_0^{1/r})=1$ and $\P(\sqrt{\E[\zeta^2\mid\bX]}\leq G_0^{1/r})=1$.
For the first term $T_1$, we have
\begin{align*}
T_1
&=\frac{A\hat q(\bS)}{\hat\pi(\bX)}\{\mu^*(\bS)-\hat\mu(\bS)\}
 +A\left\{\frac{\hat q(\bS)}{\hat\pi(\bX)}-\frac{q^*(\bS)}{\pi^*(\bX)}\right\}\varepsilon,\\
\normp{T_1}{2}
&\leq c_0^{-3}\normp{\hat\mu-\mu^*}{2}
 +G_0^{1/r}\left\{c_0^{-1}\normp{\hat q-q^*}{2}
 +c_0^{-4}\normp{\hat\pi-\pi^*}{2}\right\},
\end{align*}
where the equality follows from $A\{Y-\mu^*(\bS)\}=A\varepsilon$ under Assumption~\ref{cond:basic}(a).
The upper bound follows from Minkowski's inequality and the tower rule, using
$\E[\varepsilon^2\mid\bS]\leq G_0^{2/r}$ and the weight bounds on $\mathcal E_0$.

Therefore, we obtain
\begin{align}
    \normp{T_1}{2}=O_p(\br_{\mu}+\br_{\pi}+\br_{q}). \label{order:QR_T_1}
\end{align}

Similarly, the $L_2$-norm of $T_2$ can be bounded as
\begin{align}
    \normp{T_2}{2}\leq& c_0^{-2}\normp{(1-A)\left[\{1-\pi^*(\bX)\}\{\hat{\mu}(\bS)-\hat{\tau}_n(\bX)\}-\{1-\hat{\pi}(\bX)\}\{\mu^*(\bS)-\tau_n^*(\bX)\}\right]}{2}\nonumber\\
    \overset{(i)}{\leq}& c_0^{-2}\normp{\{1-\pi^*(\bX)\}\{\hat{\mu}(\bzS)-\mu^*(\bzS)+\varrho+\tau_{n}^*(\bX)-\hat{\tau}_n(\bX)\}-\{1-\hat{\pi}(\bX)\}\varrho}{2}\nonumber\\
    \overset{(ii)}{\leq}& c_0^{-2}\normp{\{1-\pi^*(\bX)\}\{\mu^*(\bzS)-\hat{\mu}(\bzS)\}}{2}+c_0^{-2}\normp{\{1-\pi^*(\bX)\}\{\tau_{n}^*(\bX)-\hat{\tau}_n(\bX)\}}{2}\nonumber\\
    &+c_0^{-2}\normp{\{\hat{\pi}(\bX)-\pi^*(\bX)\}\sqrt{\E[\varrho^2\mid\bX]}}{2}\nonumber\\
    =&O_p\left(\br_{\mu}+\br_{\tau_{n}}+\br_{\pi}\right).\label{order:QR_T_2}
\end{align}

Additionally, for the third term $T_3$, we have
\begin{align}
    \normp{T_3}{2}\leq& c_0^{-2}\normp{A\pi^*(\bX)\{\hat{\tau}_n(\bX)-\hat{\tau}(\bX)\}-A\hat{\pi}(\bX)\{\tau_{n}^*(\bX)-\tau^*(\bX)\}}{2}\nonumber\\
    \overset{(i)}{\leq}& c_0^{-2}\normp{\pi^*(\bX)\{\hat{\tau}_n(\bX)-\tau_{n}^*(\bX)+(\zeta-\varrho)+\tau^*(\bX)-\hat{\tau}(\bX)\}-\hat{\pi}(\bX)(\zeta-\varrho)}{2}\nonumber\\
    \overset{(ii)}{\leq}& c_0^{-2}\normp{\pi^*(\bX)\{\hat{\tau}_n(\bX)-\tau_{n}^*(\bX)\}}{2}+c_0^{-2}\normp{\pi^*(\bX)\{\tau^*(\bX)-\hat{\tau}(\bX)\}}{2}\nonumber\\
    &+c_0^{-2}\normp{\{\pi^*(\bX)-\hat{\pi}(\bX)\}\sqrt{\E[\varrho^2\mid\bX]}}{2}+c_0^{-2}\normp{\{\pi^*(\bX)-\hat{\pi}(\bX)\}\sqrt{\E[\zeta^2\mid\bX]}}{2}\nonumber\\
    \overset{(iii)}{=}&O_p\left(\br_{\tau}+\br_{\tau_{n}}+\br_{\pi}\right),\label{order:QR_T_3}
\end{align}
where (i) follows from $|A|\leq 1$, (ii) holds by Minkowski's inequality, and (iii) follows since $\P(\{c_0 \leq \pi^*(\bX)\leq 1-c_0\})=1$ under Assumption~\ref{cond:QR_basic} (a) and $\P(\{c_0\leq\hat{\pi}(\bX)\leq 1-c_0\})=1$ on the event $\mathcal{E}_0$. For the last term $T_4$, under Assumption~\ref{cond:QR_basic}~(b), we have
\begin{align*}
    \normp{T_4}{2}=O_p(\br_{\tau}).
\end{align*}
Together with \eqref{order:QR_T_1}, \eqref{order:QR_T_2} and \eqref{order:QR_T_3}, we obtain
\begin{align}
    \sqrt{\E[\{\hat{\bpsi}^{(-k)}(\bW)-\bpsi^*(\bW)\}^2]}&\leq \sum\limits_{t=1}^{4}\normp{T_t}{2}=O_p\left(\br_{\mu}+\br_{\tau}+\br_{\tau_{n}}+\br_{\pi}+\br_{q}\right).\label{order:QR_sigmahat_upper}
\end{align}
Then by Chebyshev's inequality, we obtain
\begin{align*}
    \Delta_{k,2}=O_p\left(n^{-1/2}\{\br_{\mu}+\br_{\tau}+\br_{\tau_{n}}+\br_{\pi}+\br_{q}\}\right).
\end{align*}
\end{proof}

\begin{lemma}\label{lem:Delta_k_3}
Let Assumptions~\ref{cond:basic} and \ref{cond:QR_basic} hold.
Recall that $\Delta_{k,3}$ is defined in \eqref{def:QR_delta_k_3}. Then, for each $k\in[\K]$, as $N\to\infty$,
\begin{align}
\Delta_{k,3}=O_p\big(&\br_{\mu}\br_{q}+\br_{\tau}\br_{\pi}+\br_{\tau_n}\br_{\pi}+\br_{\mu}\br_{\pi}\nonumber\\
&+\eqr_\pi\{\br_\mu+\br_{\tau_n}+\br_{\tau}\}+\eqr_q\br_{\mu}+\eqr_{\mu}\br_{q}\nonumber\\
&+\{\eqr_q\eqr_\mu+\eqr_\tau+\eqr_{\tau_n}\}\br_{\pi}\big).\label{order:QR_delta_k_2}
\end{align}
\end{lemma}
\begin{proof}
To prove \eqref{order:QR_delta_k_2}, we consider the following decomposition
\begin{align}
    \Delta_{k,3}&=\Delta_{k,4}+\Delta_{k,5}+\Delta_{k,6}+\Delta_{k,7}+\Delta_{k,8}+\Delta_{k,9}+\Delta_{k,10}+\Delta_{k,11},\label{decom:QR_delta2}
\end{align}
where
\begin{align*}
    \Delta_{k,4}&=\E\left[\frac{A}{\hat{\pi}(\bX)}\{q^*(\bS)-\hat{q}(\bS)\}\{\hat{\mu}(\bS)-\mu^*(\bS)\}\right],\\
    \Delta_{k,5}&=\E\left[\left\{1-\frac{\pi^*(\bX)}{\hat{\pi}(\bX)}\right\}\{\hat{\tau}(\bX)-\tau^*(\bX)\}\right],\\
    \Delta_{k,6}&=\E\left[\left\{\frac{1-A}{1-\hat{\pi}(\bX)}-\frac{A}{\hat{\pi}(\bX)}\right\}\{\tau^*_{n}(\bX)-\hat{\tau}_{n}(\bX)\}\right]\nonumber
    \\
    &\quad+\E\left[\frac{1-A}{1-\hat{\pi}(\bX)}\{\hat{\mu}(\bS)-\mu^*(\bS)\}-\frac{A}{\hat{\pi}(\bX)}\{\hat{\mu}(\bzS)-\mu^*(\bzS)\}\right]\nonumber\\
    &\overset{(i)}{=}\E\left[\left\{\frac{1-\pi(\bX)}{1-\hat{\pi}(\bX)}-\frac{\pi(\bX)}{\hat{\pi}(\bX)}\right\}\{\tau^*_{n}(\bX)-\hat{\tau}_{n}(\bX)+\hat{\mu}(\bzS)-\mu^*(\bzS)\}\right],\\
    \Delta_{k,7}&=\E\left[\frac{1}{\hat{\pi}(\bX)}\{\pi^*(\bX)-\pi(\bX)\}\{\hat{\tau}(\bX)-\tau^*(\bX)\}\right],\\
    \Delta_{k,8}&=\E\left[\frac{A}{\hat{\pi}(\bX)}\{\hat{\mu}(\bzS)-\mu^*(\bzS)\}-\frac{A}{\hat{\pi}(\bX)}q^*(\bS)\{\hat{\mu}(\bS)-\mu^*(\bS)\}\right],\\
    \Delta_{k,9}&=\E\left[\frac{A}{\hat{\pi}(\bX)}\{\hat{q}(\bS)-q^*(\bS)\}\{Y-\mu^*(\bS)\}\right]\nonumber\\
    &\overset{(ii)}{=}\E\left[\frac{A}{\hat{\pi}(\bX)}\{\hat{q}(\bS)-q^*(\bS)\}\{\mu(\bS)-\mu^*(\bS)\}\right],\\
    \Delta_{k,10}&=\E\left[\left\{\frac{1-A}{1-\hat{\pi}(\bX)}-\frac{1-A}{1-\pi^*(\bX)}\right\}\{\mu^*(\bS)-\tau_{n}^*(\bX)\}\right]\nonumber\\
    &\overset{(iii)}{=}\E\left[\left\{\frac{1-A}{1-\hat{\pi}(\bX)}-\frac{1-A}{1-\pi^*(\bX)}\right\}\{\tau_{n}(\bX)-\tau_{n}^*(\bX)\}\right],\\
    \Delta_{k,11}&=\E\left[\left\{\frac{A}{\hat{\pi}(\bX)}-\frac{A}{\pi^*(\bX)}\right\}\big\{q^*(\bS)\{Y-\mu^*(\bS)\}+\tau_{n}^*(\bX)-\tau^*(\bX)\big\}\right].
\end{align*}
In the above, (i) and (ii) hold by the tower rule, Assumption~\ref{cond:basic}~(a)--(b), and the identity $\mu(\bS)=\E[Y\mid A=1,\bS]$; (iii) holds by the tower rule and the definition $\tau_{n}(\bX)=\E[\mu^*(\bS)\mid A=0,\bX]$.
We then obtain upper bounds for $\Delta_{k,4}$--$\Delta_{k,11}$ using the Cauchy--Schwarz inequality, H\"older's inequality, and Minkowski's inequality.
Since
\begin{align*}
\frac{1-\pi(\bX)}{1-\hat\pi(\bX)}-\frac{\pi(\bX)}{\hat\pi(\bX)}
=\frac{\hat\pi(\bX)-\pi^*(\bX)+\pi^*(\bX)-\pi(\bX)}{\hat\pi(\bX)\{1-\hat\pi(\bX)\}},
\end{align*}
the upper bound for $|\Delta_{k,4}|+|\Delta_{k,5}|+|\Delta_{k,6}|$ is given by
\begin{align}
&|\Delta_{k,4}|+|\Delta_{k,5}|+|\Delta_{k,6}|\nonumber\\
\overset{(i)}\leq{}& c_0^{-1}\left\{\normp{\hat q(\bS)-q^*(\bS)}{2}\normp{\hat\mu(\bS)-\mu^*(\bS)}{2}
+\normp{\hat\pi(\bX)-\pi^*(\bX)}{2}\normp{\hat\tau(\bX)-\tau^*(\bX)}{2}\right\}\nonumber\\
&+c_0^{-2}\left\{\normp{\hat\pi(\bX)-\pi^*(\bX)}{2}+\normp{\pi^*(\bX)-\pi(\bX)}{2}\right\}\nonumber\\
&\qquad\times\left\{\normp{\hat\tau_n(\bX)-\tau_n^*(\bX)}{2}
+\normp{\hat\mu(\bzS)-\mu^*(\bzS)}{2}\right\}\nonumber\\
\overset{(ii)}={}&O_p\left(\br_\mu\br_q+\br_\tau\br_\pi
+\br_{\tau_n}\br_\pi+\br_\mu\br_\pi+\eqr_\pi\{\br_{\tau_n}+\br_\mu\}\right),
    \label{order:QR_delta_3+4+5}
\end{align}
where (i) follows from the Cauchy--Schwarz inequality on the event $\mathcal E_0$, which occurs with probability approaching one under Assumption~\ref{cond:QR_basic}(a), and (ii) follows from Lemma~\ref{lem:mu_relation}, Assumption~\ref{cond:QR_basic}(b), and $\|\pi^*-\pi\|_{\P,2}\leq\eqr_\pi$.

Similarly, for the term $\Delta_{k,7}$,
\begin{align}
    |\Delta_{k,7}|\leq c_0^{-1}\normp{\pi^*-\pi}{2}\normp{\hat{\tau}(\bX)-\tau^*(\bX)}{2}=O_p(\eqr_{\pi}\br_{\tau}).\label{order:QR_delta_6}
\end{align}

For the term $\Delta_{k,8}$, note that under Assumption~\ref{cond:basic},
\begin{align*}
    \Delta_{k,8}=\E\left[\frac{A}{\hat{\pi}(\bX)}\{q(\bS)-q^*(\bS)\}\{\hat{\mu}(\bS)-\mu^*(\bS)\}\right].
\end{align*}
We then have
\begin{align}
    |\Delta_{k,8}|\leq c_0^{-1}\normp{q-q^*}{2}\normp{\hat{\mu}(\bS)-\mu^*(\bS)}{2}=O_p(\eqr_{q}\br_{\mu}).\label{order:QR_delta_7}
\end{align}

Furthermore, it holds for the term $\Delta_{k,9}$ that
\begin{align}
    |\Delta_{k,9}|\leq c_0^{-1}\normp{\hat{q}(\bS)-q^*(\bS)}{2}\normp{\mu-\mu^*}{2}=O_p(\eqr_{\mu}\br_{q}).\label{order:QR_delta_8}
\end{align}

The term $\Delta_{k,10}$ can be rewritten as
\begin{align*}
    \Delta_{k,10}&=\E\left[\frac{(1-A)}{\{1-\hat{\pi}(\bX)\}\{1-\pi^*(\bX)\}}\{\hat{\pi}(\bX)-\pi^*(\bX)\}\{\tau_{n}(\bX)-\tau_{n}^*(\bX)\}\right]
\end{align*}
Therefore under Assumptions~\ref{cond:basic} and \ref{cond:QR_basic}, we have
\begin{align}
    |\Delta_{k,10}|\leq c_0^{-2}\normp{\hat{\pi}(\bX)-\pi^*(\bX)}{2}\normp{\tau_{n}-\tau_{n}^*}{2}=O_p(\eqr_{\tau_n}\br_{\pi}).\label{order:QR_delta_9}
\end{align}

By the tower rule and the definition of $q$, $\tau$, $\tau_n$ and $\bzS$, we have
\begin{align}
    \E\left[A\tau_n(\bX)\mid \bX\right]=\E\left[Aq(\bS)\mu^*(\bS)\mid \bX\right]=\E\left[A\mu^*(\bzS)\mid \bX\right]\label{eq:repre_taun}\\
    \E\left[A\tau(\bX)\mid \bX\right]=\E\left[Aq(\bS)\mu(\bS)\mid \bX\right]=\E\left[A\mu(\bzS)\mid \bX\right].\label{eq:repre_tau}
\end{align}
Therefore, the term $\Delta_{k,11}$ can be rewritten as
\begin{align*}
    \Delta_{k,11}=&\E\left[\frac{\pi(\bX)}{\hat{\pi}(\bX)\pi^*(\bX)}\left\{\pi^*(\bX)-\hat{\pi}(\bX)\right\}\Big\{\E[q^*(\bS)\{\mu(\bS)-\mu^*(\bS)\}\right.\\
    &\left.-q(\bS)\{\mu(\bS)-\mu^*(\bS)\}\mid A=1,\bX]+\tau(\bX)-\tau_n(\bX)+\tau_{n}^*(\bX)-\tau^*(\bX)\Big\}\right].
\end{align*}
It leads to
\begin{align}
    |\Delta_{k,11}|\leq& c_0^{-2}\normp{\hat{\pi}(\bX)-\pi^*(\bX)}{2}\left\{\normp{q^*-q}{4}\normp{\mu-\mu^*}{4}\right.\nonumber\\
    &\left.+\normp{\tau-\tau^*}{2}+\normp{\tau_n-\tau_n^*}{2}\right\}\nonumber\\
    =&O_p\left(\{\eqr_q \eqr_\mu+\eqr_\tau+\eqr_{\tau_n}\}\br_{\pi}\right).\label{order:QR_delta_10}
\end{align}

Plugging \eqref{order:QR_delta_3+4+5}--\eqref{order:QR_delta_9} and \eqref{order:QR_delta_10} into \eqref{decom:QR_delta2}, we obtain \eqref{order:QR_delta_k_2}.
\end{proof}

\begin{proof}[Proof of Theorem~\ref{thm:QR_consistent}]

Under Assumption~\ref{cond:QR_basic}, $\ssigma^2>c_1>0$. Further by Lemma~\ref{lem:QR_lower} (c), we conclude that $\ssigma\asymp 1$.

Recalling \eqref{decom:QR_theta}, by Lemmas~\ref{lem:Delta_0}, \ref{lem:QR_lower}, \ref{lem:Delta_k_2} and \ref{lem:Delta_k_3}, we have
\begin{align*}
    \Delta_0\;\;&=O(\eqr_\pi\{\eqr_{\tau_n}+\eqr_\tau\}+\eqr_q\eqr_\mu),\\
    \Delta_{k,1}&\overset{(i)}{=}O_p(n^{-1/2} \ssigma)=O_p(n^{-1/2}),\\
    \Delta_{k,2}&=O_p\left(n^{-1/2}\{\br_{\mu}+\br_{\tau}+\br_{\tau_{n}}+\br_{\pi}+\br_{q}\}\right),\\
\Delta_{k,3}&=O_p\big(\br_\mu\br_q+\br_\tau\br_\pi+\br_{\tau_n}\br_\pi+\br_\mu\br_\pi
+\eqr_\pi\br_\tau+\eqr_q\br_\mu+\eqr_\mu\br_q\\
&\qquad+\eqr_\pi(\br_{\tau_n}+\br_\mu)
+\{\eqr_q\eqr_\mu+\eqr_\tau+\eqr_{\tau_n}\}\br_\pi\big),
\end{align*}
where (i) holds by Markov's inequality.
Since $\br_\pi+\br_q+\br_\mu+\br_{\tau_n}+\br_\tau=o(1)$ and $n=N/\K\asymp N$, by Lemma~\ref{lem:QR_lower}, we have
\begin{align*}
    \hat{\theta}_{\mathrm{QR}}-\theta&=\Delta_0+\K^{-1}\sum\limits_{k=1}^{\K}\Delta_{k,1}+\K^{-1}\sum\limits_{k=1}^{\K}\Delta_{k,2}+\K^{-1}\sum\limits_{k=1}^{\K}\Delta_{k,3}\\
    &=O_p\Big(N^{-1/2}+\eqr_\pi \eqr_{\tau_n}+\eqr_\pi \eqr_\tau+\eqr_q\eqr_\mu+\br_{\mu}\br_{q}+\br_{\tau}\br_{\pi}+\br_{\tau_{n}}\br_{\pi}+\br_{\mu}\br_{\pi}
    \\ &\qquad+\eqr_\pi\br_{\tau}+\eqr_\pi(\br_{\tau_n}+\br_\mu)
+\eqr_q\br_{\mu}+\eqr_{\mu}\br_{q}+\{\eqr_\tau+\eqr_{\tau_n}\}\br_{\pi}\Big)\\
&=O_p(N^{-1/2}+\bar{\mathcal R}_{\mathrm{QR}}).
\end{align*}
\end{proof}

\begin{proof}[Proof of Theorem~\ref{thm:QR_normal}]
Under Assumption~\ref{cond:QR_basic} and $\bar{\mathcal R}_{\mathrm{QR}}=o(N^{-1/2})$, noting that $\ssigma>\sqrt{c_1}>0$ and $n\asymp N$, Lemmas~\ref{lem:Delta_0}--\ref{lem:Delta_k_3} yield $\Delta_0=O(\bar{\mathcal R}_{\mathrm{QR}})=o(\ssigma N^{-1/2})$, $\Delta_{k,2}=o_p(\ssigma N^{-1/2})$, and $\Delta_{k,3}=O_p(\bar{\mathcal R}_{\mathrm{QR}})=o_p(\ssigma N^{-1/2})$. Hence,
\begin{align}
    \Delta_0+\K^{-1}\sum\limits_{k=1}^{\K}(\Delta_{k,2}+\Delta_{k,3})=o_p\left(\frac{\ssigma}{\sqrt{N}}\right).\label{order:QR_ignore}
\end{align}
As in the proof of Theorem~\ref{thm:QR_consistent}, $\Delta_{k,1}=O_p(n^{-1/2})$.
By \eqref{decom:QR_theta}, we conclude that
\begin{align}
    \hat{\theta}_{\mathrm{QR}}-\theta&=\Delta_0+\K^{-1}\sum\limits_{k=1}^{\K}\Delta_{k,1}+\K^{-1}\sum\limits_{k=1}^{\K}\Delta_{k,2}+\K^{-1}\sum\limits_{k=1}^{\K}\Delta_{k,3}=O_p\left(\frac{\ssigma}{\sqrt{N}}\right).\label{order:QR_consistent}
\end{align}

Next, we justify that $\sqrt{N}\ssigma^{-1}(\hat{\theta}_{\mathrm{QR}}-\theta)\leadsto \mathcal{N}(0,1)$. By \eqref{order:QR_ignore} and Slutsky's theorem, it suffices to show that
\begin{align*}
    \sqrt{N}\ssigma^{-1}\K^{-1}\sum\limits_{k=1}^{\K}\Delta_{k,1}=\frac{1}{\sqrt{N}\ssigma}\sum\limits_{j=1}^{N}\big\{\bpsi^*(\bW_j)-\E[\bpsi^*(\bW)]\big\}\leadsto \mathcal{N}(0,1).
\end{align*}
Applying Lyapunov's central limit theorem, this holds if for some $t>0$,
\begin{align}
    \ssigma^{-t-2}\E\big[\big|\bpsi^*(\bW)-\E[\bpsi^*(\bW)]\big|^{t+2}\big]=O(1).\label{certify:QR}
\end{align}
Let $t=r-2>0$. By Minkowski's and Jensen's inequalities, $\normp{\bpsi^*(\bW)-\E[\bpsi^*(\bW)]}{2+t}\leq 2\normp{\bpsi^*(\bW)-\theta}{2+t}$. Together with Lemma~\ref{lem:QR_consistent_lower} and Assumption~\ref{cond:QR_basic}(c), we obtain
\begin{align*}
    \ssigma^{-t-2}\E\big[\big|\bpsi^*(\bW)-\E[\bpsi^*(\bW)]\big|^{t+2}\big]
    &\leq \frac{2^{6+4t}\E\big[|\xi|^{r}+|\varrho|^{r}+|\zeta|^{r}+|\varepsilon|^{r}\big]}{c_0^{6+3t}c_1^{t/2+1}}
    \leq\frac{2^{6+4t}G_0}{c_0^{6+3t}c_1^{t/2+1}},
\end{align*}
which is independent of $N$ and yields \eqref{certify:QR}.

In what follows, we prove $\hat{\sigma}^2=\ssigma^2\{1+o_p(1)\}$. For each $k\in[\K]$,
\begin{align*}
    \E_{\mathbb S_k}\left[\left\{\frac{1}{n}\sum\limits_{j\in\I_k} \{\hat{\bpsi}^{(-k)}(\bW_j)-\bpsi^*(\bW_j)\}^2\right\}^{1/2}\right]&\overset{(i)}{\leq} \left\{\E_{\mathbb S_k}\left[\frac{1}{n}\sum\limits_{j\in\I_k} \{\hat{\bpsi}^{(-k)}(\bW_j)-\bpsi^*(\bW_j)\}^2\right]\right\}^{1/2}\\
    &\overset{(ii)}{=} \left\{\E\left[\{\hat{\bpsi}^{(-k)}(\bW)-\bpsi^*(\bW)\}^2\right]\right\}^{1/2}\\
    &\overset{(iii)}{=} O_p\left(\br_{\mu}+\br_{\tau}+\br_{\tau_{n}}+\br_{\pi}+\br_{q}\right)= o_p(\ssigma).
\end{align*}
In the first line, the expectations are taken with respect to the joint distribution of $\{\bW_j\}_{j\in\I_k}$, while in the remaining lines, the expectation is taken with respect to the joint distribution of a new $\bW$. Here, (i) holds by Jensen's inequality; (ii) holds since the nuisance function estimators are independent of $\{\bW_j\}_{j\in\I_k}$ based on cross-fitting, $\{\bW_j\}_{j\in\I_k}$ are i.i.d. and $\bW$ is an independent copy of them; (iii) holds by \eqref{order:QR_sigmahat_upper}. Then by conditional Markov's inequality,
\begin{align}
    \left\{\frac{1}{n}\sum\limits_{j\in\I_k} \{\hat{\bpsi}^{(-k)}(\bW_j)-\bpsi^*(\bW_j)\}^2\right\}^{1/2}=o_p(\ssigma).\label{order:gen_sigma_decom1}
\end{align}
Next, we define $a_j=\hat{\bpsi}^{(-k)}(\bW_j)-\bpsi^*(\bW_j)-\{\hat{\theta}_{\mathrm{QR}}-\E[\bpsi^*(\bW)]\}$ and $b_j=\bpsi^*(\bW_j)-\E[\bpsi^*(\bW)]$. By the triangle inequality, it holds that
\begin{align}
    &\left|\frac{1}{n}\sum\limits_{j\in\I_k}(\hat{\bpsi}^{(-k)}(\bW_j)-\hat{\theta}_{\mathrm{QR}})^2-\frac{1}{n}\sum\limits_{j\in\I_k}\big\{\bpsi^*(\bW_j)-\E[\bpsi^*(\bW)]\big\}^2\right|\nonumber\\
    &\;\;\leq \frac{1}{n}\sum\limits_{j\in\I_k}|a_j|\cdot |a_j+2b_j|\overset{(i)}{\leq} \left[\frac{1}{n}\sum\limits_{j\in\I_k}a_j^2\right]^{\frac12}\cdot\left[\frac{1}{n}\sum\limits_{j\in\I_k}(a_j+2b_j)^2\right]^{\frac12}\nonumber\\
    &\;\;\overset{(ii)}{\leq} \left[\frac{1}{n}\sum\limits_{j\in\I_k}a_j^2\right]^{\frac12}\cdot\left[\left(\frac{1}{n}\sum\limits_{j\in\I_k}a_j^2\right)^{\frac12}+2\left(\frac{1}{n}\sum\limits_{j\in\I_k}b_j^2\right)^{\frac12}\right],\label{decom:gen_sigma}
\end{align}
where (i) holds by the Cauchy--Schwarz inequality, and (ii) holds by Minkowski's inequality. By \eqref{order:gen_sigma_decom1}, Lemma~\ref{lem:Delta_0}, and \eqref{order:QR_consistent}, since $|\hat{\theta}_{\mathrm{QR}}-\E[\bpsi^*(\bW)]|\leq |\hat{\theta}_{\mathrm{QR}}-\theta|+|\Delta_0|=o_p(\ssigma)$, we have
\begin{align}
    \left[\frac{1}{n}\sum\limits_{j\in\I_k}a_j^2\right]^{\frac12}\leq \left[\frac{1}{n}\sum\limits_{j\in\I_k}|\hat{\bpsi}^{(-k)}(\bW_j)-\bpsi^*(\bW_j)|^2\right]^{\frac12}+ \big|\hat{\theta}_{\mathrm{QR}}-\E[\bpsi^*(\bW)]\big|=o_p(\ssigma).\label{order:gen_sigma_decom2}
\end{align}
For fixed $k$, let $Z_{d_1,d_2,j}=\ssigma^{-2}\{\bpsi^*(\bW_j)-\E[\bpsi^*(\bW)]\}^2-1$ for $j\in\I_k$, where $\bW_j$ and all of the nuisance functions $\pi^*,q^*,\mu^*,\tau_n^*,\tau^*$ are possibly dependent on $(d_1,d_2)$, and therefore on $N$. As a result, $Z_{d_1,d_2,j}$ forms a row-wise independent and identically distributed triangular array. By definition, $\E[Z_{d_1,d_2,j}]=1-1=0$ holds for each $j$. In addition, since $r/2>1$,
\begin{align*}
    \left[\E\left|\frac{\{\bpsi^*(\bW)-\E[\bpsi^*(\bW)]\}^2}{\ssigma^{2}}-1\right|^{r/2}\right]^{2/r}\overset{(i)}{\leq} \left[\frac{\E\big|\bpsi^*(\bW)-\E[\bpsi^*(\bW)]\big|^{r}}{\ssigma^{r}}\right]^{2/r}+1\overset{(ii)}{\leq}C+1
\end{align*}
for some constant $C>0$. In the above, (i) holds by Minkowski's inequality, (ii) holds by \eqref{certify:QR}. Then, it follows that, for each $j$,
\begin{align*}
    \E[|Z_{d_1,d_2,j}|^{r/2}]=\E\left|\frac{\{\bpsi^*(\bW)-\E[\bpsi^*(\bW)]\}^2}{\ssigma^{2}}-1\right|^{r/2}\leq (C+1)^{r/2}.
\end{align*}
By Lemma 3 of \cite{HighdimensionalSemisupervisedLearning2022}, it follows that
\begin{align}
    \frac{1}{n}\sum\limits_{j\in\I_k}b_j^2-\ssigma^2=\frac{1}{n}\sum\limits_{j\in\I_k}\big\{\bpsi^*(\bW_j)-\E[\bpsi^*(\bW)]\big\}^2-\ssigma^2=o_p(\ssigma^2).\label{order:gen_sigma_decom3}
\end{align}
Together with \eqref{decom:gen_sigma} and \eqref{order:gen_sigma_decom2}, we obtain
\begin{align}
    \left|\frac{1}{n}\sum\limits_{j\in\I_k}(\hat{\bpsi}^{(-k)}(\bW_j)-\hat{\theta}_{\mathrm{QR}})^2-\frac{1}{n}\sum\limits_{j\in\I_k}\big\{\bpsi^*(\bW_j)-\E[\bpsi^*(\bW)]\big\}^2\right|=o_p(\ssigma^2).\label{order:gen_sigma_decom4}
\end{align}
By \eqref{order:gen_sigma_decom3} and \eqref{order:gen_sigma_decom4}, we conclude that
\begin{align*}
    \hat{\sigma}^2-\ssigma^2&=\frac{1}{\K}\sum\limits_{k=1}^{\K}\frac{1}{n}\sum\limits_{j\in\I_k}(\hat{\bpsi}^{(-k)}(\bW_j)-\hat{\theta}_{\mathrm{QR}})^2-\ssigma^2\\
    &=\frac{1}{\K}\sum\limits_{k=1}^{\K}\left(\frac{1}{n}\sum\limits_{j\in\I_k}(\hat{\bpsi}^{(-k)}(\bW_j)-\hat{\theta}_{\mathrm{QR}})^2-\frac{1}{n}\sum\limits_{j\in\I_k}b_j^2+\frac{1}{n}\sum\limits_{j\in\I_k}b_j^2-\ssigma^2\right)\\
    &=o_p(\ssigma^2).
\end{align*}
\end{proof}

\subsection{Proofs of auxiliary lemmas in Section~\ref{apdx:QR_aux}}\label{apdx:QR_aux_proofs}

\begin{proof}[Proof of Lemma~\ref{lem:mu_relation}]
We have
\begin{align*}
    \normp{\mu^*(\bS)-\hat{\mu}(\bS)}{2}
    \geq& \sqrt{\E\big[(1-A)\{\mu^*(\bS)-\hat{\mu}(\bS)\}^2\big]} \\
    \overset{(i)}=& \sqrt{\E\left[\{1-\pi(\bX)\}\E\left[\{\mu^*(\bzS)-\hat{\mu}(\bzS)\}^2 \mid A=0, \bX\right]\right]} \\
    \overset{(ii)}\geq& \sqrt{c_0\, \E\left[\E\left[\{\mu^*(\bzS)-\hat{\mu}(\bzS)\}^2 \mid \bX\right]\right]},
\end{align*}
where (i) holds by Assumption~\ref{cond:basic}~(a), and (ii) follows from Assumption~\ref{cond:basic}~(b)--(c), which implies that $1-\pi(\bX)\ge c_0>0$ holds almost surely.
After rearrangement, the desired inequality follows.
\end{proof}

\begin{proof}[Proof of Lemma~\ref{lem:QR_lower}]
\emph{(a)}
Under Assumption~\ref{cond:QR_basic} (c), we have
\begin{align*}
    \E[\varepsilon^2\mid \bS]&\overset{(i)}\leq \E[(\varepsilon^2)^{r/2}\mid\bS]^{2/r}\E[1^{\frac{r}{r-2}}\mid\bS]^{1-2/r}\leq G_0^{2/r}
\end{align*}
almost surely. Here, (i) follows from H\"older's inequality.
Similarly,
\begin{gather*}
    \E[\zeta^2\mid \bX]\leq G_0^{2/r}\mathrm{\;\;and\;\;}\E[\varrho^2\mid \bX]\leq G_0^{2/r}
\end{gather*}
hold almost surely. Then (a) holds.

\emph{(b)} It follows from (a) that $\max\{\E[\varepsilon^2],\E[\varrho^2],\E[\zeta^2]\}\leq G_0^{2/r}$.
By the tower rule and H\"older's inequality,
\begin{align*}
    \E[\xi^2]&\leq \E[\E[(\xi^2)^{r/2}\mid \bX]]^{2/r}\E[1^{\frac{r}{r-2}}]^{1-2/r}\leq G_0^{2/r}
\end{align*}
holds under Assumption~\ref{cond:QR_basic}~(c). Therefore, we complete the proof of (b).

\emph{(c)} Note that
\begin{align*}
    \ssigma=\normp{\phi^*(\bW)-\E[\phi^*(\bW)]}{2}\leq\normp{\phi^*(\bW)-\theta}{2}.
\end{align*}
By Minkowski's inequality, we have
\begin{align*}
    \normp{\bpsi^*(\bW)-\theta}{2}\leq \sum\limits_{t=1}^4\normp{I_t}{2},
\end{align*}
where
\begin{align*}
    I_1&=\frac{A}{\pi^*(\bX)}q^*(\bS)\{Y-\mu^*(\bS)\},\\
    I_2&=\frac{1-A}{1-\pi^*(\bX)}\{\mu^*(\bS)-\tau_{n}^*(\bX)\},\\
    I_3&=\frac{A}{\pi^*(\bX)}\{\tau_{n}^*(\bX)-\tau^*(\bX)\},\\
    I_4&=\tau^*(\bX)-\theta.
\end{align*}
For the first term $I_1$, we have
\begin{align*}
    \normp{I_1}{2}&\leq c_0^{-3}\sqrt{\E[A\{Y-\mu^*(\bS)\}^2]}\overset{(i)}{\leq} c_0^{-3}\sqrt{\E\big[\E[A\{Y-\mu^*(\bS)\}^2\mid \bX]\big]}\\
    &\overset{(ii)}{\leq} c_0^{-3}\sqrt{\E[A\{Y(1,\bM)-\mu^*(\bS)\}^2]}\overset{(iii)}{\leq}c_0^{-3}\normp{\varepsilon}{2},
\end{align*}
where (i) holds by the tower rule; (ii) follows from Assumption~\ref{cond:basic}~(a); (iii) holds by Assumption~\ref{cond:QR_basic} (c) and the fact that $|A|\leq1$. Similarly, we have
\begin{align*}
    \normp{I_2}{2}&\leq c_0^{-1}\sqrt{\E[(1-A)\{\mu^*(\bzS)-\tau_{n}^*(\bX)\}^2]}\leq c_0^{-1}\normp{\varrho}{2}.
\end{align*}
By Minkowski's inequality, we further have
\begin{align*}
    \normp{I_3}{2}&\leq c_0^{-1}\left\{\normp{\tau_{n}^*(\bX)-\mu^*(\bzS)}{2}+\normp{\mu^*(\bzS)-\tau^*(\bX)}{2}\right\}=c_0^{-1}\left(\normp{\varrho}{2}+\normp{\zeta}{2}\right).
\end{align*}
For the term $I_4$, we note that
\begin{align*}
    \normp{I_4}{2}&\leq \normp{\tau^*-\tau}{2}+\normp{\tau(\bX)-\theta}{2}.
\end{align*}
Next, we bound the first term on the right-hand side.

By the tower rule and Jensen's inequality, it holds that
\begin{align}
    \normp{\varepsilon}{2}^2=\E[\E[\{Y(1,\bM)-\mu^*(\bS)\}^2\mid \bS]]\geq \normp{\mu^*-\mu}{2}^2.\label{bound:mu^*-mu}
\end{align}
Note that
\begin{align*}
    \E[\{\mu(\bzS)-\tau^*(\bX)\}^2]=\E[\{\mu(\bzS)-\tau(\bX)+\tau(\bX)-\tau^*(\bX)\}^2]\geq \normp{\tau-\tau^*}{2}^2.
\end{align*}
Here, the inequality holds since
\begin{align*}
    \E[\{\mu(\bzS)-\tau(\bX)\}\{\tau(\bX)-\tau^*(\bX)\}]
    &\overset{(i)}{=}\E\big[\E[\mu(\bzS)-\tau(\bX)\mid \bX]\{\tau(\bX)-\tau^*(\bX)\}\big]\overset{(ii)}{=}0,
\end{align*}
where (i) holds by the tower rule, and (ii) follows from $\tau(\bX)=\E[\mu(\bzS)\mid \bX]$.
Additionally, note that
\begin{align*}
    \E[\{\mu(\bzS)-\tau^*(\bX)\}^2]&\overset{(i)}\leq 2\E[\{\mu(\bzS)-\mu^*(\bzS)\}^2+\{\mu^*(\bzS)-\tau^*(\bX)\}^2]\\
    &=2\E\left[\frac{1-A}{1-\pi(\bX)}\{\mu(\bS)-\mu^*(\bS)\}^2\right]+2\E[\{\mu^*(\bzS)-\tau^*(\bX)\}^2]\\
    &\overset{(ii)}\leq 2c_0^{-1}\normp{\mu-\mu^*}{2}^2+2\E[\{\mu^*(\bzS)-\tau^*(\bX)\}^2]\\
    &\overset{(iii)}\leq 2c_0^{-1}\normp{\varepsilon}{2}^2+2\normp{\zeta}{2}^2.
\end{align*}
Here, (i) holds by the fact that $(a+b)^2\leq2a^2+2b^2$; (ii) holds under Assumption~\ref{cond:basic} (c) and $|1-A|\leq 1$; (iii) holds by \eqref{bound:mu^*-mu} under Assumption~\ref{cond:QR_basic} (c).
Therefore
\begin{align}
    \normp{\tau-\tau^*}{2}^2
    &\leq 2c_0^{-1}\normp{\varepsilon}{2}^2+2\normp{\zeta}{2}^2.\label{bound:tau-tau^*}
\end{align}
It follows that
\begin{align*}
    \normp{I_4}{2}&\leq \normp{\tau^*-\tau}{2}+\normp{\tau(\bX)-\theta}{2}\overset{(i)}{\leq}(2/c_0)^{1/2}\{\normp{\varepsilon}{2}+\normp{\zeta}{2}\}+\normp{\xi}{2},
\end{align*}
where (i) follows from \eqref{bound:tau-tau^*} and $c_0<1$. Collectively, we obtain
\begin{align}
    \ssigma\leq&\normp{\bpsi^*(\bW)-\theta}{2}\nonumber\\
    \leq& c_0^{-3}\normp{\varepsilon}{2}+c_0^{-1}\normp{\varrho}{2}+c_0^{-1}(\normp{\varrho}{2}+\normp{\zeta}{2})+(2/c_0)^{1/2}\{\normp{\varepsilon}{2}+\normp{\zeta}{2}\}+\normp{\xi}{2}\nonumber\\
    \leq&2c_0^{-3}\left(\normp{\varepsilon}{2}+\normp{\varrho}{2}+\normp{\zeta}{2}+\normp{\xi}{2}\right),\label{bound:sigma_upper}
\end{align}
since $1<2c_0^{-1}<c_0^{-3}+(2/c_0)^{1/2}\leq 2c_0^{-3}$ when $c_0\in(0,1/2)$.

Combining this with Lemma~\ref{lem:QR_lower} (b), we complete the proof.
\end{proof}

\begin{proof}[Proof of Lemma~\ref{lem:QR_consistent_lower}]
Let $p:=2+t>2$. Following the proof of Lemma~\ref{lem:QR_lower}(c), we decompose $\bpsi^*(\bW)-\theta=I_1+I_2+I_3+I_4$, where
\begin{align*}
    I_1&:=\frac{Aq^*(\bS)}{\pi^*(\bX)}\{Y-\mu^*(\bS)\},\qquad
    I_2:=\frac{1-A}{1-\pi^*(\bX)}\{\mu^*(\bS)-\tau_{n}^*(\bX)\},\\
    I_3&:=\frac{A}{\pi^*(\bX)}\{\tau_{n}^*(\bX)-\tau^*(\bX)\},\qquad
    I_4:=\tau^*(\bX)-\theta.
\end{align*}
Under Assumptions~\ref{cond:basic}(a) and \ref{cond:QR_basic}(a), noting that $\tau_n^*(\bX)-\tau^*(\bX)=\zeta-\varrho$, Minkowski's inequality gives
\begin{align*}
    \normp{I_1}{p}\leq c_0^{-3}\normp{\varepsilon}{p},\qquad
    \normp{I_2}{p}\leq c_0^{-1}\normp{\varrho}{p},\qquad
    \normp{I_3}{p}\leq c_0^{-1}\left(\normp{\varrho}{p}+\normp{\zeta}{p}\right).
\end{align*}
For $I_4=\{\tau^*(\bX)-\tau(\bX)\}+\xi$, we have $\normp{I_4}{p}\leq \normp{\tau^*(\bX)-\tau(\bX)}{p}+\normp{\xi}{p}$. Since $\tau(\bX)-\tau^*(\bX)=\E[\mu(\bzS)-\mu^*(\bzS)+\zeta\mid\bX]$, Jensen's and Minkowski's inequalities yield
\begin{align*}
    \normp{\tau^*(\bX)-\tau(\bX)}{p}\leq \normp{\mu(\bzS)-\mu^*(\bzS)}{p}+\normp{\zeta}{p}.
\end{align*}
Moreover, under Assumption~\ref{cond:basic}, since $\mu(\bS)-\mu^*(\bS)=\E[\varepsilon\mid A=1,\bS]$, the tower rule and Jensen's inequality give
\begin{align*}
    \E\big[|\mu(\bzS)-\mu^*(\bzS)|^p\big]
    &=\E\left[\frac{Aq(\bS)}{\pi(\bX)}\big|\mu(\bS)-\mu^*(\bS)\big|^p\right]\\
    &\leq \E\left[\frac{Aq(\bS)}{\pi(\bX)}\E\big[|\varepsilon|^p\mid A=1,\bS\big]\right]
    \leq c_0^{-3}\E\big[|\varepsilon|^p\big],
\end{align*}
so $\normp{\mu(\bzS)-\mu^*(\bzS)}{p}\leq c_0^{-3/p}\normp{\varepsilon}{p}$. Combining the bounds for $I_1,\dots,I_4$ via Minkowski's inequality, we obtain
\begin{align*}
    \normp{\bpsi^*(\bW)-\theta}{p}
    &\leq (c_0^{-3}+c_0^{-3/p})\normp{\varepsilon}{p}+2c_0^{-1}\normp{\varrho}{p}+(c_0^{-1}+1)\normp{\zeta}{p}+\normp{\xi}{p}\\
    &\leq 2c_0^{-3}\left(\normp{\xi}{p}+\normp{\varrho}{p}+\normp{\zeta}{p}+\normp{\varepsilon}{p}\right),
\end{align*}
which yields \eqref{bound:QR_lyapunov} by Jensen's inequality.
\end{proof}

\section{Proofs of the results for the model quadruply robust (MQR) methods}\label{apdx:MQR}

This section presents the proofs of the theoretical results for the MQR framework, covering the main asymptotic properties of MQR1 in Section~\ref{sec:r_thm}, the properties of the target nuisance parameters and the nuisance convergence rates in Section~\ref{apdx:mqr_supp}, and the theoretical guarantees for MQR2 in Section~\ref{apdx:mqr2_results}.

As discussed in Section~\ref{sec:MQR}, with the proposed MQR score defined in \eqref{def:psihat-MQR} and target nuisance parameters constructed through \eqref{eq:r_loss1}--\eqref{eq:r_loss6}, the resulting $\phi(\bW,\bnu^*)$ ensures that \eqref{eq:moment} holds. Specifically, we have
\begin{align}
    \E[\bnabla_{\bepia}\mpsi(\bW,\bnu^*)]&=\tfrac12\E[\bnabla_{\betau}\ell_6(\bW,\bnu^*_6,\betaur)]=\bz_{d_1},\label{eq:r_eq_pia}\\
    \E[\bnabla_{\bepib}\mpsi(\bW,\bnu^*)]&=-\tfrac12\E[\bnabla_{\bettau}\ell_5(\bW,\bnu^*_5,\bettaur)]=\bz_{d_1},\label{eq:r_eq_pib}\\
    \E[\bnabla_{\beq}\mpsi(\bW,\bnu^*)]&=-\tfrac12\E[\bnabla_{\bemu}\ell_4(\bW,\bnu^*_4,\bemur)]=\bz_{d},\label{eq:r_eq_q}\\
    \E[\bnabla_{\bemu}\mpsi(\bW,\bnu^*)]&=-\E[\bnabla_{\beq}\ell_3(\bW,\bnu^*_3,\beqr)]=\bz_{d},\label{eq:r_eq_mu}\\
    \E[\bnabla_{\betau}\mpsi(\bW,\bnu^*)]&=\E[\bnabla_{\bepia}\ell_1(\bW,\bepiar)]=\bz_{d_1},
    \label{eq:r_eq_tau_n}\\
    \E[\bnabla_{\bettau}\mpsi(\bW,\bnu^*)]&=-\E[\bnabla_{\bepia}\ell_1(\bW,\bepiar)+\bnabla_{\bepib}\ell_2(\bW,\bepibr)]=\bz_{d_1},
    \label{eq:r_eq_tau}
\end{align}
where
\begin{align*}
    \nabla_{\bepia}\E[\ell_1(\bW,\bepia)]&=\E\big[[1-Ag^{-1}(\bX^\top\bepia)]\bX\big],\\
    \nabla_{\bepib}\E[\ell_2(\bW,\bepib)]&=\E\big[[(1-A)\{1-g(\bX^\top\bepib)\}^{-1}-1]\bX\big],\\
    \nabla_{\beq}\E[\ell_3(\bW,\bnu_3,\beq)]&=\E\big[\big\{Ag^{-1}(\bX^\top\bepia)\exp(\bS^\top\beq)\\
    &\qquad\quad-(1-A)\{1-g(\bX^\top\bepib)\}^{-1}\big\}\bS\big],\\
    \nabla_{\bemu}\E[\ell_4(\bW,\bnu_4,\bemu)]&=-2\E\big[Ag^{-1}(\bX^\top\bepia)\exp(\bS^\top\beq)(Y-\bS^\top\bemu)\bS\big],\\
    \nabla_{\bettau}\E[\ell_5(\bW,\bnu_5,\bettau)]&=-2\E\big[(1-A)\exp(\bX^\top\bepib)(\bS^\top\bemu-\bX^\top\bettau)\bX\big],\\
    \nabla_{\betau}\E[\ell_6(\bW,\bnu_6,\betau)]&=-2\E\big[A\exp(-\bX^\top\bepia)\{\exp(\bS^\top\beq)(Y-\bS^\top\bemu)\\
    &\qquad\qquad+\bX^\top\bettau-\bX^\top\betau\}\bX\big].
\end{align*}

\subsection{Auxiliary lemmas}\label{apdx:MQR_aux}

\begin{lemma}[Lemma~D.1 of \cite{HighDimensionalMEstimation2019}]\label{lem:psi2norm}
    Let $X, Y\in \R$ be random variables. If $|X|\leq|Y|$ almost surely, then $\|X\|_{\psi_2}\leq\|Y\|_{\psi_2}$ and $\|cX\|_{\psi_2}=|c|\|X\|_{\psi_2} \ \forall c\in\R $.  If $\|X\|_{\psi_2}\leq\sigma$, then $\E(|X|)\leq \sigma\sqrt\pi$ and $\E(|X|^m)\leq2\sigma^m\Gamma(m/2+1)$, $\forall m\geq 2$, where $\Gamma(a)=\int_0^\infty x^{a-1}\exp(-x)\d x$. Additionally, $\E[\exp(t\{X-\E[X]\})]\leq \exp(2\sigma^2t^2)$. Conversely, if $\E[|X|^{2k}]\leq 2\sigma^{2k}\Gamma(k+1)$, $\forall k\in\N$, then $\normg{X}\leq 2\sigma$.
    Let $\{X_i\}_{i\in[n]}$ be random variables (possibly dependent) with $\max_{i\in[n]}\|X_i\|_{\psi_2}\leq\sigma$; then $\|\max_{i\in[n]}|X_i|\|_{\psi_2}\leq\sigma(\log n+2)^{1/2}$, and $\P(\max_{i\in[n]}|X_i|>c\sigma (\log n)^{1/2})\leq 2n^{-(c^{1/2}-1)}$ for any $c>1$.
\end{lemma}

\begin{lemma}\label{lem:momentmatrix}
    Suppose Assumption~\ref{cond:basic2} holds. For any symmetric matrix $\bs{B}$, let $\lambda_{\min}(\bs{B})$ denote its smallest eigenvalue. Then there exists a constant $c_{\min}>0$ such that

    \emph{(i)} $\lambda_{\min}(\E[A\bS\bS^\top])\geq c_{\rm{\min}}$ and $\normg{A\bs{v}^\top\bS}\leq 2\sigma_{\bS}\|\bs{v}\|_2$ hold for all $\bs{v}\in\R^d$.

    \emph{(ii)} $\lambda_{\min}(\E[A\bX\bX^\top])\geq c_{\rm{\min}}$ and $\normg{A\bs{v}^\top\bX}\leq 2\sigma_{\bS}\|\bs{v}\|_2$ hold for all $\bs{v}\in\R^{d_1}$.

    \emph{(iii)} $\lambda_{\min}(\E[(1-A)\bX\bX^\top])\geq c_{\rm{\min}}$ and $\normg{(1-A)\bs{v}^\top\bX}\leq 2\sigma_{\bS}\|\bs{v}\|_2$ hold for all $\bs{v}\in\R^{d_1}$.
\end{lemma}

\begin{lemma}\label{lem:usebound}
    Let Assumptions \ref{cond:basic} and \ref{cond:basic2} hold. Let $d_0\in\{d_1,d\}$ and $\{G_{N}\}$ be a positive sequence (possibly depending on $d_1$ and $d_2$). Let $r>0$ be any positive integer. Let $\balpha$ be one of $(\bepiar,\bepibr,\beqr,\bettaur, \bemur,\betaur)$ and $\halpha$ be the corresponding estimator such that $\|\balpha-\halpha\|_2=O_p(G_{N})$ as $N$, $d_1$, $d_2\to\infty$.
    Let $\bU$ denote the generic random vector $\bX$ when $d_0=d_1$, and $\bS$ when $d_0=d$; by the preceding definition, $\bU$ is independent of $\halpha$.

    \emph{(a)} When $d_0=d_1$ and $\balpha \in \{\bepiar, \bepibr, \betaur, \bettaur\}$, or when $d_0=d$ and $\balpha \in \{\beqr,\bemur\}$, as $N$, $d_1$, $d_2\to\infty$,
    $
        \|\bU^\top(\balpha-\halpha)\|_{\P,r}=O_p(G_{N}).
    $
    Particularly, when $d_0=d$, we have $\|\bzS^\top(\balpha-\halpha)\|_{\P,r}=O_p(G_N)$.

    \emph{(b)} Suppose $G_N=O(1)$ as $N$, $d_1$, $d_2\to\infty$. When $d_0=d_1$ and $\balpha \in \{\bepiar, \bepibr\}$, or when $d_0=d$ and $\balpha=\beqr$, we have
    \begin{align*}
        \normp{\exp(-\bU^\top\balpha)-\exp(-\bU^\top\halpha)}{r}&=\normp{g^{-1}(\bU^\top\balpha)-g^{-1}(\bU^\top\halpha)}{r}=O_p(G_{N}),\\
        \normp{\exp(\bU^\top\balpha)-\exp(\bU^\top\halpha)}{r}&=\normp{\frac{1}{1-g(\bU^\top\balpha)}-\frac{1}{1-g(\bU^\top\halpha)}}{r}=O_p(G_{N}).
    \end{align*}
\end{lemma}

\begin{lemma}[Adapted from Lemma~S.5 of \cite{DynamicTreatmentEffects2023}]\label{lem:useinrsc}
    Suppose that $\mathbb{S}'=(\bU_i)_{i\in\mathcal J}$ are independent and identically distributed sub-Gaussian random vectors, i.e., $\|\bs{a}^\top\bU\|_{\psi_2}\leq\sigma_{\bU}\|\bs{a}\|_2$ for all $\bs{a}\in\R^d$ with some constant $\sigma_{\bU}>0$. Additionally, suppose $\lambda_{\min}\{\E(\bU\bU^\top)\}$ is bounded below by some constant $\lambda_{\bU}>0$. Let $M=|\mathcal J|$. For any continuous function $h:\R\to(0,\infty)$ and $\bs{\eta}\in\R^d$ satisfying $\E\{|\bU^\top\bs{\eta}|^c\}<C_1$ with some constants $c,C_1>0$, there exist constants $\kappa_1,\kappa_2,c_1,c_2>0$ such that
    \begin{align*}
    &\P_{\mathbb{S}'}\Big(M^{-1}\sum_{i\in\mathcal J}h(\bU_i^\top\{\bs{\eta}+v\bDelta\})(\bU_i^\top\bDelta)^2\geq\kappa_1\|\bDelta\|_2^2-\kappa_2\frac{\log d}{M}\|\bDelta\|_1^2,\\
    &\qquad\qquad\forall v\in[0,1],\;\forall\|\bDelta\|_2\leq1\Big)\geq1-c_1\exp(-c_2M).
    \end{align*}
\end{lemma}

\begin{lemma}[Lemma~S.7 of \cite{DynamicTreatmentEffects2023}]\label{lemma:beta1}
Let $(\bX_i)_{i\in[m]}$ be i.i.d. sub-Gaussian random vectors in $\R^d$. Then, as $m,d\to\infty$,
$$\sup_{\bDelta\in\R^{d}\setminus\{\bz\}}\frac{m^{-1}\sum_{i=1}^m(\bX_i^\top\bDelta)^2}{\|\bDelta\|_1^2m^{-1}\log d+\|\bDelta\|_2^2}=O_p(1).$$
\end{lemma}

\begin{lemma}[Lemma~S.8 of \cite{DynamicTreatmentEffects2023}]\label{lemma:sol}
    Suppose $a,b,c,x\in\R$, $a>0$, and $b,c\geq0$. Let $ax^2-bx-c\leq0$. Then
        $x\leq b/a+\sqrt{c/a}.$
\end{lemma}

\begin{lemma}\label{lem:r_common_gradient}
Let $(U_i,V_i,\bs{Z}_i)_{i=1}^m$ be i.i.d. copies of a random vector $(U,V,\bs{Z})\in\R\times\R\times\R^p$ with $p\geq2$ satisfying $\E(UV\bs{Z})=\bz_p$. Suppose there exists a constant $\sigma_Z>0$ such that $\|\bs{Z}^\top\bv\|_{\psi_2}\leq\sigma_Z\|\bv\|_2$ for all $\bv\in\R^p$.

\emph{(a)} Suppose $\|U\|_{\P,8}<\infty$ and $\|V\|_{\P,8}<\infty$. Then there exists a constant $C_1>0$ depending only on $\sigma_Z$ such that
\begin{align}
    \E\left[\left\|\frac{1}{m}\sum_{i=1}^m U_iV_i\bs{Z}_i\right\|_\infty^2\right]
    \leq C_1\frac{\log p}{m}\|U\|_{\P,8}^2\|V\|_{\P,8}^2\left(1+\frac{\log p}{\sqrt{m}}\right).
    \label{eq:r_common_multiplier}
\end{align}

\emph{(b)} Suppose $|U|\leq b$ almost surely and $\|V\|_{\psi_2}\leq\sigma_V$ for some constants $b,\sigma_V>0$. Then there exists a constant $C_2>0$ depending only on $b,\sigma_V,\sigma_Z$ such that
\begin{align*}
    \P\left(\left\|\frac{1}{m}\sum_{i=1}^m U_iV_i\bs{Z}_i\right\|_\infty\leq C_2\sqrt{\frac{\log p}{m}}\right)
    \geq 1-\frac{2}{p}-\exp(-m).
\end{align*}
\end{lemma}

\begin{lemma}
\label{lem:mqr2_weighted_quadratic}
Let $(U_i,\bs{Z}_i)_{i=1}^m$ be i.i.d. copies of a random vector $(U,\bs{Z})\in\R\times\R^p$. Suppose there exist constants $\sigma_U,\sigma_Z>0$ such that $\|U\|_{\psi_2}\leq\sigma_U$ and $\|\bs{Z}^\top\bv\|_{\psi_2}\leq\sigma_Z\|\bv\|_2$ for all $\bv\in\R^p$. For any $t\geq1$, define $a_{m,p}(t) := \sqrt{(t+\log p)/m} + (t+\log p)^2/m$. Then, there exists a constant $C>0$ depending only on $\sigma_U$ and $\sigma_Z$ such that, with probability at least $1-2\exp(-t)$,
\begin{align*}
    \frac{1}{m}\sum_{i=1}^m U_i^2(\bs{Z}_i^\top\bDelta)^2
    \leq C \left\{ \|\bDelta\|_2^2 + a_{m,p}(t)\|\bDelta\|_1^2 \right\} \quad \text{for all } \bDelta\in\R^p.
\end{align*}
\end{lemma}

\begin{lemma}\label{lem:smallball}
Let $\bU\in\R^p$ be a random vector. Suppose there exist constants $\sigma_{\bU},c_{\min}>0$ such that $\|\bU^\top\bv\|_{\psi_2}\leq\sigma_{\bU}\|\bv\|_2$ and $\E[(\bU^\top\bv)^2]\geq c_{\min}\|\bv\|_2^2$ for all $\bv\in\R^p$. Then there exists a constant $p_0>0$ depending only on $\sigma_{\bU}$ and $c_{\min}$ such that
\begin{align*}
\P\Big(|\bU^\top\bv|\geq\sqrt{c_{\min}/2}\,\|\bv\|_2\Big)\geq p_0\quad\text{for all }\bv\in\R^p.
\end{align*}
\end{lemma}

\subsection{Proofs of the results in Section \ref{sec:r_thm}}
\label{apdx:proof_r_normality}

As in \eqref{decom:QR_theta}, it holds that
\begin{align}
\hat{\theta}_{\mathrm{MQR}}-\theta&=N^{-1}\sum_{k=1}^{\K}\sum_{i\in\I_k}\mpsi(\bW_i,\hbnu^{-k})-\theta\nonumber\\
&=\Delta_0'+\K^{-1}\sum\limits_{k=1}^{\K}\Delta_{k,1}'+\K^{-1}\sum\limits_{k=1}^{\K}\Delta_{k,2}'+\K^{-1}\sum\limits_{k=1}^{\K}\Delta_{k,3}',\label{decom:r_theta}
\end{align}
where for each $k\in[\K]$,
\begin{align*}
\Delta_0'\;\;&=\E[\phi(\bW,\bnu^*)]-\theta,\\
\Delta_{k,1}'&=n^{-1}\sum\limits_{j\in\I_k}\{\mpsi(\bW_j,\bnu^*)-\E[\phi(\bW,\bnu^*)]\},\\
\Delta_{k,2}'&=n^{-1}\sum_{j\in\I_k}\left\{\mpsi(\bW_j,\hbnu^{-k})-\mpsi(\bW_j,\bnu^*)\right\}-\E\left\{\mpsi(\bW,\hbnu^{-k})-\mpsi(\bW,\bnu^*)\right\},\\
\Delta_{k,3}'&=\E\left\{\mpsi(\bW,\hbnu^{-k})-\mpsi(\bW,\bnu^*)\right\}.
\end{align*}

Unless otherwise specified, expectations $\E[\cdot]$ hereafter are evaluated over a new evaluation sample $\bW = (Y,A,\bX,\bM)$,  which is independent of the observed samples $(\bW_i)_{i\in[N]}$, and therefore the nuisance estimators. Equivalently, the expectation can be viewed as conditional on $\mathcal{F}_{-k}=\sigma(\hbnu^{-k})$.

For the analysis of the MQR estimators, we characterize the stochastic errors of the nuisance estimates by the $\ell_2$ norms of their parameter deviations from $\bnu^*$. With a slight abuse of notation, we still denote them by $\rmqr_h$ for $h\in\mathcal{H}$. Specifically, suppressing the fold index $-k$,
\begin{equation}
\begin{aligned}
&\rmqr_\pi:=\max\{\|\hbepiar-\bepiar\|_2,\|\hbepibr-\bepibr\|_2\},\;\;\;\;\;\;
\rmqr_q:=\|\hbeqr-\beqr\|_2,\\
&\rmqr_\mu:=\|\hbemur-\bemur\|_2,\;\;\;
\rmqr_{\tau_n}:=\|\hbettaur-\bettaur\|_2,\;\;\;
\rmqr_\tau:=\|\hbetaur-\betaur\|_2.
\end{aligned}\label{def:rmqr}
\end{equation}
The convergence rates of these error terms are established in Theorem~\ref{thm:r_nuisance} and Theorem~\ref{thm:mqr2_r_nuisance} for MQR1 and MQR2, respectively.

The following lemma follows the same arguments as those used in Lemma~\ref{lem:Delta_0} for the QR estimator. The main difference is that the score in \eqref{def:psihat-MQR}, rather than that in \eqref{def:QR_form}, is employed. Nevertheless, these modifications affect only the constants appearing in the bounds and leave the structure of the argument unchanged.

\begin{lemma}
    Let Assumptions~\ref{cond:basic} and \ref{cond:basic2} hold. Then we have\label{lem:r_Delta_0}
    \begin{align*}
    \Delta_{0}'=O(\bemqr_\pi\{\bemqr_{\tau_n}+\bemqr_\tau\}+\bemqr_q\bemqr_\mu).
    \end{align*}
\end{lemma}

Before proving Theorem~\ref{thm:r_bias}, we first analyze the terms $\Delta_{k,2}'$ and $\Delta_{k,3}'$ in Lemmas~\ref{lem:r_delta_2} and \ref{lem:r_delta_3}, respectively.

\begin{lemma}
Let Assumptions~\ref{cond:basic} and \ref{cond:basic2} hold. Suppose that $\rmqr_\pi+\rmqr_q+\rmqr_\mu+\rmqr_{\tau_n}+\rmqr_\tau=o_p(1)$. For each $k\in[\K]$, as $N$, $d_1$, $d_2\to\infty$, \label{lem:r_delta_2}
$
    \Delta_{k,2}'=O_p(n^{-1/2}\{\sqrt{\rmqr_{\pi}}+\sqrt{\rmqr_{q}}+\sqrt{\rmqr_{\mu}}+\sqrt{\rmqr_{\tau_n}}+\sqrt{\rmqr_{\tau}}\}).
$
\end{lemma}
\begin{proof}
First, we have $\E_{\mathbb S_k}(\Delta_{k,2}')=0$. By Chebyshev's inequality, it suffices to show
\begin{align}
    \E_{\mathbb S_k}\left(\Delta_{k,2}'^2\right)=O_p\left(n^{-1}\left\{\rmqr_{\pi}+\rmqr_{q}+\rmqr_{\mu}+\rmqr_{\tau_n}+\rmqr_{\tau}\right\}\right).\label{order:r_square1delta}
\end{align}
With some $\tilde{\bnu}=(\tbepia^{\top},\tbepib^{\top}, \tbeq^{\top},\tbemu^{\top},\tbettau^{\top},\tbetau^\top)^\top$ lying between $\bnu^*$ and $\hbnu^{-k}$, we apply Taylor's theorem to obtain
\begin{align}
    &\E_{\mathbb S_k}\left(\Delta_{k,2}'^2\right)\leq n^{-1}\E\left[\left\{\mpsi(\bW,\hbnu^{-k})-\mpsi(\bW,\bnu^*)\right\}^2\right]\nonumber\\
    &=2n^{-1}\E\left[\left\{\mpsi(\bW,\tilde{\bnu})-\mpsi(\bW,\bnu^*)\right\}\bnabla_{\bnu}\mpsi(\bW,\tilde{\bnu})^\top(\hbnu^{-k}-\bnu^*)\right]\nonumber\\
    &\leq 2n^{-1}\left\{\left\|\mpsi(\bW,\tilde{\bnu})-\bX^\top\betaur\right\|_{\P,2}+\left\|\mpsi(\bW,\bnu^*)-\bX^\top\betaur\right\|_{\P,2}\right\}\left\|\bnabla_{\bnu}\mpsi(\bW,\tilde{\bnu})^\top(\hbnu^{-k}-\bnu^*)\right\|_{\P,2}.\label{bound:Delta_k_1}
\end{align}
We now bound the right-hand side above.  Firstly, by Minkowski's inequality and H\"older's inequality,
\begin{align}
    &\normp{\mpsi(\bW,\bnu^*)-\bX^\top\betaur}{2}\nonumber\\
    \leq & \left\|\frac{1}{g(\bX^\top\bepiar)}\right\|_{\P,6}\left\|\exp(\bS^\top\beqr)\right\|_{\P,6}\left\|\varepsilon\right\|_{\P,6}+\left\|\frac{1}{1-g(\bX^\top\bepibr)}\right\|_{\P,4}\left\|\varrho\right\|_{\P,4}\nonumber\\
    &+\left\|\frac{1}{g(\bX^\top\bepiar)}\right\|_{\P,4}(\left\|\varrho\right\|_{\P,4}+\left\|\zeta\right\|_{\P,4})\nonumber\\
    \leq & c_0^{-3}\normp{\varepsilon}{6}+c_0^{-1}\normp{\varrho}{4}+c_0^{-1}(\normp{\zeta}{4}+\normp{\varrho}{4})\overset{(i)}{=}O_p(1), \label{bound:delta_k_1_part1}
\end{align}
where (i) holds under Assumption~\ref{cond:basic2} (b). Secondly, we define
\begin{align*}
    \tvarepsilon:=Y(1,\bM)-\bS^\top\tbemu,\;\;\tzeta:=\bzS^\top\tbemu-\bX^\top\tbetau,\;\;
    \tvarrho:=\bzS^\top\tbemu-\bX^\top\tbettau.
\end{align*}
Note that $\tzeta-\tvarrho=\bX^\top(\tbettau-\tbetau)$.
Although $\tilde{\bnu}$ may depend on $\bW$, convexity of the exponential function gives $e^{(1-t)u+tv}\leq e^u+e^v$ for all $u,v\in\R$ and $t\in[0,1]$, so $\exp(-\bX^\top\tbepia)$, $\exp(\bX^\top\tbepib)$, and $\exp(\bS^\top\tbeq)$ are each bounded pointwise by the sum of their values at $\bnu^*$ and $\hbnu^{-k}$. Similarly,
\begin{align*}
|\tvarepsilon|&\leq|\varepsilon|+|\bS^\top(\hbemu^{-k}-\bemur)|,\\
|\tzeta|&\leq|\zeta|+|\bzS^\top(\hbemu^{-k}-\bemur)|+|\bX^\top(\hbetau^{-k}-\betaur)|,\\
|\tvarrho|&\leq|\varrho|+|\bzS^\top(\hbemu^{-k}-\bemur)|+|\bX^\top(\hbettau^{-k}-\bettaur)|.
\end{align*}
By Lemma~\ref{lem:usebound}, nuisance consistency, and Assumption~\ref{cond:basic2}, it follows that the $L_r(\P)$ norms of $\exp(-\bX^\top\tbepia)$, $\exp(\bX^\top\tbepib)$, $\exp(\bS^\top\tbeq)$, $\tvarepsilon$, $\tzeta$, and $\tvarrho$ are all $O_p(1)$ for any fixed $r>0$.
We have
\begin{align}
    &\left\|\mpsi(\bW,\tilde{\bnu})-\bX^\top\betaur\right\|_{\P,2}\nonumber\\
    \leq&\left\|\frac{1}{g(\bX^\top\tbepia)}\right\|_{\P,6}\left\|\exp(\bS^\top\tbeq)\right\|_{\P,6}\left\|\tvarepsilon\right\|_{\P,6}+\left\|\frac{1}{1-g(\bX^\top\tbepib)}\right\|_{\P,4}\left\|\tvarrho\right\|_{\P,4}\nonumber\\
    &+\left\|\frac{1}{g(\bX^\top\tbepia)}\right\|_{\P,4}(\normp{\tvarrho}{4}+\normp{\tzeta}{4})+\left\|\bX^\top(\tbetau-\betaur)\right\|_{\P,2}\nonumber\\
    \overset{(i)}{=}&O_p\left(1+\|\hbetau^{-k}-\betaur\|_2\right)\overset{(ii)}{=}O_p(1),\label{bound:delta_k_1_part2}
\end{align}
where (i) holds by the preceding pointwise bounds and Assumption~\ref{cond:basic2}(b), and (ii) holds by nuisance consistency.

Lastly, by Minkowski's inequality, we have
\begin{align}
    &\normp{\bnabla_{\bnu}\mpsi(\bW,\tilde{\bnu})^\top(\hbnu^{-k}-\bnu^*)}{2}\nonumber\\
    &\quad\leq\normp{\bnabla_{\bepia}\mpsi(\bW,\tilde{\bnu})^\top(\hbepia^{-k}-\bepiar)}{2}+\normp{\bnabla_{\bepib}\mpsi(\bW,\tilde{\bnu})^\top(\hbepib^{-k}-\bepibr)}{2}\nonumber\\
    &\quad\quad+\normp{\bnabla_{\beq}\mpsi(\bW,\tilde{\bnu})^\top(\hbeq^{-k}-\beqr)}{2}+\normp{\bnabla_{\bemu}\mpsi(\bW,\tilde{\bnu})^\top(\hbemu^{-k}-\bemur)}{2}\nonumber\\
    &\quad\quad+\normp{\bnabla_{\bettau}\mpsi(\bW,\tilde{\bnu})^\top(\hbettau^{-k}-\bettaur)}{2}+\normp{\bnabla_{\betau}\mpsi(\bW,\tilde{\bnu})^\top(\hbetau^{-k}-\betaur)}{2}.\nonumber
\end{align}
In what follows, we provide bounds for each of these six terms on the right-hand side.
By Lemma~\ref{lem:usebound}, it holds for any integer $s>0$ that
\begin{align}
    \normp{\bX^\top(\hbepia^{-k}-\bepiar)}{s}&=O_p\left(\|\hbepia^{-k}-\bepiar\|_2\right)=O_p\left(\rmqr_\pi \right),\nonumber\\
    \normp{\bX^\top(\hbepib^{-k}-\bepibr)}{s}&=O_p\left(\|\hbepib^{-k}-\bepibr\|_2\right)=O_p\left(\rmqr_\pi \right),\nonumber\\
    \normp{\bS^\top(\hbeq^{-k}-\beqr)}{s}
    &=O_p\left(\|\hbeq^{-k}-\beqr\|_2\right)=O_p\left(\rmqr_{q}\right),\nonumber\\
    \normp{\bS^\top(\hbemu^{-k}-\bemur)}{s}
    &=O_p\left(\|\hbemu^{-k}-\bemur\|_2\right)=O_p\left(\rmqr_{\mu}\right),\nonumber\\
    \normp{\bX^\top(\hbettau^{-k}-\bettaur)}{s}&=O_p\left(\|\hbettau^{-k}-\bettaur\|_2\right)=O_p\left(\rmqr_{\tau_n}\right),\nonumber\\
    \normp{\bX^\top(\hbetau^{-k}-\betaur)}{s}
    &=O_p\left(\|\hbetau^{-k}-\betaur\|_2\right)=O_p\left(\rmqr_{\tau}\right).\label{order:gau_linear_use}
\end{align}
Firstly, we have
\begin{align*}
    &\left\|\bnabla_{\bepia}\mpsi(\bW,\tilde{\bnu})^\top(\hbepia^{-k}-\bepiar)\right\|_{\P,2}\\
    \overset{(i)}{\leq}&\left\|\exp(-\bX^\top\tbepia)\right\|_{\P,8}\left\{\left\|-\tzeta+\tvarrho\right\|_{\P,4}+\left\|\exp(\bS^\top\tbeq)\right\|_{\P,8}\left\|\tvarepsilon\right\|_{\P,8}\right\}\left\|\bX^\top(\hbepia^{-k}-\bepiar)\right\|_{\P,8}\\
    \overset{(ii)}{=}&O_p\left(\|\hbepia^{-k}-\bepiar\|_2\right)=O_p(\rmqr_\pi),
\end{align*}
where (i) holds by H\"older's inequality and Minkowski's inequality, (ii) holds by the preceding pointwise bounds and \eqref{order:gau_linear_use}. Similarly, for the second term, we have
\begin{align*}
    &\left\|\bnabla_{\bepib}\mpsi(\bW,\tilde{\bnu})^\top(\hbepib^{-k}-\bepibr)\right\|_{\P,2}\\
    \leq&\left\|\exp(\bX^\top\tbepib)\right\|_{\P,4}\left\|\tvarrho\right\|_{\P,8}\left\|\bX^\top(\hbepib^{-k}-\bepibr)\right\|_{\P,8}\overset{(i)}{=}O_p\left(\|\hbepib^{-k}-\bepibr\|_2\right)=O_p(\rmqr_\pi),
\end{align*}
where (i) follows from the preceding pointwise bounds and \eqref{order:gau_linear_use}, and the fact that $(1-A)\bS=(1-A)\bzS$.
For the remaining terms, we have
\begin{align*}
    &\left\|\bnabla_{\beq}\mpsi(\bW,\tilde{\bnu})^\top(\hbeq^{-k}-\beqr)\right\|_{\P,2}\\
    \leq& \left\{1+\left\|\exp(-\bX^\top\tbepia)\right\|_{\P,6}\right\}\left\|\exp(\bS^\top\tbeq)\right\|_{\P,6}\left\|\tvarepsilon\right\|_{\P,12}\left\|\bS^\top(\hbeq^{-k}-\beqr)\right\|_{\P,12}\\
    =&O_p\left(\|\hbeq^{-k}-\beqr\|_2\right)=O_p(\rmqr_q),\\
    &\left\|\bnabla_{\bemu}\mpsi(\bW,\tilde{\bnu})^\top(\hbemu^{-k}-\bemur)\right\|_{\P,2}\\
    \leq&\big\{(1+\|\exp(-\bX^\top\tbepia)\|_{\P,6})\|\exp(\bS^\top\tbeq)\|_{\P,6}+1+\|\exp(\bX^\top\tbepib)\|_{\P,3}\big\}\|\bS^\top(\hbemu^{-k}-\bemur)\|_{\P,6}\\
    =&O_p\left(\|\hbemu^{-k}-\bemur\|_2\right)=O_p(\rmqr_\mu),\\
    &\left\|\bnabla_{\bettau}\mpsi(\bW,\tilde{\bnu})^\top(\hbettau^{-k}-\bettaur)\right\|_{\P,2}\\
    \leq&\left\{2+\left\|\exp(-\bX^\top\tbepia)\right\|_{\P,4}+\left\|\exp(\bX^\top\tbepib)\right\|_{\P,4}\right\}\left\|\bX^\top(\hbettau^{-k}-\bettaur)\right\|_{\P,4}\\
    =&O_p\left(\|\hbettau^{-k}-\bettaur\|_2\right)=O_p(\rmqr_{\tau_n}),\\
    &\left\|\bnabla_{\betau}\mpsi(\bW,\tilde{\bnu})^\top(\hbetau^{-k}-\betaur)\right\|_{\P,2}\\
    \leq&\left\{2+\left\|\exp(-\bX^\top\tbepia)\right\|_{\P,4}\right\}\left\|\bX^\top(\hbetau^{-k}-\betaur)\right\|_{\P,4}\\
    =&O_p\left(\|\hbetau^{-k}-\betaur\|_2\right)=O_p(\rmqr_{\tau}).
\end{align*}
Combining these bounds with \eqref{bound:Delta_k_1}, \eqref{bound:delta_k_1_part1}, and \eqref{bound:delta_k_1_part2}, we obtain
\begin{align*}
    \E_{\mathbb S_k}\left(\Delta_{k,2}'^2\right)
    &=O_p\left(n^{-1}\left\{\rmqr_{\pi}+\rmqr_{q}+\rmqr_{\mu}+\rmqr_{\tau_n}+\rmqr_{\tau}\right\}\right).
\end{align*}
This establishes \eqref{order:r_square1delta}.
The same bounds in \eqref{bound:Delta_k_1} also give
\begin{align*}
\E\left[\{\mpsi(\bW,\hbnu^{-k})-\mpsi(\bW,\bnu^*)\}^2\mid\mathcal F_{-k}\right]=o_p(1).
\end{align*}
\end{proof}

\begin{lemma}\label{lem:r_delta_3}
Let Assumptions~\ref{cond:basic} and \ref{cond:basic2} hold, and let the population targets satisfy the first-order conditions of $\ell_1$--$\ell_6$ in \eqref{eq:r_loss1}--\eqref{eq:r_loss6}. Suppose that $\rmqr_\pi+\rmqr_q+\rmqr_\mu+\rmqr_{\tau_n}+\rmqr_\tau=o_p(1)$. For each $k\in[\K]$, as $N$, $d_1$, $d_2\to\infty$,
\begin{align*}
\Delta_{k,3}'=O_p\big(&\rmqr_q\rmqr_\mu+\rmqr_\pi\rmqr_\tau
+\rmqr_\pi\rmqr_{\tau_n}+\rmqr_\pi\rmqr_\mu
+\bemqr_\pi\rmqr_\pi\rmqr_\tau+\bemqr_q\rmqr_\pi\rmqr_\mu\\
&+\bemqr_\mu\{\rmqr_\pi\rmqr_q+\rmqr_q^2\}
+(\bemqr_\mu+\bemqr_{\tau_n}+\bemqr_\tau)\rmqr_\pi^2\big).
\end{align*}
\end{lemma}

\begin{proof}
For each $k\in[\K]$, we consider the following decomposition, which slightly differs from \eqref{decom:QR_delta2}:
\begin{align*}
    \Delta_{k,3}'&=\sum_{l=4}^{14}\Delta_{k,l}',
\end{align*}
where
\begin{align*}
    \Delta_{k,4}'&=\E\left[\frac{A}{g(\bX^\top\hbepia^{-k})}\{\exp(\bS^\top\beqr)-\exp(\bS^\top\hbeq^{-k})\}
    \bS^\top(\hbemu^{-k}-\bemur)\right],\\
    \Delta_{k,5}'&=\E\left[\left\{1-\frac{g(\bX^\top\bepiar)}{g(\bX^\top\hbepia^{-k})}\right\}\bX^\top(\hbetau^{-k}-\betaur)\right],\\
    \Delta_{k,6}'&=\E\left[\left\{\frac{1-A}{1-g(\bX^\top\hbepib^{-k})}-\frac{A}{g(\bX^\top\hbepia^{-k})}\right\}\bX^\top(\bettaur-\hbettau^{-k})\right],\\
    \Delta_{k,7}'&=\E\left[\left\{\frac{1-A}{1-g(\bX^\top\hbepib^{-k})}-\frac{1-A}{1-g(\bX^\top\bepibr)}\right\}\bS^\top(\hbemu^{-k}-\bemur)\right],\\
    \Delta_{k,8}'&=\E\left[\left\{\frac{A}{g(\bX^\top\bepiar)}-\frac{A}{g(\bX^\top\hbepia^{-k})}\right\}\bzS^\top(\hbemu^{-k}-\bemur)\right],\\
    \Delta_{k,9}'&=\E\left[\frac{1}{g(\bX^\top\hbepia^{-k})}\{g(\bX^\top\bepiar)-A\}\bX^\top(\hbetau^{-k}-\betaur)\right],\\
    \Delta_{k,10}'&=\E\left[\left\{\frac{A}{g(\bX^\top\hbepia^{-k})}-\frac{A}{g(\bX^\top\bepiar)}\right\}\left\{\bzS-\exp(\bS^\top\beqr)\bS\right\}^\top(\hbemu^{-k}-\bemur)\right],\\
    \Delta_{k,11}'&=\E\left[\frac{A}{g(\bX^\top\hbepia^{-k})}\{\exp(\bS^\top\hbeq^{-k})-\exp(\bS^\top\beqr)\}\{Y-\bS^\top\bemur\}\right],\\
    \Delta_{k,12}'&=\E\left[\left\{\frac{1-A}{1-g(\bX^\top\hbepib^{-k})}-\frac{1-A}{1-g(\bX^\top\bepibr)}\right\}\{\bS^\top\bemur-\bX^\top\bettaur\}\right],\\
    \Delta_{k,13}'&=\E\left[\left\{\frac{A}{g(\bX^\top\hbepia^{-k})}-\frac{A}{g(\bX^\top\bepiar)}\right\}\{\exp(\bS^\top\beqr)\varepsilon+\zeta-\varrho\}\right],\\
    \Delta_{k,14}'&=\E\left[\left\{\frac{1-A}{1-g(\bX^\top\bepibr)}-\frac{A}{g(\bX^\top\bepiar)}\exp(\bS^\top\beqr)\right\}\bS^\top(\hbemu^{-k}-\bemur)\right].
\end{align*}

By the $L_6$ definition of the approximation errors and Lyapunov's inequality, $\max_{j\in\{a,b\}}\|\pi-\pi_j^*\|_{\P,p}\leq \bemqr_\pi$ and $\|h-h^*\|_{\P,p}\leq \bemqr_h$ for $h\in\{q,\mu,\tau_n,\tau\}$ and every $1\leq p\leq6$.
First, we have
\begin{align*}
    \Delta_{k,14}'&=(\hbemu^{-k}-\bemur)^\top\E\left[\left\{\frac{1-A}{1-g(\bX^\top\bepibr)}-\frac{A}{g(\bX^\top\bepiar)}\exp(\bS^\top\beqr)\right\}\bS\right]=0,
\end{align*}
where the last equation holds by the loss function construction \eqref{eq:r_eq_mu}.

Furthermore, by Lemma~\ref{lem:usebound}, nuisance consistency and the probability-one weight bounds in Assumption~\ref{cond:basic2}(a), for any fixed integer $r\geq1$, we have
\begin{align}
    \normp{g(\bX^\top\bepiar)-g(\bX^\top\hbepia^{-k})}{r}&=O_p\left(\|\hbepia^{-k}-\bepiar\|_2\right)=O_p(\rmqr_\pi),\nonumber\\
    \normp{\exp(-\bX^\top\bepiar)-\exp(-\bX^\top\hbepia^{-k})}{r}&=O_p\left(\|\hbepia^{-k}-\bepiar\|_2\right)=O_p(\rmqr_\pi),\nonumber\\
    \normp{\exp(\bX^\top\bepibr)-\exp(\bX^\top\hbepib^{-k})}{r}&=O_p\left(\|\hbepib^{-k}-\bepibr\|_2\right)=O_p(\rmqr_\pi),\nonumber\\
    \normp{\exp(\bS^\top\beqr)-\exp(\bS^\top\hbeq^{-k})}{r}
    &=O_p\left(\|\hbeq^{-k}-\beqr\|_2\right)=O_p(\rmqr_q).
    \label{order:r_inferuse}
\end{align}
Moreover, for $(\bepi^*,\hbepi^{-k})\in\{(\bepiar,\hbepia^{-k}),(\bepibr,\hbepib^{-k})\}$, by convexity of the exponential function,
\begin{align}
&\left\|\sup_{0\leq t\leq1}\exp\{\pm\bX^\top((1-t)\bepi^*+t\hbepi^{-k})\}\right\|_{\P,12}\nonumber\\
&\leq\|\exp(\pm\bX^\top\bepi^*)\|_{\P,12}+\|\exp(\pm\bX^\top\hbepi^{-k})\|_{\P,12}=O_p(1),\nonumber\\
&\left\|\sup_{0\leq t\leq1}\exp\{\bS^\top((1-t)\beqr+t\hbeq^{-k})\}\right\|_{\P,12}\nonumber\\
&\leq\|\exp(\bS^\top\beqr)\|_{\P,12}+\|\exp(\bS^\top\hbeq^{-k})\|_{\P,12}=O_p(1).\label{order:r_taylor_exp}
\end{align}
Hence the same $O_p(1)$ bounds apply to $\exp(-\bX^\top\tbepia)$, $\exp(\bX^\top\tbepib)$, and $\exp(\bS^\top\tbeq)$ at any point lying between the target and estimated parameters in the Taylor expansions below, even when that point depends on $\bW$.

By H\"older's inequality, Minkowski's inequality and the fact that $|A|\leq1$, we have
\begin{align*}
|\Delta_{k,4}'|
&\leq\|g^{-1}(\bX^\top\hbepia^{-k})\|_{\P,4}
\|\exp(\bS^\top\hbeq^{-k})-\exp(\bS^\top\beqr)\|_{\P,4}
\|\bS^\top(\hbemu^{-k}-\bemur)\|_{\P,2}\\
&=O_p\left(\|\hbeq^{-k}-\beqr\|_2\|\hbemu^{-k}-\bemur\|_2\right)=O_p(\rmqr_q\rmqr_\mu),
\end{align*}
where the last two equations follow from \eqref{order:r_inferuse} and \eqref{order:gau_linear_use}.

Similarly, we have
\begin{align*}
    |\Delta_{k,5}'|&\leq\normp{g^{-1}(\bX^\top\hbepia^{-k})}{4}\normp{g(\bX^\top\bepiar)-g(\bX^\top\hbepia^{-k})}{4}\normp{\bX^\top(\hbetau^{-k}-\betaur)}{2}\\
    &=O_p\left(\|\hbepia^{-k}-\bepiar\|_2\|\hbetau^{-k}-\betaur\|_2\right)=O_p(\rmqr_\pi\rmqr_\tau).
\end{align*}

For the term $\Delta_{k,6}'$, the loss function construction \eqref{eq:r_eq_tau} leads to
\begin{align*}
    \E\left[\left\{\frac{A}{g(\bX^\top\bepiar)}-\frac{1-A}{1-g(\bX^\top\bepibr)}\right\}\bX\right]=\bz.
\end{align*}
Therefore, we rewrite $\Delta_{k,6}'$ as
\begin{align*}
    \Delta_{k,6}'=&\E\Bigg[\left\{\frac{1-A}{1-g(\bX^\top\hbepib^{-k})}-\frac{1-A}{1-g(\bX^\top\bepibr)}+\frac{A}{g(\bX^\top\bepiar)}-\frac{A}{g(\bX^\top\hbepia^{-k})}\right\}\\
    &\hspace{28mm}\times\bX^\top(\bettaur-\hbettau^{-k})\Bigg].
\end{align*}
By \eqref{order:r_inferuse} and Minkowski's inequality, it follows that
\begin{align*}
    |\Delta_{k,6}'|&=O_p\left(\left\{\|\hbepia^{-k}-\bepiar\|_2+\|\hbepib^{-k}-\bepibr\|_2\right\}\|\hbettau^{-k}-\bettaur\|_2\right)=O_p(\rmqr_\pi\rmqr_{\tau_n}).
\end{align*}
Similarly, the terms $\Delta_{k,7}'$ and $\Delta_{k,8}'$ can be bounded as
\begin{align*}
    |\Delta_{k,7}'|&=O_p\left(\|\hbepib^{-k}-\bepibr\|_2\|\hbemu^{-k}-\bemur\|_2\right)=O_p(\rmqr_\pi\rmqr_\mu),\\
    |\Delta_{k,8}'|&=O_p\left(\|\hbepia^{-k}-\bepiar\|_2\|\hbemu^{-k}-\bemur\|_2\right)=O_p(\rmqr_\pi\rmqr_\mu).
\end{align*}

For the term $\Delta_{k,9}'$, we define $\omega_{\pi_a}=\pi-\pi_a^*$ and note that
\begin{align*}
    \Delta_{k,9}'=-\E\left[\frac{\omega_{\pi_a}(\bX)}{g(\bX^\top\hbepia^{-k})}\bX^\top(\hbetau^{-k}-\betaur)\right].
\end{align*}
Since $\hbnu^{-k}$ can be treated as fixed conditional on $\mathcal F_{-k}$, the loss function construction \eqref{eq:r_eq_tau_n} gives
\begin{align*}
    \E\left[g^{-1}(\bX^\top\bepiar)\omega_{\pi_a}(\bX)\bX^\top(\hbetau^{-k}-\betaur)\right]=0.
\end{align*}
Therefore,
\begin{align*}
\Delta_{k,9}'&=-\E\left[\{g^{-1}(\bX^\top\hbepia^{-k})-g^{-1}(\bX^\top\bepiar)\}
\omega_{\pi_a}(\bX)\bX^\top(\hbetau^{-k}-\betaur)\right].
\end{align*}
By H\"older's inequality,
\begin{align*}
|\Delta_{k,9}'|&\leq
\|g^{-1}(\bX^\top\hbepia^{-k})-g^{-1}(\bX^\top\bepiar)\|_{\P,4}
\|\omega_{\pi_a}(\bX)\|_{\P,4}\|\bX^\top(\hbetau^{-k}-\betaur)\|_{\P,2}
=O_p(\bemqr_\pi\rmqr_\pi\rmqr_\tau).
\end{align*}
We then evaluate the term $\Delta_{k,10}'$. Define $\omega_q=q-q^*$. Note that $\E[Aq(\bS)\bS\mid\bX]=\E[A\bzS\mid\bX]$, so we can apply the tower rule to rewrite $\Delta_{k,10}'$ as
\begin{align*}
\Delta_{k,10}'&=\E\left[A\{g^{-1}(\bX^\top\hbepia^{-k})-g^{-1}(\bX^\top\bepiar)\}
\omega_q(\bS)\bS^\top(\hbemu^{-k}-\bemur)\right].
\end{align*}
Therefore, it follows from H\"older's inequality that
\begin{align*}
    |\Delta_{k,10}'|&\leq\normp{g^{-1}(\bX^\top\hbepia^{-k})-g^{-1}(\bX^\top\bepiar)}{4}\normp{\omega_q(\bS)}{4}\normp{\bS^\top(\hbemu^{-k}-\bemur)}{2}
    =O_p(\bemqr_q\rmqr_\pi\rmqr_\mu).
\end{align*}

For the term $\Delta_{k,11}'$, we define $\omega_\mu=\mu-\mu^*$. By conditioning on $(A,\bS)$ and applying the tower rule, we obtain
\begin{align*}
\Delta_{k,11}'=\E\left[Ag^{-1}(\bX^\top\hbepia^{-k})
\{\exp(\bS^\top\hbeq^{-k})-\exp(\bS^\top\beqr)\}\omega_\mu(\bS)\right].
\end{align*}
Moreover, note that $\hbeq^{-k}-\beqr$ can be treated as fixed conditional on $\mathcal F_{-k}$. The loss function construction \eqref{eq:r_eq_q} gives
$\E[Ag^{-1}(\bX^\top\bepiar)\exp(\bS^\top\beqr)\omega_\mu(\bS)\bS]=\bz$.
By Taylor's expansion, with some $\tbeq$ lying between $\beqr$ and $\hbeq^{-k}$, we have
\begin{align*}
\Delta_{k,11}'&=\E\left[A\{g^{-1}(\bX^\top\hbepia^{-k})-g^{-1}(\bX^\top\bepiar)\}
\exp(\bS^\top\beqr)\omega_\mu(\bS)\bS^\top(\hbeq^{-k}-\beqr)\right]\\
&\quad+\frac12\E\left[Ag^{-1}(\bX^\top\hbepia^{-k})
\exp(\bS^\top\tbeq)\omega_\mu(\bS)\{\bS^\top(\hbeq^{-k}-\beqr)\}^2\right].
\end{align*}

For the term $\Delta_{k,12}'$, analogous arguments hold as well. Define $\omega_{\tau_n}=\tau_n-\tau_n^*$. Note that
\begin{align*}
\E[\mu^*(\bS)-\tau_n^*(\bX)\mid A=0,\bX]
=\tau_n(\bX)-\tau_n^*(\bX)=\omega_{\tau_n}(\bX).
\end{align*}
The loss function construction \eqref{eq:r_eq_pib} leads to
\begin{align*}
\E[(1-\pi(\bX))\exp(\bX^\top\bepibr)\omega_{\tau_n}(\bX)\bX]=\bz.
\end{align*}
By Taylor's expansion, with some $\tbepib$ lying between $\bepibr$ and $\hbepib^{-k}$, we have
\begin{align*}
\Delta_{k,12}'&=\frac12\E\left[(1-\pi(\bX))
\exp(\bX^\top\tbepib)\omega_{\tau_n}(\bX)\{\bX^\top(\hbepib^{-k}-\bepibr)\}^2\right].
\end{align*}

Lastly, for the term $\Delta_{k,13}'$, define $\omega_\tau=\tau-\tau^*$. It follows from the tower rule that
\begin{align*}
&\E[q^*(\bS)\{Y-\mu^*(\bS)\}+\tau_n^*(\bX)-\tau^*(\bX)\mid A=1,\bX]\\
&\qquad=\omega_\tau(\bX)-\omega_{\tau_n}(\bX)
-\E[\omega_q(\bS)\omega_\mu(\bS)\mid A=1,\bX],
\end{align*}
because $\E[q(\bS)\omega_\mu(\bS)\mid A=1,\bX]=\tau(\bX)-\tau_n(\bX)$.
Note that \eqref{eq:r_eq_pia} holds, and that $\hbepia^{-k}-\bepiar$ can be treated as fixed conditional on $\mathcal F_{-k}$. Therefore, by Taylor's expansion, with some $\tbepia$ lying between $\bepiar$ and $\hbepia^{-k}$, we obtain
\begin{align*}
\Delta_{k,13}'&=\frac12\E\Big[\pi(\bX)
\exp(-\bX^\top\tbepia)
\{\bX^\top(\hbepia^{-k}-\bepiar)\}^2\\
&\hspace{43mm}\times
\{\omega_\tau(\bX)-\omega_{\tau_n}(\bX)
-\E[\omega_q(\bS)\omega_\mu(\bS)\mid A=1,\bX]\}\Big].
\end{align*}

Let $B(\bX):=\omega_\tau(\bX)-\omega_{\tau_n}(\bX)-\E[\omega_q(\bS)\omega_\mu(\bS)\mid A=1,\bX]$.
The probability-one bounds $|\omega_q|\leq2c_0^{-2}$ and $\pi\geq c_0$, together with Jensen's inequality, give
$\|B\|_{\P,6}\lesssim \bemqr_\mu+\bemqr_{\tau_n}+\bemqr_\tau$.
By H\"older's inequality, \eqref{order:r_inferuse}, \eqref{order:gau_linear_use}, and \eqref{order:r_taylor_exp}, we obtain
\begin{align*}
|\Delta_{k,11}'|
&\leq c_0^{-2}\|g^{-1}(\bX^\top\hbepia^{-k})-g^{-1}(\bX^\top\bepiar)\|_{\P,3}
\|\omega_\mu\|_{\P,6}\|\bS^\top(\hbeq^{-k}-\beqr)\|_{\P,2}\\
&\quad+\frac12\|g^{-1}(\bX^\top\hbepia^{-k})\|_{\P,6}
\|\exp(\bS^\top\tbeq)\|_{\P,6}\|\omega_\mu\|_{\P,6}
\|\bS^\top(\hbeq^{-k}-\beqr)\|_{\P,4}^2\\
&=O_p\left(\bemqr_\mu\{\rmqr_\pi\rmqr_q+\rmqr_q^2\}\right),\\
|\Delta_{k,12}'|
&\leq\frac12\|\exp(\bX^\top\tbepib)\|_{\P,3}\|\omega_{\tau_n}\|_{\P,6}
\|\bX^\top(\hbepib^{-k}-\bepibr)\|_{\P,4}^2
=O_p(\bemqr_{\tau_n}\rmqr_\pi^2),\\
|\Delta_{k,13}'|
&\leq\frac12\|\exp(-\bX^\top\tbepia)\|_{\P,3}\|B\|_{\P,6}
\|\bX^\top(\hbepia^{-k}-\bepiar)\|_{\P,4}^2
=O_p\left(\{\bemqr_\mu+\bemqr_{\tau_n}+\bemqr_\tau\}\rmqr_\pi^2\right).
\end{align*}
Here the reciprocal exponents are $1/3+1/6+1/2=1$ for the first term in the bound for $\Delta_{k,11}'$,
$3/6+2/4=1$ for its second term, and $1/3+1/6+2/4=1$ for the bounds for $\Delta_{k,12}'$ and $\Delta_{k,13}'$.

Collectively, we conclude that
\begin{align*}
\Delta_{k,3}'=O_p\big(&\rmqr_q\rmqr_\mu+\rmqr_\pi\rmqr_\tau
+\rmqr_\pi\rmqr_{\tau_n}+\rmqr_\pi\rmqr_\mu
+\bemqr_\pi\rmqr_\pi\rmqr_\tau+\bemqr_q\rmqr_\pi\rmqr_\mu\\
&+\bemqr_\mu\{\rmqr_\pi\rmqr_q+\rmqr_q^2\}
+(\bemqr_\mu+\bemqr_{\tau_n}+\bemqr_\tau)\rmqr_\pi^2\big).
\end{align*}
\end{proof}

\begin{proof}[Proof of Theorem~\ref{thm:r_bias}]
\;

\emph{Step 1.}
Note that
\begin{align*}
\E\left[\left\{\K^{-1}\sum\limits_{k=1}^\K \Delta_{k,1}'\right\}^2\right]
&\leq\frac{\rsigma^2}{N}.
\end{align*}
By Markov's inequality, we have
\begin{align*}
\K^{-1}\sum\limits_{k=1}^\K \Delta_{k,1}'=O_p\left(\frac{\rsigma}{\sqrt{N}}\right).
\end{align*}
We next substitute the nuisance convergence rates into Lemma~\ref{lem:r_delta_3}, and combine this bound with Lemmas~\ref{lem:r_Delta_0} and \ref{lem:r_delta_2} in \eqref{decom:r_theta}.

\emph{Step 2.}
By Theorem~\ref{thm:r_nuisance}, we have $\rmqr_\pi=O_p(r_\pi)$, $\rmqr_q=O_p(r_\pi+r_q)$, $\rmqr_\mu=O_p(\emqr_\mu\{r_\pi+r_q\}+r_\mu)$, $\rmqr_{\tau_n}=O_p(\{\emqr_\mu+\emqr_{\tau_n}\}r_\pi+\emqr_\mu r_q+r_\mu+r_{\tau_n})$ and $\rmqr_\tau=O_p(\{\emqr_\mu+\emqr_{\tau_n}+\emqr_\tau\}r_\pi+\emqr_\mu r_q+r_\mu+r_{\tau_n}+r_\tau)$.
By the infimum definition, $\emqr_h\leq\bemqr_h$ for each $h\in\mathcal{H}$. Together with \eqref{bound:nuisance_approx} and Assumption~\ref{cond:sparsity}, these bounds give $\rmqr_\pi+\rmqr_q+\rmqr_\mu+\rmqr_{\tau_n}+\rmqr_\tau=o_p(1)$.
The probability-one bounds on $\omega_{\pi_a}$ and $\omega_q$ give the alternative bounds $\Delta_{k,9}'=O_p(\rmqr_\pi\rmqr_\tau)$ and $\Delta_{k,10}'=O_p(\rmqr_\pi\rmqr_\mu)$ in Lemma~\ref{lem:r_delta_3}. Thus these terms can be absorbed into the corresponding products without any normalization of $\bemqr_\pi$ or $\bemqr_q$. Substituting the preceding nuisance rates and applying Lemma~\ref{lem:r_approx_bound} together with \eqref{bound:nuisance_approx}, we have
\begin{align*}
\rmqr_q\rmqr_\mu
&=O_p\{(r_\pi+r_q)r_\mu+\emqr_\mu(r_\pi+r_q)^2\},\\
\rmqr_\pi(\rmqr_\tau+\rmqr_{\tau_n}+\rmqr_\mu)
&=O_p\{r_\pi(r_\mu+r_{\tau_n}+r_\tau)
+(\emqr_\mu+\emqr_{\tau_n}+\emqr_\tau)r_\pi^2+\emqr_\mu r_\pi r_q\},\\
\bemqr_\mu(\rmqr_\pi\rmqr_q+\rmqr_q^2)
+(\bemqr_\mu+\bemqr_{\tau_n}+\bemqr_\tau)\rmqr_\pi^2
&=O_p\{\bemqr_\mu(r_\pi+r_q)^2+(\bemqr_{\tau_n}+\bemqr_\tau)r_\pi^2\}\\
&=O_p\{\emqr_\mu(r_\pi+r_q)^2+(\emqr_{\tau_n}+\emqr_\tau)r_\pi^2\}.
\end{align*}
It follows that $\Delta_{k,3}'=O_p\{r_\pi(r_\mu+r_{\tau_n}+r_\tau)+r_qr_\mu+\emqr_\mu(r_\pi+r_q)^2+(\emqr_{\tau_n}+\emqr_\tau)r_\pi^2\}$. Moreover, Lemma~\ref{lem:r_Delta_0}, Lemma~\ref{lem:r_approx_bound}, and \eqref{bound:nuisance_approx} yield $\Delta_0'=O\{\emqr_\pi(\emqr_{\tau_n}+\emqr_\tau)+(\emqr_q+\emqr_\pi\emqr_{\bs{m}})\emqr_\mu\}$ and $\emqr_{\tau_n}=O(\emqr_{\bs{m}})$. Combining these bounds with Step~1 gives
\begin{align*}
\hat{\theta}_{\mathrm{MQR}}-\theta
&=O_p\left(\rsigma N^{-1/2}+\mathcal R_{\mathrm{MQR}}\right).
\end{align*}

\emph{Step 3.} We now show that
\begin{align}
    \rsigma\asymp 1+\normp{\bX^\top\betaur-\theta}{2}
    \asymp 1+\|\betaur-\theta\bs{e}_1\|_2. \label{order:sigma}
\end{align}
Recall that $\bX^\top\bs{e}_1=1$ almost surely.
Write $\mpsi(\bW,\bnu^*)-\theta=R(\bW)+\bX^\top\betaur-\theta$, where
\begin{align*}
    R(\bW)&:=\frac{A\exp(\bS^\top\beqr)(Y-\bS^\top\bemur)}{g(\bX^\top\bepiar)}
    +\frac{(1-A)(\bS^\top\bemur-\bX^\top\bettaur)}{1-g(\bX^\top\bepibr)}
    +\frac{A\bX^\top(\bettaur-\betaur)}{g(\bX^\top\bepiar)}.
\end{align*}
On the probability-one event in Assumption~\ref{cond:basic2}(a), the working weights are bounded. Thus, for each fixed integer $s\geq2$, Minkowski's inequality and Assumption~\ref{cond:basic2}(b) give
\begin{align*}
    \normp{R(\bW)}{s}
    &\leq c_0^{-3}\normp{\varepsilon}{s}
    +c_0^{-1}\normp{\varrho}{s}
    +c_0^{-1}\normp{\zeta-\varrho}{s}=O(1),
\end{align*}
where the bound depends only on $s,c_0,\sigma_0$. Here we used $A(Y-\bS^\top\bemur)=A\varepsilon$, $(1-A)\bS=(1-A)\bzS$, and $\bX^\top(\bettaur-\betaur)=\zeta-\varrho$. In particular, fix $C_R>0$, depending only on $c_0,\sigma_0$, such that $\|R(\bW)\|_{\P,2}\leq C_R$. It follows that
\begin{align}
    \rsigma\leq \normp{\bX^\top\betaur-\theta}{2}+C_R,
    \qquad
    \normp{\bX^\top\betaur-\theta}{2}\leq\rsigma+C_R.
    \label{order:r_sigma}
\end{align}

To obtain a uniform positive lower bound, decompose $\mpsi(\bW,\bnu^*)-\theta=U'_{-5}+U'_5$, where
\begin{align*}
    U'_{-5}&:= \frac{Aq^*(\bS)}{\pi^*_a(\bX)}\{\mu(\bS)-\mu^*(\bS)\}
    +\frac{A}{\pi^*_a(\bX)}\{\tau_n^*(\bX)-\tau^*(\bX)\}\\
    &\qquad+\frac{1-A}{1-\pi^*_b(\bX)}\{\mu^*(\bS)-\tau_n^*(\bX)\}
    +\tau^*(\bX)-\theta,\\
    U'_5&:=\frac{Aq^*(\bS)}{\pi^*_a(\bX)}\{Y(1,\bM)-\mu(\bS)\}.
\end{align*}
Since $\E[U'_5\mid A,\bS]=0$ and $U'_{-5}$ is a function of $(A,\bS)$, $\E[U'_{-5}U'_5]=0$. Using the probability-one weight bounds in Assumption~\ref{cond:basic2}(a) and the outcome-noise bound in part~(b), we obtain
\begin{align}
    \rsigma^2
    &=\E[(U'_5)^2]+\E[(U'_{-5})^2]
    \geq \E[(U'_5)^2]\nonumber\\
    &\geq \frac{c_0^4}{(1-c_0)^2}\E[A\{Y(1,\bM)-\mu(\bS)\}^2]
    \geq \frac{c_0^4c_Y}{(1-c_0)^2}=:c_*^2>0.
    \label{order:r_sigma_lower}
\end{align}
Together with \eqref{order:r_sigma}, this yields
\begin{align*}
    1+\normp{\bX^\top\betaur-\theta}{2}
    \leq 1+C_R+\rsigma
    \leq \left\{1+\frac{1+C_R}{c_*}\right\}\rsigma,
\end{align*}
which proves the first equivalence in \eqref{order:sigma}.

By Lemma~\ref{lem:momentmatrix}(ii), $\E(\bX\bX^\top)\succeq\E(A\bX\bX^\top)\succeq c_{\min}\bm{I}_{d_1}$. Since $\bX$ is a subvector of $\bS$, Assumption~\ref{cond:basic2}(b) and Lemma~\ref{lem:psi2norm} imply
\begin{align*}
    c_{\min}\|\betaur-\theta\bs{e}_1\|_2^2
    \leq \normp{\bX^\top\betaur-\theta}{2}^2
    =\E\{\bX^\top(\betaur-\theta\bs{e}_1)\}^2
    \leq 2\sigma_{\bS}^2\|\betaur-\theta\bs{e}_1\|_2^2.
\end{align*}
This proves the second equivalence in \eqref{order:sigma}. The same moment bound, for each fixed integer $s\geq2$, gives
\begin{align}
    \normp{\mpsi(\bW,\bnu^*)-\theta}{s}
    &\leq \normp{R(\bW)}{s}
    +\{2\Gamma(s/2+1)\}^{1/s}\sigma_{\bS}\|\betaur-\theta\bs{e}_1\|_2
    \lesssim\rsigma, \label{order:r_normalneed}
\end{align}
where the comparison constant depends only on $s,c_0,\sigma_0,\sigma_{\bS},c_{\min},c_Y$. In particular, $\E|\mpsi(\bW,\bnu^*)-\theta|^4\lesssim\rsigma^4$, uniformly in $N,d_1,d_2$.
\end{proof}

\begin{proof}[Proof of Theorem~\ref{thm:r_infer}]
By Lemmas~\ref{lem:r_Delta_0} and~\ref{lem:r_approx_bound} together with the condition $\mathcal R_{\mathrm{MQR}}=o(\rsigma N^{-1/2})$, we have
\begin{align*}
\Delta_0'&=\E[\mpsi(\bW,\bnu^*)]-\theta
=O\{\bemqr_\pi(\bemqr_{\tau_n}+\bemqr_\tau)+\bemqr_q\bemqr_\mu\}\\
&=O\{\emqr_\pi(\emqr_{\tau_n}+\emqr_\tau)+(\emqr_q+\emqr_\pi\emqr_{\bs{m}})\emqr_\mu\}
=o(\rsigma N^{-1/2}).
\end{align*}
The nuisance consistency in Theorem~\ref{thm:r_nuisance} and Lemma~\ref{lem:r_delta_2} yield $\Delta_{k,2}'=o_p(N^{-1/2})$. By Lemma~\ref{lem:r_delta_3}, Step~2 of the preceding proof, and $\mathcal R_{\mathrm{MQR}}=o(\rsigma N^{-1/2})$,
$\Delta_{k,3}'=O_p(\mathcal R_{\mathrm{MQR}})=o_p(\rsigma N^{-1/2})$.
Since $\rsigma\geq c_*$ by \eqref{order:r_sigma_lower} and $\K$ is fixed,
\begin{align*}
\frac{\sqrt N}{\rsigma}\left\{\Delta_0'
+\K^{-1}\sum_{k=1}^{\K}\Delta_{k,2}'
+\K^{-1}\sum_{k=1}^{\K}\Delta_{k,3}'\right\}=o_p(1).
\end{align*}
Consequently, \eqref{decom:r_theta} gives
\begin{align*}
\frac{\sqrt N(\hat\theta_{\mathrm{MQR}}-\theta)}{\rsigma}
=\frac1{\rsigma\sqrt N}\sum_{i=1}^N
\{\mpsi(\bW_i,\bnu^*)-\E[\mpsi(\bW,\bnu^*)]\}+o_p(1).
\end{align*}
The exact variance satisfies
\begin{align*}
\Var\{\mpsi(\bW,\bnu^*)\}
=\rsigma^2-(\Delta_0')^2=\rsigma^2\{1+o(1)\}.
\end{align*}
By \eqref{order:r_normalneed}, Jensen's inequality, and the triangle inequality,
$\E|\mpsi(\bW,\bnu^*)-\E[\mpsi(\bW,\bnu^*)]|^4\lesssim\rsigma^4$.
Thus the Lyapunov ratio satisfies
\begin{align*}
\frac{\E|\mpsi(\bW,\bnu^*)-\E[\mpsi(\bW,\bnu^*)]|^4}
{N\{\rsigma^2-(\Delta_0')^2\}^2}
=O(N^{-1})\longrightarrow0.
\end{align*}
Lyapunov's central limit theorem for row-wise i.i.d.\ triangular arrays and Slutsky's theorem yield
\begin{align}
\frac1{\rsigma\sqrt N}\sum_{i=1}^N
\{\mpsi(\bW_i,\bnu^*)-\E[\mpsi(\bW,\bnu^*)]\}
\leadsto\mathcal N(0,1). \label{eq:r_normal}
\end{align}
Together with the preceding expansion, this proves $\rsigma^{-1}\sqrt N(\hat\theta_{\mathrm{MQR}}-\theta)\leadsto\mathcal N(0,1)$.

It remains to show
\begin{align}
\hat\sigma_{\mathrm m}^2=\rsigma^2\{1+o_p(1)\}. \label{approx:sigma}
\end{align}
By \eqref{order:r_normalneed} with $s=4$ and independence,
\begin{align*}
\Var\left\{\frac1{N\rsigma^2}\sum_{i=1}^N
[\mpsi(\bW_i,\bnu^*)-\theta]^2\right\}
\leq\frac{\E|\mpsi(\bW,\bnu^*)-\theta|^4}{N\rsigma^4}
\lesssim N^{-1}.
\end{align*}
The quantity inside the variance has expectation one, so Chebyshev's inequality gives
$N^{-1}\sum_{i=1}^N\{\mpsi(\bW_i,\bnu^*)-\theta\}^2
=\rsigma^2\{1+o_p(1)\}$.
The conditional $L_2$ bounds in the proof of Lemma~\ref{lem:r_delta_2} give, for each $k\in[\K]$,
\begin{align*}
\E\left[\{\mpsi(\bW,\hbnu^{-k})-\mpsi(\bW,\bnu^*)\}^2
\mid\mathcal F_{-k}\right]=o_p(1),
\end{align*}
where $\bW$ is an independent evaluation observation and
$\mathcal F_{-k}=\sigma(\hbnu^{-k})$.
Conditional Markov's inequality and the fixed number of folds imply
\begin{align*}
N^{-1}\sum_{k=1}^{\K}\sum_{i\in\I_k}
\{\mpsi(\bW_i,\hbnu^{-k})-\mpsi(\bW_i,\bnu^*)\}^2
=o_p(1)=o_p(\rsigma^2).
\end{align*}
By Theorem~\ref{thm:r_bias} and the rate conditions of Theorem~\ref{thm:r_infer}, we also have
$(\hat\theta_{\mathrm{MQR}}-\theta)/\rsigma=O_p(N^{-1/2})$.
The triangle inequality for the empirical $L_2$ norm therefore yields
\begin{align*}
\left|\hat\sigma_{\mathrm m}
-\left\{N^{-1}\sum_{i=1}^N[\mpsi(\bW_i,\bnu^*)-\theta]^2\right\}^{1/2}\right|
&\leq\left\{N^{-1}\sum_{k=1}^{\K}\sum_{i\in\I_k}
[\mpsi(\bW_i,\hbnu^{-k})-\mpsi(\bW_i,\bnu^*)]^2\right\}^{1/2}\\
&\qquad+|\hat\theta_{\mathrm{MQR}}-\theta|
=o_p(\rsigma).
\end{align*}
This proves \eqref{approx:sigma} and completes the proof.
\end{proof}

\subsection{Proofs of the results in Section~\ref{apdx:mqr_supp}}\label{apdx:mqr_supp_proofs}

\begin{proof}[Proof of Lemma~\ref{lem:unique_nu}]
The Hessian matrices of the population objective functions in \eqref{eq:r_loss1}--\eqref{eq:r_loss6} are given by
\begin{align*}
    \nabla^2_{\bepia}\E[\ell_1(\bW,\bepia)]&=\E\big[A\exp(-\bX^\top\bepia)\bX\bX^\top\big],\\
    \nabla^2_{\bepib}\E[\ell_2(\bW,\bepib)]&=\E\big[(1-A)\exp(\bX^\top\bepib)\bX\bX^\top\big],\\
    \nabla^2_{\beq}\E[\ell_3(\bW,\bnu_3^*,\beq)]&=\E\big[Ag^{-1}(\bX^\top\bepiar)\exp(\bS^\top\beq)\bS\bS^\top\big],\\
    \nabla^2_{\bemu}\E[\ell_4(\bW,\bnu_4^*,\bemu)]&=2\E\big[Ag^{-1}(\bX^\top\bepiar)\exp(\bS^\top\beqr)\bS\bS^\top\big],\\
    \nabla^2_{\bettau}\E[\ell_5(\bW,\bnu_5^*,\bettau)]&=2\E\big[(1-A)\exp(\bX^\top\bepibr)\bX\bX^\top\big],\\
    \nabla_{\betau}^2\E[\ell_6(\bW,\bnu_6^*,\betau)]&=2\E[A\exp(-\bX^\top\bepiar)\bX\bX^\top].
\end{align*}

First, we use contradiction to show that
\begin{align*}
    \nabla^2_{\bepia}\E[\ell_1(\bW,\bepia)]\succ 0 \text{,\; for any \;} \bepia\in\R^{d_1}.
\end{align*}
Suppose there exist $\bv,\bepia\in\R^{d_1}$ such that $\bv^\top\nabla^2_{\bepia}\E[\ell_1(\bW,\bepia)]\bv=0$ and $\|\bv\|_2=1$. Then we have $A\bv^\top\bX=0$ almost surely since $\exp(-\bX^\top\bepia)>0$. It follows that $A(\bv^\top\bX)^2=0$ almost surely and hence $\E[A(\bv^\top\bX)^2]=0$. Further, we have $0=\lambda_{\min}(\E[A\bX\bX^\top])\geq \lambda_{\min}(\E[A\bS\bS^\top])$, which contradicts Assumption~\ref{cond:basic2} (b). Therefore, $\nabla^2_{\bepia}\E[\ell_1(\bW,\bepia)]\succ 0$ for any $\bepia\in\R^{d_1}$, and the solution $\bepiar$ is unique and can be equivalently defined as the solution of $\nabla_{\bepia}\E[\ell_1(\bW,\bepia)]=\bz_{d_1}$.

Note that $\exp(\bX^\top\bepib)>0$ for any $\bepib\in\R^{d_1}$ and $\exp(\bS^\top\beq)>0$ for any $\beq\in\R^{d}$. Following a similar procedure as above, we also have
\begin{align*}
    \nabla^2_{\bepib}\E[\ell_2(\bW,\bepib)]\succ 0 \text{,\; for any \;} \bepib\in\R^{d_1},\\
    \nabla^2_{\beq}\E[\ell_3(\bW,\bnu_3^*,\beq)]\succ 0 \text{,\; for any \;} \beq\in\R^{d}.
\end{align*}
Otherwise, it would contradict Assumption \ref{cond:basic2} (b). Therefore, $\bepibr$ and $\beqr$ are the unique solutions of $\nabla_{\bepib}\E[\ell_2(\bW,\bepib)]=\bz_{d_1}$ and $\nabla_{\beq}\E[\ell_3(\bW,\bnu_3^*,\beq)]=\bz_{d}$, respectively, and can be equivalently defined as their solutions.

Furthermore, since
\begin{align*}
    \nabla^2_{\bemu}\E[\ell_4(\bW,\bnu_4^*,\bemu)]=2\nabla^2_{\beq}\E[\ell_3(\bW,\bnu_3^*,\beqr)]\succ 0 \text{,\; for any \;} \bemu\in\R^{d},\\
    \nabla^2_{\bettau}\E[\ell_5(\bW,\bnu_5^*,\bettau)]=2\nabla^2_{\bepib}\E[\ell_2(\bW,\bepibr)]\succ 0 \text{,\; for any \;} \bettau\in\R^{d_1},\\
    \nabla^2_{\betau}\E[\ell_6(\bW,\bnu_6^*,\betau)]=2\nabla^2_{\bepia}\E[\ell_1(\bW,\bepiar)]\succ 0 \text{,\; for any \;} \betau\in\R^{d_1},
\end{align*}
the solutions $\bemur$, $\bettaur$ and $\betaur$ are also unique.
\end{proof}

\begin{proof}[Proof of Lemma~\ref{lem:oracle_equiv}]
The lower bounds follow directly from \eqref{def:parametrization} and the definition of $e_h$. It suffices to establish $\|\beq^0\|_2=O(1)$ and the corresponding $O(\cdot)$ upper bounds.

\emph{(a)} Let $L_3^0(\beq):=\E[\{A/\pi(\bX)\}\exp(\bS^\top\beq)-\{(1-A)/(1-\pi(\bX))\}\bS^\top\beq]$. By Assumption~\ref{cond:basic}(c) and the tower rule, $L_3^0(\beq)=\E[\{A/\pi(\bX)\}\{\exp(\bS^\top\beq)-q(\bS)\bS^\top\beq\}]$. Set $a:=c_0^{-2}+c_0^2$ and $C_0:=a(\log a-1)>0$. Since $q(\bS)\in[c_0^2,c_0^{-2}]$ almost surely by Assumption~\ref{cond:basic}(c), we have $\exp(u)-q(\bS)u\geq c_0^2|u|-C_0$ for all $u\in\R$ almost surely. By $L_3^0(\beq^0)\leq L_3^0(\bz_d)=1$, $\E[A/\pi(\bX)]=1$, and $\pi(\bX)^{-1}\geq 1$,
\begin{align*}
\E\big[A|\bS^\top\beq^0|\big]
\leq\E\!\left[\frac{A}{\pi(\bX)}|\bS^\top\beq^0|\right]
\leq\frac{L_3^0(\beq^0)+C_0}{c_0^2}
\leq\frac{1+C_0}{c_0^2}=O(1).
\end{align*}
Hence,
\begin{align*}
c_{\min}\|\beq^0\|_2^2
&\overset{(i)}{\leq}\E\big[A(\bS^\top\beq^0)^2\big]\overset{(ii)}{\leq}\big(\E\big[A|\bS^\top\beq^0|\big]\big)^{2/3}\big(\E\big[(\bS^\top\beq^0)^4\big]\big)^{1/3}
\overset{(iii)}{=}O\big(\|\beq^0\|_2^{4/3}\big),
\end{align*}
where (i) holds by Lemma~\ref{lem:momentmatrix}(i), (ii) by H\"older's inequality with $|A|\leq 1$, and (iii) by Assumption~\ref{cond:basic2}(b) and Lemma~\ref{lem:psi2norm}. Thus $\|\beq^0\|_2=O(1)$. Furthermore, it follows that $\normp{\exp(\bS^\top\beq^0)}{r}=O(1)$ for any fixed $r>0$, and hence $\normp{q^0-q}{6}=O(1)$ by Assumption~\ref{cond:basic}(c).

Next, we show that $\normp{q^0-q}{6}=O(\emqr_q)$. Let $p_0>0$ be the constant from Lemma~\ref{lem:smallball} applied to $\bS$ via Lemma~\ref{lem:momentmatrix}(i), and set $c_q:=(c_0^2/2)(p_0/2)^{1/6}$. For any $\beq\in\R^d$, let $\epsilon_q:=\normp{\exp(\bS^\top\beq)-q}{6}$. When $\epsilon_q\geq c_q$, $\normp{q^0-q}{6}=O(1)=O(\epsilon_q)$ holds directly.

Otherwise, when $\epsilon_q<c_q$, we first show that $\|\beq\|_2=O(1)$.
Define the event $G_q:=\{|\bS^\top\beq|\leq\log(2c_0^{-2})\}=\{c_0^2/2\leq\exp(\bS^\top\beq)\leq 2c_0^{-2}\}$. On $G_q^c$, Assumption~\ref{cond:basic}(c) implies $|\exp(\bS^\top\beq)-q(\bS)|\geq c_0^2/2$ almost surely; hence, Markov's inequality yields $\P(G_q^c)\leq(2c_0^{-2}\epsilon_q)^6<p_0/2$, i.e., $\P(G_q)>1-p_0/2$. Meanwhile, Lemma~\ref{lem:smallball} gives $\P(|\bS^\top\beq|\geq\sqrt{c_{\min}/2}\,\|\beq\|_2)\geq p_0$. It follows that
\begin{align*}
\P\Big(G_q\cap\big\{|\bS^\top\beq|\geq\sqrt{c_{\min}/2}\,\|\beq\|_2\big\}\Big)
&\geq\P(G_q)+\P\Big(|\bS^\top\beq|\geq\sqrt{c_{\min}/2}\,\|\beq\|_2\Big)-1\\
&>\frac{p_0}{2}>0.
\end{align*}
Therefore, there exists a realization of $\bS$ in this intersection satisfying $\sqrt{c_{\min}/2}\,\|\beq\|_2\leq|\bS^\top\beq|\leq\log(2c_0^{-2})$, which implies $\|\beq\|_2\leq\sqrt{2/c_{\min}}\,\log(2c_0^{-2})=O(1)$.

Next, set $\tilde\bDelta_q:=\beq^0-\beq$. By Taylor's theorem and $\nabla L_3^0(\beq^0)=\bz_d$, there exists $\tbeq=t\beq^0+(1-t)\beq$ for some $t\in[0,1]$ such that $-\tilde\bDelta_q^\top\nabla L_3^0(\beq)=\E[\{A/\pi(\bX)\}\exp(\bS^\top\tbeq)(\bS^\top\tilde\bDelta_q)^2]$. Since $\|\beq^0\|_2=O(1)$ and $\|\beq\|_2=O(1)$, we have $\|\tbeq\|_2=O(1)$. Define the event $E_{M,q}:=\{|\bS^\top\tbeq|\leq M\}$ for a sufficiently large constant $M>0$. By Assumption~\ref{cond:basic2}(b) and Lemma~\ref{lem:psi2norm}, $\P(E_{M,q}^c)\leq 2\exp(-cM^2)$ for some constant $c>0$. Hence,
\begin{align}
-\tilde\bDelta_q^\top\nabla L_3^0(\beq)
&\overset{(i)}{\geq} e^{-M}\E\big[A(\bS^\top\tilde\bDelta_q)^2\idf_{E_{M,q}}\big]\nonumber\\
&\geq e^{-M}\Big\{\E\big[A(\bS^\top\tilde\bDelta_q)^2\big]-\P(E_{M,q}^c)^{1/2}\normp{\bS^\top\tilde\bDelta_q}{4}^2\Big\}\nonumber\\
&\overset{(ii)}{\geq}\frac{e^{-M}c_{\min}}{2}\|\tilde\bDelta_q\|_2^2,\label{eq:oracle_hess_q}
\end{align}
where (i) holds by $\pi(\bX)^{-1}\geq 1$ and $\exp(\bS^\top\tbeq)\geq e^{-M}$ on $E_{M,q}$ and (ii) by Lemma~\ref{lem:momentmatrix}(i), Lemma~\ref{lem:psi2norm}, and the choice of $M$.
On the other hand, by Assumption~\ref{cond:basic}(c), we have
\begin{align*}
\big|\tilde\bDelta_q^\top\nabla L_3^0(\beq)\big|
=\left|\E\!\left[\frac{A}{\pi(\bX)}\{\exp(\bS^\top\beq)-q(\bS)\}\bS^\top\tilde\bDelta_q\right]\right|
\leq c_0^{-1}\epsilon_q\normp{\bS^\top\tilde\bDelta_q}{2}
=O(\epsilon_q\|\tilde\bDelta_q\|_2).
\end{align*}
Combined with \eqref{eq:oracle_hess_q}, we obtain $\|\beq^0-\beq\|_2=O(\epsilon_q)$. Furthermore, since $\|\beq^0\|_2+\|\beq\|_2=O(1)$, it follows from the mean value theorem that
\begin{align*}
\normp{q^0-\exp(\bS^\top\beq)}{6}
\leq\big\{\normp{\exp(\bS^\top\beq^0)}{12}+\normp{\exp(\bS^\top\beq)}{12}\big\}\normp{\bS^\top\tilde\bDelta_q}{12}
=O(\epsilon_q).
\end{align*}
Hence, $\normp{q^0-q}{6}\leq\normp{q^0-\exp(\bS^\top\beq)}{6}+\epsilon_q=O(\epsilon_q)$ also holds when $\epsilon_q<c_q$.

Combining both cases, $\normp{q^0-q}{6}=O(\normp{\exp(\bS^\top\beq)-q}{6})$ holds uniformly for all $\beq\in\R^d$. Taking the infimum over $\beq\in\R^d$ gives $\normp{q^0-q}{6}=O(\emqr_q)$.

\emph{(b)} For any $\bemu\in\R^d$, by the construction of $\bemu^0$, Assumption~\ref{cond:basic}(c), and Lemma~\ref{lem:momentmatrix}(i), we have
\begin{align*}
c_0^2c_{\min}\|\bemu^0-\bemu\|_2^2
&\leq\E\!\left[\frac{Aq(\bS)}{\pi(\bX)}\{\bS^\top(\bemu^0-\bemu)\}^2\right]\\
&=\E\!\left[\frac{Aq(\bS)}{\pi(\bX)}\{\mu(\bS)-\bS^\top\bemu\}\bS^\top(\bemu^0-\bemu)\right]\\
&\leq c_0^{-3}\normp{\mu-\bS^\top\bemu}{2}\normp{\bS^\top(\bemu^0-\bemu)}{2}\\
&=O\big(\normp{\mu-\bS^\top\bemu}{6}\|\bemu^0-\bemu\|_2\big).
\end{align*}
Thus $\|\bemu^0-\bemu\|_2=O(\normp{\mu-\bS^\top\bemu}{6})$, and hence $\normp{\mu^0-\mu}{6}\leq\normp{\bS^\top(\bemu^0-\bemu)}{6}+\normp{\bS^\top\bemu-\mu}{6}=O(\normp{\bS^\top\bemu-\mu}{6})$. Taking the infimum over $\bemu\in\R^d$ yields $\normp{\mu^0-\mu}{6}=O(\emqr_\mu)$.

\emph{(c)} Similarly to part~(b), using $\E[\mu^*(\bS)\mid A=0,\bX]=\tau_n(\bX)$ and Lemma~\ref{lem:momentmatrix}(iii), we obtain $\normp{\tau_n^0-\tau_n}{6}=O(\emqr_{\tau_n})$. Likewise, for any $\bphi\in\R^{d_1\times d_2}$ and unit vector $\bv\in\R^{d_2}$, by the construction of $\bphi^0$ and Lemma~\ref{lem:momentmatrix}(ii), we have
\begin{align*}
c_{\min}\|(\bphi^0-\bphi)\bv\|_2^2
&\leq\E\big[\{\bX^\top(\bphi^0-\bphi)\bv\}^2\big]\leq\normp{\{\bs{m}(\bX)-\bphi^\top\bX\}^\top\bv}{2}\normp{\bX^\top(\bphi^0-\bphi)\bv}{2},
\end{align*}
so $\normp{\{\bs{m}^0(\bX)-\bs{m}(\bX)\}^\top\bv}{6}=O(\normp{\{\bs{m}(\bX)-\bphi^\top\bX\}^\top\bv}{6})$. Taking the supremum over $\|\bv\|_2=1$ and the infimum over $\bphi\in\R^{d_1\times d_2}$ yields $\sup_{\|\bv\|_2=1}\normp{\{\bs{m}^0(\bX)-\bs{m}(\bX)\}^\top\bv}{6}=O(\emqr_{\bs{m}})$.

\emph{(d)} Since $\E[q(\bS)\{Y-\mu(\bS)\}\mid A=1,\bX]=0$, a similar argument to part~(b) with Lemma~\ref{lem:momentmatrix}(ii) yields $\normp{\tau^0-\tau}{6}=O(\emqr_\tau)$.
\end{proof}

\begin{proof}[Proof of Lemma~\ref{lem:r_approx_bound}]
\;\;

\emph{(a)} Let $p_0>0$ be the constant from Lemma~\ref{lem:smallball} applied to $\bX$ via Lemma~\ref{lem:momentmatrix}(ii), and set $c_\pi:=(c_0/2)(p_0/2)^{1/6}$. For any $\bepi\in\R^{d_1}$, let $\epsilon_\pi:=\normp{g(\bX^\top\bepi)-\pi}{6}$. When $\epsilon_\pi\geq c_\pi$, $\bemqr_\pi\leq 1=O(\epsilon_\pi)$ holds directly by Assumptions~\ref{cond:basic}(c) and~\ref{cond:basic2}(a).

Otherwise, when $\epsilon_\pi<c_\pi$, we first show that $\|\bepi\|_2=O(1)$. Define the event $G_\pi:=\{|\bX^\top\bepi|\leq\log(2/c_0-1)\}=\{c_0/2\leq g(\bX^\top\bepi)\leq 1-c_0/2\}$. On $G_\pi^c$, Assumption~\ref{cond:basic}(c) implies $|g(\bX^\top\bepi)-\pi(\bX)|\geq c_0/2$ almost surely; hence, Markov's inequality yields $\P(G_\pi^c)\leq(2\epsilon_\pi/c_0)^6<p_0/2$, i.e., $\P(G_\pi)>1-p_0/2$. Meanwhile, Lemma~\ref{lem:smallball} gives $\P(|\bX^\top\bepi|\geq\sqrt{c_{\min}/2}\,\|\bepi\|_2)\geq p_0$. Similarly to the proof of Lemma~\ref{lem:oracle_equiv}(a), the intersection $G_\pi\cap\{|\bX^\top\bepi|\geq\sqrt{c_{\min}/2}\,\|\bepi\|_2\}$ has probability exceeding $p_0/2>0$, so there exists a realization of $\bX$ in this intersection satisfying $\sqrt{c_{\min}/2}\,\|\bepi\|_2\leq|\bX^\top\bepi|\leq\log(2/c_0-1)$, which implies $\|\bepi\|_2\leq\sqrt{2/c_{\min}}\,\log(2/c_0-1)=O(1)$. Likewise, $\|\bepiar\|_2=O(1)$ by Assumption~\ref{cond:basic2}(a) and Lemma~\ref{lem:momentmatrix}(ii).

Next, set $\bDelta_{\pi_a}:=\bepiar-\bepi$. Following the proof of \eqref{eq:oracle_hess_q} in Lemma~\ref{lem:oracle_equiv}(a) with Lemma~\ref{lem:momentmatrix}(ii), since $\|\bepiar\|_2=O(1)$ and $\|\bepi\|_2=O(1)$, we similarly have
\begin{align*}
-\bDelta_{\pi_a}^\top\nabla_{\bepia}\E[\ell_1(\bW,\bepi)]
\geq\frac{e^{-M}c_{\min}}{2}\|\bDelta_{\pi_a}\|_2^2
\end{align*}
for a sufficiently large constant $M>0$.
On the other hand, since $\|\bepi\|_2=O(1)$, we have $\normp{g(\bX^\top\bepi)^{-1}}{3}\leq 1+\normp{\exp(-\bX^\top\bepi)}{3}=O(1)$ by Lemma~\ref{lem:psi2norm}, and hence
\begin{align*}
\big|\bDelta_{\pi_a}^\top\nabla_{\bepia}\E[\ell_1(\bW,\bepi)]\big|
&\overset{(i)}{=}\left|\E\!\left[\frac{g(\bX^\top\bepi)-\pi(\bX)}{g(\bX^\top\bepi)}\bX^\top\bDelta_{\pi_a}\right]\right|\\
&\overset{(ii)}{\leq}\epsilon_\pi\normp{g(\bX^\top\bepi)^{-1}}{3}\normp{\bX^\top\bDelta_{\pi_a}}{2}
=O(\epsilon_\pi\|\bDelta_{\pi_a}\|_2),
\end{align*}
where (i) holds by the tower rule and (ii) by H\"older's inequality. Combining the preceding two displays yields $\|\bDelta_{\pi_a}\|_2=O(\epsilon_\pi)$. Furthermore, since $|g'(u)|\leq 1$, we have $\normp{\pi_a^*-\pi}{6}\leq\normp{\bX^\top\bDelta_{\pi_a}}{6}+\epsilon_\pi=O(\|\bDelta_{\pi_a}\|_2+\epsilon_\pi)=O(\epsilon_\pi)$. Following a similar procedure for $\ell_2$ yields $\normp{\pi_b^*-\pi}{6}=O(\epsilon_\pi)$, so $\bemqr_\pi=O(\epsilon_\pi)$ also holds when $\epsilon_\pi<c_\pi$.

Combining both cases, $\bemqr_\pi=O(\normp{g(\bX^\top\bepi)-\pi}{6})$ holds uniformly for all $\bepi\in\R^{d_1}$. Taking the infimum over $\bepi\in\R^{d_1}$ gives $\bemqr_\pi=O(\emqr_\pi)$.

\emph{(b)} By Assumption~\ref{cond:basic2}(a) and Lemma~\ref{lem:momentmatrix}(i), $\|\beqr\|_2=O(1)$, and $\|\beq^0\|_2=O(1)$ by Lemma~\ref{lem:oracle_equiv}(a). Let $\bDelta_q^0:=\beqr-\beq^0=(\bDelta_{q,\bX}^{0\top},\bDelta_{q,\bM}^{0\top})^\top$ with $\bDelta_{q,\bX}^0\in\R^{d_1}$ and $\bDelta_{q,\bM}^0\in\R^{d_2}$. Following the proof of \eqref{eq:oracle_hess_q} in Lemma~\ref{lem:oracle_equiv}(a), since $g^{-1}(\bX^\top\bepiar)\geq 1$, $\|\beqr\|_2=O(1)$, and $\|\beq^0\|_2=O(1)$, we similarly have
\begin{align}
-\bDelta_q^{0\top}\nabla_{\beq}\E[\ell_3(\bW,\bnu_3^*,\beq^0)]
\geq\frac{e^{-M}c_{\min}}{2}\|\bDelta_q^0\|_2^2\label{eq:approx_hess_q}
\end{align}
for a sufficiently large constant $M>0$.
Note that by the construction of $\bepiar$ and $\bepibr$,
\begin{align*}
\E\left[\left\{\frac{\pi(\bX)}{\pi_a^*(\bX)}-\frac{1-\pi(\bX)}{1-\pi_b^*(\bX)}\right\}\bX\right]
=\E\left[\left\{\frac{A}{g(\bX^\top\bepiar)}-\frac{1-A}{1-g(\bX^\top\bepibr)}\right\}\bX\right]=\bz_{d_1},
\end{align*}
which also implies $\E\big[\{\pi(\bX)/\pi_a^*(\bX)-(1-\pi(\bX))/(1-\pi_b^*(\bX))\}\bs{m}^0(\bX)\big]=\bz_{d_2}$ since $\bs{m}^0(\bX)=\bphi^{0\top}\bX$. Together with $\E[Aq(\bS)\bS\mid\bX]=\pi(\bX)(\bX^\top,\bs{m}(\bX)^\top)^\top$ under Assumption~\ref{cond:basic}, we obtain
\begin{align*}
\nabla_{\beq}\E[\ell_3(\bW,\bnu_3^*,\beq^0)]
={}&\E\left[\frac{A}{\pi_a^*(\bX)}\{q^0(\bS)-q(\bS)\}\bS\right]\\
&+\begin{pmatrix}\bz_{d_1}\\[1mm]
\E\left[\left\{\frac{\pi(\bX)}{\pi_a^*(\bX)}-\frac{1-\pi(\bX)}{1-\pi_b^*(\bX)}\right\}\{\bs{m}(\bX)-\bs{m}^0(\bX)\}\right]
\end{pmatrix}.
\end{align*}
By Assumption~\ref{cond:basic2} and Lemma~\ref{lem:oracle_equiv}(a) and~(c), we have
\begin{align*}
&\big|\bDelta_q^{0\top}\nabla_{\beq}\E[\ell_3(\bW,\bnu_3^*,\beq^0)]\big|\\
&\quad\leq c_0^{-1}\normp{q^0-q}{2}\normp{\bS^\top\bDelta_q^0}{2}\\
&\quad\quad+c_0^{-1}\big(\normp{\pi_a^*-\pi}{2}+\normp{\pi_b^*-\pi}{2}\big)\normp{\{\bs{m}(\bX)-\bs{m}^0(\bX)\}^\top\bDelta_{q,\bM}^0}{2}\\
&\quad=O\big((\emqr_q+\bemqr_\pi \emqr_{\bs{m}})\|\bDelta_q^0\|_2\big).
\end{align*}
Combined with \eqref{eq:approx_hess_q}, we obtain $\|\bDelta_q^0\|_2=O(\emqr_q+\bemqr_\pi \emqr_{\bs{m}})$. Similarly to Lemma~\ref{lem:oracle_equiv}(a), since $\|\beqr\|_2+\|\beq^0\|_2=O(1)$, we have $\normp{q^*-q^0}{6}=O(\|\bDelta_q^0\|_2)$. Hence, by Lemma~\ref{lem:oracle_equiv}(a) and part~(a), $\bemqr_q\leq\normp{q^*-q^0}{6}+\normp{q^0-q}{6}=O(\emqr_q+\bemqr_\pi \emqr_{\bs{m}})=O(\emqr_q+\emqr_\pi \emqr_{\bs{m}})$.

\emph{(c)} Let $\bDelta_\mu^0:=\bemur-\bemu^0$. Since $\nabla_{\bemu}\E[\ell_4(\bW,\bnu_4^*,\bemur)]=\bz_d$ and $\nabla^2_{\bemu}\E[\ell_4(\bW,\bnu_4^*,\bemu)]=2\E[Ag^{-1}(\bX^\top\bepiar)\exp(\bS^\top\beqr)\bS\bS^\top]\succeq 2c_0^2c_{\min}\bm{I}_d$ by Assumption~\ref{cond:basic2}(a) and Lemma~\ref{lem:momentmatrix}(i), we obtain
\begin{align*}
2c_0^2c_{\min}\|\bDelta_\mu^0\|_2^2
&\leq -\bDelta_\mu^{0\top}\nabla_{\bemu}\E[\ell_4(\bW,\bnu_4^*,\bemu^0)]\\
&\overset{(i)}{=}2\E\big[Ag^{-1}(\bX^\top\bepiar)\exp(\bS^\top\beqr)\{\mu(\bS)-\mu^0(\bS)\}\bS^\top\bDelta_\mu^0\big]\\
&\overset{(ii)}{\leq} 2c_0^{-3}\normp{\mu-\mu^0}{2}\normp{\bS^\top\bDelta_\mu^0}{2}=O(\emqr_\mu\|\bDelta_\mu^0\|_2),
\end{align*}
where (i) holds by the tower rule and $\E[Y\mid A=1,\bS]=\mu(\bS)$, and (ii) follows from Assumption~\ref{cond:basic2}(a) and Lemma~\ref{lem:oracle_equiv}(b). Therefore, $\|\bDelta_\mu^0\|_2=O(\emqr_\mu)$. By Lemma~\ref{lem:psi2norm} and Lemma~\ref{lem:oracle_equiv}(b), it follows that $\bemqr_\mu\leq\normp{\bS^\top\bDelta_\mu^0}{6}+\normp{\mu^0-\mu}{6}=O(\|\bDelta_\mu^0\|_2+\emqr_\mu)=O(\emqr_\mu)$.

\emph{(d)} Let $\bDelta_{\tau_n}^0:=\bettaur-\bettau^0$. Noting that $\nabla_{\bettau}\E[\ell_5(\bW,\bnu_5^*,\bettaur)]=\bz_{d_1}$ and $\E[\mu^*(\bS)\mid A=0,\bX]=\tau_n(\bX)$, a similar argument to part~(c) gives
\begin{align*}
\frac{2c_0c_{\min}}{1-c_0}\|\bDelta_{\tau_n}^0\|_2^2
\leq -\bDelta_{\tau_n}^{0\top}\nabla_{\bettau}\E[\ell_5(\bW,\bnu_5^*,\bettau^0)]
=O(\emqr_{\tau_n}\|\bDelta_{\tau_n}^0\|_2).
\end{align*}
Hence, $\|\bDelta_{\tau_n}^0\|_2=O(\emqr_{\tau_n})$, and by Lemma~\ref{lem:psi2norm} and Lemma~\ref{lem:oracle_equiv}(c), we have $\bemqr_{\tau_n}\leq\normp{\bX^\top\bDelta_{\tau_n}^0}{6}+\normp{\tau_n^0-\tau_n}{6}=O(\|\bDelta_{\tau_n}^0\|_2+\emqr_{\tau_n})=O(\emqr_{\tau_n})$.

Moreover, write $\bemur=(\bbeta_{\mu,\bX}^{*\top},\bbeta_{\mu,\bM}^{*\top})^\top$ with $\bbeta_{\mu,\bX}^*\in\R^{d_1}$ and $\bbeta_{\mu,\bM}^*\in\R^{d_2}$, and set $\bv:=((\bbeta_{\mu,\bX}^*-\betaur)^\top,\bbeta_{\mu,\bM}^{*\top})^\top\in\R^d$. Then $\bzS^\top\bv=\bzS^\top\bemur-\bX^\top\betaur=\zeta$. By Assumptions~\ref{cond:basic} and~\ref{cond:basic2}(b), together with Lemmas~\ref{lem:psi2norm} and~\ref{lem:momentmatrix}(i),
\begin{align*}
2\sigma_0^2\geq\E[\zeta^2]=\E[(\bzS^\top\bv)^2]=\E\!\left[\frac{A}{\pi(\bX)}q(\bS)(\bS^\top\bv)^2\right]\geq\frac{c_0^2}{1-c_0}\E[A(\bS^\top\bv)^2]\geq\frac{c_0^2c_{\min}}{1-c_0}\|\bv\|_2^2,
\end{align*}
which implies $\|\bbeta_{\mu,\bM}^*\|_2\leq\|\bv\|_2=O(1)$. Now define $\tilde\tau_n(\bX):=\bX^\top\bbeta_{\mu,\bX}^*+\bs{m}^0(\bX)^\top\bbeta_{\mu,\bM}^*\in\mathcal{G}_{\tau_n}$. Since $\tau_n(\bX)=\E[\bS^\top\bemur\mid A=0,\bX]=\bX^\top\bbeta_{\mu,\bX}^*+\bs{m}(\bX)^\top\bbeta_{\mu,\bM}^*$, we directly obtain by the definition of $\emqr_{\tau_n}$ and Lemma~\ref{lem:oracle_equiv}(c) that
\begin{align*}
\emqr_{\tau_n}\leq\normp{\tau_n-\tilde\tau_n}{6}=\normp{\{\bs{m}(\bX)-\bs{m}^0(\bX)\}^\top\bbeta_{\mu,\bM}^*}{6}=O(\|\bbeta_{\mu,\bM}^*\|_2 \emqr_{\bs{m}})=O(\emqr_{\bs{m}}).
\end{align*}

\emph{(e)} Let $\bDelta_\tau^0:=\betaur-\betau^0$. By the tower rule, $\E[Y-\mu^*(\bS)\mid A=1,\bS]=\mu(\bS)-\mu^*(\bS)$, and $\E[q(\bS)\{\mu(\bS)-\mu^*(\bS)\}\mid A=1,\bX]=\tau(\bX)-\tau_n(\bX)$, we have
\begin{align*}
&\nabla_{\betau}\E[\ell_6(\bW,\bnu_6^*,\betau^0)]\\
&\quad=-2\E\big[A\exp(-\bX^\top\bepiar)\big\{\exp(\bS^\top\beqr)(Y-\bS^\top\bemur)+\bX^\top\bettaur-\bX^\top\betau^0\big\}\bX\big]\\
&\quad=-2\E\Big[\pi(\bX)\exp(-\bX^\top\bepiar)\big\{\E[\{q^*(\bS)-q(\bS)\}\{\mu(\bS)-\mu^*(\bS)\}\mid A=1,\bX]\\
&\quad\quad+\tau_n^*(\bX)-\tau_n(\bX)+\tau(\bX)-\tau^0(\bX)\big\}\bX\Big].
\end{align*}
Since $\pi(\bX)\geq c_0$ with probability $1$, by Jensen's and H\"older's inequalities, we have
\begin{align*}
&\normp{\E[\{q^*(\bS)-q(\bS)\}\{\mu(\bS)-\mu^*(\bS)\}\mid A=1,\bX]}{2}\\
&\quad\leq c_0^{-1/2}\normp{q^*-q}{4}\normp{\mu^*-\mu}{4}=O(\bemqr_q \bemqr_\mu).
\end{align*}
Similarly to part~(d), by Assumption~\ref{cond:basic2}(a), Lemma~\ref{lem:momentmatrix}(ii), and Lemma~\ref{lem:oracle_equiv}(d), we obtain
\begin{align*}
\frac{2c_0c_{\min}}{1-c_0}\|\bDelta_\tau^0\|_2^2
\leq -\bDelta_\tau^{0\top}\nabla_{\betau}\E[\ell_6(\bW,\bnu_6^*,\betau^0)]
=O\big(\{\emqr_\tau+\bemqr_{\tau_n}+\bemqr_q \bemqr_\mu\}\|\bDelta_\tau^0\|_2\big).
\end{align*}
Combined with parts~(b)--(d), it follows that $\|\bDelta_\tau^0\|_2=O(\emqr_\tau+\bemqr_{\tau_n}+\bemqr_q \bemqr_\mu)=O(\emqr_\tau+\emqr_{\tau_n}+\bemqr_q \emqr_\mu)$. Hence, by Lemma~\ref{lem:psi2norm} and Lemma~\ref{lem:oracle_equiv}(d), we obtain $\bemqr_\tau\leq\normp{\bX^\top\bDelta_\tau^0}{6}+\normp{\tau^0-\tau}{6}=O(\|\bDelta_\tau^0\|_2+\emqr_\tau)=O(\emqr_\tau+\emqr_{\tau_n}+\bemqr_q \emqr_\mu)$.
Furthermore, note that $\tau(\bX)-\tau_n(\bX)=\E[\mu(\bS)-\mu^*(\bS)\mid A=0,\bX]$. By Jensen's inequality, we have $\normp{\tau-\tau_n}{6}\leq c_0^{-1/6}\bemqr_\mu=O(\emqr_\mu)$. Since $\tau^0\in\mathcal{G}_\tau=\mathcal{G}_{\tau_n}$, it follows from Lemma~\ref{lem:oracle_equiv}(d) that $\emqr_{\tau_n}\leq\normp{\tau^0-\tau_n}{6}\leq\normp{\tau^0-\tau}{6}+\normp{\tau-\tau_n}{6}=O(\emqr_\tau+\emqr_\mu)$, which implies $\emqr_\tau+\emqr_{\tau_n}=O\big(\emqr_\tau+\min\{\emqr_{\tau_n},\emqr_\mu\}\big)$. Combining this with part~(b), we conclude that
\begin{align*}
\bemqr_\tau=O\big(\emqr_\tau+\min\{\emqr_{\tau_n},\emqr_\mu\}+\bemqr_q \emqr_\mu\big)
=O\big(\emqr_\tau+\min\{\emqr_{\tau_n},\emqr_\mu\}+(\emqr_q+\emqr_\pi \emqr_{\bs{m}})\emqr_\mu\big).
\end{align*}
\end{proof}

In what follows, we turn to the proof of Theorem~\ref{thm:r_nuisance}.
For each $k\in[\K]$, let $\mathbb{S}_{\pi}=(\bW_j)_{j\in\mathcal I_{\pi}}$, $\mathbb{S}_{q}=(\bW_j)_{j\in\mathcal I_{q}}$, $\mathbb{S}_{\mu}=(\bW_j)_{j\in\mathcal I_{\mu}}$, $\mathbb{S}_{\tau_n}=(\bW_j)_{j\in\mathcal I_{\tau_n}}$ and $\mathbb{S}_{\tau}=(\bW_j)_{j\in\mathcal I_{\tau}}$ be the subsets of $\mathbb{S}$ corresponding to the index sets $\mathcal I_{\pi}$, $\mathcal I_{q}$, $\mathcal I_{\mu}$, $\mathcal I_{\tau_n}$ and $\mathcal I_{\tau}$, respectively.
Define
\begin{align*}
    \bar{\ell}_1(\bepia):=M^{-1}\supiar \ell_1(\bW_i,\bepia),&\qquad\bar{\ell}_2(\bepib):=M^{-1}\supibr\ell_2(\bW_i,\bepib),\\
    \bar{\ell}_3(\bnu_3,\beq):=M^{-1}\suqr \ell_3(\bW_i,\bnu_3,\beq),&\qquad\bar{\ell}_4(\bnu_4,\bemu):=M^{-1}\sumur \ell_4(\bW_i,\bnu_4,\bemu),\\
    \bar\ell_5(\bnu_5,\bettau):=M^{-1}\suttaur \ell_5(\bW_i,\bnu_5,\bettau),&\qquad\bar{\ell}_6(\bnu_6,\betau):=M^{-1}\sutaur \ell_6(\bW_i,\bnu_6,\betau),
\end{align*}
where the loss functions are defined in \eqref{eq:r_loss1}--\eqref{eq:r_loss6}. Then for any $\bDelta\in\R^{d_1}$, the Taylor-series errors corresponding to $\ell_1$, $\ell_2$, $\ell_5$ and $\ell_6$ are defined as
\begin{align*}
\delta\bar{\ell}_1(\bepia,\bDelta):=&\bar{\ell}_1(\bepia+\bDelta)-\bar{\ell}_1(\bepia)-\bnabla_{\bepia}\bar{\ell}_1(\bepia)^\top\bDelta, \\
\delta\bar{\ell}_2(\bepib,\bDelta):=&\bar{\ell}_2(\bepib+\bDelta)-\bar{\ell}_2(\bepib)-\bnabla_{\bepib}\bar{\ell}_2(\bepib)^\top\bDelta, \\
\delta \bar{\ell }_5(\bnu_5,\bettau,\bDelta):=&\bar{\ell}_5(\bnu_5,\bettau+\bDelta)-\bar{\ell}_5(\bnu_5,\bettau)-\bnabla_{\bettau}\bar{\ell}_5(\bnu_5,\bettau)^\top\bDelta, \\
\delta \bar{\ell }_6(\bnu_6,\betau,\bDelta):=&\bar{\ell}_6(\bnu_6, \betau+\bDelta)-\bar{\ell}_6(\bnu_6,\betau)-\bnabla_{\betau}\bar{\ell}_6(\bnu_6, \betau)^\top\bDelta.
\end{align*}
Similarly, for any $\bDelta\in\mathbb{R}^{d}$, we define
\begin{align}
\delta\bar{\ell}_3(\bnu_3,\beq,\bDelta)&:=\bar{\ell}_3(\bnu_3,\beq+\bDelta)-\bar{\ell}_3(\bnu_3,\beq)-\bnabla_{\beq}\bar{\ell}_3(\bnu_3,\beq)^\top\bDelta, \nonumber\\
\delta\bar{\ell}_4(\bnu_4,\bemu,\bDelta)&:=\bar{\ell}_4(\bnu_4,\bemu+\bDelta)-\bar{\ell}_4(\bnu_4,\bemu)-\bnabla_{\bemu}\bar{\ell}_4(\bnu_4,\bemu)^\top\bDelta.\nonumber
\end{align}

\begin{lemma}\label{lem:r_imputation}
        Let the assumptions of Theorem~\ref{thm:r_nuisance} hold. Let $s_{\pi}=o(N/\log d_1)$, $s_{q}=o(N/\log d)$, and consider some $\lambda_{\pi_a}, \lambda_{\pi_b}\asymp\sqrt{\log d_1/N}$. Let $\lambda_q=4\sigma_q\sqrt{\log d/M}$ for a constant $\sigma_q>0$ (specified in Lemma~\ref{lem:r_gradient}).
        For any fixed $t>0$, define
        \begin{align*}
            \mathcal B_1:=&\{\|\bnabla_{\beq}\bar\ell_3(\bnu_3^*,\beqr)\|_\infty\leq\lambda_{q}/2\},\\
            \mathcal B_2:=&\left\{|R_1(\bDelta)|\leq c_t r_{\pi}\left(\sqrt{\frac{\log d}{N}}\|\bDelta\|_1+\|\bDelta\|_2\right),\;\;\forall\bDelta\in\R^d\right\},\\
            \mathcal B_3:=&\left\{\delta\bar\ell_3(\hat{\bnu}_3,  \beqr,\bDelta)\geq\kappa_1\|\bDelta\|_2^2-\kappa_2\frac{\log d}{M}\|\bDelta\|_1^2,\;\;\forall\bDelta\in\R^d:\|\bDelta\|_2\leq1\right\},
        \end{align*}
        where
        $R_1(\bDelta):=\{\bnabla_{\beq}\bar\ell_3(\hat{\bnu}_3, \beqr)-\bnabla_{\beq}\bar\ell_3(\bepiar, \bepibr, \beqr)\}^\top\bDelta$ and $c_t>0$ is some constant. Let $\bar s_{q}:=s_{\pi}\log d_1/\log d+s_{q}$. Further define $\bDelta_q:=\hbeqr-\beqr$ and $\Ctil(s,k)=\{\bDelta\in\R^d:\|\bDelta\|_1\leq k\sqrt{s}\|\bDelta\|_2\}$. Then:

        \emph{(a)} Let $S_{\beqr}$ be the support of $\beqr$. On the event $\mathcal B_1$, we have
        \begin{align*}
            2\delta\bar\ell_3(\hat{\bnu}_3, \beqr,\bDelta_q)+\lambda_{q}\|\bDelta_q\|_1&\leq 4\lambda_{q}\|\bDelta_{q,S_{\beqr}}\|_1+2|R_1(\bDelta_q)|.
        \end{align*}

        \emph{(b)} On the event $\mathcal B_1\cap\mathcal B_2$, when $N>N_0$, we have $\bDelta_q\in\Ctil(\bar s_{q},k_0)$, where $N_0$ and $k_0$ are some positive constants.

        \emph{(c)} On the event $\mathcal B_1\cap\mathcal B_2\cap\mathcal B_3$, if $\bDelta_q\in\Ctil(\bar s_{q},k_0)$, then $\|\bDelta_q\|_2\leq 1$ for all sufficiently large $N$.
\end{lemma}

\begin{proof}[Proof of Lemma~\ref{lem:r_imputation}]
\;\;

\emph{(a)} Define
\begin{align}
   \mathcal F(\bDelta)&:=\bar{\ell}_3(\hat{\bnu}_3,\beqr+\bDelta)+\lambda_{q}\|\beqr+\bDelta\|_1-\bar{\ell}_3(\hat{\bnu}_3,\beqr)-\lambda_{q}\|\beqr\|_1\nonumber\\
   &=\delta\bar\ell_3(\hat{\bnu}_3, \beqr,\bDelta)+\lambda_{q}\|\beqr+\bDelta\|_1+\bnabla_{\beq}\bar\ell_3(\hat{\bnu}_3, \beqr)^\top\bDelta-\lambda_{q}\|\beqr\|_1\nonumber\\
   &=\delta\bar\ell_3(\hat{\bnu}_3, \beqr,\bDelta)+\lambda_{q}\|\beqr+\bDelta\|_1+\bnabla_{\beq}\bar\ell_3(\bnu_3^*,\beqr)^\top\bDelta+R_1(\bDelta)-\lambda_{q}\|\beqr\|_1.
   \label{def:F_Delta}
\end{align}
Then conditioning on $\mathcal B_1$, we have
\begin{align}
   2\mathcal F(\bDelta)&\overset{(i)}{\geq}2\delta\bar\ell_3(\hat{\bnu}_3, \beqr,\bDelta)+2\lambda_{q}\|\beqr+\bDelta\|_1-2\lambda_{q}\|\beqr\|_1-2|R_1(\bDelta)|-\lambda_{q}\|\bDelta\|_1 \nonumber
   \\ &\overset{(ii)}{\geq} 2\delta\bar\ell_3(\hat{\bnu}_3, \beqr,\bDelta)+2\lambda_{q}\|\bDelta_{S_{\beqr}^c}\|_1-2\lambda_{q}\|\bDelta_{S_{\beqr}}\|_1-2|R_1(\bDelta)|-\lambda_{q}\|\bDelta\|_1 \nonumber
   \\ &\overset{(iii)}{\geq} 2\delta\bar\ell_3(\hat{\bnu}_3, \beqr,\bDelta)+\lambda_{q}\|\bDelta\|_1-4\lambda_{q}\|\bDelta_{S_{\beqr}}\|_1-2|R_1(\bDelta)|,\label{bound:lower_F}
\end{align}
where (i) holds by $R_1(\bDelta)\geq-|R_1(\bDelta)|$ and the construction of $\mathcal B_1$; (ii) follows from the decomposability of the $\ell_1$ norm over $S_{\beqr}$ and $S_{\beqr}^c$; and (iii) follows from $\|\bDelta_{S_{\beqr}^c}\|_1=\|\bDelta\|_1-\|\bDelta_{S_{\beqr}}\|_1$. Since $\bDelta_q:=\hbeqr-\beqr$, by the construction of $\hbeq$, we further have $\mathcal F(\bDelta_q)\leq 0$. It follows that, on the event $\mathcal B_1$,
\begin{align*}
   2\delta\bar\ell_3(\hat{\bnu}_3, \beqr,\bDelta_q)+\lambda_{q}\|\bDelta_q\|_1&\leq 4\lambda_{q}\|\bDelta_{q,S_{\beqr}}\|_1+2|R_1(\bDelta_q)|.
\end{align*}

\emph{(b)} Conditional on $\mathcal B_1\cap\mathcal B_2$, we have
\begin{align*}
4\mathcal F(\bDelta)&\overset{(i)}{\geq}4\delta\bar\ell_3(\hat{\bnu}_3, \beqr,\bDelta)+2\lambda_{q}\|\bDelta\|_1-8\lambda_{q}\|\bDelta_{S_{\beqr}}\|_1-4|R_1(\bDelta)|\\
&\overset{(ii)}{\geq} 4\delta\bar\ell_3(\hat{\bnu}_3, \beqr,\bDelta)+\left(2\lambda_{q}-4c_t r_{\pi}\sqrt{\frac{\log d}{N}}\right)\|\bDelta\|_1-(4c_t r_{\pi}+8\lambda_{q}\sqrt{s_{q}})\|\bDelta\|_2,
\end{align*}
where (i) holds by \eqref{bound:lower_F}, and (ii) follows from the construction of $\mathcal B_2$ and $\|\bDelta_{S_{\beqr}}\|_1\leq \sqrt{s_{q}}\|\bDelta\|_2$. Under the assumptions of Theorem~\ref{thm:r_nuisance}, we have $r_\pi=\sqrt{s_\pi\log d_1/N}=o(1)$. Note that $\lambda_{q}\asymp \sqrt{\log d/N}$. Then for large enough $N_0>0$, $4c_t r_\pi\sqrt{\log d/N}\leq \lambda_{q}$ holds for all $N>N_0$. Therefore,
\begin{align}
4\delta\bar\ell_3(\hat{\bnu}_3, \beqr,\bDelta)+\lambda_{q}\|\bDelta\|_1&\leq (4c_t r_{\pi}+8\lambda_{q}\sqrt{s_{q}})\|\bDelta\|_2+4\mathcal F(\bDelta).\label{bound:taylorerror}
\end{align}
Let $\bDelta=\bDelta_q$. Since $\delta\bar\ell_3(\hat{\bnu}_3, \beqr,\bDelta_q)\geq 0$, $\mathcal F(\bDelta_q)\leq 0$, $\lambda_{q}\asymp \sqrt{\log d/N}$, and $r_\pi=\sqrt{s_{\pi}\log d_1/N}$, there exists some constant $k_0>0$ such that
\begin{align*}
\|\bDelta_q\|_1&\leq \left(\frac{4c_t r_\pi}{\lambda_q}+8\sqrt{s_q}\right)\|\bDelta_q\|_2
\leq k_0\sqrt{\bar s_{q}}\|\bDelta_q\|_2.
\end{align*}

\emph{(c)} We prove (c) by contradiction. Let $\Ktil(k_0,\bar s_{q},1):=\Ctil(\bar s_{q},k_0)\cap \{\bDelta\in\R^d:\|\bDelta\|_2=1\}$. If $\|\bDelta_q\|_2\geq 1$, write $u=\|\bDelta_q\|_2\geq 1$ and $\bDelta_0=\bDelta_q/u$. Then $\|\bDelta_0\|_2=1$ and
\begin{align*}
\|\bDelta_0\|_1=\|\bDelta_q\|_1/\|\bDelta_q\|_2\leq k_0\sqrt{\bar s_{q}}=k_0\sqrt{\bar s_{q}}\|\bDelta_0\|_2;
\end{align*}
hence $\bDelta_0\in \Ktil(k_0,\bar s_{q},1)$. By \eqref{def:F_Delta}, $\mathcal F(\cdot)$ is convex, so we have
\begin{align*}
\mathcal F(\bDelta_0)=\mathcal F\left(\frac{1}{u}\bDelta_q+(1-\frac{1}{u})\bz\right)\leq \frac{1}{u}\mathcal F(\bDelta_q)+\left(1-\frac{1}{u}\right)\mathcal F(\bz)\overset{(i)}{=}\frac{1}{u}\mathcal F(\bDelta_q)\overset{(ii)}{\leq} 0,
\end{align*}
where (i) holds by $\mathcal F(\bz)=0$, and (ii) holds by the construction of $\hbeq$.
It remains to show that, on $\mathcal B_1\cap \mathcal B_2 \cap \mathcal B_3$,
\begin{align*}
\mathcal{F}(\bDelta)>0\text{, } \forall \bDelta\in \Ktil(k_0,\bar s_{q},1),
\end{align*}
which contradicts $\mathcal{F}(\bDelta_0)\leq 0$ at $\bDelta=\bDelta_0$ and hence implies $\|\bDelta_q\|_2\leq 1$.

For any $\bDelta\in \Ktil(k_0,\bar s_{q},1)$, we have $\|\bDelta\|_2=1$ and $\|\bDelta\|_1\leq k_0\sqrt{\bar s_{q}}$. Hence, on $\mathcal B_3$,
\begin{align*}
\delta\bar\ell_3(\hat{\bnu}_3, \beqr,\bDelta)\geq\kappa_1\|\bDelta\|_2^2-\kappa_2\frac{\log d}{M}\|\bDelta\|_1^2\geq \kappa_1-\kappa_2k_0^2\frac{\bar s_{q}\log d}{M}.
\end{align*}
Conditional on $\mathcal B_1\cap \mathcal B_2$, since \eqref{bound:taylorerror} holds for any $\bDelta\in\R^d$, we have
\begin{align*}
 4\mathcal F(\bDelta)& \geq 4\delta\bar\ell_3(\hat{\bnu}_3, \beqr,\bDelta)+\lambda_{q}\|\bDelta\|_1-4c_t r_{\pi}-8\lambda_{q}\sqrt{s_{q}}\\
 &\geq 4\Big(\kappa_1-\kappa_2k_0^2\frac{\bar s_{q}\log d}{M}\Big)+\lambda_{q}\|\bDelta\|_1-4c_t r_{\pi}-8\lambda_{q}\sqrt{s_{q}}.
\end{align*}
Note that as $N$, $d_1$, $d_2\to\infty$, we have $\bar s_{q}\log d/M=O(r_{\pi}^2+r_{q}^2)=o(1)$ and $r_{\pi}=o(1)$. Since $\lambda_q\sqrt{s_q}=4\sigma_q\sqrt{s_q\log d/M}\asymp r_q=o(1)$ and $c_t$ is fixed, $4c_t r_\pi+8\lambda_q\sqrt{s_q}=o(1)$. Therefore,
\begin{align*}
    4\mathcal F(\bDelta)\geq \frac{\kappa_1}{2}> 0.
\end{align*}
\end{proof}

\begin{lemma}[Restricted strong convexity conditions]\label{lem:r_rsc}
Let Assumptions \ref{cond:basic} and \ref{cond:basic2} hold. Define $f_{M, d_1}(\bDelta):=\kappa_1\|\bDelta\|_2^2-\kappa_2\|\bDelta\|_1^2\log d_1/M$, for any $\bDelta\in \R^{d_1}$,  and $f_{M,d}(\bDelta):=\kappa_1\|\bDelta\| _2^2-\kappa _2\|\bDelta\|_1^2\log d/ M$, for any $\bDelta\in \R^d$. Then, with some constants $\kappa_{1}$, $\kappa_{2}$, $c_{1}$, $c_{2}>0$ and noting that $M\asymp N$, we have
\begin{align}
    \P_{\mathbb{S}_{\pi}}\left(\delta\bar{\ell}_1(\bepiar,\bDelta)\geq f_{M,d_1}(\bDelta),\:\forall\|\bDelta\|_2\leq1\right)&\geq1-c_1\exp(-c_2M),\label{eq:r_RSC1}\\
    \P_{\mathbb{S}_{\pi}}\left(\delta\bar{\ell}_2(\bepibr,\bDelta)\geq f_{M,d_1}(\bDelta),\:\forall\|\bDelta\|_2\leq1\right)&\geq1-c_1\exp(-c_2M),\label{eq:r_RSC2}\\
     \P_{\mathbb{S}_{q}}\left(\delta\bar{\ell}_3(\hat{\bnu}_3,  \beqr,\bDelta)\geq f_{M,d}(\bDelta),\:\forall\|\bDelta\|_{2}\leq1\right)&\geq1-c_{1}\exp(-c_{2}M).\label{eq:r_RSC3}
\end{align}
In \eqref{eq:r_RSC3}, we only consider the randomness in $\mathbb{S}_{q}$, and $\hbepiar$ is treated as fixed (or conditional on).
Further, let $\|\hbepibr-\bepibr\|_{2}\leq1$. Then
\begin{equation}
    \P_{\mathbb{S}_{\tau_n}}\left(\delta\bar{\ell}_5(\hat{\bnu}_5, \bettaur,\bDelta)\geq f_{M,d_1}(\bDelta),\:\forall\bDelta\in\R^{d_1}\right)\geq1-c_{1}\exp(-c_{2}M),\label{eq:r_RSC5}
\end{equation}
where we only consider the randomness in $\mathbb{S}_{\tau_n}$, and $\hbepibr$ is treated as fixed.
On the other hand, if $\|\hbepiar-\bepiar\|_{2}\leq1$, we have
\begin{align}
    \P_{\mathbb{S}_{\tau}}\left(\delta\bar{\ell}_6(\hat{\bnu}_6, \betaur,\bDelta)\geq f_{M,d_1}(\bDelta),\:\forall\bDelta\in\R^{d_1}\right)&\geq1-c_1\exp(-c_2M).\label{eq:r_RSC6}
\end{align}
In \eqref{eq:r_RSC6}, $\hbepiar,\hbeqr$, $\hbemur$ and $\hbettaur$ are all treated as fixed.

If $\|\hbepiar-\bepiar\|_{2}\leq1$ and $\|\hbeqr-\beqr\|_{2}\leq1$, we have
\begin{equation}
    \P_{\mathbb{S}_{\mu}}\left(\delta\bar{\ell}_4(\hat{\bnu}_4, \bemur,\bDelta)\geq f_{M,d}(\bDelta),\:\forall\bDelta\in\R^d\right)\geq1-c_1\exp(-c_2M),\label{eq:r_RSC4}
\end{equation}
where $\hbepiar$ and $\hbeqr$ are treated as fixed.
\end{lemma}

\begin{proof}
For any $\bDelta\in\R^{d_1}$, by Taylor's theorem there exist $v_1,v_1'\in(0,1)$ (possibly depending on $\bDelta$ and the sample) such that
\begin{align*}
    \delta\bar{\ell}_1(\bepiar,\bDelta)&=\frac1{2M}\supiar A_i\exp(-\bX_i^\top\{\bepiar+v_1\bDelta\})(\bX_i^\top\bDelta)^2,\\
    \delta\bar{\ell}_2(\bepibr,\bDelta)&=\frac1{2M}\supibr (1-A_i)\exp(\bX_i^\top\{\bepibr+v_1'\bDelta\})(\bX_i^\top\bDelta)^2,\\
    \delta\bar{\ell}_5(\hat{\bnu}_5, \bettaur,\bDelta)&=\frac1{M}\suttaur (1-A_i)\exp(\bX_i^\top\hbepibr)(\bX_i^\top\bDelta)^2,\\
    \delta\bar{\ell}_6(\hat{\bnu}_6, \betaur,\bDelta)&=\frac1{M}\sutaur A_i\exp(-\bX_i^\top\hbepiar)(\bX_i^\top\bDelta)^2.
\end{align*}
Similarly, for any $\bDelta\in\R^{d}$, there exists $v_2\in(0,1)$ (possibly depending on $\bDelta$ and the sample) such that
\begin{align*}
    \delta\bar{\ell}_3(\hat{\bnu}_3,  \beqr,\bDelta)&=\frac1{2M}\suqr \frac{A_i}{g(\bX_i^\top\hbepiar)}\exp(\bS_i^\top\{\beqr+v_2\bDelta\})(\bS_i^\top\bDelta)^2,\\
    \delta\bar{\ell}_4(\hat{\bnu}_4, \bemur,\bDelta)&=\frac1{M}\sumur \frac{A_i}{g(\bX_i^\top\hbepiar)}\exp(\bS_i^\top\hbeqr)(\bS_i^\top\bDelta)^2.
\end{align*}

\textbf{Part 1.} Let $\bU=A\bX$, $\mathbb S'=(A_{i}\bX_i)_{i\in\I_{\pi}}$, $h(u)=\exp(-u)$, and $\bs{\eta}=\bepiar$. Under Assumption~\ref{cond:basic2}, $\E[|\bU^\top\bs{\eta}|]\leq \E[|\bX^\top\bepiar|]<C$ with some constant $C>0$. By Lemmas~\ref{lem:momentmatrix} and~\ref{lem:useinrsc} (which holds uniformly over $v\in[0,1]$), \eqref{eq:r_RSC1} holds.

\textbf{Part 2.} Similarly, let $\bU=(1-A)\bX$, $\mathbb S'=(\{1-A_{i}\}\bX_i)_{i\in\I_{\pi}}$, $h(u)=\exp(u)$, and $\bs{\eta}=\bepibr$. Under Assumption~\ref{cond:basic2}, $\E[|\bU^\top\bs{\eta}|]\leq \E[|\bX^\top\bepibr|]<C$ with some constant $C>0$. By Lemmas~\ref{lem:momentmatrix} and~\ref{lem:useinrsc} (uniformly over $v\in[0,1]$), \eqref{eq:r_RSC2} holds.

\textbf{Part 3.}
Since $g^{-1}(u)=1+\exp(-u)>1$, we have
\begin{align*}
    \delta\bar{\ell}_3(\hat{\bnu}_3,  \beqr,\bDelta)&>\frac1{2M}\suqr A_i\exp(\bS_i^\top\{\beqr+v_2\bDelta\})(\bS_i^\top\bDelta)^2.
\end{align*}
Let $\bU=A\bS$, $\mathbb S'=(A_{i}\bS_i)_{i\in\I_{q}}$, $h(u)=\exp(u)$, and $\bs{\eta}=\beqr$. Under Assumption~\ref{cond:basic2}, $\E[|\bU^\top\bs{\eta}|]\leq \E[|\bS^\top\beqr|]<C$ with some constant $C>0$. By Lemmas~\ref{lem:momentmatrix} and~\ref{lem:useinrsc} (uniformly over $v\in[0,1]$), we have
\begin{align*}
    \frac1{2M}\suqr A_i\exp(\bS_i^\top\{\beqr+v_2\bDelta\})(\bS_i^\top\bDelta)^2 \geq \kappa_1\|\bDelta\|_2^2-\kappa_2\|\bDelta\|_1^2\frac{\log d}{M}\text{, }\;\;\;\forall \|\bDelta\|_2\leq1,
\end{align*}
with probability $\P_{\mathbb{S}_q}$ at least $1-c_1\exp(-c_2M)$ and some constants $\kappa_1,\kappa_2,c_1,c_2>0$.

\textbf{Part 4.}
Treat both $\hbepiar$ and $\hbeqr$ as fixed (or conditional on) and suppose that $\|\bepiar-\hbepiar\|_2\leq 1$ and $\|\beqr-\hbeqr\|_2\leq 1$. Note that
\begin{align*}
    \delta\bar{\ell}_4(\hat{\bnu}_4, \bemur,\bDelta)&=\frac1{M}\sumur A_i g^{-1}(\bX_i^\top\hbepiar)\exp(\bS_i^\top\hbeqr)(\bS_i^\top\bDelta)^2\\
    &\geq\frac1{M}\sumur A_i\exp(\bS_i^\top\hbeqr)(\bS_i^\top\bDelta)^2,
\end{align*}
where the inequality follows from $g^{-1}(u)=1+\exp(-u)>1$. Let $\bU=A\bS$, $\mathbb S'=(A_{i}\bS_i)_{i\in\I_{\mu}}$, $h(u)=\exp(u)$, $v=0$ and $\bs{\eta}=\hbeqr$. Under $\|\hbeqr-\beqr\|_2\leq 1$, it holds that $\E[|\bU^\top\bs{\eta}|]\leq \E[|\bS^\top\beqr|]+\E[|\bS^\top(\hbeqr-\beqr)|]<C$ with some constant $C>0$. By Lemmas \ref{lem:momentmatrix} and \ref{lem:useinrsc}, we have
\begin{align*}
    \frac1{M}\sumur A_i\exp(\bS_i^\top\hbeqr)(\bS_i^\top\bDelta)^2 \geq \kappa_1'\|\bDelta\|_2^2-\kappa_2'\|\bDelta\|_1^2\frac{\log d}{M}\text{, }\;\;\;\forall \|\bDelta\|_2\leq1,
\end{align*}
with probability $\P_{\mathbb{S}_\mu}$ at least $1-c_1'\exp(-c_2'M)$ and some constants $\kappa_1',\kappa_2',c_1',c_2'>0$.

Combining the preceding two displays yields \eqref{eq:r_RSC4} for every $\bDelta\in\R^d$ with $\|\bDelta\|_2\leq1$.

Note that $\delta\bar{\ell}_4(\hat{\bnu}_4, \bemur,\bDelta)$ is based on a weighted squared loss. For any $\bDelta\in\R^d$ with $\|\bDelta\|_2>1$, applying the same bound to $\bDelta/\|\bDelta\|_2$ and multiplying both sides by $\|\bDelta\|_2^2$ yields \eqref{eq:r_RSC4}.

\textbf{Part 5.}
Treat $\hbepibr$ as fixed (or conditional on) and suppose that $\|\bepibr-\hbepibr\|_2\leq 1$.
Let $\bU=(1-A)\bX$, $\mathbb S'=(\{1-A_{i}\}\bX_i)_{i\in\I_{\tau_n}}$, $h(u)=\exp(u)$, $v=0$ and $\bs{\eta}=\hbepibr$. On the event $\{\|\hbepibr-\bepibr\|_2\leq 1\}$, under Assumption~\ref{cond:basic2}, we have $\E[|\bU^\top\bs{\eta}|]\leq \E[|\bX^\top\bepibr|]+\E[|\bX^\top(\hbepibr-\bepibr)|]<C$ with some constant $C>0$. By Lemmas~\ref{lem:momentmatrix} and~\ref{lem:useinrsc}, analogously to Part~4, we obtain \eqref{eq:r_RSC5}.

\textbf{Part 6.}
Treat $\hbepiar$ as fixed (or conditional on) and suppose that $\|\bepiar-\hbepiar\|_2\leq 1$.
Let $\bU=A\bX$, $\mathbb S'=(A_i\bX_i)_{i\in\I_{\tau}}$, $h(u)=\exp(-u)$, $v=0$ and $\bs{\eta}=\hbepiar$. On the event $\{\|\hbepiar-\bepiar\|_2\leq 1\}$, under Assumption~\ref{cond:basic2}, we have $\E[|\bU^\top\bs{\eta}|]\leq \E[|\bX^\top\bepiar|]+\E[|\bX^\top(\hbepiar-\bepiar)|]<C$ with some constant $C>0$. Note that $\delta\bar{\ell}_6(\hat{\bnu}_6, \betaur,\bDelta)$ is also constructed based on a weighted squared loss. By Lemmas~\ref{lem:momentmatrix} and~\ref{lem:useinrsc}, similarly to Part~4, we obtain \eqref{eq:r_RSC6}.
\end{proof}

\begin{lemma}[Gradient bounds for the MQR method
]\label{lem:r_gradient}
Let Assumption~\ref{cond:basic2} hold. Let $\sigma_{\pi_a}$, $\sigma_{\pi_b}$, $\sigma_{q}$, $\sigma_{\mu}$, $\sigma_{\tau_n}$, $\sigma_{\tau}>0$ be some constants and note that $M\asymp N$. Then, for any $t>0$,
\begin{align}
&\P_{\mathbb{S}_{\pi}}\left(\left\|\bnabla_{\bepia}\bar\ell_1(\bepiar)\right\|_\infty\leq\sigma_{\pi_a}\sqrt{(t+\log d_1)/M}\right)\geq1-2\exp(-t),\label{eq:r_gradient1}\\
&\P_{\mathbb{S}_{\pi}}\left(\left\|\bnabla_{\bepib}\bar\ell_2(\bepibr)\right\|_\infty\leq\sigma_{\pi_b}\sqrt{(t+\log d_1)/M}\right)\geq1-2\exp(-t).\label{eq:r_gradient2}
\end{align}
Further, let Assumption~\ref{cond:basic} hold. Then for any $t>0$,
\begin{align}
&\P_{\mathbb{S}_{q}}\left(\left\|\bnabla_{\beq}\bar\ell_3(\bnu_3^*,\beqr)\right\|_\infty\leq\sigma_{q}\sqrt{(t+\log d)/M}\right)\geq1-2\exp(-t),\label{eq:r_gradient3}\\
&\P_{\mathbb{S}_{\mu}}\left(\left\|\bnabla_{\bemu}\bar\ell_4(\bnu_4^*,\bemur)\right\|_\infty\leq\sigma_{\mu}\left\{2\sqrt{(t+\log d)/M}+(t+\log d)/M\right\}\right)\geq1-2\exp(-t),\label{eq:r_gradient4}\\
&\P_{\mathbb{S}_{\tau_n}}\left(\left\|\bnabla_{\bettau}\bar\ell_5(\bnu_5^*,\bettaur)\right\|_\infty\leq\sigma_{\tau_n}\left\{2\sqrt{(t+\log d_1)/M}+(t+\log d_1)/M\right\}\right)\geq1-2\exp(-t),\label{eq:r_gradient5}\\
&\P_{\mathbb{S}_{\tau}}\left(\left\|\bnabla_{\betau}\bar\ell_6(\bnu_6^*,\betaur)\right\|_\infty\leq\sigma_{\tau}\left\{2\sqrt{(t+\log d_1)/M}+(t+\log d_1)/M\right\}\right)\geq1-2\exp(-t).\label{eq:r_gradient6}
\end{align}
\end{lemma}

\begin{proof}
Define $\varepsilon_i:=Y_i(1,\bM_i)-\bS_i^\top\bemur$, $\zeta_i:=\bziS^\top\bemur-\bX_i^\top\betaur$ and $\varrho_i:=\bziS^\top\bemur-\bX_i^\top\bettaur$. Note that $\zeta_i-\varrho_i=\bX_i^\top\bettaur-\bX_i^\top\betaur$.
We have
\begin{align*}
    \bnabla_{\bepia}\bar{\ell}_1(\bepiar)&=\frac1{M}\supiar [1-A_ig^{-1}(\bX_i^\top\bepiar)]\bX_i,\\
    \bnabla_{\bepib}\bar{\ell}_2(\bepibr)&=\frac1{M}\supibr [(1-A_i)\{1-g(\bX_i^\top\bepibr)\}^{-1}-1]\bX_i,\\
    \bnabla_{\beq}\bar{\ell}_3(\bnu_3^*, \beqr)&=\frac1{M}\suqr \left\{A_i\exp(\bS_i^\top\beqr)g^{-1}(\bX_i^\top\bepiar)-(1-A_i)\{1-g(\bX_i^\top\bepibr)\}^{-1}\right\}\bS_i,\\
    \bnabla_{\bemu}\bar{\ell}_4(\bnu_4^*,\bemur)&=-\frac2M\sumur A_i\exp(\bS_i^\top\beqr)g^{-1}(\bX_i^\top\bepiar)\varepsilon_i\bS_i,\\
    \bnabla_{\bettau}\bar{\ell}_5(\bnu_5^*,\bettaur)&=-\frac2M\suttaur (1-A_i)\exp(\bX_i^\top\bepibr)\varrho_{i}\bX_i,\\
    \bnabla_{\betau}\bar{\ell}_6(\bnu_6^*,\betaur)&=-\frac2M\sutaur A_i\exp(-\bX_i^\top\bepiar)(\exp(\bS_i^\top\beqr)\varepsilon_i+\zeta_{i}-\varrho_{i})\bX_i.
\end{align*}

Recalling the moment equations \eqref{eq:r_eq_pia}--\eqref{eq:r_eq_tau}, we next prove the above inequalities using concentration bounds.

For \eqref{eq:r_gradient1}, under Assumption~\ref{cond:basic2}, $|1-Ag^{-1}(\bX^\top\bepiar)|\leq (1+c_0^{-1})$ holds with probability $1$, so that for each $j\in[d_1]$,
$$
    \normg{[1-Ag^{-1}(\bX^\top\bepiar)]\bX^\top\e_j}\leq (1+c_0^{-1})\normg{\bX^\top\e_j}\leq (1+c_0^{-1})\sigma_{\bS}.
$$
Let $\sigma_{\pi_a}:=\sqrt{8}(1+c_0^{-1})\sigma_{\bS}$. By the Hoeffding bound provided in Lemma~D.2 of \cite{HighDimensionalMEstimation2019}, we have for each $j\in[d_1]$ and any $t>0$,
$$
    \P_{\mathbb{S}_{\pi}}\left(|\bnabla_{\bepia}\bar\ell_1(\bepiar)^\top \e_j|>\sigma_{\pi_a}\sqrt{(t+\log d_1)/M}\right)\leq 2\exp(-t-\log d_1).
$$
It follows that
\begin{align*}
    &\P_{\mathbb{S}_{\pi}}\left(\|\bnabla_{\bepia}\bar\ell_1(\bepiar)\|_\infty>\sigma_{\pi_a}\sqrt{(t+\log d_1)/M}\right)\\
    \leq& \suj{d_1}\P_{\mathbb{S}_{\pi}}\left(|\bnabla_{\bepia}\bar\ell_1(\bepiar)^\top \e_j|>\sigma_{\pi_a}\sqrt{(t+\log d_1)/M}\right)\leq 2d_1\exp(-t-\log d_1)=2\exp(-t).
\end{align*}
Similar justification holds for \eqref{eq:r_gradient2}. Specifically, $|(1-A)\{1-g(\bX^\top\bepibr)\}^{-1}-1|\leq (1+c_0^{-1})$ under Assumption~\ref{cond:basic2}, so that for each $j\in[d_1]$,
$$
    \normg{[(1-A)\{1-g(\bX^\top\bepibr)\}^{-1}-1]\bX^\top\e_j}\leq (1+c_0^{-1})\normg{\bX^\top\e_j}\leq (1+c_0^{-1})\sigma_{\bS}.
$$
Let $\sigma_{\pi_b}:=\sqrt{8}(1+c_0^{-1})\sigma_{\bS}$. By the Hoeffding bound provided in Lemma~D.2 of \cite{HighDimensionalMEstimation2019}, for each $j\in[d_1]$ and any $t>0$,
$$
    \P_{\mathbb{S}_{\pi}}\left(|\bnabla_{\bepib}\bar\ell_2(\bepibr)^\top \e_j|>\sigma_{\pi_b}\sqrt{(t+\log d_1)/M}\right)\leq 2\exp(-t-\log d_1).
$$
It follows that
\begin{align*}
    &\P_{\mathbb{S}_{\pi}}\left(\|\bnabla_{\bepib}\bar\ell_2(\bepibr)\|_\infty>\sigma_{\pi_b}\sqrt{(t+\log d_1)/M}\right)\\
    \leq& \suj{d_1}\P_{\mathbb{S}_{\pi}}\left(|\bnabla_{\bepib}\bar\ell_2(\bepibr)^\top \e_j|>\sigma_{\pi_b}\sqrt{(t+\log d_1)/M}\right) \leq 2d_1\exp(-t-\log d_1)=2\exp(-t).
\end{align*}
For \eqref{eq:r_gradient3}, $|Ag^{-1}(\bX^\top\bepiar)\exp(\bS^\top\beqr)-(1-A)\{1-g(\bX^\top\bepibr)\}^{-1}|\leq c_0^{-3}+c_0^{-1}$ under Assumption~\ref{cond:basic2}; hence for each $j\in[d]$,
$$
    \normg{\left\{\frac{A\exp(\bS^\top\beqr)}{g(\bX^\top\bepiar)}-\frac{1-A}{1-g(\bX^\top\bepibr)}\right\}\bS^\top\e_j}\leq (c_0^{-3}+c_0^{-1})\normg{\bS^\top\e_j}\leq (c_0^{-3}+c_0^{-1})\sigma_{\bS}.
$$
Let $\sigma_{q}:=\sqrt{8}(c_0^{-3}+c_0^{-1})\sigma_{\bS}$. By the Hoeffding bound provided in Lemma~D.2 of \cite{HighDimensionalMEstimation2019}, we have for each $j\in[d]$ and any $t>0$,
$$
    \P_{\mathbb{S}_{q}}\left(|\bnabla_{\beq}\bar\ell_3(\bnu_3^*,\beqr)^\top \e_j|>\sigma_{q}\sqrt{(t+\log d)/M}\right)\leq 2\exp(-t-\log d).
$$
It follows that
\begin{align*}
    &\P_{\mathbb{S}_{q}}\left(\|\bnabla_{\beq}\bar\ell_3(\bnu_3^*,\beqr)\|_\infty>\sigma_{q}\sqrt{(t+\log d)/M}\right)\\
    \leq& \suj{d}\P_{\mathbb{S}_{q}}\left(|\bnabla_{\beq}\bar\ell_3(\bnu_3^*,\beqr)^\top \e_j|>\sigma_{q}\sqrt{(t+\log d)/M}\right)\leq 2d\exp(-t-\log d)=2\exp(-t).
\end{align*}

For \eqref{eq:r_gradient4}, under Assumption~\ref{cond:basic2}, we have $|Ag^{-1}(\bX^\top\bepiar)\exp(\bS^\top\beqr)|\leq c_0^{-3}$, so that for each $j\in[d]$,
$$
    \norme{\frac{A\exp(\bS^\top\beqr)}{g(\bX^\top\bepiar)}\varepsilon\bS^\top\e_j}\leq c_0^{-3}\normg{\varepsilon}\normg{\bS^\top\e_j}\leq c_0^{-3}\sigma_0\sigma_{\bS}.
$$
Let $\sigma_{\mu}:=2c_0^{-3}\sigma_0\sigma_{\bS}$. By Bernstein's inequality provided in Lemma~D.4 of \cite{HighDimensionalMEstimation2019}, we have for each $j\in[d]$ and any $t>0$,
$$
    \P_{\mathbb{S}_{\mu}}\left(|\bnabla_{\bemu}\bar\ell_4(\bnu_4^*,\bemur)^\top \e_j|>\sigma_{\mu}\left\{2\sqrt{(t+\log d)/M}+(t+\log d)/M\right\}\right)\leq 2\exp(-t-\log d).
$$
It follows that
\begin{align*}
    &\P_{\mathbb{S}_{\mu}}\left(\|\bnabla_{\bemu}\bar\ell_4(\bnu_4^*,\bemur)\|_\infty>\sigma_{\mu}\left\{2\sqrt{(t+\log d)/M}+(t+\log d)/M\right\}\right)\\
    \leq& \suj{d}\P_{\mathbb{S}_{\mu}}\left(|\bnabla_{\bemu}\bar\ell_4(\bnu_4^*,\bemur)^\top \e_j|>\sigma_{\mu}\left\{2\sqrt{(t+\log d)/M}+(t+\log d)/M\right\}\right)\\
    \leq& 2d\exp(-t-\log d)=2\exp(-t).
\end{align*}

For \eqref{eq:r_gradient5}, we have $|(1-A)\exp(\bX^\top\bepibr)|\leq c_0^{-1}$ under Assumption~\ref{cond:basic2}. Hence, for each $j\in[d_1]$,
$$
    \norme{(1-A)\exp(\bX^\top\bepibr)\varrho\bX^\top\e_j}\leq c_0^{-1}\normg{\varrho}\normg{\bX^\top\e_j}\leq c_0^{-1}\sigma_0\sigma_{\bS}.
$$
Let $\sigma_{\tau_n}:=2c_0^{-1}\sigma_0\sigma_{\bS}$. By Bernstein's inequality provided in Lemma~D.4 of \cite{HighDimensionalMEstimation2019}, we have for each $j\in[d_1]$ and any $t>0$,
\begin{align*}
    &\P_{\mathbb{S}_{\tau_n}}\left(|\bnabla_{\bettau}\bar\ell_5(\bnu_5^*,\bettaur)^\top \e_j|>\sigma_{\tau_n}\left\{2\sqrt{(t+\log d_1)/M}+(t+\log d_1)/M\right\}\right)\\
    &\qquad\leq 2\exp(-t-\log d_1).
\end{align*}
It follows that
\begin{align*}
    &\P_{\mathbb{S}_{\tau_n}}\left(\|\bnabla_{\bettau}\bar\ell_5(\bnu_5^*,\bettaur)\|_\infty>\sigma_{\tau_n}\left\{2\sqrt{(t+\log d_1)/M}+(t+\log d_1)/M\right\}\right)\\
    \leq& \suj{d_1}\P_{\mathbb{S}_{\tau_n}}\left(|\bnabla_{\bettau}\bar\ell_5(\bnu_5^*,\bettaur)^\top \e_j|>\sigma_{\tau_n}\left\{2\sqrt{(t+\log d_1)/M}+(t+\log d_1)/M\right\}\right)\\
    \leq& 2d_1\exp(-t-\log d_1)=2\exp(-t).
\end{align*}

For \eqref{eq:r_gradient6}, both $|A\exp(-\bX^\top\bepiar)|\leq c_0^{-1}-1$ and $|A\exp(-\bX^\top\bepiar)\exp(\bS^\top\beqr)|\leq (c_0^{-3}-c_0^{-2})$ hold under Assumption~\ref{cond:basic2}; hence for each $j\in[d_1]$,
\begin{align*}
    &\norme{[A\exp(-\bX^\top\bepiar)\{(-\zeta+\varrho)-\exp(\bS^\top\beqr)\varepsilon\}]\bX^\top\e_j}\\
    \leq& (c_0^{-1}-1)(\normg{\zeta}+\normg{\varrho}+c_0^{-2}\normg{\varepsilon})\sigma_{\bS}\leq(c_0^{-1}-1)(2+c_0^{-2})\sigma_0\sigma_{\bS}.
\end{align*}
Let $\sigma_{\tau}:=2(c_0^{-1}-1)(2+c_0^{-2})\sigma_0\sigma_{\bS}$. By Bernstein's inequality provided in Lemma~D.4 of \cite{HighDimensionalMEstimation2019}, we have for each $j\in[d_1]$ and any $t>0$,
$$
    \P_{\mathbb{S}_{\tau}}\left(|\bnabla_{\betau}\bar\ell_6(\bnu_6^*,\betaur)^\top \e_j|>\sigma_{\tau}\left\{2\sqrt{(t+\log d_1)/M}+(t+\log d_1)/M\right\}\right)\leq 2\exp(-t-\log d_1).
$$
It follows that
\begin{align*}
    &\P_{\mathbb{S}_{\tau}}\left(\|\bnabla_{\betau}\bar\ell_6(\bnu_6^*,\betaur)\|_\infty>\sigma_{\tau}\left\{2\sqrt{(t+\log d_1)/M}+(t+\log d_1)/M\right\}\right)\\
    \leq& \suj{d_1}\P_{\mathbb{S}_{\tau}}\left(|\bnabla_{\betau}\bar\ell_6(\bnu_6^*,\betaur)^\top \e_j|>\sigma_{\tau}\left\{2\sqrt{(t+\log d_1)/M}+(t+\log d_1)/M\right\}\right)\\
    \leq& 2d_1\exp(-t-\log d_1)=2\exp(-t).
\end{align*}
\end{proof}

\begin{lemma}\label{lem:r_useevent}
    Let the assumptions in Theorem~\ref{thm:r_nuisance} hold. For the MQR1 estimator, with some constant $C>0$, define
    \begin{align}
    \mathcal{E}_1:=&\Big\{\|\hbepiar-\bepiar\|_2\leq1\text{ and }\|g^{-1}(\bX^\top\bepi)\|_{\P,12}\leq C,\nonumber\\
    &\qquad\forall\bepi\in\{w\bepiar+(1-w)\hbepiar:w\in[0,1]\}\Big\},\label{def:E1}\\
    \mathcal{E}_2:=&\Big\{\|\hbepibr-\bepibr\|_2\leq1\text{ and }\|g^{-1}(\bX^\top\bepi)\|_{\P,12}\leq C,\nonumber\\
    &\qquad\forall\bepi\in\{w\bepibr+(1-w)\hbepibr:w\in[0,1]\}\Big\},\label{def:E2}\\
    \mathcal{E}_3:=&\Big\{\|\hbeqr-\beqr\|_2\leq1\text{ and }\|\exp(\bS^\top\beq)\|_{\P,12}\leq C,\nonumber\\
    &\qquad\forall\beq\in\{w\beqr+(1-w)\hbeqr:w\in[0,1]\}\Big\}.\label{def:E3}
    \end{align}
    Then, as $N$, $d_1$, $d_2\to\infty$,
    $\P_{\mathbb S_{\pi}\cup \mathbb S_{q}}(\mathcal{E}_1 \cap\mathcal{E}_2 \cap \mathcal{E}_3)=1-o(1).$

    Moreover, let $r'\in[1,12]$ be some integer.
    On the event $\mathcal{E}_1$, for $\bepi\in\{w\bepiar+(1-w)\hbepiar:w\in[0,1]\}$ and some constant $C_1>0$, we also have
    \begin{align*}
        \left\|g^{-1}(\bX^\top\bepi)\right\|_{\P,r'}\leq C_1,\quad\left\|\exp(-\bX^\top\bepi)\right\|_{\P,r'}\leq C_1,\quad\left\|\exp(\bX^\top\bepi)\right\|_{\P,r'}\leq C_1.
    \end{align*}
    On the event $\mathcal{E}_2$, for $\bepi\in\{w\bepibr+(1-w)\hbepibr:w\in[0,1]\}$ and some constant $C_2>0$, we have
    \begin{align}
        \left\|g^{-1}(\bX^\top\bepi)\right\|_{\P,r'}\leq C_2,\quad\left\|\exp(-\bX^\top\bepi)\right\|_{\P,r'}\leq C_2,\quad\left\|\exp(\bX^\top\bepi)\right\|_{\P,r'}\leq C_2.
        \label{event2}
    \end{align}
    On the event $\mathcal{E}_3$, for $\beq\in\{w\beqr+(1-w)\hbeqr:w\in[0,1]\}$ and some constant $C_3>0$, we have
    \begin{align}
        \left\|\exp(\bS^\top\beq)\right\|_{\P,r'}\leq C_3,\quad\left\|\exp(-\bS^\top\beq)\right\|_{\P,r'}\leq C_3.
        \label{event3}
    \end{align}
\end{lemma}

\begin{proof}[Proof of Lemma~\ref{lem:r_useevent}]
First, consider any fixed vector $\bDelta\in\R^d$ with $\|\bDelta\|_2\leq 1$. Under Assumption~\ref{cond:basic2}(b), we have $\normg{\bS^\top\bDelta}\leq\sigma_{\bS}\|\bDelta\|_2\leq\sigma_{\bS}$. By Lemma~\ref{lem:psi2norm}, we obtain $|\E[\bS^\top\bDelta]|\leq\E[|\bS^\top\bDelta|]\leq\sqrt{\pi}\sigma_{\bS}$ and $\E[\exp(t\{\bS^\top\bDelta-\E[\bS^\top\bDelta]\})]\leq\exp(2\sigma_{\bS}^2 t^2)$ for all $t\in\R$. Hence, for any integer $r'\in[1,12]$ and either sign $\pm$,
\begin{align*}
    \normp{\exp(\pm\bS^\top\bDelta)}{r'}
    &=\Big\{\exp(\pm r'\E[\bS^\top\bDelta])\,\E\big[\exp\big(\pm r'\{\bS^\top\bDelta-\E[\bS^\top\bDelta]\}\big)\big]\Big\}^{1/r'}\\
    &\leq\exp(\sqrt{\pi}\sigma_{\bS}+2r'\sigma_{\bS}^2)\leq K_0,
\end{align*}
where $K_0:=\exp(\sqrt{\pi}\sigma_{\bS}+24\sigma_{\bS}^2)$. Since $\bX$ is a subvector of $\bS$, the same bound $\normp{\exp(\pm\bX^\top\bDelta)}{r'}\leq K_0$ holds for any fixed $\bDelta\in\R^{d_1}$ with $\|\bDelta\|_2\leq 1$.

Now suppose $\|\hbepiar-\bepiar\|_2\leq 1$. For any $\bepi=w\bepiar+(1-w)\hbepiar$ with $w\in[0,1]$, let $\bDelta:=\bepi-\bepiar=(1-w)(\hbepiar-\bepiar)$, so that $\|\bDelta\|_2\leq 1$. Under Assumption~\ref{cond:basic2}(a), we have $\P(c_0\leq g(\bX^\top\bepiar)\leq 1-c_0)=1$, and hence $\exp(\pm\bX^\top\bepiar)\leq c_0^{-1}-1$ almost surely. Consequently, for any integer $r'\in[1,12]$,
\begin{align*}
    \normp{\exp(\pm\bX^\top\bepi)}{r'}
    &=\normp{\exp(\pm\bX^\top\bepiar)\exp(\pm\bX^\top\bDelta)}{r'}\\
    &\leq(c_0^{-1}-1)\normp{\exp(\pm\bX^\top\bDelta)}{r'}\leq(c_0^{-1}-1)K_0.
\end{align*}
Since $g^{-1}(u)=1+\exp(-u)$, it follows from Minkowski's inequality that $\normp{g^{-1}(\bX^\top\bepi)}{r'}\leq 1+(c_0^{-1}-1)K_0=:C_1$.

Under Assumption~\ref{cond:basic2}(a), $\P(c_0\leq g(\bX^\top\bepibr)\leq 1-c_0)=1$ also holds. By the same argument as above, on the event $\{\|\hbepibr-\bepibr\|_2\leq 1\}$, we have $\normp{g^{-1}(\bX^\top\bepi)}{r'}\leq C_2$ and $\normp{\exp(\pm\bX^\top\bepi)}{r'}\leq C_2$ for all $\bepi\in\{w\bepibr+(1-w)\hbepibr:w\in[0,1]\}$ and any integer $r'\in[1,12]$, with $C_2:=C_1$. Similarly, under Assumption~\ref{cond:basic2}(a), we have $\P(c_0^2\leq\exp(\bS^\top\beqr)\leq c_0^{-2})=1$, and hence $\exp(\pm\bS^\top\beqr)\leq c_0^{-2}$ almost surely. Therefore, on the event $\{\|\hbeqr-\beqr\|_2\leq 1\}$, for any $\beq\in\{w\beqr+(1-w)\hbeqr:w\in[0,1]\}$ and any integer $r'\in[1,12]$, we obtain
\begin{align*}
    \normp{\exp(\pm\bS^\top\beq)}{r'}
    &=\normp{\exp(\pm\bS^\top\beqr)\exp\{\pm\bS^\top(\beq-\beqr)\}}{r'}\\
    &\leq c_0^{-2}K_0=:C_3.
\end{align*}

Let $C:=\max\{C_1,C_3\}$. Setting $r'=12$ in the upper bounds on $\normp{g^{-1}(\bX^\top\bepi)}{r'}$ and $\normp{\exp(\bS^\top\beq)}{r'}$ above, we have $\{\|\hbepiar-\bepiar\|_2\leq 1\}\subseteq\mathcal{E}_1$, $\{\|\hbepibr-\bepibr\|_2\leq 1\}\subseteq\mathcal{E}_2$, and $\{\|\hbeqr-\beqr\|_2\leq 1\}\subseteq\mathcal{E}_3$. By parts~(a) and~(b) of Theorem~\ref{thm:r_nuisance} and Assumption~\ref{cond:sparsity}, we have $\|\hbepiar-\bepiar\|_2=O_p(r_\pi)=o_p(1)$, $\|\hbepibr-\bepibr\|_2=O_p(r_\pi)=o_p(1)$, and $\|\hbeqr-\beqr\|_2=O_p(r_\pi+r_q)=o_p(1)$. Therefore,
\begin{align*}
    &\P_{\mathbb S_{\pi}\cup \mathbb S_{q}}(\mathcal{E}_1\cap\mathcal{E}_2\cap\mathcal{E}_3)\\
    &\qquad\geq\P_{\mathbb S_{\pi}\cup \mathbb S_{q}}\big(\|\hbepiar-\bepiar\|_2\leq 1,\;\|\hbepibr-\bepibr\|_2\leq 1,\;\|\hbeqr-\beqr\|_2\leq 1\big)=1-o(1).
\end{align*}
Finally, by \eqref{def:E1}--\eqref{def:E3}, we have $\mathcal{E}_1\subseteq\{\|\hbepiar-\bepiar\|_2\leq 1\}$, $\mathcal{E}_2\subseteq\{\|\hbepibr-\bepibr\|_2\leq 1\}$, and $\mathcal{E}_3\subseteq\{\|\hbeqr-\beqr\|_2\leq 1\}$. Consequently, for any integer $r'\in[1,12]$, we obtain $\normp{g^{-1}(\bX^\top\bepi)}{r'}\leq C_1$ and $\normp{\exp(\pm\bX^\top\bepi)}{r'}\leq C_1$ for all $\bepi\in\{w\bepiar+(1-w)\hbepiar:w\in[0,1]\}$ on $\mathcal{E}_1$, while \eqref{event2} holds on $\mathcal{E}_2$ and \eqref{event3} holds on $\mathcal{E}_3$.
\end{proof}

\begin{proof}[Proof of Theorem~\ref{thm:r_nuisance}]

Under Assumption~\ref{cond:basic}(c), it follows from Bayes' rule that
$\P(A=1\mid\bS)=\pi(\bX)/[\pi(\bX)+\{1-\pi(\bX)\}q(\bS)]\geq c_0^3$.
Since $\mu(\bS)-\mu^*(\bS)=\E[\varepsilon\mid A=1,\bS]$, Jensen's inequality bounds $\bemqr_\mu$.
Applying the same inequality to $\tau_n$ and $\tau$, and using Minkowski's inequality, we obtain
\begin{align*}
\bemqr_\mu&\leq c_0^{-1/2}\normp{\varepsilon}{6},\qquad
\normp{\mu(\bzS)-\mu^*(\bzS)}{6}\leq c_0^{-1/6}\bemqr_\mu,\\
\bemqr_{\tau_n}&\leq\normp{\varrho}{6},\qquad
\bemqr_\tau\leq\normp{\zeta}{6}+c_0^{-1/6}\bemqr_\mu.
\end{align*}
These bounds, together with Assumption~\ref{cond:basic2}(b) and the true and working positivity conditions, imply
\begin{equation}
\bemqr_\pi+\bemqr_q+\bemqr_\mu+\bemqr_{\tau_n}+\bemqr_\tau=O(1)
\label{bound:nuisance_approx}
\end{equation}
uniformly in $N,d_1,d_2$. Under Assumption~\ref{cond:sparsity}, we then have
$\{\bemqr_\mu+\bemqr_{\tau_n}+\bemqr_\tau\}(r_\pi+r_q)=o(1)$.

Here and in Section~\ref{apdx:mqr2_proofs}, we refer to the empirical loss gradients evaluated at the target parameters $\bnu^*$ as oracle gradients (e.g., $\bnabla_{\bepia}\bar\ell_1(\bepiar)$ and $\bnabla_{\bemu}\bar\ell_4(\bnu_4^*,\bemur)$).

\emph{(a)} Take $\lambda_{\pi_a}=4\sigma_{\pi_a}\sqrt{\log d_1/M}$ and $\lambda_{\pi_b}=4\sigma_{\pi_b}\sqrt{\log d_1/M}$. Lemma~\ref{lem:r_gradient} with $t=\log d_1$ bounds each oracle gradient by its tuning parameter divided by two with probability at least $1-2/d_1$. By Lemmas \ref{lem:r_gradient} and \ref{lem:r_rsc}, as well as Corollary~9.20 of \cite{HighDimensionalStatisticsNonAsymptotic2019}, we have
\begin{align*}
    \|\hbepiar-\bepiar\|_2=O_p\left(\sqrt{\frac{s_{\pi_a}\log d_1}{M}}\right)\text{, } \|\hbepiar-\bepiar\|_1=O_p\left(s_{\pi_a}\sqrt{\frac{\log d_1}{M}}\right),
\end{align*}
where $\sqrt{s_{\pi_a}\log d_1/M}=O_p(r_{\pi})$ since $M\asymp N$.

Similarly, we have
\begin{align*}
    \|\hbepibr-\bepibr\|_2=O_p\left(\sqrt{\frac{s_{\pi_b}\log d_1}{M}}\right)\text{, } \|\hbepibr-\bepibr\|_1=O_p\left(s_{\pi_b}\sqrt{\frac{\log d_1}{M}}\right),
\end{align*}
where $\sqrt{s_{\pi_b}\log d_1/M}=O_p(r_{\pi})$ since $M\asymp N$.

\emph{(b)}
Choosing some $\lambda_{\pi_a}$, $\lambda_{\pi_b} \asymp \sqrt{\log d_1/N}$, we obtain from Theorem~\ref{thm:r_nuisance} (a) that $\rmqr_\pi=O_p(r_\pi)=o_p(1)$, where $r_\pi=\sqrt{s_\pi\log d_1/N}=o(1)$.
Recall that in Lemma~\ref{lem:r_imputation}, $s_{\pi}=s_{\pi_a}+s_{\pi_b}$ and
\begin{align*}
    \mathcal B_1:=&\{\|\bnabla_{\beq}\bar\ell_3(\bnu_3^*,\beqr)\|_\infty\leq\lambda_{q}/2\},\\
    \mathcal B_2:=&\left\{|R_1(\bDelta)|\leq c_t r_{\pi}\left(\|\bDelta\|_1\sqrt{\frac{\log d}{N}}+\|\bDelta\|_2\right),\;\; \forall\bDelta\in\R^d\right\},\\
    \mathcal B_3:=&\left\{\delta\bar\ell_3(\hat{\bnu}_3,  \beqr,\bDelta)\geq\kappa_1\|\bDelta\|_2^2-\kappa_2\frac{\log d}{M}\|\bDelta\|_1^2,\;\;\forall\bDelta\in\R^d:\|\bDelta\|_2\leq1\right\},
\end{align*}
where
\begin{align}
    R_1(\bDelta):=\{\bnabla_{\beq}\bar{\ell}_3(\hat{\bnu}_3, \beqr)-\bnabla_{\beq}\bar{\ell}_3(\bnu_3^*,\beqr)\}^\top\bDelta=(Q_{1}-Q_{2})\label{decom:R1}
\end{align}
with
\begin{align*}
    Q_{1}&:= M^{-1}\suqr A_i[g^{-1}(\bX_i^\top\hbepiar)-g^{-1}(\bX_i^\top\bepiar)]\exp(\bS_i^\top\beqr)\bS_i^\top\bDelta,\\
    Q_{2}&:= M^{-1}\suqr (1-A_i)[\{1-g(\bX_i^\top\hbepibr)\}^{-1}-\{1-g(\bX_i^\top\bepibr)\}^{-1}]\bS_i^\top\bDelta.
\end{align*}
Fix $\lambda_q=4\sigma_q\sqrt{\log d/M}$. Lemma~\ref{lem:r_gradient} with $t=\log d$ gives $\P_{\mathbb S_q}(\mathcal B_1)\geq1-2/d=1-o(1)$. We then analyze the probability of the event $\mathcal B_2$.
By Assumption~\ref{cond:basic2}, we define the event
\begin{align}
    \mathcal{E}^*:=\{c_0\leq\pi_a^*(\bX_i), \pi_b^*(\bX_i)\leq1-c_0\text{ and }c_0^2\leq\exp(\bS_i^\top\beqr)\leq c_0^{-2} \text{ for all }i\in[N]\}.\label{def:E*}
\end{align}
It holds that $\P_{\mathbb S}(\mathcal{E}^*)=1$.
Conditional on $\mathcal{E}^*$, the term $Q_1$ can be bounded as:
\begin{align}
    |Q_1|&\overset{(i)}{\leq} \frac1{M}\suqr \left|\left[\frac{1}{g(\bX_i^\top\hbepiar)}-\frac{1}{g(\bX_i^\top\bepiar)}\right]\exp(\bS_i^\top\beqr)\bS_i^\top\bDelta\right|\nonumber\\
    &\overset{(ii)}{\leq} c_0^{-2}\frac1{M}\suqr \left|\left[\frac{1}{g(\bX_i^\top\hbepiar)}-\frac{1}{g(\bX_i^\top\bepiar)}\right]\bS_i^\top\bDelta\right|\nonumber\\
    &\overset{(iii)}{\leq} c_0^{-2}\sqrt{\frac1{M}\suqr\left[\frac{1}{g(\bX_i^\top\hbepiar)}-\frac{1}{g(\bX_i^\top\bepiar)}\right]^2}\sqrt{\frac1{M}\suqr(\bS_i^\top\bDelta)^2},\label{bound:Q1}
\end{align}
where (i) holds by the triangle inequality and $|A|\leq 1$; (ii) holds on $\mathcal{E}^*$; and (iii) holds by the Cauchy--Schwarz inequality. Since $\P_{\mathbb S}(\mathcal{E}^*)=1$, \eqref{bound:Q1} holds with probability $1$.
By Lemma~\ref{lem:usebound} and Theorem~\ref{thm:r_nuisance}(a), the first term on the right-hand side of \eqref{bound:Q1} has expectation
\begin{align*}
    \E_{\mathbb{S}_{q}}\left[\frac1{M}\suqr\left[\frac{1}{g(\bX_i^\top\hbepiar)}-\frac{1}{g(\bX_i^\top\bepiar)}\right]^2\right]&=\E\left[\left[\frac{1}{g(\bX^\top\hbepiar)}-\frac{1}{g(\bX^\top\bepiar)}\right]^2\right]=O_p(\rmqr_{\pi}^2).
\end{align*}
In the above, the second expectation is taken with respect to the independent copy $(A, \bX)$, as the estimator $\hbepiar$ can be treated as fixed.
Hence by Markov's inequality,
\begin{align*}
    \frac1{M}\suqr\left[\frac{1}{g(\bX_i^\top\hbepiar)}-\frac{1}{g(\bX_i^\top\bepiar)}\right]^2=O_p(\rmqr_{\pi}^2)=O_p(r_{\pi}^2).
\end{align*}
For the other term, by Lemma~\ref{lemma:beta1},
\begin{align*}
\sup_{\bDelta\in\R^{d}\setminus\{\bz\}}\frac{M^{-1}\sum_{i\in\I_{q}}(\bS_{i}^\top\bDelta)^2}{N^{-1}\log d\|\bDelta\|_1^2+\|\bDelta\|_2^2}=O_p(1).
\end{align*}
It follows that, for any $\bDelta\in\R^d$,
\begin{align*}
    \frac1{M}\suqr(\bS_i^\top\bDelta)^2
    &\leq \left(\frac{\log d}{N}\|\bDelta\|_1^2+\|\bDelta\|_2^2\right)\sup_{\bDelta'\in\R^{d}\setminus\{\bz\}}\frac{M^{-1}\sum_{i\in\I_{q}}(\bS_{i}^\top\bDelta')^2}{N^{-1}\log d\|\bDelta'\|_1^2+\|\bDelta'\|_2^2}\\
    &=O_p\left(\frac{\log d}{N}\|\bDelta\|_1^2+\|\bDelta\|_2^2\right),
\end{align*}
since when $\bDelta=0$, the left-hand side becomes zero. We then conclude that
\begin{align*}
    |Q_1|=O_p\left(r_{\pi}\left(\sqrt{\frac{\log d}{N}}\|\bDelta\|_1+\|\bDelta\|_2\right)\right).
\end{align*}

For the term $Q_2$, note that
\begin{align*}
    |Q_2|&\overset{(i)}{\leq} \frac1{M}\suqr \left|\left[\frac{1}{1-g(\bX_i^\top\hbepibr)}-\frac{1}{1-g(\bX_i^\top\bepibr)}\right]\bS_i^\top\bDelta\right|\\
    &\overset{(ii)}{\leq} \sqrt{\frac1{M}\suqr\left[\frac{1}{1-g(\bX_i^\top\hbepibr)}-\frac{1}{1-g(\bX_i^\top\bepibr)}\right]^2}\sqrt{\frac1{M}\suqr(\bS_i^\top\bDelta)^2}
\end{align*}
holds conditional on $\mathcal{E}^*$. In the above, (i) holds by the triangle inequality and $|1-A|\leq 1$; and (ii) follows from the Cauchy--Schwarz inequality. Since $\P_{\mathbb S}(\mathcal{E}^*)=1$, by Lemma~\ref{lem:usebound} and Theorem~\ref{thm:r_nuisance} (a), we have
\begin{align*}
    \E_{\mathbb{S}_{q}}\left[\frac1{M}\suqr\left[\frac{1}{1-g(\bX_i^\top\hbepibr)}-\frac{1}{1-g(\bX_i^\top\bepibr)}\right]^2\right]&=\E\left[\left[\exp(\bX^\top\hbepibr)-\exp(\bX^\top\bepibr)\right]^2\right]\\
    &=O_p(\rmqr_{\pi}^2)=O_p(r_{\pi}^2).
\end{align*}
Similarly to the argument for $Q_1$, we have
\begin{align*}
    |Q_2|=O_p\left(r_{\pi}\left(\sqrt{\frac{\log d}{N}}\|\bDelta\|_1+\|\bDelta\|_2\right)\right).
\end{align*}

Combined with \eqref{decom:R1}, we conclude that, when $N$ is large enough, for any $t>0$, there exists some constant $c_t>0$ such that $\P_{\mathbb{S}_{\pi}\cup \mathbb{S}_{q}}(\mathcal B_2)\geq1-1/t$.

By \eqref{eq:r_RSC3} in Lemma~\ref{lem:r_rsc}, we have $\P_{\mathbb{S}_{\pi}\cup \mathbb{S}_{q}}(\mathcal{B}_3)\geq 1-c_1\exp(-c_{2} M)=1-o(1)$. Hence we conclude that
\begin{align*}
\P_{\mathbb{S}_{\pi} \cup \mathbb{S}_{q}}(\mathcal B_1\cap\mathcal B_2 \cap \mathcal B_3)&\geq 1-1/t-o(1).
\end{align*}

In the remaining proof of (b), we write $\bDelta:=\hbeqr-\beqr$. Then by Lemma~\ref{lem:r_imputation}(a), on the event $\mathcal B_1$, we have
\begin{align*}
    2\delta\bar\ell_3(\hat{\bnu}_3, \beqr,\bDelta)+\lambda_{q}\|\bDelta\|_1&\leq 4\lambda_{q}\|\bDelta_{S_{\beqr}}\|_1+2|R_1(\bDelta)|.
\end{align*}
Conditional on $\mathcal B_1\cap\mathcal B_2\cap\mathcal B_3$, by Lemma~\ref{lem:r_imputation}, we have $\bDelta\in\Ctil(\bar s_{q},k_0)$ and $\|\bDelta\|_2\leq 1$. Since $\lambda_q=4\sigma_q\sqrt{\log d/M}$, $M\asymp N$, and $r_\pi=o(1)$, for each fixed $t>0$ we have $c_tr_\pi\sqrt{\log d/N}=o(\lambda_q)$, so that $c_tr_\pi\sqrt{\log d/N}\leq\lambda_q/2$ for all sufficiently large $N$.
For large enough $N$,
\begin{align*}
    \left(2\lambda_{q}\sqrt{s_{q}}+c_tr_{\pi}\right)\|\bDelta\|_2 &\overset{(i)}{\geq} 2\lambda_{q}\|\bDelta_{S_{\beqr}}\|_1+|R_1(\bDelta)|-\frac12\lambda_{q}\|\bDelta\|_1\\
    &\overset{(ii)}{\geq}\delta\bar\ell_3(\hat{\bnu}_3, \beqr,\bDelta)\\
    &\overset{(iii)}{\geq}\kappa_1\|\bDelta\|_2^2-\kappa_2k_0^2\frac{\bar s_{q}\log d}{M}\|\bDelta\|_2^2\\
    &\overset{(iv)}{\geq}\frac{\kappa_1}{2}\|\bDelta\|_2^2,
\end{align*}
where (i) holds by $\|\bDelta_{S_{\beqr}}\|_1\leq \sqrt{s_{q}}\|\bDelta_{S_{\beqr}}\|_2\leq \sqrt{s_{q}}\|\bDelta\|_2$ and the fact that on the event $\mathcal B_2$, $|R_1(\bDelta)|\leq c_tr_{\pi}(\|\bDelta\|_1\sqrt{\log d/N}+\|\bDelta\|_2)\leq
(\lambda_{q}/2)\|\bDelta\|_1+c_tr_{\pi}\|\bDelta\|_2$; (ii) holds by Lemma~\ref{lem:r_imputation}; (iii) holds by the construction of $\mathcal B_3$, $\|\bDelta\|_2\leq 1$ and $\bDelta\in\Ctil(\bar s_{q},k_0)$; and (iv) holds since $\bar s_{q}\log d/M=s_{\pi}\log d_1/M+s_{q}\log d/M=o(1)$.

Since the selected $\lambda_q=4\sigma_q\sqrt{\log d/M}\asymp\sqrt{\log d/N}$, we have
$$\|\bDelta\|_2\leq\frac{4\lambda_{q}\sqrt{s_{q}}}{\kappa_1}+\frac{2c_t}{\kappa_1}r_{\pi}=O_p\left(r_{\pi}+r_{q}\right)$$
on the event $\mathcal B_1\cap \mathcal B_2\cap \mathcal B_3$.
Since $\bDelta\in\Ctil(\bar s_{q},k_0)$, it follows that
$$\|\bDelta\|_1\leq k_0\sqrt{\bar s_{q}}\|\bDelta\|_2=O_p\left(s_{\pi}\sqrt\frac{(\log d_1)^2}{N\log d}+s_{q}\sqrt\frac{\log d}{N}\right).$$
For each fixed $t>0$, the preceding deterministic bounds give a constant $C_t>0$, depending only on $t$ and the fixed regularity constants, such that
\begin{align*}
\limsup_{N\to\infty}\P_{\mathbb S_\pi\cup\mathbb S_q}
\{\|\hbeqr-\beqr\|_2>C_t(r_\pi+r_q)\}\leq1/t.
\end{align*}
The same argument applies to the $\ell_1$ bound. Letting $t\to\infty$ proves the stated stochastic orders.

\emph{(c)} Choosing some $\lambda_{\pi_a}, \lambda_{\pi_b}\asymp\sqrt{\log d_1/N}$ and $\lambda_{q}\asymp\sqrt{\log d/N}$, we obtain from Theorem~\ref{thm:r_nuisance} (a) and (b) that
\begin{align*}
     \|\hbepia-\bepiar\|_2=O_p\left(r_{\pi}\right) \text{ and }
    \|\hbeq-\beqr\|_2=O_p\left(r_{\pi}+r_{q}\right).
\end{align*}

In the remaining proof of (c), let $\bDelta:=\hbemur-\bemur$.
Consider the decomposition
\begin{align}
    \bnabla_{\bemu}\bar\ell_4(\hat{\bnu}_4, \bemur)^\top\bDelta=\bnabla_{\bemu}\bar\ell_4(\bnu_4^*,\bemur)^\top\bDelta+R_2+R_3,\label{decom:R2R3}
\end{align}
where
\begin{align*}
    R_2&:=M^{-1}\sumur -2A_i\left[\frac{\exp(\bS_i^\top\hbeqr)}{g(\bX_i^\top\hbepiar)}-\frac{\exp(\bS_i^\top\beqr)}{g(\bX_i^\top\bepiar)}\right]\omega_{\mu,i}\bS_i^\top\bDelta,\\
    R_3&:=M^{-1}\sumur W_{\mu,i}^\top\bDelta
\end{align*}
with
\begin{align*}
    W_{\mu,i}:=-2A_i\left\{\frac{\exp(\bS_i^\top\hbeqr)}{g(\bX_i^\top\hbepiar)}-\frac{\exp(\bS_i^\top\beqr)}{g(\bX_i^\top\bepiar)}\right\}\{Y_i(1,\bM_i)-\mu(\bS_i)\}\bS_i.
\end{align*}

It follows from Young's inequality $2ab\leq a^2/2+2b^2$ that
\begin{equation}
    |R_2|\leq \frac{1}{2}\delta\bar\ell_4(\hat{\bnu}_4, \bemur,\bDelta)+2\bar{R}_2,\label{rela:R2R3}
\end{equation}
where
\begin{align*}
    \bar{R}_2&:=\frac1{M}\sumur \frac{g(\bX_i^\top\hbepiar)}{\exp(\bS_{i}^\top\hbeqr)}\big[g^{-1}(\bX_i^\top\hbepiar)\exp(\bS_{i}^\top\hbeqr)-g^{-1}(\bX_i^\top\bepiar)\exp(\bS_{i}^\top\beqr)\big]^2\omega_{\mu,i}^2.
\end{align*}
By H\"older's inequality with exponents $(3,3,3)$ and $g\leq1$, we have
\begin{align*}
    \E_{\mathbb S_{\mu}}[\bar{R}_2]=&\E\left[\frac{g(\bX^\top\hbepiar)}{\exp(\bS^\top\hbeqr)}\left[\frac{\exp(\bS^\top\hbeqr)}{g(\bX^\top\hbepiar)}
    -\frac{\exp(\bS^\top\beqr)}{g(\bX^\top\bepiar)}\right]^2\omega_{\mu}^2\right]\nonumber\\
    \leq& \normp{\exp(-\bS^\top\hbeqr)}{3}\left\|\frac{\exp(\bS^\top\hbeqr)}{g(\bX^\top\hbepiar)}-\frac{\exp(\bS^\top\beqr)}{g(\bX^\top\bepiar)}\right\|_{\P,6}^2\|\omega_\mu\|_{\P,6}^2.
\end{align*}

Taking $d_0=d_1$, $\balpha=\bepiar$ and $\halpha=\hbepiar$, we obtain from Lemma~\ref{lem:usebound} and Theorem~\ref{thm:r_nuisance} (a) that
\begin{align}
    \normp{\exp(-\bX^\top\hbepiar)-\exp(-\bX^\top\bepiar)}{r}&=\normp{g^{-1}(\bX^\top\hbepiar)-g^{-1}(\bX^\top\bepiar)}{r}\nonumber\\
    &=O_p(\|\hbepia-\bepiar\|_2)=O_p\left(r_{\pi}\right) \label{order:r_betara_imputation}
\end{align}
holds for any integer $r>0$. Similarly, letting $d_0=d$, $\balpha=\beqr$ and $\halpha=\hbeqr$, we have
\begin{align}
    \normp{\exp(\bS^\top\hbeqr)-\exp(\bS^\top\beqr)}{r}=O_p(\|\beqr-\hbeq\|_2)=O_p\left(r_{\pi}+r_{q}\right). \label{order:r_q_imputation}
\end{align}
By Lemma~\ref{lem:usebound}(b) and Assumption~\ref{cond:basic2}(a), we have $\|\exp(-\bS^\top\hbeqr)\|_{\P,4}=O_p(1)$ and $\|\exp(\bS^\top\hbeqr)\|_{\P,16}=O_p(1)$. Therefore, by \eqref{order:r_betara_imputation}, \eqref{order:r_q_imputation} and the probability-one weight bounds in Assumption~\ref{cond:basic2}(a),
\begin{align}
    \left\|\frac{\exp(\bS^\top\hbeqr)}{g(\bX^\top\hbepiar)}-\frac{\exp(\bS^\top\beqr)}{g(\bX^\top\bepiar)}\right\|_{\P,8}
    &\leq\normp{g^{-1}(\bX^\top\bepiar)}{16}\normp{\left\{\exp(\bS^\top\hbeqr)-\exp(\bS^\top\beqr)\right\}}{16}\nonumber\\
    &\quad+\normp{\exp(\bS^\top\hbeqr)}{16}\normp{g^{-1}(\bX^\top\hbepiar)-g^{-1}(\bX^\top\bepiar)}{16}\nonumber\\
    &=O_p\left(r_{\pi}+r_{q}\right).\label{eq:r_beq+bepi_bound}
\end{align}
By Lyapunov's inequality, the preceding $L_8$ bound also controls the $L_6$ weight difference, and $\|\omega_\mu\|_{\P,6}=\bemqr_\mu$. We conclude that
\begin{align}
    \E_{\mathbb S_{\mu}}[\bar{R}_2]&=O_p\left(\bemqr_{\mu}^2\{r_{\pi}^2+r_{q}^2\}\right), \text{ and }\bar{R}_2=O_p\left(\bemqr_{\mu}^2\{r_{\pi}^2+r_{q}^2\}\right)\label{bound:R3}.
\end{align}

Define
\begin{align*}
    \mathcal B_4:=&\{\|\bnabla_{\bemu}\bar\ell_4(\bnu_4^*,\bemur)+M^{-1}\sumur W_{\mu,i}\|_\infty \leq\lambda_{\mu}/2\},\\
    \mathcal B_5:=&\left\{\delta\bar\ell_4(\hat{\bnu}_4, \bemur,\bDelta)\geq\kappa_1\|\bDelta\|_2^2-\kappa_2\frac{\log d}{M}\|\bDelta\|_1^2,\;\;\forall\bDelta\in\R^d\right\}.
\end{align*}
We then prove that $\P_{\mathbb{S}_\pi \cup \mathbb{S}_q \cup \mathbb{S}_\mu}(\mathcal B_4\cap \mathcal B_5)\geq 1-o(1)$.

We first bound the oracle gradients used in parts~(c)--(e).
Their means are zero by the population first-order conditions.
Assumption~\ref{cond:basic2} bounds the weights almost surely and controls the sub-Gaussian norms of the residuals
$\varepsilon$, $\varrho$, and $q^*(\bS)\varepsilon+\zeta-\varrho$.
By Lemma~\ref{lem:r_common_gradient}(b), we may choose
$\lambda_\mu=C_\mu\sqrt{\log d/M}$ and $\lambda_h=C_h\sqrt{\log d_1/M}$ for $h\in\{\tau_n,\tau\}$,
with sufficiently large fixed constants depending only on $c_0,\sigma_{\bS},\sigma_0$,
so that each oracle gradient is bounded by its tuning parameter divided by four with probability $1-o(1)$.

Conditional on $\mathbb S_\pi\cup\mathbb S_q$, the tower rule gives $\E_{\mathbb S_\mu}W_{\mu,i}=\bz_d$.
By \eqref{eq:r_beq+bepi_bound}, the weight difference in $W_{\mu,i}$ has $L_8$ norm $O_p(r_\pi+r_q)$.
Jensen's inequality gives
$\|A\{Y(1,\bM)-\mu(\bS)\}\|_{\P,8}\leq2\|\varepsilon\|_{\P,8}=O(1)$.
By Lemma~\ref{lem:r_common_gradient}(a) applied to $W_{\mu,i}$ with design $\bS_i$, we obtain
\begin{equation*}
\E_{\mathbb S_\mu}\left\|M^{-1}\sumur W_{\mu,i}\right\|_\infty^2
=O_p\left[\frac{\log d}{N}(r_\pi^2+r_q^2)\left(1+\frac{\log d}{\sqrt N}\right)\right]
=o_p(\log d/N),
\end{equation*}
where the last equality follows from Assumption~\ref{cond:sparsity}, so conditional Markov's inequality implies $\|M^{-1}\sum_{i\in\I_\mu} W_{\mu,i}\|_\infty=o_p(\sqrt{\log d/N})$.
Combining the two gradient bounds with Lemma~\ref{lem:r_rsc} and parts~(a)--(b) of this theorem, we obtain
$\P_{\mathbb S_\pi\cup\mathbb S_q\cup\mathbb S_\mu}(\mathcal B_4\cap\mathcal B_5)=1-o(1)$.

Following arguments similar to those used in the proof of Lemma~\ref{lem:r_imputation}(a), we have, on the event $\mathcal B_4$,
\begin{align}
    2\delta\bar\ell_4(\hat{\bnu}_4, \bemur,\bDelta)+\lambda_{\mu}\|\bDelta\|_1&\leq 4\lambda_{\mu}\|\bDelta_{S_{\bemur}}\|_1+2|R_2|.\label{eq:r_key2}
\end{align}
Based on \eqref{rela:R2R3} and the fact that $\|\bDelta_{S_{\bemur}}\|_1\leq\sqrt{s_{\mu}}\|\bDelta_{S_{\bemur}}\|_2\leq\sqrt{s_{\mu}}\|\bDelta\|_2$, \eqref{eq:r_key2} implies that
\begin{equation*}
    \delta\bar\ell_4(\hat{\bnu}_4, \bemur,\bDelta)+\lambda_{\mu}\|\bDelta\|_1\leq 4\lambda_{\mu}\|\bDelta_{S_{\bemur}}\|_1+4\bar{R}_2\leq 4\lambda_{\mu}\sqrt{s_{\mu}}\|\bDelta\|_2+4\bar{R}_2.
\end{equation*}
By the convexity of $\ell_4$, we have $\delta\bar\ell_4(\hat{\bnu}_4, \bemur,\bDelta)\geq 0$, so that
\begin{equation}\label{eq:r_deltacone}
    \|\bDelta\|_1\leq4\sqrt{s_{\mu}}\|\bDelta\|_2+4\lambda_{\mu}^{-1}\bar{R}_2.
\end{equation}
Therefore, on the event $\mathcal{E}_1\cap\mathcal{E}_3\cap\mathcal B_4\cap\mathcal B_5$, the following inequalities hold for sufficiently large $N$,
\begin{align*}
    4\lambda_{\mu}\sqrt{s_{\mu}}\|\bDelta\|_2+4\bar{R}_2
    &\overset{(i)}{\geq}\delta\bar\ell_4(\hat{\bnu}_4, \bemur,\bDelta)
    \overset{(ii)}{\geq} \kappa_1\|\bDelta\|_2^2-\kappa_2\frac{\log d}{M}\|\bDelta\|_1^2\\
    &\overset{(iii)}{\geq} \kappa_1\|\bDelta\|_2^2-2\kappa_2\frac{\log d}{M}\left(16s_{\mu}\|\bDelta\|_2^2+\frac{16\bar{R}_2^2}{\lambda_{\mu}^2}\right)\\
    &\overset{(iv)}{\geq} \frac{\kappa_1}{2}\|\bDelta\|_2^2-32\kappa_2\bar{R}_2^2\frac{ \log d }{M\lambda_{\mu}^2},
\end{align*}
where (i) holds since $\|\bDelta\|_1 \geq0$; (ii) holds on $\mathcal B_5$; (iii) holds by \eqref{eq:r_deltacone} and also the fact that $(a+b)^2\leq2a^2+2b^2$; (iv) holds for large enough $N$, since $s_{\mu}\log d/M =o(1)$.
By Lemma~\ref{lem:r_useevent} and the preceding arguments, we have $\P_{\mathbb S_{\pi}\cup\mathbb S_q\cup\mathbb S_\mu}(\mathcal E_1\cap\mathcal E_3\cap\mathcal B_4\cap\mathcal B_5)\geq 1-o(1)$.
Since the selected tuning parameter $\lambda_{\mu}\asymp\sqrt{\log d/N}$, it follows from Lemma~\ref{lemma:sol} that
\begin{align}
    \|\bDelta\|_2
    &\leq \frac{8\lambda_{\mu}\sqrt{s_{\mu}}}{\kappa_1}+\sqrt{64\bar{R}_2^2\frac{ \kappa_2\log d }{\kappa_1 M\lambda_{\mu}^2}+\frac{8\bar{R}_2}{\kappa_1}}\nonumber\\
    &\overset{(i)}{=}O_p\left(r_{\mu}+\bemqr_{\mu}^2\{r_{\pi}+r_{q}\}^2+\bemqr_{\mu}\{r_{\pi}+r_{q}\} \right)\overset{(ii)}{=}O_p\left(r_{\mu}+\bemqr_{\mu}\{r_{\pi}+r_{q}\}\right),\label{eq:r_mu_bar_rate}
\end{align}
where (i) holds by $\lambda_{\mu}\sqrt{s_{\mu}}\asymp\sqrt{s_{\mu}\log d/N}$ and \eqref{bound:R3}, and (ii) follows from $\bemqr_{\mu}\{r_{\pi}+r_{q}\}=o(1)$ by \eqref{bound:nuisance_approx} and Assumption~\ref{cond:sparsity}.
Combined with \eqref{eq:r_deltacone}, we further have
\begin{align*}
    \|\bDelta\|_1=O_p\left(\sqrt{\frac{N}{\log d}}\left\{\bemqr_{\mu}(r_{\pi}+r_q)+r_{\mu}\right\}^2\right).
\end{align*}
By Lemma~\ref{lem:r_approx_bound}(c), we have $\bemqr_\mu=O(\emqr_\mu)$, which yields the stated bound $\|\hbemur-\bemur\|_2=O_p(r_{\mu}+\emqr_{\mu}\{r_{\pi}+r_{q}\})$.

\emph{(d)}
Choosing some $\lambda_{\pi_a}, \lambda_{\pi_b} \asymp\sqrt{\log d_1/N}$ and some $\lambda_{q}, \lambda_{\mu}\asymp\sqrt{\log d/N}$, we obtain from part~(a) and \eqref{eq:r_mu_bar_rate} in part~(c) that
\begin{align}
     \|\hbepib-\bepibr\|_2=O_p\left(r_{\pi}\right) \text{ and }
    \|\hbemu-\bemur\|_2=O_p\left(r_{\mu}+\bemqr_{\mu}\{r_\pi+r_q\}\right).\label{rate:tau_n_previous}
\end{align}

In the remaining proof of (d), let $\bDelta=\hbettaur-\bettaur$.
Define $\omega_{\tau_n,i}=\tau_n(\bX_i)-\tau_n^*(\bX_i)$ and $\omega_{\mu,i}=\mu(\bS_i)-\mu^*(\bS_i)$.
Consider the decomposition
\begin{align*}
    \bnabla_{\bettau}\bar\ell_5(\hat{\bnu}_5, \bettaur)^\top\bDelta=\bnabla_{\bettau}\bar\ell_5(\bnu_5^*,\bettaur)^\top\bDelta+R_4+R_5,
\end{align*}
where
\begin{align*}
    R_4:=&-2M^{-1}\suttaur (Q_{3,i}-Q_{4,i})\bX_i^\top\bDelta,\\
    R_5:=&M^{-1}\suttaur W_{\tau_n,i}^\top\bDelta
\end{align*}
with
\begin{align*}
    Q_{3,i}&:=(1-A_i)\{\exp(\bX_i^\top\hbepibr)-\exp(\bX_i^\top\bepibr)\}\omega_{\tau_n,i},\\
    Q_{4,i}&:=(1-A_i)\exp(\bX_i^\top\hbepibr)\bS_i^\top\{\bemur-\hbemur\},\\
    W_{\tau_n,i}&:=2(1-A_i)\left[\exp(\bX_i^\top\hbepibr)-\exp(\bX_i^\top\bepibr)\right]\{\tau_n(\bX_i)-\mu^*(\bS_i)\}\bX_i.
\end{align*}
Applying Young's inequality $2ab\leq a^2/2+2b^2$, we obtain
\begin{equation}
    |R_4|\leq \frac{1}{2}\delta\bar\ell_5(\hat{\bnu}_5, \bettaur,\bDelta)+2\bar{R}_4,\label{rela:R4R5}
\end{equation}
where
\begin{align*}
    \bar{R}_4=&\frac1{M}\suttaur \exp(-\bX_i^\top\hbepibr)\left\{Q_{3,i}-Q_{4,i}\right\}^2\leq \frac2{M}\suttaur \exp(-\bX_i^\top\hbepibr)\left\{Q_{3,i}^2+Q_{4,i}^2\right\}.
\end{align*}
Let $d_0=d_1$, $\balpha=\bepibr$ and $\halpha=\hbepibr$. By (a), we have $\|\balpha-\halpha\|_2=O_p(r_{\pi})$; hence by Lemma~\ref{lem:usebound}, it holds for any integer $r>0$ that
    \begin{align}
    \normp{\frac{1}{1-g(\bX^\top\hbepibr)}-\frac{1}{1-g(\bX^\top\bepibr)}}{r}=\normp{\exp(\bX^\top\hbepibr)-\exp(\bX^\top\bepibr)}{r}=O_p\left(r_{\pi}\right).
    \label{order:r_betabr_imputation}
\end{align}
Similarly, let $d_0=d$, $\balpha=\bemur$ and $\halpha=\hbemur$. By Lemma~\ref{lem:usebound} and (c), we have, for any integer $r>0$,
\begin{align}
    \normp{\bS^\top(\bemur-\hbemur)}{r}=O_p(\|\bemur-\hbemur\|_2)=O_p\left(\bemqr_{\mu}\{r_{\pi}+r_{q}\}+r_{\mu}\right).
    \label{order:r_mu_imputation}
\end{align}
By Lemma~\ref{lem:r_useevent}, on the event $\mathcal{E}_2$ defined in \eqref{def:E2}, there exists a constant $C_2>0$ such that $\|\exp(\bX^\top\hbepib)\|_{\P,r'}\leq C_2$ and $\|\exp(-\bX^\top\hbepib)\|_{\P,r'}\leq C_2$  for every integer $r'\in[1,12]$. Furthermore, $\mathcal{E}_2$ occurs with probability approaching 1. Therefore, for each $i\in\mathcal{I}_{\tau_n}$,
\begin{align*}
    &\E_{\mathbb{S}_{\tau_n}}[\exp(-\bX_i^\top\hbepibr)Q_{3,i}^2]\leq \E[\exp(-\bX^\top\hbepibr)\{\exp(\bX^\top\hbepibr)-\exp(\bX^\top\bepibr)\}^2\omega_{\tau_n}^2]\\
    &\overset{(i)}{\leq} \normp{\exp(-\bX^\top\hbepibr)}{6}\normp{\exp(\bX^\top\hbepibr)-\exp(\bX^\top\bepibr)}{4}^2\normp{\omega_{\tau_n}}{6}^2\\
    &\overset{(ii)}{=}O_p(\bemqr_{\tau_n}^2r_{\pi}^2),
\end{align*}
where (i) follows from H\"older's inequality together with $|1-A|\leq1$;
(ii) follows from \eqref{rate:tau_n_previous} and \eqref{order:r_betabr_imputation}.
Furthermore,
\begin{align*}
    &\E_{\mathbb{S}_{\tau_n}}[\exp(-\bX_i^\top\hbepibr)Q_{4,i}^2]\leq \E[\exp(\bX^\top\hbepibr)\{\bS^\top(\bemur-\hbemur)\}^2]\\
    &\leq \normp{\exp(\bX^\top\hbepibr)}{2}\normp{\bS^\top(\bemur-\hbemur)}{4}^2=O_p(r_{\mu}^2+\bemqr_{\mu}^2\{r_{\pi}^2+r_{q}^2\}).
\end{align*}
Collectively, we have
\begin{align*}
    \E_{\mathbb{S}_{\tau_n}}[\bar{R}_4]=O_p(\bemqr_{\tau_n}^2r_{\pi}^2+\bemqr_{\mu}^2\{r_{\pi}^2+r_{q}^2\}+r_{\mu}^2)
\end{align*}
and
\begin{align}
    \bar{R}_4=O_p(\bemqr_{\tau_n}^2r_{\pi}^2+\bemqr_{\mu}^2\{r_{\pi}^2+r_{q}^2\}+r_{\mu}^2).\label{bound:R_11}
\end{align}
Define
\begin{align*}
    \mathcal B_6:=&\left\{\left\|\bnabla_{\bettau}\bar\ell_5(\bnu_5^*,\bettaur)+M^{-1}\suttaur W_{\tau_n,i}\right\|_\infty \leq\lambda_{\tau_n}/2\right\},\\
    \mathcal B_7:=&\left\{\delta\bar\ell_5(\hat{\bnu}_5, \bettaur,\bDelta)\geq\kappa_1\|\bDelta\|_2^2-\kappa_2\frac{\log d_1}{M}\|\bDelta\|_1^2,\;\;\forall\bDelta\in\R^{d_1}\right\}.
\end{align*}
We then prove that $\P_{\mathbb{S}_{\pi}\cup \mathbb{S}_q \cup \mathbb{S}_\mu \cup \mathbb{S}_{\tau_n}}(\mathcal{B}_6 \cap \mathcal{B}_7)=1-o(1)$.

Conditional on $\mathbb S_\pi\cup\mathbb S_q\cup\mathbb S_\mu$, the identity
$\E[\mu^*(\bS)\mid A=0,\bX]=\tau_n(\bX)$ gives $\E_{\mathbb S_{\tau_n}}W_{\tau_n,i}=\bz_{d_1}$.
By \eqref{order:r_betabr_imputation}, the weight difference in $W_{\tau_n,i}$ has $L_8$ norm $O_p(r_\pi)$, while Jensen's inequality gives
$\|(1-A)\{\tau_n(\bX)-\mu^*(\bS)\}\|_{\P,8}\leq2\|\varrho\|_{\P,8}=O(1)$.
Lemma~\ref{lem:r_common_gradient}(a) and conditional Markov's inequality, with design $\bX_i$, therefore yield
$\|M^{-1}\suttaur W_{\tau_n,i}\|_\infty=o_p(\sqrt{\log d_1/N})$
under Assumption~\ref{cond:sparsity}, since $d_1\leq d$.
Combining this bound with the oracle-gradient bound in part~(c) and Lemma~\ref{lem:r_rsc}, we obtain
$\P_{\mathbb S_\pi\cup\mathbb S_q\cup\mathbb S_\mu\cup\mathbb S_{\tau_n}}(\mathcal B_6\cap\mathcal B_7)=1-o(1)$.

Following the proof of Lemma~\ref{lem:r_imputation}(a), on the event $\mathcal B_6$ we have
\begin{equation}
    2\delta\bar\ell_5(\hat{\bnu}_5,\bettaur,\bDelta)+\lambda_{\tau_n}\|\bDelta\|_1
    \leq4\lambda_{\tau_n}\|\bDelta_{S_{\bettaur}}\|_1+2|R_4|.\label{eq:r_key3}
\end{equation}
By \eqref{rela:R4R5}, $\|\bDelta_{S_{\bettaur}}\|_1\leq\sqrt{s_{\tau_n}}\|\bDelta\|_2$, and the convexity of $\ell_5$, this implies
\begin{align}
    \delta\bar\ell_5(\hat{\bnu}_5,\bettaur,\bDelta)+\lambda_{\tau_n}\|\bDelta\|_1
    &\leq4\lambda_{\tau_n}\sqrt{s_{\tau_n}}\|\bDelta\|_2+4\bar R_4,\nonumber\\
    \|\bDelta\|_1
    &\leq4\sqrt{s_{\tau_n}}\|\bDelta\|_2+4\lambda_{\tau_n}^{-1}\bar R_4.\label{eq:r_varphicone}
\end{align}
On $\mathcal E_2\cap\mathcal B_6\cap\mathcal B_7$, the RSC inequality and \eqref{eq:r_varphicone} therefore yield, for sufficiently large $N$,
\begin{align*}
    4\lambda_{\tau_n}\sqrt{s_{\tau_n}}\|\bDelta\|_2+4\bar R_4
    &\geq\delta\bar\ell_5(\hat{\bnu}_5,\bettaur,\bDelta)\\
    &\geq\frac{\kappa_1}{2}\|\bDelta\|_2^2
    -32\kappa_2\bar R_4^2\frac{\log d_1}{M\lambda_{\tau_n}^2},
\end{align*}
where the last inequality uses $s_{\tau_n}\log d_1/M=o(1)$. Lemma~\ref{lemma:sol}, \eqref{bound:R_11}, and $\lambda_{\tau_n}\asymp\sqrt{\log d_1/N}$ now give
\begin{align}
    \|\bDelta\|_2
    &\leq \frac{8\lambda_{\tau_n}\sqrt{s_{\tau_n}}}{\kappa_1}
    +\sqrt{64\bar R_4^2\frac{\kappa_2\log d_1}{\kappa_1M\lambda_{\tau_n}^2}
    +\frac{8\bar R_4}{\kappa_1}}\nonumber\\
    &=O_p\left(\{\bemqr_{\mu}+\bemqr_{\tau_n}\}r_{\pi}+\bemqr_{\mu}r_q+r_{\mu}+r_{\tau_n}\right).\label{eq:r_taun_bar_rate}
\end{align}
Substituting this bound and \eqref{bound:R_11} into \eqref{eq:r_varphicone} further yields
\begin{align*}
    \|\bDelta\|_1=O_p\left(\sqrt{\frac{N}{\log d_1}}
    \{(\bemqr_{\mu}+\bemqr_{\tau_n})r_{\pi}+\bemqr_{\mu}r_q+r_{\mu}+r_{\tau_n}\}^2\right).
\end{align*}
By Lemma~\ref{lem:r_approx_bound}(c)--(d), we have $\bemqr_\mu=O(\emqr_\mu)$ and $\bemqr_{\tau_n}=O(\emqr_{\tau_n})$, which yields the stated bound $\|\hbettaur-\bettaur\|_2=O_p(\{\emqr_{\mu}+\emqr_{\tau_n}\}r_{\pi}+\emqr_{\mu}r_{q}+r_{\mu}+r_{\tau_n})$.

\emph{(e)}
Choosing some $\lambda_{\pi_a}, \lambda_{\pi_b}$ and $\lambda_{\tau_n}\asymp\sqrt{\log d_1/N}$ and $\lambda_{q}$, $\lambda_{\mu}\asymp\sqrt{\log d/N}$, we obtain from parts~(a)--(b) together with \eqref{eq:r_mu_bar_rate} and \eqref{eq:r_taun_bar_rate} that
\begin{align*}
    \|\hbepia-\bepiar\|_2=O_p(r_{\pi}), &\quad
    \|\hbeq-\beqr\|_2=O_p(r_{\pi}+r_q),\\
    \|\hbemu-\bemur\|_2=O_p(\bemqr_{\mu}\{r_\pi+r_q\}+r_\mu), &\quad
    \|\hbettau-\bettaur\|_2=O_p(\{\bemqr_{\mu}+\bemqr_{\tau_n}\}r_\pi+\bemqr_\mu r_q+r_\mu+r_{\tau_n}).
\end{align*}

In the remaining proof of (e), let $\bDelta=\hbetaur-\betaur$. Define $\omega_{\tau,i}=\tau(\bX_i)-\tau^*(\bX_i)$, $\omega_{q,i}=q(\bS_i)-q^*(\bS_i)$ and $\varepsilon_i:=Y_i(1,\bM_i)-\bS_i^\top\bemur$. We consider the decomposition
\begin{align*}
    \bnabla_{\betau}\bar\ell_6(\hat{\bnu}_6,\betaur)^\top\bDelta=\bnabla_{\betau}\bar\ell_6(\bnu_6^*, \betaur)^\top\bDelta+R_6+R_7,
\end{align*}
where
\begin{align*}
    R_6&:= \frac2M\sutaur (Q_{5,i}+Q_{6,i}-Q_{7,i}) \bX_i^\top\bDelta,\\
    R_7&:= \frac2M\sutaur (W_{\tau,1,i}+W_{\tau,2,i}-W_{\tau,3,i})^\top\bDelta
\end{align*}
with
\begin{align*}
    Q_{5,i}:=& A_i\{\exp(-\bX_i^\top\hbepiar)-\exp(-\bX_i^\top\bepiar)\}(-\omega_{\tau,i}+\omega_{\tau_n,i}+\omega_{q,i}\omega_{\mu,i}),\\
    Q_{6,i}:=& A_i\exp(-\bX_i^\top\hbepiar)\{\exp(\bS_i^\top\hbeq)\bS_i^\top(\hbemu-\bemur)+\bX_i^\top(\bettaur-\hbettau)\},\\
    Q_{7,i}:=& A_i\exp(-\bX_i^\top\hbepiar)\{\exp(\bS_i^\top\hbeq)-\exp(\bS_i^\top\beqr)\}\omega_{\mu,i},
\end{align*}
and
\begin{align*}
    W_{\tau,1,i}&:= A_i\{\exp(-\bX_i^\top\hbepiar)-\exp(-\bX_i^\top\bepiar)\}[\tau(\bX_i)-\tau_n(\bX_i)-q(\bS_i)\varepsilon_i]\bX_i,\\
    W_{\tau,2,i}&:= A_i\{\exp(-\bX_i^\top\hbepiar)-\exp(-\bX_i^\top\bepiar)\}\omega_{q,i}\{Y_i-\mu(\bS_i)\}\bX_i,\\
    W_{\tau,3,i}&:= A_i\exp(-\bX_i^\top\hbepiar)\{\exp(\bS_i^\top\hbeq)-\exp(\bS_i^\top\beqr)\}\{Y_i-\mu(\bS_i)\}\bX_i.
\end{align*}
It follows from Young's inequality $2ab\leq a^2/2+2b^2$ that
\begin{equation}
    |R_6|\leq \frac{1}{2}\delta\bar\ell_6(\hat{\bnu}_6, \betaur,\bDelta)+2\bar{R}_{6},\label{rela:R6R7}
\end{equation}
where
\begin{align*}
    \bar{R}_{6}&=M^{-1}\sutaur \exp(\bX_i^\top\hbepiar)\big[Q_{5,i}+Q_{6,i}-Q_{7,i}\big]^2\\
    &\leq 3M^{-1}\sutaur \exp(\bX_i^\top\hbepiar)\big[Q_{5,i}^2+Q_{6,i}^2+Q_{7,i}^2\big].
\end{align*}
Let $d_0=d_1$, $\balpha=\bettaur$ and $\halpha=\hbettau$. By Lemma~\ref{lem:usebound} and (d), we have
\begin{align}
    &\normp{\bX^\top(\hbettaur-\bettaur)}{4}=O_p(\|\hbettaur-\bettaur\|_2)=O_p\left(\{\bemqr_{\mu}+\bemqr_{\tau_n}\}r_{\pi}+\bemqr_{\mu}r_{q}+r_{\mu}+r_{\tau_n}\right).\label{order:r_tau_n_imputation}
\end{align}
By Lemma~\ref{lem:r_useevent}, on the event $\mathcal{E}_1\cap\mathcal{E}_3$ defined in \eqref{def:E1} and \eqref{def:E3}, there exist constants $C_1,C_3>0$ such that $\|\exp(\pm\bX^\top\hbepia)\|_{\P,r'}\leq C_1$ and $\|\exp(\bS^\top\hbeq)\|_{\P,r'}\leq C_3$ for every integer $r'\in[1,12]$.
Let $\Omega:=-\omega_\tau(\bX)+\omega_{\tau_n}(\bX)+\omega_q(\bS)\omega_\mu(\bS)$.
By H\"older's and Minkowski's inequalities, $\|\Omega\|_{\P,3}\leq \bemqr_\tau+\bemqr_{\tau_n}+\bemqr_q\bemqr_\mu$.
For each $i\in\mathcal I_\tau$, applying H\"older's inequality again and using \eqref{order:r_betara_imputation}, \eqref{order:r_mu_imputation}, \eqref{order:r_tau_n_imputation}, and \eqref{order:r_q_imputation}, we have
\begin{align*}
    \E_{\mathbb{S}_\tau}[\exp(\bX_i^\top\hbepiar)Q_{5,i}^2]
    &\leq\|\exp(\bX^\top\hbepiar)\|_{\P,6}
    \|\exp(-\bX^\top\hbepiar)-\exp(-\bX^\top\bepiar)\|_{\P,12}^2\|\Omega\|_{\P,3}^2\\
    &=O_p\left(\{\bemqr_{\tau}^2+\bemqr_{\tau_n}^2+\bemqr_q^2\bemqr_\mu^2\}r_\pi^2\right),\\
    \E_{\mathbb{S}_\tau}[\exp(\bX_i^\top\hbepiar)Q_{6,i}^2]
    &\leq2\|\exp(-\bX^\top\hbepiar)\|_{\P,3}\|\exp(\bS^\top\hbeqr)\|_{\P,6}^2
    \|\bS^\top(\hbemu-\bemur)\|_{\P,6}^2\\
    &\quad+2\|\exp(-\bX^\top\hbepiar)\|_{\P,2}
    \|\bX^\top(\bettaur-\hbettau)\|_{\P,4}^2\\
    &=O_p\left(\{\bemqr_\mu^2+\bemqr_{\tau_n}^2\}r_\pi^2+\bemqr_\mu^2r_q^2+r_\mu^2+r_{\tau_n}^2\right),\\
    \E_{\mathbb{S}_\tau}[\exp(\bX_i^\top\hbepiar)Q_{7,i}^2]
    &\leq\|\exp(-\bX^\top\hbepiar)\|_{\P,3}
    \|\exp(\bS^\top\hbeq)-\exp(\bS^\top\beqr)\|_{\P,6}^2\|\omega_\mu\|_{\P,6}^2\\
    &=O_p\left(\bemqr_\mu^2r_\pi^2+\bemqr_\mu^2r_q^2\right).
\end{align*}
The H\"older exponents satisfy $1/6+2/12+2/3=1$, $1/3+2/6+2/6=1$, or $1/2+2/4=1$.
Thus only the estimated weight difference requires an $L_{12}$ norm; all approximation-error norms have order at most six.
Consequently,
\begin{align}
    \E_{\mathbb S_\tau}[\bar R_6]
    =O_p\left(\{\bemqr_\mu+\bemqr_{\tau_n}+\bemqr_\tau\}^2r_\pi^2
    +\bemqr_\mu^2r_q^2+r_\mu^2+r_{\tau_n}^2\right).\label{bound:R_13}
\end{align}

Define
\begin{align*}
    \mathcal B_8:=&\left\{\|\bnabla_{\betau}\bar\ell_6(\bnu_6^*,\betaur)+2M^{-1}\sutaur (W_{\tau,1,i}+W_{\tau,2,i}-W_{\tau,3,i})\|_\infty \leq\lambda_{\tau}/2\right\},\\
    \mathcal B_9:=&\left\{\delta\bar\ell_6(\hat{\bnu}_6,  \betaur,\bDelta)\geq\kappa_1\|\bDelta\|_2^2-\kappa_2\frac{\log d_1}{M}\|\bDelta\|_1^2,\;\;\forall\bDelta\in\R^{d_1}\right\}.
\end{align*}
We then prove that $\P_{\mathbb{S}_\pi\cup\mathbb{S}_q\cup\mathbb{S}_\mu\cup\mathbb{S}_{\tau_n}\cup\mathbb{S}_\tau}(\mathcal B_8\cap \mathcal B_9)=1-o(1)$.

Conditional on $\mathbb S_\pi\cup\mathbb S_q\cup\mathbb S_\mu\cup\mathbb S_{\tau_n}$, the identities
$\E[q(\bS)\varepsilon\mid A=1,\bX]=\tau(\bX)-\tau_n(\bX)$ and
$\E[Y-\mu(\bS)\mid A=1,\bS]=0$ give $\E_{\mathbb S_\tau}W_{\tau,j,i}=\bz_{d_1}$ for $j\in[3]$.
Jensen's inequality and $q\leq c_0^{-2}$ bound the $L_8$ norms of the residual factors by
$2c_0^{-2}\|\varepsilon\|_{\P,8}$ for $W_{\tau,1,i}$ and $2\|\varepsilon\|_{\P,8}$ for $W_{\tau,2,i}$ and $W_{\tau,3,i}$.
The $L_8$ norms of their weight factors are $O_p(r_\pi)$ for the first two terms and $O_p(r_\pi+r_q)$ for the third,
by Lemma~\ref{lem:usebound}, \eqref{order:r_betara_imputation}, \eqref{order:r_q_imputation}, and $|\omega_q|\leq2c_0^{-2}$.
Applying Lemma~\ref{lem:r_common_gradient}(a) and conditional Markov's inequality to each term, with design $\bX_i$, and using Assumption~\ref{cond:sparsity}, we obtain
$\|2M^{-1}\sutaur(W_{\tau,1,i}+W_{\tau,2,i}-W_{\tau,3,i})\|_\infty=o_p(\sqrt{\log d_1/N})$.
The oracle-gradient bound in part~(c) and Lemma~\ref{lem:r_rsc} then give
$\P_{\mathbb S_\pi\cup\mathbb S_q\cup\mathbb S_\mu\cup\mathbb S_{\tau_n}\cup\mathbb S_\tau}(\mathcal B_8\cap\mathcal B_9)=1-o(1)$.

Following the proof of Lemma~\ref{lem:r_imputation}(a), on the event $\mathcal B_8$ we have
\begin{equation}
    2\delta\bar\ell_6(\hat{\bnu}_6,\betaur,\bDelta)+\lambda_\tau\|\bDelta\|_1
    \leq4\lambda_\tau\|\bDelta_{S_{\betaur}}\|_1+2|R_6|.\label{eq:r_key4}
\end{equation}
Using \eqref{rela:R6R7}, $\|\bDelta_{S_{\betaur}}\|_1\leq\sqrt{s_\tau}\|\bDelta\|_2$, and the convexity of $\ell_6$, we obtain
\begin{align}
    \delta\bar\ell_6(\hat{\bnu}_6,\betaur,\bDelta)+\lambda_\tau\|\bDelta\|_1
    &\leq4\lambda_\tau\sqrt{s_\tau}\|\bDelta\|_2+4\bar R_6,\nonumber\\
    \|\bDelta\|_1
    &\leq4\sqrt{s_\tau}\|\bDelta\|_2+4\lambda_\tau^{-1}\bar R_6.\label{eq:r_taucone}
\end{align}
On $\mathcal E_1\cap\mathcal E_3\cap\mathcal B_8\cap\mathcal B_9$, the RSC inequality and \eqref{eq:r_taucone} imply, for sufficiently large $N$,
\begin{align*}
    4\lambda_\tau\sqrt{s_\tau}\|\bDelta\|_2+4\bar R_6
    &\geq\delta\bar\ell_6(\hat{\bnu}_6,\betaur,\bDelta)\\
    &\geq\frac{\kappa_1}{2}\|\bDelta\|_2^2
    -32\kappa_2\bar R_6^2\frac{\log d_1}{M\lambda_\tau^2}.
\end{align*}
Here the last inequality uses $s_\tau\log d_1/M=o(1)$. Lemma~\ref{lemma:sol}, \eqref{bound:R_13}, and $\lambda_\tau\asymp\sqrt{\log d_1/N}$ now give
\begin{align*}
    \|\bDelta\|_2
    &\leq\frac{8\lambda_\tau\sqrt{s_\tau}}{\kappa_1}
    +\sqrt{64\bar R_6^2\frac{\kappa_2\log d_1}{\kappa_1M\lambda_\tau^2}
    +\frac{8\bar R_6}{\kappa_1}}\\
    &=O_p\left(\{\bemqr_\mu+\bemqr_{\tau_n}+\bemqr_\tau\}r_\pi
    +\bemqr_\mu r_q+r_\mu+r_{\tau_n}+r_\tau\right).
\end{align*}
Substituting this bound and \eqref{bound:R_13} into \eqref{eq:r_taucone} further yields
\begin{align*}
    \|\bDelta\|_1=O_p\left(\sqrt{\frac{N}{\log d_1}}
    \{(\bemqr_\mu+\bemqr_{\tau_n}+\bemqr_\tau)r_\pi+\bemqr_\mu r_q+r_\mu+r_{\tau_n}+r_\tau\}^2\right).
\end{align*}
Finally, by Lemma~\ref{lem:r_approx_bound}(c)--(e) and $\bemqr_q=O(1)$ from \eqref{bound:nuisance_approx}, we have $\bemqr_\mu=O(\emqr_\mu)$, $\bemqr_{\tau_n}=O(\emqr_{\tau_n})$, and $\bemqr_\tau=O(\emqr_\tau+\emqr_\mu)$, which yields the stated bound $\|\hbetaur-\betaur\|_2=O_p(\{\emqr_{\mu}+\emqr_{\tau_n}+\emqr_{\tau}\}r_{\pi}+\emqr_{\mu}r_{q}+r_{\mu}+r_{\tau_n}+r_{\tau})$.
\end{proof}

\subsection{Proofs of the results in Section~\ref{apdx:mqr2_results}}\label{apdx:mqr2_proofs}
Here, the nuisance estimators and their $\ell_2$ stochastic errors in~\eqref{def:rmqr} correspond to MQR2 in Algorithm~\ref{alg:mqr2}.
We first introduce the notation and auxiliary results used in the proof of Theorem~\ref{thm:mqr2_r_nuisance}.
For each $k\in[\K]$, let $\mathbb{S}_{-k}=(\bW_j)_{j\in\mathcal I_{-k}}$ denote the subset of $\mathbb{S}$ corresponding to the index set $\mathcal I_{-k}$.
Since MQR2 obtains all nuisance estimates on the same fold $\mathcal I_{-k}$, we use the superscript ``$\mathrm{sf}$'' (same-fold) to distinguish the empirical losses on $\mathcal I_{-k}$ from $\bar\ell_1,\ldots,\bar\ell_6$ in Section~\ref{apdx:mqr_supp_proofs}. For any $\bepia$, $\bepib$, $\bettau$, $\betau\in\mathbb{R}^{d_1}$ and $\beq$, $\bemu\in\mathbb{R}^{d}$, we define
\begin{align*}
    \bar{\ell}^{\mathrm{sf}}_1(\bepia):=|\mathcal{I}_{-k}|^{-1}\suk \ell_1(\bW_i,\bepia),\qquad&
    \bar{\ell}^{\mathrm{sf}}_2(\bepib):=|\mathcal{I}_{-k}|^{-1}\suk\ell_2(\bW_i,\bepib),\\
    \bar{\ell}^{\mathrm{sf}}_3(\bnu_3,\beq):=|\mathcal{I}_{-k}|^{-1}\suk \ell_3(\bW_i,\bnu_3,\beq),\qquad&
    \bar{\ell}^{\mathrm{sf}}_4(\bnu_4,\bemu):=|\mathcal{I}_{-k}|^{-1}\suk \ell_4(\bW_i,\bnu_4,\bemu),\\
    \bar{\ell}^{\mathrm{sf}}_5(\bnu_5,\bettau):=|\mathcal{I}_{-k}|^{-1}\suk \ell_5(\bW_i,\bnu_5,\bettau),\qquad&
    \bar{\ell}^{\mathrm{sf}}_6(\bnu_6,\betau):=|\mathcal{I}_{-k}|^{-1}\suk \ell_6(\bW_i,\bnu_6,\betau),
\end{align*}
where the loss functions are defined in \eqref{eq:r_loss1}--\eqref{eq:r_loss6}, respectively. Then for any $\bDelta\in\R^{d_1}$, their corresponding Taylor-series errors are defined as
\begin{align*}
\delta\bar{\ell}^{\mathrm{sf}}_1(\bepia,\bDelta):=&\bar{\ell}^{\mathrm{sf}}_1(\bepia+\bDelta)-\bar{\ell}^{\mathrm{sf}}_1(\bepia)-\bnabla_{\bepia}\bar{\ell}^{\mathrm{sf}}_1(\bepia)^\top\bDelta, \\
\delta\bar{\ell}^{\mathrm{sf}}_2(\bepib,\bDelta):=&\bar{\ell}^{\mathrm{sf}}_2(\bepib+\bDelta)-\bar{\ell}^{\mathrm{sf}}_2(\bepib)-\bnabla_{\bepib}\bar{\ell}^{\mathrm{sf}}_2(\bepib)^\top\bDelta, \\
\delta\bar{\ell}^{\mathrm{sf}}_5(\bnu_5,\bettau,\bDelta):=&\bar{\ell}^{\mathrm{sf}}_5(\bnu_5,\bettau+\bDelta)-\bar{\ell}^{\mathrm{sf}}_5(\bnu_5,\bettau)-\bnabla_{\bettau}\bar{\ell}^{\mathrm{sf}}_5(\bnu_5,\bettau)^\top\bDelta, \\
\delta\bar{\ell}^{\mathrm{sf}}_6(\bnu_6,\betau,\bDelta):=&\bar{\ell}^{\mathrm{sf}}_6(\bnu_6, \betau+\bDelta)-\bar{\ell}^{\mathrm{sf}}_6(\bnu_6,\betau)-\bnabla_{\betau}\bar{\ell}^{\mathrm{sf}}_6(\bnu_6, \betau)^\top\bDelta.
\end{align*}
Similarly, for any $\bDelta\in\mathbb{R}^{d}$, we define
\begin{align*}
\delta\bar{\ell}^{\mathrm{sf}}_3(\bnu_3,\beq,\bDelta)&:=\bar{\ell}^{\mathrm{sf}}_3(\bnu_3,\beq+\bDelta)-\bar{\ell}^{\mathrm{sf}}_3(\bnu_3,\beq)-\bnabla_{\beq}\bar{\ell}^{\mathrm{sf}}_3(\bnu_3,\beq)^\top\bDelta,\\
\delta\bar{\ell}^{\mathrm{sf}}_4(\bnu_4,\bemu,\bDelta)&:=\bar{\ell}^{\mathrm{sf}}_4(\bnu_4,\bemu+\bDelta)-\bar{\ell}^{\mathrm{sf}}_4(\bnu_4,\bemu)-\bnabla_{\bemu}\bar{\ell}^{\mathrm{sf}}_4(\bnu_4,\bemu)^\top\bDelta.
\end{align*}

The next lemma establishes restricted strong convexity (RSC) for the empirical losses on $\mathcal I_{-k}$ used by the MQR2 method. It is the counterpart of Lemma~\ref{lem:r_rsc}, which establishes the corresponding RSC conditions for the MQR1 method in Section~\ref{apdx:mqr_supp_proofs}.

\begin{lemma}[Restricted strong convexity conditions for the MQR2 method]\label{lem:mqr2_rsc}
Let Assumptions~\ref{cond:basic} and~\ref{cond:basic2} hold, and suppose that $N=O(d_1)$.
Define $f_{-k, d_1}(\bDelta):=\kappa_1\|\bDelta\|_2^2-\kappa_2\|\bDelta\|_1^2\log d_1/|\mathcal I_{-k}|$, for any $\bDelta\in \R^{d_1}$,  and $f_{-k,d}(\bDelta):=\kappa_1\|\bDelta\| _2^2-\kappa _2\|\bDelta\|_1^2\log d/|\mathcal I_{-k}|$, for any $\bDelta\in \R^d$. Then, with some constants $\kappa_{1}$, $\kappa_{2}$, $c_{1}$, $c_{2}>0$, we have
\begin{align}
    \P_{\mathbb{S}}\left(\delta\bar{\ell}^{\mathrm{sf}}_1(\bepiar,\bDelta)\geq f_{-k,d_1}(\bDelta),\:\forall\|\bDelta\|_2\leq1\right)&\geq1-c_1\exp(-c_2|\mathcal I_{-k}|),\label{eq:mqr2_RSC1}\\
    \P_{\mathbb{S}}\left(\delta\bar{\ell}^{\mathrm{sf}}_2(\bepibr,\bDelta)\geq f_{-k,d_1}(\bDelta),\:\forall\|\bDelta\|_2\leq1\right)&\geq1-c_1\exp(-c_2|\mathcal I_{-k}|),\label{eq:mqr2_RSC2}\\
    \P_{\mathbb{S}}\left(\delta\bar{\ell}^{\mathrm{sf}}_3(\hat{\bnu}_3,  \beqr,\bDelta)\geq f_{-k,d}(\bDelta),\:\forall\|\bDelta\|_{2}\leq1\right)&\geq1-c_{1}\exp(-c_{2}|\mathcal I_{-k}|).\label{eq:mqr2_RSC3}
\end{align}
Furthermore, suppose that $\|\hbepib-\bepibr\|_1\leq C_{\pi_b}/\sqrt{\log d_1}$ for some constant $C_{\pi_b}>0$.
Then
\begin{equation}
    \P_{\mathbb{S}}\left(\delta\bar{\ell}^{\mathrm{sf}}_5(\hat{\bnu}_5, \bettaur,\bDelta)\geq f_{-k,d_1}(\bDelta),\:\forall\bDelta\in\R^{d_1}\right)\geq1-c_{1}\exp(-c_{2}|\mathcal I_{-k}|)-\frac{2}{d_1|\mathcal I_{-k}|}.\label{eq:mqr2_RSC5}
\end{equation}
Moreover, suppose that there exist constants $C_{\pi_a},C_q>0$ such that $\|\hbepia-\bepiar\|_1\leq C_{\pi_a}/\sqrt{\log d_1}$ and $\|\hbeq-\beqr\|_1\leq C_q/\sqrt{\log d}$. Then
\begin{align}
    \P_{\mathbb{S}}\left(\delta\bar{\ell}^{\mathrm{sf}}_4(\hat{\bnu}_4, \bemur,\bDelta)\geq f_{-k,d}(\bDelta),\:\forall\bDelta\in\R^d\right)&\geq1-c_1\exp(-c_2|\mathcal I_{-k}|)-\frac{2}{d|\mathcal I_{-k}|},\label{eq:mqr2_RSC4}\\
    \P_{\mathbb{S}}\left(\delta\bar{\ell}^{\mathrm{sf}}_6(\hat{\bnu}_6, \betaur,\bDelta)\geq f_{-k,d_1}(\bDelta),\:\forall\bDelta\in\R^{d_1}\right)&\geq1-c_1\exp(-c_2|\mathcal I_{-k}|)-\frac{2}{d_1|\mathcal I_{-k}|}.\label{eq:mqr2_RSC6}
\end{align}
\end{lemma}

\begin{proof}
For any $\bDelta\in\R^{d_1}$, by Taylor's theorem there exist $v_1,v_1'\in(0,1)$ (possibly depending on $\bDelta$ and the sample) such that
\begin{align*}
    \delta\bar{\ell}^{\mathrm{sf}}_1(\bepiar,\bDelta)&=\frac1{2|\mathcal{I}_{-k}|}\suk A_i\exp(-\bX_i^\top\{\bepiar+v_1\bDelta\})(\bX_i^\top\bDelta)^2,\\
    \delta\bar{\ell}^{\mathrm{sf}}_2(\bepibr,\bDelta)&=\frac1{2|\mathcal{I}_{-k}|}\suk (1-A_i)\exp(\bX_i^\top\{\bepibr+v_1'\bDelta\})(\bX_i^\top\bDelta)^2,\\
    \delta\bar{\ell}^{\mathrm{sf}}_5(\hat{\bnu}_5, \bettaur,\bDelta)&=\frac1{|\mathcal{I}_{-k}|}\suk (1-A_i)\exp(\bX_i^\top\hbepibr)(\bX_i^\top\bDelta)^2,\\
    \delta\bar{\ell}^{\mathrm{sf}}_6(\hat{\bnu}_6, \betaur,\bDelta)&=\frac1{|\mathcal{I}_{-k}|}\suk A_i\exp(-\bX_i^\top\hbepiar)(\bX_i^\top\bDelta)^2.
\end{align*}
Similarly, for any $\bDelta\in\R^{d}$, there exists $v_2\in(0,1)$ (possibly depending on $\bDelta$ and the sample) such that
\begin{align*}
    \delta\bar{\ell}^{\mathrm{sf}}_3(\hat{\bnu}_3,  \beqr,\bDelta)&=\frac1{2|\mathcal{I}_{-k}|}\suk \frac{A_i}{g(\bX_i^\top\hbepiar)}\exp(\bS_i^\top\{\beqr+v_2\bDelta\})(\bS_i^\top\bDelta)^2,\\
    \delta\bar{\ell}^{\mathrm{sf}}_4(\hat{\bnu}_4, \bemur,\bDelta)&=\frac1{|\mathcal{I}_{-k}|}\suk \frac{A_i}{g(\bX_i^\top\hbepiar)}\exp(\bS_i^\top\hbeqr)(\bS_i^\top\bDelta)^2.
\end{align*}

The proofs of \eqref{eq:mqr2_RSC1}--\eqref{eq:mqr2_RSC2} follow Parts~1--2 of
Lemma~\ref{lem:r_rsc}, with $M$ replaced by $|\I_{-k}|$. For
\eqref{eq:mqr2_RSC3}--\eqref{eq:mqr2_RSC6}, we bound the plug-in weights directly
because the preceding nuisance estimators are fitted on $\I_{-k}$.

\textbf{Part 3.}
Since $g^{-1}(u)=1+\exp(-u)>1$,
\begin{align}
    \delta\bar{\ell}^{\mathrm{sf}}_3(\hat{\bnu}_3,  \beqr,\bDelta)&\geq\frac1{2|\mathcal{I}_{-k}|}\suk A_i\exp(\bS_i^\top\{\beqr+v_2\bDelta\})(\bS_i^\top\bDelta)^2.\label{eq:MQR2_rsc3_1}
\end{align}
Let $\bU=A\bS$, $\mathbb S'=(A_{i}\bS_i)_{i\in\I_{-k}}$, $h(u)=\exp(u)$,
and $\bs{\eta}=\beqr$. Under Assumption~\ref{cond:basic2},
$\E|\bU^\top\bs{\eta}|\leq \E|\bS^\top\beqr|<C$ for some constant $C>0$.
By Lemmas~\ref{lem:momentmatrix} and~\ref{lem:useinrsc} (uniformly over $v\in[0,1]$), we obtain
\begin{align*}
    \frac1{2|\mathcal I_{-k}|}\suk A_i\exp(\bS_i^\top\{\beqr+v_2\bDelta\})(\bS_i^\top\bDelta)^2 \geq \kappa_1'\|\bDelta\|_2^2-\kappa_2'\|\bDelta\|_1^2\frac{\log d}{|\mathcal I_{-k}|}\text{, }\;\;\;\forall \|\bDelta\|_2\leq1,
\end{align*}
with probability at least $1-c_1'\exp(-c_2'|\mathcal I_{-k}|)$ for some
$\kappa_1',\kappa_2',c_1',c_2'>0$. Together with
\eqref{eq:MQR2_rsc3_1}, this proves \eqref{eq:mqr2_RSC3}.

\textbf{Part 4.}
Since $|\I_{-k}|<N$, $N=O(d_1)$, and $d_1\leq d$, there exists a constant $c_3>0$ such that $\log(|\I_{-k}|d_1)\leq c_3^2\log d_1$ and $\log(|\I_{-k}|d)\leq c_3^2\log d$ (for $d_1\geq2$). Define
\begin{align*}
    \mathcal{E}_{q,-k}:=\left\{\max_{i\in \I_{-k}}|\bS_i^\top(\hbeqr-\beqr)|\leq 4c_3\sigma_{\bS}C_{q}\right\}\cap \mathcal{E}^*.
\end{align*}
Since $\P_{\mathbb S}(\mathcal E^*)=1$ and
$\|\hbeq-\beqr\|_1\leq C_q/\sqrt{\log d}$, Lemma~\ref{lem:psi2norm} and
H\"older's inequality yield
$\P_{\mathbb S}(\mathcal E_{q,-k})\geq1-2/(d|\I_{-k}|)$. On this event,
\begin{align*}
    \exp(\bS_i^\top\hbeq)
    &\geq c_0^2\exp(-4c_3\sigma_{\bS}C_q)=:c_\mu,
    \qquad i\in\I_{-k}.
\end{align*}
Since $g^{-1}(u)>1$,
\begin{align*}
    \delta\bar{\ell}^{\mathrm{sf}}_4(\hat{\bnu}_4, \bemur,\bDelta)
    &\geq c_\mu\frac1{|\mathcal I_{-k}|}\suk A_i(\bS_i^\top\bDelta)^2.
\end{align*}
Let $\bU=A\bS$, $\mathbb S'=(A_i\bS_i)_{i\in\I_{-k}}$, $h(u)\equiv1$, $v=0$, and $\bs{\eta}=\bz_{d}$. By Lemmas~\ref{lem:momentmatrix} and \ref{lem:useinrsc}, the event
\begin{align*}
    \frac1{|\mathcal I_{-k}|}\suk A_i(\bS_i^\top\bDelta)^2
    =\frac1{|\mathcal I_{-k}|}\suk (\bU_i^\top\bDelta)^2
    \geq \kappa_1'\|\bDelta\|_2^2-\kappa_2'\|\bDelta\|_1^2\frac{\log d}{|\mathcal I_{-k}|},\quad \forall\|\bDelta\|_2\leq1
\end{align*}
holds with probability at least $1-c_1'\exp(-c_2'|\mathcal I_{-k}|)$.
On the intersection of these two events, for all $\|\bDelta\|_2\leq1$,
\begin{align*}
    \delta\bar{\ell}^{\mathrm{sf}}_4(\hat{\bnu}_4, \bemur,\bDelta)
    &\geq c_\mu\kappa_1'\|\bDelta\|_2^2
    -c_\mu\kappa_2'\|\bDelta\|_1^2\frac{\log d}{|\mathcal I_{-k}|}
    =\kappa_1\|\bDelta\|_2^2-\kappa_2\|\bDelta\|_1^2\frac{\log d}{|\mathcal I_{-k}|},
\end{align*}
where $\kappa_1:=c_\mu\kappa_1'$ and $\kappa_2:=c_\mu\kappa_2'$. Since
$\delta\bar{\ell}^{\mathrm{sf}}_4$ is quadratic in $\bDelta$, applying the
preceding inequality to $\bDelta/\|\bDelta\|_2$ extends it to
$\|\bDelta\|_2>1$. Moreover, the intersection has probability at least
$1-c_1'\exp(-c_2'|\mathcal I_{-k}|)-2/(d|\I_{-k}|)$. This proves
\eqref{eq:mqr2_RSC4}.

\textbf{Part 5.}
Define
\begin{align*}
    \mathcal E_{\tau_n,-k}:=
    \left\{\max_{i\in\I_{-k}}|\bX_i^\top(\hbepibr-\bepibr)|
    \leq4c_3\sigma_{\bS}C_{\pi_b}\right\}\cap\mathcal E^*.
\end{align*}
As in Part~4, $\P_{\mathbb S}(\mathcal E_{\tau_n,-k})\geq
1-2/(d_1|\I_{-k}|)$. On this event,
\begin{align*}
    \delta\bar{\ell}^{\mathrm{sf}}_5(\hat{\bnu}_5, \bettaur,\bDelta)
    &\geq c_{\tau_n}\frac1{|\mathcal I_{-k}|}\suk(1-A_i)(\bX_i^\top\bDelta)^2,
\end{align*}
where
$c_{\tau_n}:=c_0(1-c_0)^{-1}\exp(-4c_3\sigma_{\bS}C_{\pi_b})$.
With $\bU=(1-A)\bX$, Lemmas~\ref{lem:momentmatrix} and
\ref{lem:useinrsc} give, with probability at least
$1-c_1'\exp(-c_2'|\mathcal I_{-k}|)$,
\begin{align*}
    \frac1{|\mathcal I_{-k}|}\suk(1-A_i)(\bX_i^\top\bDelta)^2
    &\geq\kappa_1'\|\bDelta\|_2^2-\kappa_2'\|\bDelta\|_1^2
    \frac{\log d_1}{|\mathcal I_{-k}|},\qquad \forall\|\bDelta\|_2\leq1.
\end{align*}
Therefore, on the intersection of these events,
\begin{align*}
    \delta\bar{\ell}^{\mathrm{sf}}_5(\hat{\bnu}_5, \bettaur,\bDelta)
    &\geq c_{\tau_n}\kappa_1'\|\bDelta\|_2^2
    -c_{\tau_n}\kappa_2'\|\bDelta\|_1^2\frac{\log d_1}{|\mathcal I_{-k}|}\\
    &=\kappa_1\|\bDelta\|_2^2-\kappa_2\|\bDelta\|_1^2\frac{\log d_1}{|\mathcal I_{-k}|},
\end{align*}
where $\kappa_1:=c_{\tau_n}\kappa_1'$ and
$\kappa_2:=c_{\tau_n}\kappa_2'$. As in Part~4, the quadratic form of
$\delta\bar{\ell}^{\mathrm{sf}}_5$ extends the preceding inequality to
$\|\bDelta\|_2>1$. The intersection occurs with probability at least
$1-c_1'\exp(-c_2'|\mathcal I_{-k}|)-2/(d_1|\I_{-k}|)$, and hence
\eqref{eq:mqr2_RSC5} follows.

\textbf{Part 6.}
Define
\begin{align*}
    \mathcal E_{\pi_a,-k}:=
    \left\{\max_{i\in\I_{-k}}|\bX_i^\top(\hbepiar-\bepiar)|
    \leq4c_3\sigma_{\bS}C_{\pi_a}\right\}\cap\mathcal E^*.
\end{align*}
As in Part~5, $\P_{\mathbb S}(\mathcal E_{\pi_a,-k})\geq
1-2/(d_1|\I_{-k}|)$. On this event,
\begin{align*}
    \delta\bar{\ell}^{\mathrm{sf}}_6(\hat{\bnu}_6, \betaur,\bDelta)
    &\geq c_{\pi_a}\frac1{|\mathcal I_{-k}|}\suk A_i(\bX_i^\top\bDelta)^2,
\end{align*}
where
$c_{\pi_a}:=c_0(1-c_0)^{-1}\exp(-4c_3\sigma_{\bS}C_{\pi_a})$.
With $\bU=A\bX$, Lemmas~\ref{lem:momentmatrix} and \ref{lem:useinrsc} give,
with probability at least $1-c_1'\exp(-c_2'|\mathcal I_{-k}|)$,
\begin{align*}
    \frac1{|\mathcal I_{-k}|}\suk A_i(\bX_i^\top\bDelta)^2
    &\geq\kappa_1'\|\bDelta\|_2^2-
    \kappa_2'\|\bDelta\|_1^2\frac{\log d_1}{|\mathcal I_{-k}|},
    \qquad \forall\|\bDelta\|_2\leq1.
\end{align*}
Combining the preceding two displays and applying the same scaling argument as
in Part~5 gives the required lower bound for all $\bDelta\in\R^{d_1}$. Since
$\P_{\mathbb S}(\mathcal E_{\pi_a,-k})\geq1-2/(d_1|\I_{-k}|)$,
\eqref{eq:mqr2_RSC6} follows.
\end{proof}

\begin{proof}[Proof of Theorem~\ref{thm:mqr2_r_nuisance}]
The proof is similar to that of Theorem~\ref{thm:r_nuisance}, so we only sketch the key steps.

\emph{(a)}
The losses $\bar{\ell}^{\mathrm{sf}}_1$ and $\bar{\ell}^{\mathrm{sf}}_2$ do not depend on preceding nuisance estimators and coincide with the corresponding losses for the MQR1 estimator, except that the empirical averages are evaluated over $\mathcal{I}_{-k}$. Since $|\mathcal{I}_{-k}| \asymp N$, the arguments in the proof of Theorem~\ref{thm:r_nuisance} carry over directly via Lemmas~\ref{lem:mqr2_rsc} and \ref{lem:r_gradient}. Therefore, taking $\lambda_{\pi_a}=4\sigma_{\pi_a}\sqrt{\log d_1/|\mathcal I_{-k}|}$ and $\lambda_{\pi_b}=4\sigma_{\pi_b}\sqrt{\log d_1/|\mathcal I_{-k}|}$, with $t=\log d_1$ in Lemma~\ref{lem:r_gradient}, gives
\begin{align}
    \|\hbepiar-\bepiar\|_2&=O_p\left(\sqrt{\frac{s_{\pi_a}\log d_1}{N}}\right),&
    \|\hbepiar-\bepiar\|_1&=O_p\left(s_{\pi_a}\sqrt{\frac{\log d_1}{N}}\right),\label{rate:mqr2_pi_a}\\
    \|\hbepibr-\bepibr\|_2&=O_p\left(\sqrt{\frac{s_{\pi_b}\log d_1}{N}}\right),&
    \|\hbepibr-\bepibr\|_1&=O_p\left(s_{\pi_b}\sqrt{\frac{\log d_1}{N}}\right).\label{rate:mqr2_pi_b}
\end{align}
Hence $\rmqr_\pi=O_p(r_\pi)=o_p(1)$, where $r_\pi=\sqrt{s_\pi \log d_1/N}=o(1)$.

\emph{(b)}
We condition on the probability-one event $\mathcal E^*$ defined in
\eqref{def:E*}. Thus, for every $i\in\mathcal I_{-k}$,
\begin{align*}
    \exp(-\bX_i^\top\bepiar)\leq c_0^{-1},\qquad
    \exp(\bX_i^\top\bepibr)\leq c_0^{-1},\qquad
    \exp(\bS_i^\top\beqr)\leq c_0^{-2}.
\end{align*}
Define the events
\begin{align*}
    \mathcal B_1' &:= \left\{\|\bnabla_{\beq}\bar{\ell}^{\mathrm{sf}}_3(\bnu_3^*,\beqr)\|_\infty \leq \lambda_q/2 \right\},\\
    \mathcal B_2' &:= \left\{|R_1'(\bDelta)| \leq c_t r_{\pi} \left(\|\bDelta\|_1\sqrt{\frac{\log d}{N}} + \|\bDelta\|_2\right),\; \forall \bDelta \in \R^d \right\},\\
    \mathcal B_3' &:= \left\{\delta\bar{\ell}^{\mathrm{sf}}_3(\hat{\bnu}_3, \beqr, \bDelta) \geq f_{-k,d}(\bDelta),\; \forall \|\bDelta\|_2 \leq 1 \right\},
\end{align*}
where
\begin{align*}
    R_1'(\bDelta)
    &:=
    \{\bnabla_{\beq}\bar{\ell}^{\mathrm{sf}}_3(\hat{\bnu}_3,\beqr)
    -\bnabla_{\beq}\bar{\ell}^{\mathrm{sf}}_3(\bnu_3^*,\beqr)\}^{\top}\bDelta
    =Q_1'(\bDelta)-Q_2'(\bDelta),\\
    Q_1'(\bDelta)
    &:=
    \frac{1}{|\I_{-k}|}\sum_{i\in\mathcal I_{-k}}A_i
    \{g^{-1}(\bX_i^\top\hbepiar)-g^{-1}(\bX_i^\top\bepiar)\}
    \exp(\bS_i^\top\beqr)\bS_i^\top\bDelta,\\
    Q_2'(\bDelta)
    &:=
    \frac{1}{|\I_{-k}|}\sum_{i\in\mathcal I_{-k}}(1-A_i)
    \left[\{1-g(\bX_i^\top\hbepibr)\}^{-1}
    -\{1-g(\bX_i^\top\bepibr)\}^{-1}\right]\bS_i^\top\bDelta.
\end{align*}
Fix $\lambda_q=4\sigma_q\sqrt{\log d/|\mathcal I_{-k}|}$. Lemma~\ref{lem:r_gradient} with $t=\log d$ and Lemma~\ref{lem:mqr2_rsc}, applied to the average over $\I_{-k}$, give
$\P_{\mathbb S}(\mathcal B_1')\geq1-2/d$ and
$\P_{\mathbb S}(\mathcal B_3')=1-o(1)$, respectively.
For each fixed $t>0$, we next show that, for a sufficiently large constant $c_t>0$, $\P_{\mathbb S}(\mathcal B_2')\geq1-1/t$ for all sufficiently large $N$. Since $\bX$ is a subvector of the sub-Gaussian vector $\bS$, by Lemma~\ref{lem:psi2norm}, a union bound over $i\in\mathcal I_{-k}$ and the $d_1$ coordinates gives
\begin{align*}
    \max_{i\in\mathcal I_{-k}}\|\bX_i\|_\infty
    =O_p\left[\{\log(|\mathcal I_{-k}|d_1)\}^{1/2}\right].
\end{align*}
Let
\begin{align*}
    R_{\pi_a,\infty}&:=\max_{i\in\mathcal I_{-k}}|\bX_i^\top(\hbepiar-\bepiar)|,&
    R_{\pi_b,\infty}&:=\max_{i\in\mathcal I_{-k}}|\bX_i^\top(\hbepibr-\bepibr)|.
\end{align*}
Together with \eqref{rate:mqr2_pi_a}--\eqref{rate:mqr2_pi_b},
$|\I_{-k}|\asymp N$, $N=O(d_1)$, and Assumption~\ref{cond:mqr2_sparsity}, the preceding bound gives
\begin{align}
    R_{\pi_a,\infty}
    &\leq \max_{i\in\mathcal I_{-k}}\|\bX_i\|_\infty
    \|\hbepiar-\bepiar\|_1
    =O_p\left(\frac{s_{\pi_a}\log d_1}{\sqrt N}\right)=o_p(1),\nonumber\\
    R_{\pi_b,\infty}
    &\leq \max_{i\in\mathcal I_{-k}}\|\bX_i\|_\infty
    \|\hbepibr-\bepibr\|_1
    =O_p\left(\frac{s_{\pi_b}\log d_1}{\sqrt N}\right)=o_p(1).
    \label{eq:mqr2_pi_max_prediction}
\end{align}

Hence $\exp(R_{\pi_a,\infty})=\exp(R_{\pi_b,\infty})=1+o_p(1)$.
This step uses $r_\pi=o(N^{-1/4})$, which is stronger than the condition in
Section~\ref{apdx:mqr_supp_proofs}.

We control $Q_1'(\bDelta)$ and $Q_2'(\bDelta)$ together. Since
$g^{-1}(u)=1+\exp(-u)$ and $\{1-g(u)\}^{-1}=1+\exp(u)$,
\begin{align*}
    g^{-1}(\bX_i^\top\hbepiar)-g^{-1}(\bX_i^\top\bepiar)
    &=\exp(-\bX_i^\top\bepiar)
    [\exp\{-\bX_i^\top(\hbepiar-\bepiar)\}-1],\\
    \{1-g(\bX_i^\top\hbepibr)\}^{-1}
    -\{1-g(\bX_i^\top\bepibr)\}^{-1}
    &=\exp(\bX_i^\top\bepibr)
    [\exp\{\bX_i^\top(\hbepibr-\bepibr)\}-1].
\end{align*}
Thus, conditional on $\mathcal E^*$, the inequality
$|\exp(u)-1|\leq |u|\exp(|u|)$ and Cauchy--Schwarz give
\begin{align*}
    |Q_1'(\bDelta)|
    &\leq c_0^{-3}\exp(R_{\pi_a,\infty})
    \left\{\frac{1}{|\I_{-k}|}\sum_{i\in\I_{-k}}
    [\bX_i^\top(\hbepiar-\bepiar)]^2\right\}^{1/2}
    \left\{\frac{1}{|\I_{-k}|}\sum_{i\in\I_{-k}}
    (\bS_i^\top\bDelta)^2\right\}^{1/2},\\
    |Q_2'(\bDelta)|
    &\leq c_0^{-1}\exp(R_{\pi_b,\infty})
    \left\{\frac{1}{|\I_{-k}|}\sum_{i\in\I_{-k}}
    [\bX_i^\top(\hbepibr-\bepibr)]^2\right\}^{1/2}
    \left\{\frac{1}{|\I_{-k}|}\sum_{i\in\I_{-k}}
    (\bS_i^\top\bDelta)^2\right\}^{1/2}.
\end{align*}
Although the preceding stochastic errors depend on the observations in
$\I_{-k}$, Lemma~\ref{lemma:beta1} applies to data-dependent arguments and yields
\begin{align*}
    \frac{1}{|\I_{-k}|}\sum_{i\in\I_{-k}}
    [\bX_i^\top(\hbepiar-\bepiar)]^2
    &=O_p\left(\frac{s_{\pi_a}\log d_1}{N}\right),\\
    \frac{1}{|\I_{-k}|}\sum_{i\in\I_{-k}}
    [\bX_i^\top(\hbepibr-\bepibr)]^2
    &=O_p\left(\frac{s_{\pi_b}\log d_1}{N}\right),\\
    \frac{1}{|\I_{-k}|}\sum_{i\in\I_{-k}}(\bS_i^\top\bDelta)^2
    &=O_p\left(\frac{\log d}{N}\|\bDelta\|_1^2+\|\bDelta\|_2^2\right)
\end{align*}
uniformly over $\bDelta\in\R^d$. Since $s_\pi=s_{\pi_a}+s_{\pi_b}$, these
bounds and the preceding estimates of $Q_1'$ and $Q_2'$ imply
\begin{align*}
    |R_1'(\bDelta)|
    =O_p\left[r_\pi\left\{\sqrt{\frac{\log d}{N}}\|\bDelta\|_1+\|\bDelta\|_2\right\}\right]
\end{align*}
uniformly over $\bDelta\in\R^d$.

Thus $\mathcal E^*\cap\mathcal B_1'\cap\mathcal B_2'\cap\mathcal B_3'$ has
probability at least $1-1/t-o(1)$. On this event, letting
$\bDelta:=\hbeqr-\beqr$, Lemma~\ref{lem:r_imputation} (b)--(c) and the
deterministic argument in the proof of Theorem~\ref{thm:r_nuisance} (b) give
$\|\bDelta\|_1\leq k_0\sqrt{\bar s_q}\|\bDelta\|_2$, $\|\bDelta\|_2\leq1$, and
\begin{align*}
    \frac{\kappa_1}{2}\|\bDelta\|_2^2
    \leq \{2\lambda_q\sqrt{s_q}+c_t r_{\pi}\}\|\bDelta\|_2.
\end{align*}
For each fixed $t>0$, the preceding inequality gives a constant $C_t>0$, depending only on $t$ and the fixed regularity constants, such that
\begin{align*}
\limsup_{N\to\infty}\P_{\mathbb S}\{\|\hbeqr-\beqr\|_2>C_t(r_\pi+r_q)\}\leq1/t.
\end{align*}
The same argument applies to the $\ell_1$ bound. Since $t>0$ is arbitrary, it follows that
\begin{align}
    \|\hbeqr-\beqr\|_2 &= O_p(r_{\pi}+r_q),\notag\\
    \|\hbeqr-\beqr\|_1 &= O_p\left(s_{\pi}\frac{\log d_1}{\sqrt{N\log d}} + s_q\sqrt{\frac{\log d}{N}}\right).
    \label{eq:mqr2_first_l1_rates}
\end{align}

\emph{(c)} Choose $\lambda_{\pi_a},\lambda_{\pi_b}\asymp
\sqrt{\log d_1/N}$ and $\lambda_q\asymp\sqrt{\log d/N}$ as in parts~(a) and~(b), and take $\lambda_\mu=C_\mu\sqrt{\log d/|\mathcal I_{-k}|}\asymp\sqrt{\log d/N}$ with $C_\mu>0$ as in the proof of Theorem~\ref{thm:r_nuisance}(c).
In the remainder of the proof of (c), let $\bDelta:=\hbemur-\bemur$ and define
\begin{align*}
    R_{2,3}'
    &:=\bnabla_{\bemu}\bar\ell_4^{\mathrm{sf}}(\hat{\bnu}_4, \bemur)^\top\bDelta-\bnabla_{\bemu}\bar\ell_4^{\mathrm{sf}}(\bnu_4^*,\bemur)^\top\bDelta=-\frac{2}{|\mathcal I_{-k}|}\sum_{i\in\mathcal I_{-k}}
    A_i(\widehat h_i-h_i^*)\varepsilon_i\bS_i^\top\bDelta,
\end{align*}
where $\varepsilon_i:=Y_i(1,\bM_i)-\bS_i^\top\bemur$ and
\begin{align*}
    \widehat h_i:=\frac{\exp(\bS_i^\top\hbeqr)}
    {g(\bX_i^\top\hbepiar)},\qquad
    h_i^*:=\frac{\exp(\bS_i^\top\beqr)}
    {g(\bX_i^\top\bepiar)}.
\end{align*}
This is the counterpart of $R_2+R_3$ in \eqref{decom:R2R3}; it is kept as one
term because $\hbepiar$ and $\hbeqr$ are estimated from the same observations
entering $\bar\ell^{\mathrm{sf}}_4$. Define
\begin{align*}
    \mathcal B_4':=&\{\|\bnabla_{\bemu}\bar\ell_4^{\mathrm{sf}}(\bnu_4^*,\bemur)\|_\infty \leq\lambda_{\mu}/2\},\\
    \mathcal B_5':=&\left\{\delta\bar\ell_4^{\mathrm{sf}}(\hat{\bnu}_4, \bemur,\bDelta)\geq\kappa_1\|\bDelta\|_2^2-\kappa_2\frac{\log d}{|\mathcal I_{-k}|}\|\bDelta\|_1^2,\;\;\forall\bDelta\in\R^d\right\}.
\end{align*}
Since the oracle gradient $\bnabla_{\bemu}\bar\ell_4^{\mathrm{sf}}(\bnu_4^*,\bemur)$ shares the same population target, weights, and residual as in the proof of Theorem~\ref{thm:r_nuisance}(c) while averaging over $\mathcal I_{-k}$, Lemma~\ref{lem:r_common_gradient}(b) gives $\P_{\mathbb S}(\mathcal B_4')=1-o(1)$. Moreover,
\eqref{rate:mqr2_pi_a}, \eqref{eq:mqr2_first_l1_rates}, and
Assumption~\ref{cond:mqr2_sparsity} imply
\begin{align*}
    \sqrt{\log d_1}\|\hbepiar-\bepiar\|_1
    +\sqrt{\log d}\|\hbeqr-\beqr\|_1
    =O_p\{\sqrt N(r_\pi^2+r_q^2)\}=o_p(1).
\end{align*}
Thus the condition of Lemma~\ref{lem:mqr2_rsc} holds with probability tending
to one, and hence $\P_{\mathbb S}(\mathcal B_5')=1-o(1)$. On
$\mathcal B_4'\cap\mathcal B_5'$, the deterministic argument in the proof of
Theorem~\ref{thm:r_nuisance} (c) gives
\begin{align*}
        2\delta\bar\ell_4^{\mathrm{sf}}(\hat{\bnu}_4, \bemur,\bDelta)+\lambda_{\mu}\|\bDelta\|_1&\leq 4\lambda_{\mu}\|\bDelta_{S_{\bemur}}\|_1+2|R_{2,3}'|.
\end{align*}
Similarly, by Young's inequality,
\begin{align*}
    |R_{2,3}'|
    \leq\frac12\delta\bar\ell^{\mathrm{sf}}_4
    (\hat{\bnu}_4,\bemur,\bDelta)+2\bar R_{2,3}',
\end{align*}
where
\begin{align*}
    \bar R_{2,3}'
    :=\frac{1}{|\mathcal I_{-k}|}\sum_{i\in\mathcal I_{-k}}
    \varepsilon_i^2\frac{(\widehat h_i-h_i^*)^2}{\widehat h_i}.
\end{align*}

Because the preceding estimators use $\mathcal I_{-k}$, we bound
$\bar R_{2,3}'$ directly. Let $r_{\mu,i}:=\log(\widehat h_i/h_i^*)$ and
$R_{\mu,\infty}:=\max_{i\in\mathcal I_{-k}}|r_{\mu,i}|$. Since
$(\log g)'(\cdot)=1-g(\cdot)\in(0,1)$,
\begin{align}
    |r_{\mu,i}|
    \leq |\bS_i^\top(\hbeq-\beqr)|
    +|\bX_i^\top(\hbepiar-\bepiar)|.
    \label{eq:mqr2_mu_log_ratio}
\end{align}
As in \eqref{eq:mqr2_pi_max_prediction}, a union bound,
Assumption~\ref{cond:mqr2_sparsity}, \eqref{rate:mqr2_pi_a}, and
\eqref{eq:mqr2_first_l1_rates} give
\begin{align*}
    R_{\mu,\infty}
    =O_p\{\sqrt N(r_\pi^2+r_q^2)\}=o_p(1).
\end{align*}
Hence $\exp(R_{\mu,\infty})=1+o_p(1)$.

The mean value theorem gives, for every $i\in\mathcal I_{-k}$,
\begin{align}
    \frac{(\widehat h_i-h_i^*)^2}{\widehat h_i}
    =h_i^*\{\exp(r_{\mu,i}/2)-\exp(-r_{\mu,i}/2)\}^2\leq h_i^*\exp(|r_{\mu,i}|)r_{\mu,i}^2.
    \label{eq:mqr2_weight_ratio_bound}
\end{align}
Conditional on $\mathcal E^*$, we have $\max_{i\in\I_{-k}} h_i^*\leq c_0^{-3}$ and therefore
\begin{align}
    \bar R_{2,3}'
    \leq 2c_0^{-3}\exp(R_{\mu,\infty})\bigg\{
    &\frac{1}{|\mathcal I_{-k}|}\sum_{i\in\mathcal I_{-k}}
    \varepsilon_i^2(\bS_i^\top(\hbeq-\beqr))^2+\frac{1}{|\mathcal I_{-k}|}\sum_{i\in\mathcal I_{-k}}
    \varepsilon_i^2[\bX_i^\top(\hbepiar-\bepiar)]^2\bigg\},
    \label{eq:mqr2_R23_quadratic}
\end{align}
by \eqref{eq:mqr2_mu_log_ratio} and $(a+b)^2\leq2a^2+2b^2$.

Under Assumption~\ref{cond:basic2}~(b), apply
Lemma~\ref{lem:mqr2_weighted_quadratic} with $m=|\mathcal I_{-k}|$ and
$U_i=\varepsilon_i$, first taking $\bs Z_i=\bS_i$ and
$\bDelta=\hbeq-\beqr$, and then $\bs Z_i=\bX_i$ and
$\bDelta=\hbepiar-\bepiar$. By \eqref{eq:mqr2_first_l1_rates},
\eqref{rate:mqr2_pi_a}, and Assumption~\ref{cond:mqr2_sparsity},
\begin{align*}
    \frac{1}{|\mathcal I_{-k}|}\sum_{i\in\mathcal I_{-k}}
    \varepsilon_i^2(\bS_i^\top(\hbeq-\beqr))^2
    &=O_p\{(r_\pi+r_q)^2\},\\
    \frac{1}{|\mathcal I_{-k}|}\sum_{i\in\mathcal I_{-k}}
    \varepsilon_i^2[\bX_i^\top(\hbepiar-\bepiar)]^2
    &=O_p(r_\pi^2).
\end{align*}
Together with \eqref{eq:mqr2_R23_quadratic}, this proves
\begin{align*}
    \bar R_{2,3}'=O_p\{(r_\pi+r_q)^2\}.
\end{align*}

Since $|\mathcal I_{-k}|\asymp N$, the deterministic argument following
\eqref{eq:r_key2}, with $M$ and $\bar R_2$ replaced by $|\I_{-k}|$ and
$\bar R_{2,3}'$, together with the preceding bound, gives
\begin{align}
    \|\hbemur-\bemur\|_2
    &=O_p(r_\pi+r_q+r_\mu),\nonumber\\
    \|\hbemur-\bemur\|_1
    &=O_p\left[\sqrt{\frac{N}{\log d}}
    \{r_\pi+r_q+r_\mu\}^2\right].
    \label{eq:mqr2_mu_rates}
\end{align}

\emph{(d)} Choose $\lambda_{\pi_a},\lambda_{\pi_b}\asymp\sqrt{\log d_1/N}$ and $\lambda_q,\lambda_\mu\asymp\sqrt{\log d/N}$ as in parts~(a)--(c), and take $\lambda_{\tau_n}=C_{\tau_n}\sqrt{\log d_1/|\mathcal I_{-k}|}\asymp\sqrt{\log d_1/N}$ with $C_{\tau_n}>0$ as in the proof of Theorem~\ref{thm:r_nuisance}(c).
In the remainder of the proof of (d), let $\bDelta:=\hbettaur-\bettaur$ and define
\begin{align*}
    R_{4,5}'
    &:=[\bnabla_{\bettau}\bar\ell_5^{\mathrm{sf}}(\hat{\bnu}_5,\bettaur)
    -\bnabla_{\bettau}\bar\ell_5^{\mathrm{sf}}(\bnu_5^*,\bettaur)]^\top\bDelta=\frac{2}{|\I_{-k}|}\sum_{i\in\mathcal I_{-k}}(1-A_i)D_{\tau_n,i}
    \bX_i^\top\bDelta,
\end{align*}
where $\varrho_i:=\bziS^\top\bemur-\bX_i^\top\bettaur$ and
\begin{align*}
    D_{\tau_n,i}:={}&-(\widehat w_{\tau_n,i}-w_{\tau_n,i}^*)\varrho_i
    -\widehat w_{\tau_n,i}\bS_i^\top(\hbemur-\bemur),\\
    \widehat w_{\tau_n,i}:={}&\exp(\bX_i^\top\hbepibr),\qquad
    w_{\tau_n,i}^*:=\exp(\bX_i^\top\bepibr).
\end{align*}
Here, $R_{4,5}'$ corresponds to $R_4+R_5$ in the proof of
Theorem~\ref{thm:r_nuisance} (d). Define
\begin{align*}
    \mathcal B_6'&:=\{\|\bnabla_{\bettau}\bar\ell_5^{\mathrm{sf}}
    (\bnu_5^*,\bettaur)\|_\infty\leq\lambda_{\tau_n}/2\},\\
    \mathcal B_7'&:=\left\{\delta\bar\ell_5^{\mathrm{sf}}
    (\hat{\bnu}_5,\bettaur,\bDelta)\geq\kappa_1\|\bDelta\|_2^2
    -\kappa_2\frac{\log d_1}{|\I_{-k}|}\|\bDelta\|_1^2,
    \;\forall\bDelta\in\R^{d_1}\right\}.
\end{align*}
By Lemma~\ref{lem:r_common_gradient}(b) applied to the oracle gradient as in part~(c), together with Lemma~\ref{lem:mqr2_rsc},
\eqref{rate:mqr2_pi_b}, and Assumption~\ref{cond:mqr2_sparsity}, we obtain
$\P_{\mathbb S}(\mathcal B_6'\cap\mathcal B_7')=1-o(1)$. On this event, the
deterministic argument in the proof of
Theorem~\ref{thm:r_nuisance} (d) gives
\begin{align*}
    2\delta\bar\ell_5^{\mathrm{sf}}(\hat{\bnu}_5,\bettaur,\bDelta)
    +\lambda_{\tau_n}\|\bDelta\|_1
    \leq4\lambda_{\tau_n}\|\bDelta_{S_{\bettaur}}\|_1
    +2|R_{4,5}'|.
\end{align*}
Similarly, by Young's inequality,
\begin{align*}
    |R_{4,5}'|
    \leq\frac12\delta\bar\ell_5^{\mathrm{sf}}
    (\hat{\bnu}_5,\bettaur,\bDelta)+2\bar R_{4,5}',
\end{align*}
where
\begin{align*}
    \bar R_{4,5}':=\frac{1}{|\I_{-k}|}\sum_{i\in\mathcal I_{-k}}(1-A_i)
    \frac{D_{\tau_n,i}^2}{\widehat w_{\tau_n,i}}.
\end{align*}
On the probability-one event $\mathcal E^*$, let
\begin{align*}
    r_{\tau_n,i}:=\log(\widehat w_{\tau_n,i}/w_{\tau_n,i}^*)
    =\bX_i^\top(\hbepibr-\bepibr),\qquad
    R_{\tau_n,\infty}:=\max_{i\in\mathcal I_{-k}}|r_{\tau_n,i}|.
\end{align*}
By \eqref{eq:mqr2_pi_max_prediction}, \eqref{rate:mqr2_pi_b}, and
Assumption~\ref{cond:mqr2_sparsity},
\begin{align*}
    R_{\tau_n,\infty}=O_p(\sqrt N r_\pi^2)=o_p(1)\text{ and }\exp(R_{\tau_n,\infty})=1+o_p(1).
\end{align*}
The argument in \eqref{eq:mqr2_weight_ratio_bound}, applied to
$D_{\tau_n,i}$, gives
\begin{align}
    \bar R_{4,5}'\leq
    &\frac{2c_0^{-1}\exp(R_{\tau_n,\infty})}{|\I_{-k}|}\sum_{i\in\mathcal I_{-k}}(1-A_i)\bigg\{\varrho_i^2
    [\bX_i^\top(\hbepibr-\bepibr)]^2+
    [\bS_i^\top(\hbemur-\bemur)]^2\bigg\}.
    \label{eq:mqr2_R45_quadratic}
\end{align}
Under Assumption~\ref{cond:basic2}~(b),
$\|(1-A_i)\varrho_i\|_{\psi_2}\leq\sigma_0$. Hence, as in part~(c),
Lemmas~\ref{lem:mqr2_weighted_quadratic} and \ref{lemma:beta1}, together with
\eqref{rate:mqr2_pi_b}, \eqref{eq:mqr2_mu_rates}, and
Assumption~\ref{cond:mqr2_sparsity}, show that the two terms in braces in
\eqref{eq:mqr2_R45_quadratic} are $O_p(r_\pi^2)$ and
$O_p\{(r_\pi+r_q+r_\mu)^2\}$, respectively. Thus,
\begin{align*}
    \bar R_{4,5}'=O_p\{(r_\pi+r_q+r_\mu)^2\}.
\end{align*}
Since $|\I_{-k}|\asymp N$, the deterministic argument following
\eqref{eq:r_key3}, with $M$ and $\bar R_4$ replaced by $|\I_{-k}|$ and
$\bar R_{4,5}'$, together with the preceding bound, gives
\begin{align}
    \|\hbettaur-\bettaur\|_2
    &=O_p(r_\pi+r_q+r_\mu+r_{\tau_n}),\nonumber\\
    \|\hbettaur-\bettaur\|_1
    &=O_p\left[\sqrt{\frac{N}{\log d_1}}
    \{r_\pi+r_q+r_\mu+r_{\tau_n}\}^2\right].
    \label{eq:mqr2_taun_rates}
\end{align}

\emph{(e)} Choose $\lambda_{\pi_a},\lambda_{\pi_b},\lambda_{\tau_n}\asymp\sqrt{\log d_1/N}$ and $\lambda_q,\lambda_\mu\asymp\sqrt{\log d/N}$ as in parts~(a)--(d), and take $\lambda_\tau=C_\tau\sqrt{\log d_1/|\mathcal I_{-k}|}\asymp\sqrt{\log d_1/N}$ with $C_\tau>0$ as in the proof of Theorem~\ref{thm:r_nuisance}(c).
In the remainder of the proof of (e), let $\bDelta:=\hbetaur-\betaur$ and define
\begin{align*}
    R_{6,7}'
    &:=[\bnabla_{\betau}\bar\ell_6^{\mathrm{sf}}(\hat{\bnu}_6,\betaur)
    -\bnabla_{\betau}\bar\ell_6^{\mathrm{sf}}(\bnu_6^*,\betaur)]^\top\bDelta=-\frac{2}{|\I_{-k}|}\sum_{i\in\mathcal I_{-k}}A_iD_{\tau,i}\bX_i^\top\bDelta,
\end{align*}
where $\zeta_i:=\bziS^\top\bemur-\bX_i^\top\betaur$ and
\begin{align*}
    D_{\tau,i}:={}&
    \{\exp(-\bX_i^\top\hbepiar)-\exp(-\bX_i^\top\bepiar)\}
    \{\exp(\bS_i^\top\beqr)\varepsilon_i+\zeta_i-\varrho_i\}\\
    &+\exp(-\bX_i^\top\hbepiar)\big[
    \{\exp(\bS_i^\top\hbeqr)-\exp(\bS_i^\top\beqr)\}\varepsilon_i\\
    &\qquad-\exp(\bS_i^\top\hbeqr)\bS_i^\top(\hbemur-\bemur)
    +\bX_i^\top(\hbettaur-\bettaur)\big].
\end{align*}
Here, $R_{6,7}'$ corresponds to $R_6+R_7$ in the proof of
Theorem~\ref{thm:r_nuisance} (e).
Define
\begin{align*}
    \mathcal B_8'&:=\{\|\bnabla_{\betau}\bar\ell_6^{\mathrm{sf}}
    (\bnu_6^*,\betaur)\|_\infty\leq\lambda_\tau/2\},\\
    \mathcal B_9'&:=\left\{\delta\bar\ell_6^{\mathrm{sf}}
    (\hat{\bnu}_6,\betaur,\bDelta)\geq\kappa_1\|\bDelta\|_2^2
    -\kappa_2\frac{\log d_1}{|\I_{-k}|}\|\bDelta\|_1^2,
    \;\forall\bDelta\in\R^{d_1}\right\}.
\end{align*}
By Lemma~\ref{lem:r_common_gradient}(b) applied to the oracle gradient as in part~(c), together with Lemma~\ref{lem:mqr2_rsc},
\eqref{rate:mqr2_pi_a}, \eqref{eq:mqr2_first_l1_rates}, and
Assumption~\ref{cond:mqr2_sparsity}, we obtain
$\P_{\mathbb S}(\mathcal B_8'\cap\mathcal B_9')=1-o(1)$. On this event, the
deterministic argument in the proof of
Theorem~\ref{thm:r_nuisance} (e) gives
\begin{align*}
    2\delta\bar\ell_6^{\mathrm{sf}}(\hat{\bnu}_6,\betaur,\bDelta)
    +\lambda_\tau\|\bDelta\|_1
    \leq4\lambda_\tau\|\bDelta_{S_{\betaur}}\|_1
    +2|R_{6,7}'|.
\end{align*}
Similarly, by Young's inequality,
\begin{align*}
    |R_{6,7}'|
    \leq\frac12\delta\bar\ell_6^{\mathrm{sf}}
    (\hat{\bnu}_6,\betaur,\bDelta)+2\bar R_{6,7}',
\end{align*}
where
\begin{align*}
    \bar R_{6,7}':=\frac{1}{|\I_{-k}|}\sum_{i\in\mathcal I_{-k}}A_i
    \exp(\bX_i^\top\hbepiar)D_{\tau,i}^2.
\end{align*}

On the probability-one event $\mathcal E^*$, let
\begin{align*}
    R_{\tau,\infty}&:=\max_{i\in\mathcal I_{-k}}
    \{|\bX_i^\top(\hbepiar-\bepiar)|
    +|\bS_i^\top(\hbeqr-\beqr)|\}.
\end{align*}
By \eqref{rate:mqr2_pi_a}, \eqref{eq:mqr2_first_l1_rates}, and
Assumption~\ref{cond:mqr2_sparsity},
\begin{align*}
    R_{\tau,\infty}=O_p\{\sqrt N(r_\pi^2+r_q^2)\}=o_p(1).
\end{align*}
The exponential-difference argument in \eqref{eq:mqr2_weight_ratio_bound},
applied to $D_{\tau,i}$, gives
\begin{align}
    \bar R_{6,7}'\leq6c_0^{-5}\exp(3R_{\tau,\infty})\bigg\{
    &\frac{1}{|\I_{-k}|}\sum_{i\in\mathcal I_{-k}}A_i
    \{\exp(\bS_i^\top\beqr)\varepsilon_i+\zeta_i-\varrho_i\}^2
    [\bX_i^\top(\hbepiar-\bepiar)]^2\nonumber\\
    &+\frac{1}{|\I_{-k}|}\sum_{i\in\mathcal I_{-k}}A_i\varepsilon_i^2
    [\bS_i^\top(\hbeqr-\beqr)]^2+\frac{1}{|\I_{-k}|}\sum_{i\in\mathcal I_{-k}}A_i
    [\bS_i^\top(\hbemur-\bemur)]^2\nonumber\\
    &+\frac{1}{|\I_{-k}|}\sum_{i\in\mathcal I_{-k}}A_i
    [\bX_i^\top(\hbettaur-\bettaur)]^2\bigg\}.
    \label{eq:mqr2_R67_quadratic}
\end{align}
Under Assumption~\ref{cond:basic2}~(b), Lemma~\ref{lem:psi2norm} gives
$\|A_i\varepsilon_i\|_{\psi_2}\leq\sigma_0$ and
$\|A_i\{\exp(\bS_i^\top\beqr)\varepsilon_i+\zeta_i-\varrho_i\}\|_{\psi_2}
\leq(c_0^{-2}+2)\sigma_0$. As in part~(c),
Lemmas~\ref{lem:mqr2_weighted_quadratic} and \ref{lemma:beta1}, together with
\eqref{rate:mqr2_pi_a}, \eqref{eq:mqr2_first_l1_rates},
\eqref{eq:mqr2_mu_rates}, \eqref{eq:mqr2_taun_rates}, and
Assumption~\ref{cond:mqr2_sparsity}, show that the four terms in braces in
\eqref{eq:mqr2_R67_quadratic} are, respectively,
$O_p(r_\pi^2)$, $O_p\{(r_\pi+r_q)^2\}$,
$O_p\{(r_\pi+r_q+r_\mu)^2\}$, and
$O_p\{(r_\pi+r_q+r_\mu+r_{\tau_n})^2\}$. Hence,
\begin{align*}
    \bar R_{6,7}'=O_p\{(r_\pi+r_q+r_\mu+r_{\tau_n})^2\}.
\end{align*}
Since $|\I_{-k}|\asymp N$, the deterministic argument following
\eqref{eq:r_key4}, with $M$ and $\bar R_6$ replaced by $|\I_{-k}|$ and
$\bar R_{6,7}'$, together with the preceding bound, gives
\begin{align*}
    \|\hbetaur-\betaur\|_2
    &=O_p(r_\pi+r_q+r_\mu+r_{\tau_n}+r_\tau),\\
    \|\hbetaur-\betaur\|_1
    &=O_p\left[\sqrt{\frac{N}{\log d_1}}
    \{r_\pi+r_q+r_\mu+r_{\tau_n}+r_\tau\}^2\right].
\end{align*}
\end{proof}

We now turn to Theorems~\ref{thm:mqr2_bias} and~\ref{thm:mqr2_infer}. In Algorithm~\ref{alg:mqr2}, $\hat\theta_{\mathrm{MQR2}}$ shares the same score function $\mpsi$, target parameter $\bnu^*$, and cross-fitting partition $(\I_k)_{k\in[\K]}$ as MQR1, and $\hbnu^{-k}$ is trained solely on $\I_{-k}$ and is therefore independent of $(\bW_i)_{i\in\I_k}$. Because the decomposition~\eqref{decom:r_theta} and Lemmas~\ref{lem:r_Delta_0}--\ref{lem:r_delta_3} in Section~\ref{apdx:proof_r_normality} depend on $\hbnu^{-k}$ only through its independence from $\I_k$ and the consistency condition $\rmqr_\pi+\rmqr_q+\rmqr_\mu+\rmqr_{\tau_n}+\rmqr_\tau=o_p(1)$, without using within-fold sample splitting, they continue to hold for $\hat\theta_{\mathrm{MQR2}}$ with the MQR2 nuisance estimator $\hbnu^{-k}$ and corresponding $\ell_2$ stochastic errors $(\rmqr_\pi,\allowbreak\rmqr_q,\allowbreak\rmqr_\mu,\allowbreak\rmqr_{\tau_n},\allowbreak\rmqr_\tau)$.

\begin{proof}[Proof of Theorem~\ref{thm:mqr2_bias}]
The proof parallels that of Theorem~\ref{thm:r_bias}, differing only in the nuisance convergence rates substituted in Step~2. Specifically, Theorem~\ref{thm:mqr2_r_nuisance} gives
\begin{align*}
&\rmqr_\pi=O_p(r_\pi),\qquad \rmqr_q=O_p(r_\pi+r_q),\qquad \rmqr_\mu=O_p(r_\pi+r_q+r_\mu),\\
&\rmqr_{\tau_n}=O_p(r_\pi+r_q+r_\mu+r_{\tau_n}),\qquad \rmqr_\tau=O_p(r_\pi+r_q+r_\mu+r_{\tau_n}+r_\tau),
\end{align*}
which satisfy $\rmqr_\pi+\rmqr_q+\rmqr_\mu+\rmqr_{\tau_n}+\rmqr_\tau=o_p(1)$ under Assumption~\ref{cond:mqr2_sparsity}. Substituting these rates into Lemma~\ref{lem:r_delta_3}, together with \eqref{bound:nuisance_approx}, yields
$\Delta_{k,3}'=O_p\{r_\pi(r_\mu+r_{\tau_n}+r_\tau)+r_qr_\mu+r_\pi^2+r_q^2\}$.
Combining this bound with Step~1 of the proof of Theorem~\ref{thm:r_bias}, Lemma~\ref{lem:r_approx_bound}, and Lemmas~\ref{lem:r_Delta_0}--\ref{lem:r_delta_2} in \eqref{decom:r_theta}, we obtain
\begin{align*}
    \hat{\theta}_{\mathrm{MQR2}}-\theta
    =O_p(\rsigma N^{-1/2}+\mathcal R_{\mathrm{MQR2}}).
\end{align*}
Repeating Step~3 of the proof of Theorem~\ref{thm:r_bias} yields $\rsigma\asymp 1+\|\betaur-\theta\bs{e}_1\|_2$, which completes the proof.
\end{proof}

\begin{proof}[Proof of Theorem~\ref{thm:mqr2_infer}]
For each $k\in[\K]$, it follows from Lemmas~\ref{lem:r_Delta_0} and~\ref{lem:r_approx_bound} together with the stated condition $\mathcal R_{\mathrm{MQR2}}=o(\rsigma N^{-1/2})$ that $\Delta_0'=o(\rsigma N^{-1/2})$. Combining Theorem~\ref{thm:mqr2_r_nuisance} with Lemmas~\ref{lem:r_delta_2} and~\ref{lem:r_delta_3} as well as \eqref{order:r_sigma_lower}, we further obtain $\Delta_{k,3}'=O_p\{r_\pi(r_\mu+r_{\tau_n}+r_\tau)+r_qr_\mu+r_\pi^2+r_q^2\}$ and hence $\Delta_{k,2}'+\Delta_{k,3}'=o_p(\rsigma N^{-1/2})$. Substituting these bounds into \eqref{decom:r_theta} gives
\begin{align*}
    \rsigma^{-1}N^{1/2}(\hat\theta_{\mathrm{MQR2}}-\theta)=\frac{1}{\rsigma\sqrt N}\sum_{i=1}^N\{\bpsi(\bW_i,\bnu^*)-\E[\bpsi(\bW,\bnu^*)]\}+o_p(1),
\end{align*}
which, together with \eqref{eq:r_normal} and \eqref{order:sigma}, yields $\rsigma^{-1}N^{1/2}(\hat\theta_{\mathrm{MQR2}}-\theta)\leadsto\mathcal N(0,1)$ with $\rsigma\asymp 1+\|\betaur-\theta\bs{e}_1\|_2$. Moreover, by an argument analogous to the proof of \eqref{approx:sigma}, Lemma~\ref{lem:r_delta_2} and Theorem~\ref{thm:mqr2_bias} imply that $N^{-1}\sum_{k=1}^{\K}\sum_{i\in\I_k}\{\bpsi(\bW_i,\hbnu^{-k})-\bpsi(\bW_i,\bnu^*)\}^2=o_p(\rsigma^2)$ and $|\hat\theta_{\mathrm{MQR2}}-\theta|=o_p(\rsigma)$, from which it follows that $\hat{\sigma}_{\mathrm{m},2}^2=\rsigma^2\{1+o_p(1)\}$.
\end{proof}

\subsection{Proofs of auxiliary lemmas in Section~\ref{apdx:MQR_aux}}\label{apdx:MQR_aux_proofs}
\begin{proof}[Proof of Lemma~\ref{lem:momentmatrix}]
\emph{(i)} Under Assumption~\ref{cond:basic2}, since $\E[A\bS\bS^\top]=\Var[A\bS]+\E[A\bS]\E[A\bS]^\top$ and $\E[A\bS]\E[A\bS]^\top$ is positive semidefinite, we have $\lambda_{\min}(\E[A\bS\bS^\top])\geq \lambda_{\min}(\Var[A\bS])\geq c_{\min}$. Additionally, for any $\bs{v}\in\R^d$ and $k\in\N$,
\begin{align*}
    \E[|A\bs{v}^\top\bS|^{2k}]\leq \E[|\bs{v}^\top\bS|^{2k}]\overset{(i)}{\leq} 2(\sigma_{\bS}\|\bs{v}\|_2)^{2k}\Gamma(k+1),
\end{align*}
where (i) holds by Lemma~\ref{lem:psi2norm}. Then by Lemma~\ref{lem:psi2norm},
\begin{align*}
    \normg{A\bs{v}^\top\bS}\leq 2\sigma_{\bS}\|\bs{v}\|_2\text{, }\;\;\;\forall\bs{v}\in\R^d.
\end{align*}

\emph{(ii)} Note that $\E[A\bX\bX^\top]$ is the upper-left principal block of $\E[A\bS\bS^\top]$. By the Cauchy interlacing theorem and part (i), $\lambda_{\min}\{\E[A\bX\bX^\top]\}\geq \lambda_{\min}\{\E[A\bS\bS^\top]\}\geq c_{\min}$. Additionally, for any $k\in\N$ and $\bs{v}\in\R^{d_1}$,
\begin{align*}
    \E[|A\bs{v}^\top\bX|^{2k}]\leq \E[|\bs{v}^\top\bX|^{2k}] = \E[|\bar{\bs{v}}^\top\bS|^{2k}] \overset{(i)}{\leq} 2(\sigma_{\bS}\|\bs{v}\|_2)^{2k}\Gamma(k+1),
\end{align*}
where $\bar{\bs{v}}=(\bs{v}^\top,\bz^\top)^\top\in\R^d$ and (i) holds by Lemma~\ref{lem:psi2norm} and the fact that $\|\bar{\bs{v}}\|_2=\|\bs{v}\|_2$. Then by Lemma~\ref{lem:psi2norm},
\begin{align*}
    \normg{A\bs{v}^\top\bX}\leq 2\sigma_{\bS}\|\bs{v}\|_2\text{, }\;\;\;\forall\bs{v}\in\R^{d_1}.
\end{align*}

\emph{(iii)} Under Assumption~\ref{cond:basic2}, since $\E[(1-A)\bX\bX^\top]=\Var[(1-A)\bX]+\E[(1-A)\bX]\E[(1-A)\bX]^\top$ and $\E[(1-A)\bX]\E[(1-A)\bX]^\top$ is positive semidefinite, we have $\lambda_{\min}(\E[(1-A)\bX\bX^\top])\geq \lambda_{\min}(\Var[(1-A)\bX])\geq c_{\min}$. Additionally, for any $\bs{v}\in\R^{d_1}$ and $k\in\N$,
\begin{align*}
    \E[|(1-A)\bs{v}^\top\bX|^{2k}]\leq \E[|\bs{v}^\top\bX|^{2k}]\overset{(i)}{\leq} 2(\sigma_{\bS}\|\bs{v}\|_2)^{2k}\Gamma(k+1),
\end{align*}
where (i) holds by Lemma~\ref{lem:psi2norm}. Then by Lemma~\ref{lem:psi2norm},
\begin{align*}
    \normg{(1-A)\bs{v}^\top\bX}\leq 2\sigma_{\bS}\|\bs{v}\|_2\text{, }\;\;\;\forall\bs{v}\in\R^{d_1}.
\end{align*}
\end{proof}

\begin{proof}[Proof of Lemma~\ref{lem:usebound}]
\;\;

\emph{(a)}
Let $d_0=d_1$, $\bU=\bX$ and $\balpha\in\R^{d_0}$. Note that $\bX$ is a subvector of the random vector $\bS$.
First, under Assumption~\ref{cond:basic2} (b), $\bU^\top(\balpha-\halpha)$ is a sub-Gaussian random variable with $\normg{\bU^\top(\balpha-\halpha)}\leq \sigma_{\bS}\|\balpha-\halpha\|_2$. Then by Lemma~\ref{lem:psi2norm}, we have
\begin{align*}
    \left\|\bU^\top(\balpha-\halpha)\right\|_{\P,r}=O_p\left(\normg{\bU^\top(\balpha-\halpha)}\right)=O_p(\|\balpha-\halpha\|_2)=O_p(G_N).
\end{align*}
If $d_0=d$, $\bU=\bS$ and $\balpha\in\R^{d_0}$, then under Assumption~\ref{cond:basic2}, a similar justification still holds.

Particularly, we let $d_0=d$ and $\bU=\bzS$. If $\halpha=\balpha$, the bound holds trivially; otherwise, let $\bv:=(\balpha-\halpha)/\|\balpha-\halpha\|_2$. Under Assumption~\ref{cond:basic}, there exists some $c_{0}\in(0,1/2)$ such that $\pi(\bX)\in[c_0,1-c_0]$. Then by the definition of $\|\cdot\|_{\psi_2}$ and the fact that $0\leq 1-A\leq1$, we have
\begin{align*}
    \E[\exp(\{\bzS^\top\bv\}^2/\sigma_{\bS}^2)]=\E\left[\frac{1-A}{1-\pi(\bX)}\exp(\{\bS^\top\bv\}^2/\sigma_{\bS}^2)\right]\leq 2c_0^{-1}.
\end{align*}
Since $\log_2 c_0<-1$ and $f(\cdot)=(\cdot)^{1/(1-\log_2c_0)}$ is concave, we obtain
\begin{align*}
    \E\left[\exp\left(\frac{\{\bzS^\top\bv\}^2}{\{1-\log_2 c_0\}\sigma_{\bS}^2}\right)\right]&=\E\left[\left\{\exp(\{\bzS^\top\bv\}^2/\sigma_{\bS}^2)\right\}^{\frac1{1-\log_2 c_0}}\right]\\
    &\overset{(i)}{\leq} \left\{\E\left[\exp(\{\bzS^\top\bv\}^2/\sigma_{\bS}^2)\right]\right\}^{\frac1{1-\log_2 c_0}}\\
    &\overset{(ii)}{\leq} (2c_0^{-1})^{\frac1{1-\log_2 c_0}}=2,
\end{align*}
where (i) holds by Jensen's inequality. Therefore, define $\sigma_{\bzS}=\sqrt{1-\log_2 c_0}\sigma_{\bS}$; then $\normg{\bzS^\top\bv}\leq\sigma_{\bzS}$, and hence $\bzS^\top(\balpha-\halpha)$ is a sub-Gaussian random variable with $\normg{\bzS^\top(\balpha-\halpha)}\leq \sigma_{\bzS}\|\balpha-\halpha\|_2$. By Lemma~\ref{lem:psi2norm}, we have
\begin{align*}
    \normp{\bzS^\top(\balpha-\halpha)}{r}=O_p\left(\normg{\bzS^\top(\balpha-\halpha)}\right)=O_p(\|\balpha-\halpha\|_2)=O_p(G_N).
\end{align*}

\emph{(b)} Let $d_0=d_1$, $\bU=\bX$ and $\balpha$ be $\bepiar$ or $\bepibr$.
Then we have, with some $v\in[0,1]$,
\begin{align*}
    \normp{\exp(-\bU^\top\balpha)-\exp(-\bU^\top\halpha)}{r}&=\normp{\exp(-\bU^\top\balpha)\big\{1-\exp\{\bU^\top(\balpha-\halpha)\}\big\}}{r}\\
    &\overset{(i)}{\leq}(c_0^{-1}-1)\normp{1-\exp\{\bU^\top(\balpha-\halpha)\}}{r}\\
    &\overset{(ii)}{=}(c_0^{-1}-1)\normp{\exp\big\{v\bU^\top(\balpha-\halpha)\big\}\bU^\top(\balpha-\halpha)}{r},
\end{align*}
where (i) holds since $\P(\exp(-\bU^\top\balpha)\leq c_0^{-1}-1)=1$ under Assumption~\ref{cond:basic2}(a); (ii) holds by Taylor's theorem.

By the monotonicity of $\exp(\cdot)$, it holds for any $v\in[0,1]$ that
\begin{align*}
   \exp\big\{v\bU^\top(\balpha-\halpha)\big\}\leq 1+\exp\big\{\bU^\top(\balpha-\halpha)\big\}.
\end{align*}
With $\halpha$ held fixed, Lemma~\ref{lem:psi2norm} gives, for either choice of sign,
\begin{align*}
&\normp{\exp\{\pm\bU^\top(\halpha-\balpha)\}}{2r}\\
&\quad=\exp\{\pm\E[\bU]^\top(\halpha-\balpha)\}
\left[\E\exp\{\pm2r(\bU-\E[\bU])^\top(\halpha-\balpha)\}\right]^{1/(2r)}\\
&\quad\leq\exp\{\pm\E[\bU]^\top(\halpha-\balpha)
+4r\sigma_{\bS}^2\|\halpha-\balpha\|_2^2\}\\
&\quad\leq\exp\{\sqrt{2}\sigma_{\bS}\|\halpha-\balpha\|_2
+4r\sigma_{\bS}^2\|\halpha-\balpha\|_2^2\}=O_p(1),
\end{align*}
where
\begin{align*}
|\E[\bU]^\top(\halpha-\balpha)|
\leq\normp{\bU^\top(\halpha-\balpha)}{2}
\leq\sqrt{2}\sigma_{\bS}\|\halpha-\balpha\|_2,
\end{align*}
and $\|\halpha-\balpha\|_2=O_p(G_N)=O_p(1)$.
Note that by (a), we have $\|\bU^\top(\balpha-\halpha)\|_{\P,2r}=O_p(\|\halpha-\balpha\|_2)$. Hence by H\"older's inequality, for any $v\in[0,1]$, as $N$, $d_1$, $d_2\to\infty$,
\begin{align*}
    &\normp{\exp\big\{v\bU^\top(\balpha-\halpha)\big\}\bU^\top(\balpha-\halpha)}{r}\\
    \leq&\normp{\bU^\top(\balpha-\halpha)}{r}+\normp{\exp\big\{\bU^\top(\balpha-\halpha)\big\}}{2r}\normp{\bU^\top(\balpha-\halpha)}{2r}=O_p(\|\balpha-\halpha\|_2).
\end{align*}
We then have
\begin{align*}
    \normp{\exp(-\bU^\top\balpha)-\exp(-\bU^\top\halpha)}{r}&=O_p(\|\balpha-\halpha\|_2).
\end{align*}
Similarly, it holds for any $v\in[0,1]$ that,
\begin{align*}
    &\normp{\exp\big\{v\bU^\top(\halpha-\balpha)\big\}\bU^\top(\halpha-\balpha)}{r}\\
    \leq& \normp{\bU^\top(\halpha-\balpha)}{r}+\normp{\exp\big\{\bU^\top(\halpha-\balpha)\big\}}{2r}\normp{\bU^\top(\halpha-\balpha)}{2r}=O_p(\|\balpha-\halpha\|_2).
\end{align*}
Then for some $v'\in[0,1]$,
\begin{align*}
    \normp{\exp(\bU^\top\balpha)-\exp(\bU^\top\halpha)}{r}&=\normp{\exp(\bU^\top\balpha)\big\{1-\exp\{\bU^\top(\halpha-\balpha)\}\big\}}{r}\\
    &\leq(c_0^{-1}-1)\normp{1-\exp\{\bU^\top(\halpha-\balpha)\}}{r}\\
    &\leq(c_0^{-1}-1)\normp{\exp\big\{v'\bU^\top(\halpha-\balpha)\big\}\bU^\top(\halpha-\balpha)}{r}\\
    &=O_p(\|\balpha-\halpha\|_2).
\end{align*}

Let $d_0=d$, $\bU=\bS$ and $\balpha=\beqr$. Since $\P(c_0^2\leq\exp(\bS^\top\beqr)\leq c_0^{-2})=1$, a similar argument applies.
\end{proof}

\begin{proof}[Proof of Lemma~\ref{lem:r_common_gradient}]
\emph{(a)} Write $\bs{\xi}_i:=U_iV_i\bs{Z}_i=(\xi_{i1},\ldots,\xi_{ip})^\top$ and $Z_{ij}:=\bs{Z}_i^\top\e_j$ for $i\in[m]$ and $j\in[p]$. Let $(\bs{\xi}_i')_{i=1}^m$ be an independent copy of $(\bs{\xi}_i)_{i=1}^m$, and let $(\epsilon_i)_{i=1}^m$ be i.i.d. Rademacher random variables, i.e., $\P(\epsilon_i=1)=\P(\epsilon_i=-1)=1/2$, independent of $(\bs{\xi}_i,\bs{\xi}_i')_{i=1}^m$. Since $\E(\bs{\xi}_i')=\bz_p$, we have
\begin{align}
    \E\left[\left\|\frac{1}{m}\sum_{i=1}^m \bs{\xi}_i\right\|_\infty^2\right]
    &=\E\left[\left\|\frac{1}{m}\sum_{i=1}^m \left\{\bs{\xi}_i-\E(\bs{\xi}_i')\right\}\right\|_\infty^2\right]
    \overset{(i)}{\leq}\E\left[\left\|\frac{1}{m}\sum_{i=1}^m (\bs{\xi}_i-\bs{\xi}_i')\right\|_\infty^2\right]\nonumber\\
    &\overset{(ii)}{=}\E\left[\left\|\frac{1}{m}\sum_{i=1}^m \epsilon_i(\bs{\xi}_i-\bs{\xi}_i')\right\|_\infty^2\right]
    \overset{(iii)}{\leq}4\E\left[\left\|\frac{1}{m}\sum_{i=1}^m \epsilon_i\bs{\xi}_i\right\|_\infty^2\right],
    \label{eq:r_sym_bound}
\end{align}
where (i) holds by Jensen's inequality and the convexity of $\|\cdot\|_\infty^2$; (ii) holds since $(\bs{\xi}_i-\bs{\xi}_i')_{i=1}^m$ is symmetric and independent across $i\in[m]$; and (iii) follows from $\|\bs{a}-\bs{b}\|_\infty^2\leq2\|\bs{a}\|_\infty^2+2\|\bs{b}\|_\infty^2$.

Let $\E_\epsilon$ and $\P_\epsilon$ denote the expectation and probability with respect to $(\epsilon_i)_{i=1}^m$, treating $(\bs{\xi}_i)_{i=1}^m$ as fixed. Define $S_j:=m^{-1}\sum_{i=1}^m\epsilon_i\xi_{ij}$ and $v_{\max}^2:=\max_{j\in[p]}m^{-2}\sum_{i=1}^m\xi_{ij}^2$. If $v_{\max}=0$, then $\max_{j\in[p]}S_j^2=0$. Suppose $v_{\max}>0$. Since $\cosh(u)\leq\exp(u^2/2)$ for all $u\in\R$, for each $j\in[p]$ and any $t>0$ we have
\begin{align*}
    \E_\epsilon\{\exp(tS_j)\}
    =\prod_{i=1}^m\cosh\left(\frac{t\xi_{ij}}{m}\right)
    \leq\exp\left(\frac{t^2}{2m^2}\sum_{i=1}^m\xi_{ij}^2\right)
    \leq\exp\left(\frac{t^2v_{\max}^2}{2}\right).
\end{align*}
By Markov's inequality applied to $\exp(\pm tS_j)$ with $t=u/v_{\max}^2$, we obtain $\P_\epsilon(|S_j|>u)\leq2\exp\{-u^2/(2v_{\max}^2)\}$ for all $u>0$, and hence
\begin{align*}
    \E_\epsilon\left[\exp\left(\frac{S_j^2}{6v_{\max}^2}\right)-1\right]
    =\int_0^\infty\P_\epsilon\left\{|S_j|>v_{\max}\sqrt{6\log(1+x)}\right\}\d x
    \leq2\int_0^\infty(1+x)^{-3}\,\d x=1.
\end{align*}
By the convexity of $\exp(\cdot)$ and Jensen's inequality, it follows that
\begin{align*}
    \exp\left(\frac{\E_\epsilon[\max_{j\in[p]}S_j^2]}{6v_{\max}^2}\right)
    \leq\E_\epsilon\left[\max_{j\in[p]}\exp\left(\frac{S_j^2}{6v_{\max}^2}\right)\right]
    \leq\sum_{j=1}^p\E_\epsilon\left[\exp\left(\frac{S_j^2}{6v_{\max}^2}\right)\right]
    \leq2p,
\end{align*}
which implies $\E_\epsilon[\max_{j\in[p]}S_j^2]\leq6\log(2p)\,v_{\max}^2\leq12(\log p)v_{\max}^2$ since $p\geq2$. Taking expectations over $(\bs{\xi}_i)_{i=1}^m$ and combining with \eqref{eq:r_sym_bound}, we obtain
\begin{align}
    \E\left[\left\|\frac{1}{m}\sum_{i=1}^m\bs{\xi}_i\right\|_\infty^2\right]
    \leq\frac{48\log p}{m}\E\left[\max_{j\in[p]}\frac{1}{m}\sum_{i=1}^m\xi_{ij}^2\right].
    \label{eq:r_max_var}
\end{align}

By the Cauchy--Schwarz and H\"older inequalities, the expectation on the right-hand side of \eqref{eq:r_max_var} satisfies
\begin{align}
    \E\left[\max_{j\in[p]}\frac{1}{m}\sum_{i=1}^m\xi_{ij}^2\right]
    &\leq\E\left[\left(\frac{1}{m}\sum_{i=1}^m U_i^4V_i^4\right)^{1/2}\left(\max_{j\in[p]}\frac{1}{m}\sum_{i=1}^m Z_{ij}^4\right)^{1/2}\right]\nonumber\\
    &\leq\left\{\E\left(\frac{1}{m}\sum_{i=1}^m U_i^4V_i^4\right)\right\}^{1/2}
    \left\{\E\left(\max_{j\in[p]}\frac{1}{m}\sum_{i=1}^m Z_{ij}^4\right)\right\}^{1/2}\nonumber\\
    &\leq\|U\|_{\P,8}^2\|V\|_{\P,8}^2
    \left\{\E\left(\max_{j\in[p]}\frac{1}{m}\sum_{i=1}^m Z_{ij}^4\right)\right\}^{1/2}.
    \label{eq:r_xi_holder}
\end{align}
Next, we bound $\E(\max_{j\in[p]}m^{-1}\sum_{i=1}^m Z_{ij}^4)$. For each $j\in[p]$, write $Z_j:=\bs{Z}^\top\e_j$ and define the centered random variable $Q_{ij}:=Z_{ij}^4-\E(Z_j^4)$. Taking $\bv=\e_j$ in $\|\bs{Z}^\top\bv\|_{\psi_2}\leq\sigma_Z\|\bv\|_2$ gives $\|Z_j\|_{\psi_2}\leq\sigma_Z$. By Lemma~\ref{lem:psi2norm}, we have $\E(Z_j^4)\leq4\sigma_Z^4$ and $\|Z_j\|_{\P,4r}\leq2\sigma_Z\sqrt{r}$ for every integer $r\geq2$. Minkowski's inequality then yields, for every integer $r\geq2$,
\begin{align*}
    \|Q_{ij}\|_{\P,r}
    \leq2\|Z_j^4\|_{\P,r}
    =2\|Z_j\|_{\P,4r}^4
    \leq32\sigma_Z^4 r^2.
\end{align*}
This implies that $Q_{ij}$ is a sub-Weibull random variable of order $1/2$ with $\|Q_{ij}\|_{\psi_{1/2}}\leq c_1\sigma_Z^4$ for some constant $c_1>0$. Consequently, for each $j\in[p]$, invoking Theorem~3.1 of \citet{MovingBeyondSubGaussianity2022} with $\alpha=1/2$ and $a_i=1/m$ implies that for all $x\geq1$, there exists a constant $c_2>0$ such that
\begin{align*}
    \P\left[\left|\frac{1}{m}\sum_{i=1}^m Q_{ij}\right|>c_2\sigma_Z^4\left(\sqrt{\frac{x}{m}}+\frac{x^2}{m}\right)\right]
    \leq2\exp(-x).
\end{align*}
Taking $x=1+t+\log p$ for $t\geq0$ and applying a union bound over $j\in[p]$, we obtain
\begin{align*}
    \P\left[\max_{j\in[p]}\left|\frac{1}{m}\sum_{i=1}^m Q_{ij}\right|
    >c_2\sigma_Z^4\left\{\sqrt{\frac{1+t+\log p}{m}}+\frac{(1+t+\log p)^2}{m}\right\}\right]
    \leq2\exp(-t).
\end{align*}
Since $1+\log p\leq3\log p$ for $p\geq2$, we have $\sqrt{(1+t+\log p)/m}+(1+t+\log p)^2/m\leq b_{m,p}+h(t)$ for all $t\geq0$, where $b_{m,p}:=\sqrt{3\log p/m}+18(\log p)^2/m$ and $h(t):=\sqrt{t/m}+2t^2/m$. Letting $T_m:=(c_2\sigma_Z^4)^{-1}\max_{j\in[p]}|m^{-1}\sum_{i=1}^m Q_{ij}|$ and $(T_m-b_{m,p})_+:=\max\{T_m-b_{m,p},0\}$, we obtain
\begin{align*}
    \E[(T_m-b_{m,p})_+]
    &=\int_0^\infty\P\{(T_m-b_{m,p})_+>h(t)\}h'(t)\,\d t\\
    &\leq2\int_0^\infty\exp(-t)\left(\frac{1}{2\sqrt{mt}}+\frac{4t}{m}\right)\d t
    =\sqrt{\frac{\pi}{m}}+\frac{8}{m}\leq10.
\end{align*}
Since $\sqrt{3\log p/m}\leq1+2(\log p)^2/m$ for $p\geq2$, it follows that $\E(T_m)\leq b_{m,p}+10\leq20\{1+(\log p)^2/m\}$. Together with $\E(Z_j^4)\leq4\sigma_Z^4$ and $\sqrt{a+b}\leq\sqrt{a}+\sqrt{b}$ for $a,b\geq0$, this gives
\begin{align*}
    \left\{\E\left(\max_{j\in[p]}\frac{1}{m}\sum_{i=1}^m Z_{ij}^4\right)\right\}^{1/2}
    \leq\sqrt{4+20c_2}\,\sigma_Z^2\left(1+\frac{\log p}{\sqrt{m}}\right).
\end{align*}
Combining this bound with \eqref{eq:r_max_var} and \eqref{eq:r_xi_holder} yields \eqref{eq:r_common_multiplier} with $C_1:=48\sqrt{4+20c_2}\,\sigma_Z^2$.

\emph{(b)} Let $\sigma_H:=b\sigma_V\sigma_Z$. For each $j\in[p]$, since $|UVZ_j|\leq b|VZ_j|$ almost surely, Lemma~D.1 of \cite{HighDimensionalMEstimation2019} gives
\begin{align*}
    \norme{UVZ_j}\leq b\normg{V}\normg{Z_j}\leq\sigma_H.
\end{align*}
When $\log p\leq m$, since $\E(UVZ_j)=0$, by Bernstein's inequality provided in Lemma~D.4 of \cite{HighDimensionalMEstimation2019} with $t=2\log p$ and a union bound over $j\in[p]$, we obtain
\begin{align*}
    &\P\left(\left\|\frac{1}{m}\sum_{i=1}^m U_iV_i\bs{Z}_i\right\|_\infty
    >\sigma_H\left\{2\sqrt{\frac{2\log p}{m}}+\frac{2\log p}{m}\right\}\right)\\
    &\quad\leq\sum_{j=1}^p\P\left(\left|\frac{1}{m}\sum_{i=1}^m U_iV_iZ_{ij}\right|
    >\sigma_H\left\{2\sqrt{\frac{2\log p}{m}}+\frac{2\log p}{m}\right\}\right)
    \leq2p\exp(-2\log p)=\frac{2}{p}.
\end{align*}
Since $\log p\leq m$ implies $2\sqrt{2\log p/m}+2\log p/m\leq5\sqrt{\log p/m}$, the stated bound holds with $C_2:=5\sigma_H$ when $\log p\leq m$.

When $\log p>m$, the triangle inequality, $|U_i|\leq b$, and the Cauchy--Schwarz inequality give
\begin{align}
    \left\|\frac{1}{m}\sum_{i=1}^m U_iV_i\bs{Z}_i\right\|_\infty
    \leq b\left(\frac{1}{m}\sum_{i=1}^m V_i^2\right)^{1/2}
    \left(\max_{j\in[p]}\frac{1}{m}\sum_{i=1}^m Z_{ij}^2\right)^{1/2}.
    \label{eq:r_oracle_cs}
\end{align}
Since $\|V_i\|_{\psi_2}\leq\sigma_V$ and $\|Z_{ij}\|_{\psi_2}\leq\sigma_Z$, we have $\E\{\exp(V_i^2/\sigma_V^2)\}\leq2$ and $\E\{\exp(Z_{ij}^2/\sigma_Z^2)\}\leq2$ for all $i\in[m]$ and $j\in[p]$. By independence across $i\in[m]$ and Markov's inequality,
\begin{align*}
    \P\left(\frac{1}{m}\sum_{i=1}^m V_i^2>2\sigma_V^2\right)
    \leq\exp(-2m)\prod_{i=1}^m\E\left\{\exp\left(\frac{V_i^2}{\sigma_V^2}\right)\right\}
    \leq(2e^{-2})^m\leq\exp(-m),
\end{align*}
where the last inequality uses $2e^{-2}<e^{-1}$. Similarly, Markov's inequality and a union bound over $j\in[p]$ give
\begin{align*}
    \P\left\{\max_{j\in[p]}\frac{1}{m}\sum_{i=1}^m Z_{ij}^2
    >\sigma_Z^2\left(\log 2+\frac{2\log p}{m}\right)\right\}
    &\leq\sum_{j=1}^p\exp(-m\log 2-2\log p)\prod_{i=1}^m\E\left\{\exp\left(\frac{Z_{ij}^2}{\sigma_Z^2}\right)\right\}\\
    &\leq\frac{1}{p}.
\end{align*}
When $\log p>m$, we have $\log 2+2\log p/m<3\log p/m$. Substituting these two bounds into \eqref{eq:r_oracle_cs} shows that, with probability at least $1-1/p-\exp(-m)\geq1-2/p-\exp(-m)$,
\begin{align*}
    \left\|\frac{1}{m}\sum_{i=1}^m U_iV_i\bs{Z}_i\right\|_\infty
    \leq\sqrt{6}\,\sigma_H\sqrt{\frac{\log p}{m}}
    \leq C_2\sqrt{\frac{\log p}{m}},
\end{align*}
where $C_2=5\sigma_H$. Combining the two cases completes the proof.
\end{proof}

\begin{proof}[Proof of Lemma~\ref{lem:mqr2_weighted_quadratic}]
Define the empirical and population covariance-type matrices as
\begin{align*}
    \widehat{\bs{\Gamma}} := \frac{1}{m}\sum_{i=1}^m U_i^2\bs{Z}_i\bs{Z}_i^\top
    \quad \text{and} \quad
    \bs{\Gamma} := \E(U^2\bs{Z}\bs{Z}^\top).
\end{align*}

For any fixed $\bDelta\in\R^p$, by H\"older's inequality and the assumption $\|\bs{Z}^\top\bDelta\|_{\psi_2}\leq\sigma_Z\|\bDelta\|_2$, we obtain the population bound
\begin{align}
    \bDelta^\top\bs{\Gamma}\bDelta
    = \E\{U^2(\bs{Z}^\top\bDelta)^2\}
    \leq \|U\|_{\P,4}^2\|\bs{Z}^\top\bDelta\|_{\P,4}^2
    \leq 4\sigma_U^2\sigma_Z^2\|\bDelta\|_2^2.
    \label{eq:weighted_quadratic_population}
\end{align}

Next, we bound the entries of $\widehat{\bs{\Gamma}}-\bs{\Gamma}$. For any $j,l\in[p]$, define the centered random variable $V_{i,jl} := U_i^2Z_{ij}Z_{il}-\E(U^2Z_jZ_l)$. By Lemma~\ref{lem:psi2norm}, we have $\|U\|_{\P,4r}\leq 2\sigma_U\sqrt{r}$ and $\|Z_j\|_{\P,4r}\leq 2\sigma_Z\sqrt{r}$ for every integer $r\geq2$, where the latter follows from taking $\bDelta=\e_j$. It then follows from H\"older's inequality that, for any integer $r\geq2$,
\begin{align*}
    \|V_{i,jl}\|_{\P,r}
    \leq 2\|U^2Z_jZ_l\|_{\P,r}
    \leq 2\|U\|_{\P,4r}^2 \|Z_j\|_{\P,4r}\|Z_l\|_{\P,4r}
    \leq 32\sigma_U^2\sigma_Z^2 r^2.
\end{align*}
This implies that $V_{i,jl}$ is a sub-Weibull random variable (of order $1/2$) with its Orlicz norm bounded by $\|V_{i,jl}\|_{\psi_{1/2}} \leq C_1\sigma_U^2\sigma_Z^2$ for some constant $C_1>0$.
Consequently, for each fixed index pair $(j,l)$, invoking Theorem~3.1 of \citet{MovingBeyondSubGaussianity2022} with $\alpha=1/2$ and $a_i=1/m$ implies that for all $x\geq1$, there exists a universal constant $C_2>0$ such that
\begin{align*}
    \P\left[
    \left| \frac{1}{m}\sum_{i=1}^m V_{i,jl} \right|
    > C_2\sigma_U^2\sigma_Z^2 \left( \sqrt{\frac{x}{m}} + \frac{x^2}{m} \right)
    \right]
    \leq 2e^{-x}.
\end{align*}
Furthermore, since
$(\widehat{\bs{\Gamma}}-\bs{\Gamma})_{jl}
=m^{-1}\sum_{i=1}^mV_{i,jl}$,
we have
\begin{align*}
    \P\left[
    \|\widehat{\bs{\Gamma}}-\bs{\Gamma}\|_{\max}
    >
    C_2\sigma_U^2\sigma_Z^2
    \left\{
    \sqrt{\frac{t+2\log p}{m}}
    +
    \frac{(t+2\log p)^2}{m}
    \right\}
    \right]
    \leq2e^{-t}
\end{align*}
by taking $x=t+2\log p$ and applying a
union bound over $(j,l)\in[p]^2$.
Since $t+2\log p\leq2(t+\log p)$, with probability at least
$1-2e^{-t}$,
\begin{align}
    \|\widehat{\bs{\Gamma}}-\bs{\Gamma}\|_{\max}
    \leq
    4C_2\sigma_U^2\sigma_Z^2a_{m,p}(t).
    \label{eq:weighted_quadratic_matrix}
\end{align}
Moreover, note that for any
$\bDelta\in\R^p$,
\begin{align*}
    \frac1m\sum_{i=1}^mU_i^2(\bs{Z}_i^\top\bDelta)^2
    &=\bDelta^\top\widehat{\bs{\Gamma}}\bDelta\leq\bDelta^\top\bs{\Gamma}\bDelta
    +|\bDelta^\top(\widehat{\bs{\Gamma}}-\bs{\Gamma})\bDelta|\\
    &\leq 4\sigma_U^2\sigma_Z^2\|\bDelta\|_2^2
    +\|\widehat{\bs{\Gamma}}-\bs{\Gamma}\|_{\max}\|\bDelta\|_1^2.
\end{align*}
Combining this with \eqref{eq:weighted_quadratic_population} and \eqref{eq:weighted_quadratic_matrix}, we conclude that with probability at least $1-2\exp(-t)$,
\begin{align*}
    \frac1m\sum_{i=1}^mU_i^2(\bs{Z}_i^\top\bDelta)^2&\leq 4\sigma_U^2\sigma_Z^2\max\{1,C_2\}\{\|\bDelta\|_2^2+a_{m,p}(t)\|\bDelta\|_1^2\}.
\end{align*}
\end{proof}

\begin{proof}[Proof of Lemma~\ref{lem:smallball}]
When $\bv=\bz_p$, the inequality holds directly. Now consider any $\bv\in\R^p$ with $\bv\neq\bz_p$, and let $Z:=(\bU^\top\bv)^2$. By Lemma~\ref{lem:psi2norm}, there exists a constant $C_4>0$ depending only on $\sigma_{\bU}$ such that $\E[Z^2]=\normp{\bU^\top\bv}{4}^4\leq C_4\|\bv\|_2^4$. Note that
\begin{align*}
\E[Z]
&=\E\big[Z\idf_{\{Z<\frac12\E[Z]\}}\big]+\E\big[Z\idf_{\{Z\geq\frac12\E[Z]\}}\big]\\
&\overset{(i)}{\leq}\frac12\E[Z]+\{\E[Z^2]\}^{1/2}\big\{\P\big(Z\geq\tfrac12\E[Z]\big)\big\}^{1/2},
\end{align*}
where (i) holds by the Cauchy--Schwarz inequality. Hence, by $\E[Z]\geq c_{\min}\|\bv\|_2^2$, we obtain
\begin{align*}
\P\Big(|\bU^\top\bv|\geq\sqrt{c_{\min}/2}\,\|\bv\|_2\Big)
\geq\P\big(Z\geq\tfrac12\E[Z]\big)
\geq\frac{(\E[Z])^2}{4\E[Z^2]}
\geq\frac{c_{\min}^2}{4C_4}=:p_0>0.
\end{align*}
\end{proof}

\end{appendices}
\end{document}